\documentclass[a4paper,11pt]{article}
\usepackage{amssymb}
\usepackage{amsfonts}
\usepackage{amsmath}
\usepackage{geometry}
\usepackage{setspace}
\usepackage[round]{natbib}
\usepackage[pdftex]{graphicx}
\usepackage{epstopdf}
\usepackage{framed}
\usepackage{rotating,booktabs}
\usepackage{color}
\usepackage{algorithm}
\usepackage{multirow}
\usepackage{caption}
\usepackage{subcaption}
\usepackage{algpseudocode}
\newtheorem{theorem}{Theorem}
\newtheorem{assumption}{Assumption}

\newtheorem{lemma}{Lemma}

\newenvironment{proof}[1][Proof]{\textbf{#1.} }{\ \rule{0.5em}{0.5em}}

\begin{document}
\onehalfspacing
\title{\textbf{Estimation of the Local Conditional Tail
		Average Treatment Effect}\thanks{We are grateful to the editor, Ivan Canay, the associate editor and two anonymous reviewers for valuable comments and suggestions on
		previous versions of the paper. We thank seminar participants in 2019 macroeconometric modelling workshop (Academia Sinica), 2020 Annual Meeting of Taiwan Econometric Society, 2021 Delhi Winter School-the Econometric Society, The 5th International Conference on Econometrics and Statistics (EcoSta 2022), 2022 Asian Meeting of the Econometric Society, CRETA Seminar, National Chengchi University and National Taiwan University for helpful comments.
	}}
\author{Le-Yu Chen\thanks{%
		Institute of Economics, Academia Sinica, 128, Section 2, Academia Road,
		Nankang, Taipei 115, Taiwan. E-mail: \texttt{lychen@econ.sinica.edu.tw}} \\
	%EndAName
	Academia Sinica \and Yu-Min Yen\thanks{%
		Department of International Business, National Chengchi University, 64,
		Section 2, Zhi-nan Road, Wenshan, Taipei 116, Taiwan. E-mail: \texttt{%
			yyu\_min@nccu.edu.tw}} \\
	%EndAName
	 National Chengchi University}
\date{\today}
\maketitle
\begin{abstract}The conditional tail average treatment effect (CTATE) is defined as a difference between the conditional tail expectations of potential outcomes, which can capture heterogeneity and deliver aggregated local information on treatment effects over different quantile levels and is closely related to the notion of second-order stochastic dominance and the Lorenz curve. These properties render it a valuable tool for policy evaluation. In this paper, we study estimation of the CTATE locally for a group of compliers (local CTATE or LCTATE) under the two-sided noncompliance framework. We consider a semiparametric treatment effect framework under endogeneity for the LCTATE estimation using a newly introduced class of consistent loss functions jointly for the conditional tail expectation and quantile. We establish the asymptotic theory of our proposed LCTATE estimator and provide an efficient algorithm for its implementation. We then apply the method to evaluate the effects of participating in programs under the Job Training Partnership Act in the US.
\end{abstract}
\clearpage

\doublespacing
\section{Introduction}

The treatment effect of a policy change is often heterogeneous among individuals. Accounting for such heterogeneity is crucial for evaluating the effects of and understanding the mechanisms underlying the policy \citep{HSC_1997}. To capture the heterogeneity of treatment effects, the quantile treatment effect (QTE) is a frequently used measure, which is computed as the difference between a given quantile of the distribution of the potential outcome subject to the policy being changed and that of the potential outcome under the policy had it not been changed. In this paper, we consider an alternative measure for evaluating heterogeneous treatment effects: the conditional tail average treatment effect (CTATE), defined as a difference between the conditional tail expectations (CTEs) of two potential outcomes at a given quantile level. The CTATE amounts to a (rescaled) integral of the QTE over a specified range of quantiles and is thus useful for delivering aggregated local information on treatment effect.   

Recently, \citet{FZ_2016} proposed a class of consistent loss functions (henceforth the FZ loss) for jointly estimating the quantile and CTE of a random variable. Through the FZ loss, we develop a semiparametric estimation procedure to jointly estimate the CTATE and QTE for the group of compliers under the two-sided noncompliance framework \citep[][Chapter~24]{IR_2015}, which is commonly employed in addressing endogeneity with instrumental variables and in developing various estimators for the local treatment effects in the causal inference literature \citep[e.g.,][]{IA_1994, AIR_1996, AAI_2002, Abadie_2003, FM_2013, DHL_2014, BCFH_2017, FH_2017, CHW_2020, FFHL_2020, WPZF_2021, HLL_2022}. Under this framework, if some regularity conditions hold \citep{AAI_2002, Abadie_2003}, endogeneity can be eliminated within samples of compliers. Using this fact, we demonstrate that both the QTE and CTATE locally for compliers can be estimated through a weighted FZ loss minimization-based estimation approach. We then establish the asymptotic theory for the resulting local CTATE (LCTATE) estimator and provide an efficient and stable algorithm for its implementation.

We now further remarks on the usefulness of the CTATE in empirical policy evaluation research. The conditional tail expectation (CTE) is related to the notion of second order stochastic dominance (SOSD). Stochastic dominance is a uniform order relation between the distributions of stochastic outcomes and is often used in ranking individuals' preferences under uncertainty. Conventionally, the formulations of first and second order stochastic dominances (FOSD and SOSD) are based on the properties of the cumulative distributions of competing random outcomes. Such formulations can also be equivalently established using their quantiles and CTEs \citep{Levy_2016}. For policy evaluation, if the QTE is nonnegative over all and positive over some quantile levels, the policy being changed first order stochastically dominates that that had not been changed. Likewise, if the CTATE is nonnegative over all and positive over some quantile levels, we can then attest that the policy being changed second order stochastically dominates that that had not been changed. Therefore we can use the CTATE to rank counterfactual outcome distributions in the SOSD sense when ranking by the FOSD criterion is empirically inconclusive. In addition, the CTATE has recently attracted growing interest in the study of optimal treatment assignment policies. \cite{LSW_2023} show the importance of the CTATE in constructing optimal regret policies in the worst-case welfare when sample selection is biased. Their study demonstrates the usefulness of the CTATE for robust policy learning. 

The CTATEs at two quantile levels can be used to calculate an average of the QTEs between the two quantile levels, which we term the inter-quantile average treatment effect (IQATE). The IQATE is also useful for summarizing heterogeneous treatment effects to reveal an overall picture especially when the QTEs exhibit a large fluctuation over a specified range of quantiles of interest. Moreover, related to the distributional comparison of economic outcomes, the CTE is directly connected to the Lorenz curve for comparing income inequalities. Hence our proposed procedure for estimating the CTATE may be used to quantify the Lorenz effect, which captures the shift in the Lorenz curve due to changes in policy regimes \citep{CFM_2013}.      

%\footnote{For an interested potential outcome, the Lorenz curve can be computed as a ratio of the conditional tail expectation of the potential outcome to its mean, scaled by the quantile level of the conditional tail expectation.} 
% The proposed estimation procedure for the CTATE can estimate conditional tail expectations of potential outcomes. Together with their estimated means, the proposed estimation procedure helps to quantify a shift of the Lorenz curve due to a policy change (Lorenz effect, \citet{CFM_2013}).

In this paper, we estimate the CTEs of potential outcomes by minimizing the FZ loss \citep{FZ_2016}. To the best of our knowledge, the FZ loss is the only known class of consistent loss functions for estimating the CTE of a random variable. Using consistent loss to estimate parameters of statistical functionals results in an M-estimation problem, which facilitates computation and statistical analysis of parameter estimators through empirical risk minimization. The FZ loss is non-smooth in model parameters, hence rendering the estimation problem computationally challenging. To address this issue, we propose an iterative scheme that decomposes the estimation problem into a smooth optimization and a weighted quantile regression estimation problem. This then yields a computation algorithm that can be efficiently and stably implemented in practice. We apply our proposed method to analyze the effects of the Job Training Partnership Act (JTPA) program participation. Our empirical findings indicate that, when endogeneity is considered, the estimated LCTATE on the earnings for adult women is nonnegative over all and statistically significantly positive over some quantile levels, but that for adult men is not. Our results also reveal that for both adult men and women, evidence for the positivity of local QTEs on earnings is much weaker than that of LCTATEs over all quantile levels. For adult women, conditional on the group of compliers, there hence appears to be second order stochastic dominance of the distribution of potential earnings from participating in the JTPA programs over that from not participating. However, the SOSD results do not seem to hold for the case of adult men. As a risk-averse individual prefers the distribution of a random outcome that second order stochastically dominates that of another competing outcome, our empirical results suggest that at least, the JTPA could be beneficial for risk-averse female workers who complied with assignment of the JTPA offer.

Conditional tail expectation (CTE), as a parameter of interest, has already received considerable attention in finance and risk management. In particular, CTE of an asset's return, known as the expected shortfall, is now a benchmark risk assessment measure in the industry. In the previous literature, the estimation of expected shortfalls is often based on estimated truncated means of the assets' returns \citep[see e.g.,][]{LX_2013,Hill_2015}. Recently, there has been a growing interest in structural models for forecasting the expected shortfall using the FZ loss \citep{PFC_2019, Taylor_2019, MT_2020, CYY_2022}. These studies address the problem of finding the best model to predict the expected shortfall. In contrast, our paper, which also builds upon the FZ loss minimization approach, focuses on the CTATE estimation in a causal inference framework and thus has a fundamentally different objective from these previous works. Finally, in a study independent of ours, \citet{WTH_2024} developed a method for estimating the causal effect related to the CTE of the potential outcome for compliers. However, their estimation relies on a plug-in conditional quantile estimate and is quite different from ours, which uses the FZ loss.

The remainder of this paper is organized as follows. In Section 2, we introduce the notion of CTATE, set forth the causal framework, and present our estimation approach as well as the implementation algorithm. In Section 3, we establish the asymptotic properties of our CTATE estimator. In Section 4, we conduct simulation experiments to examine the finite-sample performance of the proposed method. In Section 5, we illustrate the usefulness of our approach in an empirical application using the JTPA data. Finally, we conclude the paper in Section 6. The proofs of all the technical results are provided in the appendix.
\section{Methodology}
\subsection{Conditional tail average treatment effect}
We follow the Rubin causal model with potential outcomes \citep[see, e.g.,][Chapter~1]{IR_2015}. Let $D$ denote the treatment status having a value in $\{0,1\}$ (binary variable). Following convention, here an individual is deemed treated if and only if her treatment status $D$ takes a value of $1$, and $0$ otherwise. We use $Y_{d}$ to denote the potential outcome under the $d$th treatment, where $d\in\{0,1\}$, and $Y:=DY_{1}+(1-D)Y_{1}$ to denote the observed outcome. To formally define the CTATE, let $Q_{Y_{d}|X}(\tau)$ denote the $\tau$-quantile of the distribution of the potential outcome $Y_{d}$ conditional on $X$. Let $CTE_{Y_{d}|X}\left( \tau \right) :=E\left[Y_{d}|X,Y_{d}\leq Q_{Y_{d}|X}(\tau)\right]$ denote the CTE of $Y_{d}$ at the $\tau$-quantile level, given $X$. Note that, for any continuous random
variable $Y_{d}$,
\begin{equation}
CTE_{Y_{d}|X}\left(\tau\right) =\frac{1}{\tau}\int\nolimits_{0}^{\tau
}Q_{Y_{d}|X}(u)du. \label{CTE}
\end{equation}
We define the CTATE at the $\tau $th quantile level as 
\begin{equation}
CTATE(\tau):=CTE_{Y_{1}|X}\left( \tau \right)
-CTE_{Y_{0}|X}\left( \tau \right) .  \label{CTATE}
\end{equation}%
In the literature, the QTE conditional on $X$ at the $\tau $th quantile is defined as
\begin{equation}
QTE(\tau):=Q_{Y_{1}|X}(\tau )-Q_{Y_{0}|X}(\tau).
\label{QTE}
\end{equation}%
Using (\ref{CTE}), (\ref{CTATE}) and (\ref{QTE}), we note that the CTATE is related to QTE through the equation: 
\begin{equation}
CTATE(\tau)=\frac{1}{\tau}\int\nolimits_{0}^{\tau
}QTE(u)du.  \label{CTATE and QTE}
\end{equation} Therefore, if the function $Q_{Y_{d}|X}(\tau )$ is linear in quantile-specific parameters (e.g., \citet{CH_2006}, \citet{CHJ_2007}, \citet{CH_2008} and \citet{CHJ_2009}), then both $CTE_{Y_{d}|X}\left(\tau\right)$ and $CTATE(\tau )$ are also linear in quantile-specific parameters. Accordingly, for such a linear specification setting, if the treatment status $D$ is also exogenous, we can consistently estimate the structural parameters of the function $Q_{Y_{d}|X}(\tau )$ by minimizing the check loss for the standard quantile regression of \citet{KB_1978} and estimate those of $CTE_{Y_{d}|X}\left(\tau\right)$ by minimizing the FZ loss. This then allows for consistent estimation of the corresponding QTE and CTATE. However, in the presence of endogeneity, such an estimation procedure may be invalid and may bias the results. 

\subsection{Estimating the CTATE under two-sided noncompliance}
We consider treatment effect estimation in a two-sided noncompliance framework in which the treatment status is potentially endogenous. Assume that there is also a binary instrumental variable $Z\in \{0,1\}$, which acts as an indicator for a random assignment of eligibility for receiving a treatment. Given noncompliance, there are four types ($T$) of individuals in the population: compliers $(T=c)$, always-takers $(T=a)$, never-takers $(T=ne)$ and defiers $(T=de)$. Always-takers will and never-takers will not be treated regardless of their eligibility to receive the treatment. Compliers will take the treatment only if they are eligible and defiers will do so only if they are not eligible. Let $D_{z}$, $z\in \{0,1\}$, denote potential treatment statuses under eligibility assignment $z$. We can then summarize the relations between type $T$ and the potential treatment statuses $D_{1}$ and $D_{0}$ as follows: (1) $D_{1}=D_{0}=1$ if and only if $T=a$; (2) $D_{1}=D_{0}=0$ if and only if $T=ne$; (3) $D_{1}=1$ and $D_{0}=0$ if and only if $T=c$; (4) $D_{1}=0$ and $D_{0}=1$ if and only if $T=de$.

Following \citet{AAI_2002}, we make the following assumptions to identify the treatment effects for \textit{compliers} under the two-sided noncompliance framework:
\begin{assumption}[\citet{AAI_2002}]For almost every realization of $X$:
\label{Assumption AAI}
	\item[1.] $(Y_{1},Y_{0},D_{1},D_{0})\perp Z|X$.
	%\item[2.] $(Y_{dz},D_{z})\perp Z|X$ for $\left( d,z\right) \in \{0,1\}\times\{0,1\}$.
	\item[2.] $0<P(Z=1|X)<1$.
	\item[3.] $P(D_{1}=1|X)\neq P(D_{0}=1|X)$.
	\item[4.] $P(D_{1}\geq D_{0}|X)=1$.
\end{assumption}

Assumptions similar to those stated above are common in the literature on the Rubin causal model with instrumental variables \citep[see e.g.,][Chapter 24]{IA_1994, AIR_1996, Abadie_2003, IR_2015}. Assumption 1.1 is analogous to an exclusion restriction requiring that, given $X$, the instrument $Z$ should not directly affect both the potential outcomes and potential treatment statuses. Thus the effect of $Z$ on outcome $Y$ can only occur through that on the observed treatment status $D$. This assumption holds when $Z$ is randomly assigned conditional on the covariates $X$. Assumptions 1.2 and 1.3 are mild and ensure that the relevance condition that $Z$ and $D$ are correlated conditional on $X$ holds. Assumption 1.3 also guarantees that $P\left(T=c|X\right)>0$ \citep{Abadie_2003} so that there is a nonzero proportion of compliers in the population. Assumption 1.4 is a monotonicity condition, implying that the potential treatment status $D_{z}$ is weakly increasing with $z$, thereby ruling out the existence of defiers. This condition is plausible in various applications. See, for example, \citet{AAI_2002} and \citet{Abadie_2003} for further discussions on these assumptions.

Under Assumption 1, \citet{AAI_2002} show that unconfounded treatment selection holds within the group of compliers:
\begin{equation}
\left(Y_{1},Y_{0}\right)\perp D|X,T=c.
\label{unconfounded_assignment_condition}
\end{equation} 
Therefore, for the group of compliers, the distribution of potential outcome $Y_{d}$ given $X$ amounts to that of the observed outcome $Y$ given $X$ and $D=d$, which enables us to derive the causal parameters of interest for the compilers. In the present paper, we focus on the CTATE estimation for compliers, or local CTATE (LCTATE). Let $CTATE(\tau,T=c)$ and $QTE(\tau, T=c)$ denote the CTATE and QTE for compliers at the $\tau$th quantile, respectively, given $X$. If the complier quantile regression model is specified as linear-in-parameters, 
\begin{equation}
Q_{Y|D,X,T=c}(\tau )=\alpha _{1,\tau }D+X^\top\boldsymbol{\beta}_{1,\tau },
\label{linear_q_Y}
\end{equation}%
it then follows from (\ref{CTE}) that the complier CTE given $D$ and $X$ is also linear in parameters such that
\begin{equation}
CTE_{Y|D,X,T=c}(\tau )=\alpha_{2,\tau}D+X^\top\boldsymbol{\beta}_{2,\tau},
\label{linear_e_Y}
\end{equation}
where
\[
\alpha_{2,\tau} =\frac{1}{\tau}\int\nolimits_{0}^{\tau }\alpha
_{1,u}du,\text{ } \boldsymbol{\beta}_{2,\tau } =\frac{1}{\tau}\int\nolimits_{0}^{\tau }\boldsymbol{\beta}_{1,u}du.
\]

From above, it can be seen that $CTATE(\tau,T=c)=\alpha_{2,\tau}$ and $QTE(\tau, T=c)=\alpha_{1,\tau}$. In (\ref{linear_e_Y}), we are mainly concerned with the scalar parameter $\alpha_{2,\tau}$, which delivers the $\tau$th quantile level CTATE for compliers. However, this model is not directly applicable for estimation because the group of compliers are generally not observed in the data. To address this issue, we follow the weighting scheme of \citet{AAI_2002} and \citet{Abadie_2003}, who show that under Assumptions 1, given any function $h(Y,D,X)$ with a finite mean, the following result holds:
\begin{equation}
E\left[h\left(Y,D,X\right)\right|T=c]=\frac{1}{P(T=c)}E\left[K(D,Z,X)h\left(Y,D,X\right)\right],
\label{AAI_weight_loss}
\end{equation}where
\begin{eqnarray}
K(D,Z,X) &=&1-\frac{D(1-Z)}{1-\pi(X)}-\frac{(1-D)Z}{\pi(X)},
\label{AAI_weight}\\
\pi(X)&=&P(Z=1|X)\label{pi_X}.
\end{eqnarray}

Exploiting this result, we can estimate the parameters in (\ref{linear_e_Y}) by minimizing a sample analog of the numerator term of (\ref{AAI_weight_loss}) with the function $h$ being a loss function for the estimation of the CTE. For this purpose, we propose using the FZ loss \citep{FZ_2016}, which leads to a weighted empirical FZ loss minimization-based estimation problem. Before formally setting forth the FZ loss in the next section, we note that the weight $K\left(D,Z,X\right)$ may not be non-negative for all realizations of $\left(D,Z,X\right)$, thus rendering the estimation problem computationally difficult. Following \citet{AAI_2002}, we replace $K\left(D,Z,X\right)$ with \begin{eqnarray}
\bar{K}\left(Y,D,X\right)& = & E\left[K\left(D,Z,X\right)|Y,D,X\right]\nonumber\\
& = & 1-\frac{D(1-E\left[Z|Y,D,X\right])}{1-\pi(X)}-\frac{(1-D)E\left[Z|Y,D,X\right]}{\pi(X)}
\label{AAI_star}
\end{eqnarray} 
It is straightforward to see that
\begin{equation}
E\left[\bar{K}(Y,D,X)h\left(Y,D,X\right)\right] = E\left[K(D,Z,X)h\left(Y,D,X\right)\right].
\label{key_equation} 
\end{equation}
In addition, under Assumption 1, we have that $\bar{K}\left(Y,D,X\right) = P\left(T=c|Y,D,X\right)$ so that the projected weight $\bar{K}\left(Y,D,X\right)$ is almost surely non-negative. 

For practical implementation, we plug in consistent estimators of the unknown components $\pi(X)$ and $E\left(Z|Y,D,X\right)$ of (\ref{AAI_star}) to estimate $\bar{K}\left(Y,D,X\right)$. Let $\hat{\bar{K}}_{i}:=\hat{\bar{K}}(Y_{i},D_{i},X_{i})$ denote such a plug-in estimator evaluated at the $i$th observation. We will show in Section 3 that $\hat{\bar{K}}$ is consistent for $\bar{K}$ under certain regularity conditions. However, the estimated weight $\hat{\bar{K}}_{i}$ might not be guaranteed to lie within the interval $(0,1)$ for every observation in the finite samples. If $\hat{\bar{K}}_{i}>1$ or $\hat{\bar{K}}_{i}<0$, we will then further truncate it by resetting its value to be $1$ or 0 accordingly.
 
Before we proceed to parameter estimation, it is worth comparing the setting of our paper to that of the instrumental variable quantile regression (IVQR) framework of \cite{CH_2005, CH_2006}. In the IVQR setting, the outcome $Y$ is assumed to be strictly increasing in a structural unobservable that can correlate with the treatment status $D$ but is uniformly distributed conditional on the covariates $X$ and instruments $Z$. Such a setting allows for identification of the quantiles of potential outcomes $Q_{Y_{d}|X}(\tau)$ for $d\in\{0,1\}$ as well as the resulting QTE and CTATE, which are conditional on $X$ but unconditional over the agent's compliance status. By contrast, in our framework the QTE and CTATE are only identified for compliers. As the agent's compliance decision can be endogenous, the QTE and CTATE  unconditionally over the agent's compliance status are generally different from their local counterparts conditional on the group of compilers.

There are other notable differences between the IVQR setting and the causal framework of our paper. The IVQR leaves the dimensionality of unobserved heterogeneity in the treatment selection equation unrestricted. However, in the outcome equation, all unobserved heterogeneities should be summarized by a scalar unobservable that needs to satisfy the rank similarity condition, a restriction on variations of this unobservable across different treatment statuses. By contrast, our setting does not restrict dimensionality of the unobservables in the outcome equation but its monotonicity assumption (Assumption 1.4) essentially restricts the unobservables in the treatment selection equation to being a scalar. Thus the IVQR and local QTE models are nonnested and may complement each other in empirical analysis \citep{Melly_2017}.

Despite the differences between the IVQR and local QTE, \cite{Wuthrich_2020} show that the QTE in the IVQR amounts to that for compliers at adjusted quantile levels under certain regularity conditions. This finding has interesting implications for exploring the connection between the two estimation frameworks. For example, if the QTE for compliers is constant across quantiles, then the CTATE for compliers and the IVQR-based QTE are all constant across quantiles. Moreover, the CTATE for compliers and the IVQR-based QTE will all be of the same sign if the QTE for complier is positive (or negative) across all quantiles.

%This finding has interesting implications for exploring the connection between the CTATE estimands under the two frameworks. For example, if the QTE for compliers is constant across quantiles, then the CTATE for compliers and the IVQR-based QTE and CTATE are all constant across quantiles. Moreover, the CTATE for compliers and the IVQR-based QTE and CTATE will all be of the same sign if the QTE for complier is positive (or negative) across all quantiles.

\subsection{FZ loss function}

The FZ loss \citep{FZ_2016} is a class of consistent loss functions for eliciting both the quantile and CTE of a random variable, meaning that these two statistical functionals can be jointly identified by minimizing the expectation of the FZ loss. Interestingly, unlike the quantile, which can be elicited with some consistent loss functions (say, the check loss), there does not exist a consistent loss function solely for eliciting the CTE \citep{FZ_2016}. 

Let $\mathcal{F}$ be a class of distribution functions on real numbers $\mathbb{R}$ with finite first moments and unique $\tau$-quantiles. For a random variable $Y$ following a distribution $F\in\mathcal{F}$, let $Q_{Y}\left(\tau\right)$ denote the $\tau$-quantile and $CTE_{Y}\left(\tau\right)$ denote the corresponding CTE. \citet{FZ_2016} show that
\begin{equation}
\left(Q_{Y}\left(\tau\right),CTE_{Y}\left(\tau\right)\right)=\arg\min_{\left(q,e\right)}E\left[FZ_{\tau}\left(q,e,Y\right)\right],\label{FZ_minimization}
\end{equation}where\begin{eqnarray}
	FZ_{\tau}\left(q,e,y\right) & = & 1\left\{ y\leq q\right\} \left[G_{1}\left(q\right)-G_{1}\left(y\right)\right]-\tau G_{1}\left(q\right)\nonumber \\
	&  & +G_{2}^{\prime}\left(e\right)\left[e+LQ_{\tau}\left(q,y\right)\right] -G_{2}\left(e\right)+\eta\left(y\right)\label{FZ_loss}
\end{eqnarray}
is the FZ loss. Here $LQ_{\tau}\left(q,y\right):=\tau^{-1}\max\left(q-y,0\right)-q$, $\tau\in\left(0,1\right)$ is the quantile level, $(q,e,y)\in\mathbb{R}^{3}$ and $e\leq q$. $G_{1}\left(.\right)$,
$G_{2}\left(.\right)$ and $\eta\left(.\right)$ are functions of
real numbers and $G_{2}^{\prime}\left(.\right)$ is the subgradient
of $G_{2}\left(.\right)$. It is required that $1_{\left[-\infty,q\right)}G_{1}\left(.\right)$ be $F$ integrable for all $q\in\mathbb{R}$ and all $F\in\mathcal{F}$, and
$G_{2}^{\prime}\left(.\right)$ and $\eta(.)$ be $F$ integrable for all $F\in\mathcal{F}$. In addition, $G_{1}\left(.\right)$ should be increasing and $G_{2}\left(.\right)$ should be increasing and convex. If $G_{2}\left(.\right)$ is both strictly increasing and strictly convex, the minimizer in (\ref{FZ_minimization}) is unique and the loss function $FZ_{\tau}\left(q,e,y\right)$
is said to be \textit{strictly} consistent\footnote{Let $L(x,y)$ denote a loss function for obtaining a statistical functional of a random variable. In our case, $L=F_{\tau}^{sp}$ and $x=(q,e)$. 
Let $\mathcal{F}$ denote a class of distribution functions and $F$ be an element in $\mathcal{F}$. Let $\lambda:\mathcal{F}\mapsto \mathbb{S}$ denote a statistical functional which maps $F\in \mathcal{F}$ to a set $\mathbb{S}$, which in our case amounts to the set $\mathbb{R}^{2}$. $L(x,y)$ is consistent for the statistical functional $\lambda(F)$ if $ E_{F}\left[L\left(\lambda(F),Y\right)\right]\leq E_{F}\left[L\left(x,Y\right)\right]$ for all $F\in\mathcal{F}$ and all $x\in\mathbb{S}$ and a random variable $Y$ following the distribution $F$. If a loss function is consistent and $E_{F}\left[L\left(\lambda(F),Y\right)\right]=E_{F}\left[L\left(x,Y\right)\right]$ implies $x=\lambda\left(F\right)$, the loss function is said to be strictly consistent.} for the $\tau$-quantile and the corresponding CTE of all $F\in\mathcal{F}$ \citep{FZ_2016}.

%\footnote{\textcolor{blue}{Let $L(x,y)$ denote a loss function for obtaining a statistical functional of a random variable. In our case, $L=F_{\tau}^{sp}$ and $x=(q,e)$, $y=y$. 
%Let $\mathcal{F}$ denote a class of distribution functions and $F$ be an element in $\mathcal{F}$. Let $\lambda:\mathcal{F}\mapsto \mathbb{R}$ denote a statistical functional which maps $F\in \mathcal{F}$ to $\mathbb{R}$. $L(x,y)$ is consistent for a statistical functional $\lambda(F)$ if $ E_{F}\left[L\left(x^{*},Y\right)\right]\leq E_{F}\left[L\left(x,Y\right)\right]$ for all $F\in\mathcal{F}$, all $x^{*}\in\lambda\left(F\right)$, all $x\in\mathbb{R}$ and a random variable $Y$ following distribution $F$. If a loss function is consistent and $E_{F}\left[L\left(x^{*},Y\right)\right]=E_{F}\left[L\left(x,Y\right)\right]$ implies $x\in\lambda\left(F\right)$, the loss function is called strictly consistent.}} for the $\tau$-quantile and the corresponding CTE of all $F\in\mathcal{F}$ \citep{FZ_2016}.

Let $\boldsymbol{\theta}_{1,\tau}=(\alpha_{1,\tau},\boldsymbol{\beta}_{1,\tau}^\top)^\top$ and $\boldsymbol{\theta}_{2,\tau}=(\alpha_{2,\tau},\boldsymbol{\beta}_{2,\tau}^\top)^\top$ be the vectors that collect the true parameters of models (\ref{linear_q_Y}) and (\ref{linear_e_Y}), respectively. Using the weighted loss minimization approach of Section 2.2, we can estimate $\boldsymbol{\theta}_{1,\tau}$ and $\boldsymbol{\theta}_{2,\tau}$ by minimizing the sample analog of the left-hand side of (\ref{key_equation}), with the function $h$ being replaced by the FZ loss (\ref{FZ_loss}). Specifically, let $\hat{\boldsymbol{\theta}}_{1,\tau}$ and $\hat{\boldsymbol{\theta}}_{2,\tau}$ denote the estimators of $\boldsymbol{\theta}_{1,\tau}$ and $\boldsymbol{\theta}_{2,\tau}$. Assume that the data consist of a sample of $n$ observations $\left( Y_{i},D_{i},X_{i},Z_{i}\right), i=1,\ldots,n$. The estimators $\hat{\boldsymbol{\theta}}_{1,\tau}$ and $\hat{\boldsymbol{\theta}}_{2,\tau}$ can then be computed as the solution to the problem
\begin{equation}
	\min_{\left(\boldsymbol{\theta}_{1},\boldsymbol{\theta}_{2}\right)}\frac{1}{n}\sum_{i=1}^{n}\tilde{K}_{i}FZ_{\tau}\left(q_{i}\left(\boldsymbol{\theta}_{1}\right),e_{i}\left(\boldsymbol{\theta}_{2}\right),Y_{i}\right),\label{empirical_loss}
\end{equation}
where $q_{i}(\boldsymbol{\theta}_{1})=q(D_{i},X_{i},\boldsymbol{\theta}_{1}):= \alpha _{1}D_{i}+X_{i}^\top\boldsymbol{\beta}_{1}$ and $e_{i}(\boldsymbol{\theta}_{2})=e(D_{i},X_{i},\boldsymbol{\theta}_{1}):=\alpha _{2}D_{i}+X_{i}^\top\boldsymbol{\beta}_{2}$ are regression functions of (\ref{linear_q_Y}) and (\ref{linear_e_Y}) evaluated at parameter values $\boldsymbol{\theta}_{1}=(\alpha_{1},\boldsymbol{\beta}_{1}^\top)^\top$ and $\boldsymbol{\theta}_{2}=(\alpha_{2},\boldsymbol{\beta}_{2}^\top)^\top$, and $\tilde{K}_{i}=\min\left\{\max\left\{\hat{\bar{K}}_{i},0\right\},1\right\}$ is an estimate of $\bar{K}\left(Y_{i},D_{i},X_{i}\right)$ truncated to take value in the interval $[0,1]$.

\subsection{Specification of the FZ loss and an implementation algorithm}
To implement the estimation practically, we must specify the functional
forms of $G_{1}\left(.\right)$, $G_{2}\left(.\right)$ and $\eta\left(.\right)$.
In the present paper, we set $G_{1}\left(t\right)=0$, $G_{2}\left(t\right)=\eta\left(t\right)=\ln\left(1+\exp\left(t\right)\right)$ (the softplus function)
for $t\in\mathbb{R}$, which enables us to solve the problem (\ref{empirical_loss}) through a simple iterative computation algorithm.  
Under these settings, the loss function $FZ_{\tau}\left(q,e,y\right)$
becomes 
\begin{equation}
	FZ_{\tau}^{sp}\left(q,e,y\right)=\frac{\exp\left(e\right)}{1+\exp\left(e\right)}\left[e+LQ_{\tau}\left(q,y\right)\right]-\ln\left[1+\exp\left(e\right)\right]+\ln\left[1+\exp\left(y\right)\right].\label{FZ_sp}
\end{equation}
The loss function (\ref{FZ_sp}) is defined for all $\left(q,e,y\right)\in\mathbb{R}^{3}$. As $G_{2}(.)$ is the softplus function, it is straightforward to verify that $FZ_{\tau}^{sp}\left(q,e,y\right)$ is a strictly consistent loss function. %\citep{Gneiting_2011}}
Moreover, because $FZ_{\tau}(q,e,y)\geq0$ whenever $\eta(y)=\tau G_{1}(y)+G_{2}(y)$ \citep{DB_2019}, it follows that $FZ_{\tau}^{sp}\left(q,e,y\right)$ is a non-negative function as well. 

For our specifications, $G_{1}(t)=0$ is the most frequently used specification in practice \citep{FZ_2016,DB_2019,PFC_2019}.
Previous studies also have found that such a choice of $G_{1}(t)$ is a good candidate to make the FZ loss fulfill the positively homogeneous property, which means that the ordering of the losses will be unaffected by changing the unit of measurement of the data. Notice that although $G_{1}(q)=G_{1}(y)=0$, the check loss is still implicitly kept in $LQ_{\tau}(q,y)$ in (\ref{FZ_sp}) and therefore estimating the quantile function is still viable through $FZ_{\tau}^{sp}$. As for $G_{2}(t)=\ln(1+\exp(t))$, our setting makes its first and second order derivatives bounded for all $t\in \mathbb{R}$. This property helps us to establish the relevant asymptotic results and asymptotic covariance matrix of the estimators in a parsimonious and less restrictive manner (see Section 3). In addition, this setting does not restrict the sign of the CTE estimate, which is crucial for applications in microeconometrics. In Appendix A.10, we provide a more detailed discussion on issues of the specifications. We also show that using alternative specifications for $G_{1}(t)$ and $G_{2}(t)$ does not seem to affect our empirical results.

With the FZ loss specification (\ref{FZ_sp}), we solve the following minimization problem
to obtain estimates of $\boldsymbol{\theta}_{1,\tau}$ and $\boldsymbol{\theta}_{2,\tau}$:
\begin{equation}
	\min_{\left(\boldsymbol{\theta}_{1},\boldsymbol{\theta}_{2}\right)}\frac{1}{n}\sum_{i=1}^{n}\tilde{K}_{i}FZ_{\tau}^{sp}\left(q_{i}\left(\boldsymbol{\theta}_{1}\right),e_{i}\left(\boldsymbol{\theta}_{2}\right),Y_{i}\right).\label{FZ_sp_emp_loss}
\end{equation}
However, joint minimization over $\left(\boldsymbol{\theta}_{1},\boldsymbol{\theta}_{2}\right)$ in (\ref{FZ_sp_emp_loss}) could be computationally nontrivial, as the function $LQ_{\tau}\left(q,y\right)$ involves the term $\max\left(q-y,0\right)$, which has a kink and is thus not everywhere differentiable. \citet{PFC_2019} adopted a smoothed approximation of the objective function in their FZ loss minimization problem. They solved the smoothed problem to obtain an approximate solution, which was then used as an initial estimate in a non-gradient based numerical procedure for the joint minimization of the original non-smoothed objective function in the estimation problem. Here we propose an alternative computational approach, which can be efficiently and stably implemented to solve the problem (\ref{FZ_sp_emp_loss}). 

Our computational algorithm builds on the following iterative scheme. Let $\hat{\boldsymbol{\theta}}_{1}^{(k)}$ and $\hat{\boldsymbol{\theta}}_{2}^{(k)}$ denote the computed estimates of $\boldsymbol{\theta}_{1}$ and $\boldsymbol{\theta}_{2}$ at
the $k$-th iteration step. Note that, fixing the value of $\boldsymbol{\theta}_{1}$ at $\hat{\boldsymbol{\theta}}_{1}^{(k)}$, minimization of the objective function of (\ref{FZ_sp_emp_loss}) with respect to $\boldsymbol{\theta}_{2}$ is a smooth optimization problem, which can be easily solved using commonly used numerical optimization solvers (e.g., the \texttt{optim} function of \texttt{R}). 
Next, fixing the value of $\boldsymbol{\theta}_{2}$ at $\hat{\boldsymbol{\theta}}_{2}^{(k)}$, minimization of the objective function of (\ref{FZ_sp_emp_loss}) with respect to $\boldsymbol{\theta}_{1}$ reduces to the following weighted quantile regression estimation problem
\begin{equation}
	\min_{\boldsymbol{\theta}_{1}}\frac{1}{n}\sum_{i=1}^{n}\hat{\xi}_{i}^{\left(k\right)}c_{\tau}\left(Y_{i}-q_{i}\left(\boldsymbol{\theta}_{1}\right)\right),\label{WLAD}
\end{equation}
where $c_{\tau}(u) = (\tau - 1\{u<0\})u$ is the check loss and the weight
\begin{equation}
\hat{\xi}_{i}^{\left(k\right)}=\tilde{K}_{i}\frac{\exp\left[e_{i}\left(\hat{\boldsymbol{\theta}}_{2}^{\left(k\right)}\right)\right]}{1+\exp\left[e_{i}\left(\hat{\boldsymbol{\theta}}_{2}^{\left(k\right)}\right)\right]}
\label{weight_xi_hat}
\end{equation}
is nonnegative. To see this, note that fixing $\boldsymbol{\theta}_{2}$ at $\hat{\boldsymbol{\theta}}_{2}^{(k)}$, the only term involving with $\boldsymbol{\theta}_{1}$ in (\ref{FZ_sp_emp_loss}) is $\hat{\xi}_{i}^{(k)}LQ_{\tau}\left(q_{i}(\boldsymbol{\theta}_{1}),Y_{i}\right)$. Furthermore, it can be shown that 
$\hat{\xi}_{i}^{(k)}LQ_{\tau}\left(q_{i}(\boldsymbol{\theta}_{1}),Y_{i}\right) \propto \hat{\xi}_{i}^{(k)}c_{\tau}(Y_{i}-q_{i}(\boldsymbol{\theta}_{1}))$. Problem (\ref{WLAD}) can be solved efficiently using well-known algorithms for quantile regression estimation (e.g., the \texttt{quantreg} package of \texttt{R}).

The proposed iterated scheme is an example of alternating optimization (or block coordinate descent algorithm). This type of optimization procedure is designed to replace the joint optimization of a multivariate function over all variables, which is sometimes difficult to solve, with a sequence of easily solved optimizations involving grouped subsets of the variables. %\footnote{Some well known algorithms in econometrics and statistics belong to the class of alternating optimization, for example, c-means clustering algorithms and the expectation-maximization (EM) algorithm. Convergence analysis on alternating optimization can be found in \citet{BH_2002} and \citet{Bertsekas_2016}.} 
Note that in practice, $\hat{\boldsymbol\theta}_{1}^{(1)}$ used for the first iteration step can be obtained from directly solving a weighted quantile regression estimation problem with weight $\tilde{K}_{i}$. Using the arguments above, our estimators $\hat{\boldsymbol{\theta}}_{1,\tau}$ and $\hat{\boldsymbol{\theta}}_{2,\tau}$ are thus computed as the solution upon the convergence of the iterative algorithm.

\subsection{Comparison with an integrated-QTE based estimator}
From (\ref{CTATE and QTE}), a more straightforward approach to estimate the CTATE for compliers is to integrate estimates of the QTE for compliers from zero to a specific quantile level $\tau$. Suppose (\ref{linear_q_Y}) holds and let 
\begin{equation}
	\widehat{\text{IntQ}}^{\star}(\tau) = \frac{1}{\tau}\int_{0}^{\tau}\hat{\alpha}_{1,u}du
	\label{IntQ}
\end{equation}denote such an integrated-QTE based estimator for estimating the CTATE for compliers at quantile level $\tau\in (0,1)$, where $\hat{\alpha}_{1,u}$ is an estimate of $\alpha_{1,u}$, the QTE for compliers at quantile level $u$. The integrated-QTE based estimator is derived from the difference between two integrated-quantile based estimators (e.g., \cite{LW_2018}) for estimating the CTE of $Y_d$ at quantile level $\tau$. One advantage of (\ref{IntQ}) is its ability to accommodate various models of QTE, such as those for compliers under the two-sided non-compliance framework \citep{AAI_2002}, or the instrumental variable quantile regression (IVQR) framework of \citet{CH_2004, CH_2005, CH_2006, CH_2008}.

Our proposed estimator for estimating the CTATE requires a specification of the functions $G_1$ and $G_2$ in the FZ loss. %By contrast, using (\ref{IntQ}) to directly integrate estimates of the QTEs as an estimation for the CTATE does not require that. 
Nonetheless, our proposed estimator using the FZ loss still has several advantages over the integrated-QTE based estimator. One of them is that it is easier to derive theoretical properties of the proposed estimator. For example, to establish ``pointwise'' consistency of $\widehat{\text{IntQ}}^{\star}(\tau)$ for estimating the CTATE at $\tau\in (0,1)$, since this estimator integrates QTE estimates $\hat{\alpha}_{1,u}$ over $u\in (0,\tau)$, typically one would have to derive uniform consistency of $\hat{\alpha}_{1,u}$ %for estimating $\alpha_{1,u}$
over $u\in (0,\tau)$. However, as shown in Theorem 1 in Section 3, such uniform approximation property is not required to establish the pointwise consistency of our proposed estimator.

Furthermore, since the integration domain in (\ref{IntQ}) includes low quantile indices that fall around the neighborhood of zero, conducting inference based on the estimator $\widehat{\text{IntQ}}^{\star}(\tau)$ requires addressing the statistical behavior of the QTE estimators evaluated at extreme quantiles. The literature on extremal quantile estimation and inference indicates that the finite sample distribution of the quantile regression estimator evaluated at extremal quantile levels cannot be accurately approximated by the normal distribution \citep{Chernozhukov_2005, CF_2011}. Consequently, the asymptotic analysis of $\widehat{\text{IntQ}}^{\star}(\tau)$ is nonstandard as the process of estimated QTEs over $u\in (0,\tau)$ does not generally converge weakly to a Gaussian process.

In practice, one might consider using a ``trimmed'' version of (\ref{IntQ}) to circumvent the issues associated with estimation at extreme quantiles, such as
\begin{equation}
	\widehat{\text{IntQ}}(\tau)=\frac{1}{\tau}\int_{\underline{\tau}}^{\tau}\hat{\alpha}_{1,u}du,
	\label{trimmed_IntQ}
\end{equation}where $\underline{\tau}$ is a small constant and $0<\underline{\tau}<\tau$. It is evident that (\ref{trimmed_IntQ}) will incur a truncation bias in the estimation of the CTATE ($\text{IntQ}^{\star}(\tau)$) at quantile level $\tau$, but it is hoped that this bias will not be excessively large if the trimming constant $\underline{\tau}$ is set to a very small value. However, as demonstrated in simulation results in Appendix A.9, even for a very small $\underline{\tau}$, %as small as 0.005,
the resulting bias can remain substantial.

Finally, the integrated-QTE based estimator (\ref{IntQ}) relies on the computation of $\hat{\alpha}_{1,u}$ over a finely distributed collection of quantile indices for numerical integration. The computational burden of the QTE-based approach can be substantially reduced by employing a more sophisticated numerical integration algorithm than the brute-force one of using the estimated QTEs evaluated at many different quantiles. %comparing to the estimator using the FZ loss, 
	%Thus this alternative estimator could be relatively computational intensive. %For estimating the CTATE at a specific quantile level $\tau$, as one needs to estimate QTE's at many quantile levels $u\in[0,\tau]$ for the numerical integration, 
	However, differences in computational costs between the two competing estimators could still occur as the complexity of the regression models and the sample size increase.

%Such a difference would be further enlarged in a resampling based procedure (e.g., bootstrap) for conducting statistical inferences. These could also compromise the usefulness of the integrated-QTE based estimator in a data rich environment currently faced by econometricians.

\subsection{Inter-quantile average treatment effect} 

We conclude Section 2 by remarking here on one related causal parameters that build upon the CTATE. Define inter-quantile expectation (IQE) of a potential outcome $Y_{d}$ between quantile levels $\tau^{\prime}$ and
$\tau$ with $0<\tau^{\prime}<\tau<1$ as
\begin{eqnarray*}
	IQE\left(\tau^{\prime},\tau,d\right) & := & E\left[Y_{d}|Q_{Y_{d}|X}\left(\tau^{\prime}\right)\leq Y_{d}\leq Q_{Y_{d}|X}\left(\tau\right),X\right]\\
	&=& \frac{\tau CTE_{Y_{d}|X}\left(\tau\right)-\tau^{\prime}CTE_{Y_{d}|X}\left(\tau^{\prime}\right)}{\tau-\tau^{\prime}}. 
\end{eqnarray*}
Using the notion of IQE, we can define the inter-quantile average treatment effect (IQATE), which is the difference
between inter-quantile expectations of two potential outcomes:
\begin{eqnarray}
IQATE(\tau^{\prime},\tau) &:=& IQE\left(\tau^{\prime},\tau,1\right)-IQE\left(\tau^{\prime},\tau,0\right)\nonumber\\
&=&\frac{\tau CTATE(\tau)-\tau^{\prime}CTATE(\tau^{\prime})}{\tau-\tau^{\prime}}.
\label{IQATE}
\end{eqnarray}
Through (\ref{CTATE and QTE}), the IQATE can be viewed as an aggregate of quantile treatment effects over a range of quantile levels and would thus be useful for providing a summary of heterogeneous treatment effects locally over a specified quantile range. Under Assumption 1 and model (\ref{linear_e_Y}), we can deduce from (\ref{IQATE}) that the IQATE for compliers in the binary treatment causal framework of Section 2.2 reduces to 
\[
\frac{\tau\alpha_{2,\tau}-\tau^{\prime}\alpha_{2,\tau^{\prime}}}{\tau-\tau^{\prime}},    
\]
which can be readily estimated using the proposed estimator for estimating CTATE for compliers.

%Conditional tail expectations of potential outcomes $CTE_{Y_{d}|X}\left(\tau\right)$ are also relevant for deriving the Lorenz curves of the corresponding distributions. Suppose the potential outcome $Y_{d}$ has non-zero mean conditional on $X$. The Lorenz curve, an often used measure for degree of income or wealth inequality, is defined as a ratio of partial mean to the overall mean of $Y_{d}$ :\[
%	LO\left(\tau,d\right) := \frac{\tau CTE_{Y_{d}|X}(\tau)}{E\left[Y_{d}|X\right]} = \frac{\tau CTE_{Y_{d}|X}(\tau)}{\lim_{\tau\rightarrow 1} \tau CTE_{Y_{d}|X}(\tau)}.
%\] The Lorenz effect \citep{CFM_2013} is then defined as \[LO\left(\tau,1\right)-LO\left(\tau,0\right),\]which is useful on measuring how the degree of inequality of an outcome of interest in a population changes across two treatment regimes. The Lorenz effect for the group of compliers can also be calculated empirically using the estimated parameters of model (\ref{linear_e_Y}) in the causal effect framework of this paper.
\section{Asymptotic Result}
In this section we establish asymptotic properties of the estimators $\hat{\boldsymbol{\theta}}_{1,\tau}$ and $\hat{\boldsymbol{\theta}}_{2,\tau}$. We assume that the probability $\pi\left(X\right)$ in (\ref{pi_X}) depends on some parameters $\boldsymbol{\gamma}_{0}$ and rewrite it as
$\pi\left(X,\boldsymbol{\gamma}_{0}\right)$. Let $W:=\left(D, X^\top\right)^\top$, $V:=\left(Y,W\right)$ and $v_{0}\left(V\right):=E[Z|V]$. Let $\mathcal{V}$ denote the support of $V$ and $\boldsymbol{\Theta}$ be the parameter space of $\left(\boldsymbol{\theta}_{1,\tau},\boldsymbol{\theta}_{2,\tau}\right)$. Define
\begin{eqnarray}
K\left(W,Z;\boldsymbol{\gamma}\right) &:=&  1-\frac{D(1-Z)}{1-\pi(X,\boldsymbol{\gamma})}-\frac{(1-D)Z}{\pi(X,\boldsymbol{\gamma})}, \label{AAI_star2}\\
\bar{K}\left(V;v,\boldsymbol{\gamma}\right) &:=&  1-\frac{D(1-v\left(V\right))}{1-\pi(X,\boldsymbol{\gamma})}-\frac{(1-D)v\left(V\right)}{\pi(X,\boldsymbol{\gamma})}, \label{AAI_star1}\\
\tilde{K}\left(V;v,\boldsymbol{\gamma}\right)&:=&\min\left\{\max\left\{\bar{K}\left(V;v,\boldsymbol{\gamma}\right),0\right\},1\right\}.\label{AAI_star3}
\end{eqnarray}
With (\ref{AAI_star2}) to (\ref{AAI_star3}), it can be seen that $K\left(D,Z,X\right)$ in (\ref{AAI_weight}) can be expressed as $ K\left(W,Z;\boldsymbol{\gamma}_{0}\right)$, and $\bar{K}\left(Y,D,X\right)$ in (\ref{AAI_star}) can be expressed as $\bar{K}\left(V;v_{0},\boldsymbol{\gamma}_{0}\right)$. The truncated version of $\bar{K}\left(Y,D,X\right)$ is then given by $\tilde{K}\left(V;v_{0},\boldsymbol{\gamma}_{0}\right)$.

In the following, we derive asymptotic results for a more general setup where $Q_{Y_{i}|W_{i},T=c}\left(\tau\right)$ and $ CTE_{Y_{i}|W_{i},T=c}\left(\tau\right)$ take parametric structural forms $q_{i}\left(\boldsymbol{\theta}_{1}\right):=q\left(W_{i},\boldsymbol{\theta}_{1}\right)$ and $e_{i}\left(\boldsymbol{\theta}_{2}\right):=e\left(W_{i},\boldsymbol{\theta}_{2}\right)$, which could be nonlinear but are known up to some finite dimensional vectors of parameters. Let $\hat{\boldsymbol{\gamma}}$ and $\hat{v}$ denote some estimators of $\boldsymbol{\gamma}_{0}$ and $v_{0}$. We are interested in the parameter values $\boldsymbol{\theta}_{1,\tau}$ and $\boldsymbol{\theta}_{2,\tau}$, which are estimated by $\hat{\boldsymbol{\theta}}_{1,\tau}$ and $\hat{\boldsymbol{\theta}}_{2,\tau}$, where   
\begin{equation}
\left(\hat{\boldsymbol{\theta}}_{1,\tau},\hat{\boldsymbol{\theta}}_{2,\tau}\right)=\arg\min_{\left(\boldsymbol{\theta}_{1},\boldsymbol{\theta}_{2}\right)\in\boldsymbol{\Theta}}\frac{1}{n}\sum_{i=1}^{n}\tilde{K}\left(V_{i};\hat{v},\hat{\boldsymbol{\gamma}}\right)FZ_{\tau}\left(q_{i}\left(\boldsymbol{\theta}_{1}\right),e_{i}\left(\boldsymbol{\theta}_{2}\right),Y_{i}\right),
\label{sample_analogue}
\end{equation}
and
\begin{equation}
 \left(\boldsymbol{\theta}_{1,\tau},\boldsymbol{\theta}_{2,\tau}\right)  =\arg\min_{\left(\boldsymbol{\theta}_{1},\boldsymbol{\theta}_{2}\right)\in\boldsymbol{\Theta}}E\left[FZ_{\tau}\left(q\left(W,\boldsymbol{\theta}_{1}\right),e\left(W,\boldsymbol{\theta}_{2}\right),Y\right)|T=c\right].\label{estimation_problem2}
\end{equation} 
We make the following regularity assumptions for the consistency of the parameter estimators.

\begin{assumption}
\item[1.] The data $\left(Y_{i},D_{i},X_{i},Z_{i}\right)$, $i=1,\ldots,n$, are i.i.d..
\item[2.] The distribution of $Y$ conditional on $W$ and $T=c$ is absolutely continuous with respect to the Lebesgue measure.
\item[3.] The parameter space $\boldsymbol{\Theta}$ is compact.
\item[4.] $q\left(W,\boldsymbol{\theta}_{1}\right)$ and $e\left(W,\boldsymbol{\theta}_{2}\right)$ are continuous in $\boldsymbol{\theta}_{1}$ and $\boldsymbol{\theta}_{2}$ respectively for every $\left(\boldsymbol{\theta}_{1},\boldsymbol{\theta}_{2}\right)\in\boldsymbol{\Theta}$.
\item[5.] The components $G_{1}\left(.\right)$, $G_{2}\left(.\right)$ and $G_{2}^{\prime}\left(.\right)$ in (\ref{FZ_loss}) are continuous differentiable functions with $G_{1}^{\prime}\left(.\right)\geq0$,
$G_{2}^{\prime}\left(.\right)>0$ and $G_{2}^{\prime\prime}\left(.\right)>0$.
\item[6.] $P\left(\left\{ q\left(W,\boldsymbol{\theta}_{1}\right)=q\left(W,\boldsymbol{\theta}_{1,\tau}\right)\right\} \cap\left\{ e\left(W,\boldsymbol{\theta}_{2}\right)=e\left(W,\boldsymbol{\theta}_{2,\tau}\right)\right\} |T=c\right)=1$
implies $\left(\boldsymbol{\theta}_{1},\boldsymbol{\theta}_{2}\right)=\left(\boldsymbol{\theta}_{1,\tau},\boldsymbol{\theta}_{2,\tau}\right).$
\item[7.] There is a function $b\left(V\right)$ with $E\left[b\left(V\right)\right]<\infty$ such that
$b\left(V\right)>\left| FZ_{\tau}\left(q\left(W,\boldsymbol{\theta}_{1}\right),e\left(W,\boldsymbol{\theta}_{2}\right),Y\right)\right| $ for all $\left(\boldsymbol{\theta}_{1},\boldsymbol{\theta}_{2}\right)\in\boldsymbol{\Theta}$.
\item[8.] The estimated weight $\bar{K}(V;\hat{\boldsymbol{\gamma}},\hat{v})$ is uniformly consistent such that
\[\sup_{V\in\mathcal{V}}\left|\bar{K}\left(V;\hat{\boldsymbol{\gamma}},\hat{v}\right)-\bar{K}\left(V;\boldsymbol{\gamma}_{0},v_{0}\right)\right|=o_{p}\left(1\right).\]
\end{assumption}
%\item[7.] $0<\bar{K}\left(v_{0},\boldsymbol{\gamma}_{0}\right)=P(T=c|V)\leq 1$.
%\item[8.]  The estimators $\hat{\boldsymbol{\gamma}}$ and $\hat{v}$ satisfy that $\hat{\boldsymbol{\gamma}}\stackrel{p}{\rightarrow}\boldsymbol{\gamma}_{0}$ and \[\sup_{\mathrm{v}\in\mathcal{V}}\left|\hat{v}\left(\mathrm{v}\right)-v_{0}\left(\mathrm{v}\right)\right|=o_{p}\left(1\right).\]

Assumptions 2.1, 2.3 and 2.4 are standard. Assumption 2.2 is sufficient for the uniqueness of any $\tau$-quantile of the distribution of $Y$ given $W$ and $T=c$. Assumption 2.5 requires that the FZ loss function (\ref{FZ_loss}) be strictly consistent for eliciting quantiles and CTEs. Assumption 2.6 is the rank condition for parameter identification. For (\ref{linear_q_Y}) and (\ref{linear_e_Y}), this condition immediately holds when, conditional on $T=c$, the support of $W$ is not contained in any proper linear subspace of $\mathbb{R}^{k}$ where $k$ denotes the dimension of $W$. Assumptions 2.5 and 2.6 ensure that the solution $\left(\boldsymbol{\theta}_{1,\tau},\boldsymbol{\theta}_{2,\tau}\right)$ is unique in the expected FZ loss minimization problem (\ref{estimation_problem2}). Assumptions 2.7 and 2.8 are imposed to establish uniform convergence of the empirical objective function in (\ref{sample_analogue}) to its population counterpart. Assumption 2.7 is a dominance condition, which can hold under the compactness of $\boldsymbol{\Theta}$ and the aforementioned requirements for the FZ loss function, provided that $W$ has a bounded support. Assumption 2.8 hinges on the consistency of the plug-in estimators $\hat{\boldsymbol{\gamma}}$ and $\hat{v}$. With Assumptions 1 and 2, we can establish the following result.
\begin{theorem}
If Assumptions 1 and 2 hold, then, for any given $\tau\in\left(0,1\right)$,
\[\left(\hat{\boldsymbol{\theta}}_{1,\tau},\hat{\boldsymbol{\theta}}_{2,\tau}\right)\stackrel{p}{\rightarrow}\left(\boldsymbol{\theta}_{1,\tau},\boldsymbol{\theta}_{2,\tau}\right).\] 	
\end{theorem}

In our numerical studies, we estimate the function $v_{0}$ using the power series estimator \citep{Newey_1997} and $\gamma_{0}$ with some parametric binary choice model. Let $W=(W_{d},W_{c})$, where $W_{d}$ and $W_{c}$ denote the discrete and continuous components of $W$ respectively. Let $\mathcal{W}_{d}$ denote the support of $W_{d}$, which takes only a finite number of possible values. The estimator $\hat{v}(V)$ takes the form
\begin{equation}
\hat{v}(V):=\sum_{m\in\mathcal{W}_{d}}1\{W_{d}=m\}\hat{v}_{m}(Y,W_{c}),
\label{v_hat}
\end{equation}
where, for each $m\in\mathcal{W}_{d}$, $\hat{v}_{m}(Y,W_{c})$ is a power series estimator of the conditional expectation $v_{0,m}(Y,W_{c}):=E[Z|Y,W_{d}=m,W_{c}]$. Let $\kappa_{m}$ denote the number of series terms in the power series approximation of $v_{0,m}$, $s_{m}$ be the order of continuous derivatives of $v_{0,m}$ and $r$ be the dimension of $W_{c}$. The next assumption allows us to verify the uniform consistency of the estimated weight $\hat{\bar{K}}(.)$.

\begin{assumption} %(Uniform convergence of the weight)
\item[1.] The support of $(Y,W_{c})$ conditional on $W_{d}$ is a Cartesian product of compact connected intervals on which $(Y,W_{c})$ has a probability density function that is bounded away from zero. 
\item[2.] $v_{0,m}(Y,W_{c})$ is continuously differentiable of order $s_{m}$ on the support of $(Y,W_{c})$. 
\item[3.] There is a constant $\varepsilon>0$ such that, for the neighborhood $\left\Vert \boldsymbol{\gamma}-\boldsymbol{\gamma}_{0}\right\Vert \leq\varepsilon$,
$\pi\left(x,\boldsymbol{\gamma}\right)$ is bounded away from 0 and 1 and has a partial derivative $\nabla_{\boldsymbol{\gamma}}\pi\left(x,\boldsymbol{\gamma}\right)$, which is uniformly bounded for $\boldsymbol{\gamma}$ in this neighborhood and for $x$ over the support of $X$.
\item[4.] $\hat{\boldsymbol{\gamma}}\stackrel{p}{\rightarrow}\boldsymbol{\gamma}_{0}$. 
\end{assumption}

Assumptions 3.1 and 3.2 are standard conditions for power series estimators \citep{Newey_1997}. Assumptions 3.3 is also a mild condition on the smoothness of $\pi(X,\boldsymbol{\gamma})$. Assumptions 3.4 requires that the estimator $\hat{\boldsymbol{\gamma}}$ be consistent for $\boldsymbol{\gamma}_{0}$. Under Assumption 3, we have the following result.

\begin{lemma}
	Suppose Assumptions 2.1, 2.3 and 3 hold. Then Assumption 2.8 also holds under the series approximation conditions that $r+1<s_{m}$, $\kappa_{m}\rightarrow\infty$ and $\kappa_{m}^{3}/n\rightarrow0$ for each $m\in\mathcal{W}_{d}$.
\end{lemma}

We now provide further regularity assumptions for deriving the asymptotic distribution of the estimators.

\begin{assumption}%(Asymptotic covariance matrix of the estimated parameters)
\item[1.] $\hat{\boldsymbol{\gamma}}$ is asymptotically linear with influence
function $\psi\left(X\right)$:
\[
\sqrt{n}\left(\hat{\boldsymbol{\gamma}}-\boldsymbol{\gamma}_{0}\right)=\frac{1}{\sqrt{n}}\sum_{i=1}^{n}\psi\left(X_{i}\right)+o_{p}\left(1\right),
\]
where $E\left[\psi\left(X\right)\right]=0$ and $E\left| \psi\left(X\right)\right| ^{2}<\infty$.
\item[2.] $E\left[\left(A_{k,\tau}\left(V\right)\right)^{2}|T=c\right]<\infty$,
where $A_{k,\tau}\left(V\right)$, $k=1,\ldots,5$ are the terms defined
in Lemma A.1 in Appendix A.3.	
\item[3.] The matrix 
\[
\nabla_{\boldsymbol{\theta\theta}}E\left[FZ_{\tau}\left(q\left(W,\boldsymbol{\theta}_{1,\tau}\right),e\left(W,\boldsymbol{\theta}_{2,\tau}\right),Y\right)|T=c\right]
\] is nonsingular, where $\boldsymbol{\theta}:=\left(\boldsymbol{\theta}_{1},\boldsymbol{\theta}_{2}\right)$.
\item[4.] The conditional density $f_{Y|W,T=c}\left(q\left(W,\boldsymbol{\theta}_{1,\tau}\right)\right)$ is bounded away from 0 over the support of $W$.
\item[5.] The estimated weight $\bar{K}(V;\hat{\boldsymbol{\gamma}},\hat{v})$ satisfies that \[\sup_{V\in\mathcal{V}}\left|\bar{K}\left(V;\hat{\boldsymbol{\gamma}},\hat{v}\right)-\bar{K}\left(V;\boldsymbol{\gamma}_{0},v_{0}\right)\right|=o_{p}\left(n^{-1/4}\right).\]
\end{assumption}

Assumption 4.1 can be easily fulfilled for various parametric estimators in econometrics. Assumption 4.2 is a technical condition related to the local Lipschitz continuity of the FZ loss on neighborhood of the true parameter value. Assumptions 4.3 and 4.4 are for the existence of the asymptotic covariance matrix of our estimator.  Assumption 4.5, which strengthens Assumption 2.8, requires that $\bar{K}(V;\hat{v},\hat{\boldsymbol{\gamma}})$ should converge uniformly at a rate faster than $n^{-1/4}$. For the power series estimator (\ref{v_hat}), this assumption holds under the conditions of Lemma 1 with the growth rate of $\kappa_{m}$ being further restricted such that $\kappa_{m}^{6}/n\rightarrow0$ and $n^{1/4}\kappa_{m}^{1-s_{m}/(r+1)}\rightarrow0$.\par 
The next theorem establishes the asymptotic normality of the proposed estimator and provides the form of its asymptotic covariance matrix.

\begin{theorem}
If Assumptions 1, 2, and 4 hold, then 
\[
\sqrt{n}\left(\left(\hat{\boldsymbol{\theta}}_{1,\tau},\hat{\boldsymbol{\theta}}_{2,\tau}\right)-\left(\boldsymbol{\theta}_{1,\tau},\boldsymbol{\theta}_{2,\tau}\right)\right)\stackrel{d}{\rightarrow}N\left(\mathbf{0},\mathbf{H}_{\tau}^{-1}\boldsymbol{\Omega}_{\tau}\mathbf{H}_{\tau}^{-1}\right),
\]
where
\begin{eqnarray*}
\mathbf{H}_{\tau} & = & \nabla_{\boldsymbol{\theta\theta}}E\left[FZ_{\tau}\left(q\left(W,\boldsymbol{\theta}_{1,\tau}\right),e\left(W,\boldsymbol{\theta}_{2,\tau}\right),Y\right)|T=c\right]\times P\left(T=c\right),\\
\boldsymbol{\Omega}_{\tau} & = & E\left[\mathbf{J}_{\tau}\mathbf{J}_{\tau}^\top\right],\\
\mathbf{J}_{\tau} & = & K\left(W,Z;\boldsymbol{\gamma}_{0}\right)\nabla_{\boldsymbol{\theta}}FZ_{\tau}\left(q\left(W,\boldsymbol{\theta}_{1,\tau}\right),e\left(W,\boldsymbol{\theta}_{2,\tau}\right),Y\right)+\mathbf{M}_{\tau}\psi\left(X\right),\\
\mathbf{M}_{\tau} & = & E\left[\nabla_{\boldsymbol{\theta}}FZ_{\tau}\left(q\left(W,\boldsymbol{\theta}_{1,\tau}\right),e\left(W,\boldsymbol{\theta}_{2,\tau}\right),Y\right)\left[\nabla_{\boldsymbol{\gamma}}K\left(W,Z;\boldsymbol{\gamma}_{0}\right)\right]^\top\right].
\end{eqnarray*}
\end{theorem}

For the asymptotic covariance matrix of the estimators, the effect of the presence of
the estimated parameter $\hat{\boldsymbol{\gamma}}$ is taken into
account. Ignoring the effect of the estimated nuisance parameters
will cause the asymptotic covariance matrix to be inconsistent, leading
to invalid confidence interval constructions \citep{NM_1994}. Without such an effect, the vector $\mathbf{J}_{\tau}$ will only have the first term. For estimating the asymptotic covariance matrix, we can use a plug-in estimator by replacing $\left(\boldsymbol{\theta}_{1,\tau},\boldsymbol{\theta}_{2,\tau},v_{0},\boldsymbol{\gamma}_{0}\right)$ in the formula of the asymptotic covariance matrix with their estimates $\left(\hat{\boldsymbol{\theta}}_{1,\tau},\hat{\boldsymbol{\theta}}_{2,\tau},\hat{v},\boldsymbol{\hat{\gamma}}\right)$. In Appendix A.1, we detail the estimation procedures and derive the asymptotic result of the plug-in estimator when $FZ_{\tau}^{sp}\left(q,e,y\right)$ is used and the structural forms of $q\left(W,\boldsymbol{\theta}_{1,\tau}\right)$ and $e\left(W,\boldsymbol{\theta}_{2,\tau}\right)$ are linear in parameters. The resulting estimated asymptotic covariance matrix is then used to construct the pointwise confidence bands of the estimated CTATE and QTE for compliers in the empirical application in Section 5.  

The theoretical results of the proposed estimator above can be used immediately to derive the theoretical properties for estimating the IQATE in (\ref{IQATE}). Again, we focus on the case when the models are linear in parameters. Divide the interval $\left(0,1\right)$ into $L+1$ subintervals with $L$ break points $\tau_{1},\ldots,\tau_{L}$, where $0<\tau_{1}<\tau_{2}\ldots\tau_{L}<1$. For $l,l^{\prime}\in\left\{ 1,\ldots,L\right\} $, and $l^{\prime}<l$,  the IQATE between quantile levels $\tau_{l^{\prime}}$ and $\tau_{l}$ is
\[
IQATE\left(\tau_{l},\tau_{l^{\prime}}\right)=\frac{\tau_{l}\alpha_{2,\tau_{l}}-\tau_{l^{\prime}}\alpha_{2,\tau_{l^{\prime}}}}{\tau_{l}-\tau_{l^{\prime}}}.
\]
If Assumptions 1, 2 and 4 hold, using Theorem 1 we can obtain the pointwise consistency: 
\begin{equation}
\frac{\tau_{l}\hat{\alpha}_{2,\tau_{l}}-\tau_{l^{\prime}}\hat{\alpha}_{2,\tau_{l^{\prime}}}}{\tau_{l}-\tau_{l^{\prime}}}\stackrel{p}{\rightarrow}\frac{\tau_{l}\alpha_{2,\tau_{l}}-\tau_{l^{\prime}}\alpha_{2,\tau_{l^{\prime}}}}{\tau_{l}-\tau_{l^{\prime}}}.
\label{IQATE_consistency}
\end{equation}
The standard deviation of the above estimator for IQATE with $\hat{\alpha}_{2,\tau_{l}}$ and $\hat{\alpha}_{2,\tau_{l^{\prime}}}$ is given by
\[\sqrt{
	\frac{1}{\left(\tau_{l}-\tau_{l^{\prime}}\right)^{2}}\left[\tau_{l}^{2}Var\left(\hat{\alpha}_{2,\tau_{l}}\right)+\tau_{l^{\prime}}^{2}Var\left(\hat{\alpha}_{2,\tau_{l^{\prime}}}\right)-2\tau_{l}\tau_{l^{\prime}}Cov\left(\hat{\alpha}_{2,\tau_{l}},\hat{\alpha}_{2,\tau_{l^{\prime}}}\right)\right]}.
\]
If $\dim\left(X\right)=p\times1$, $Var\left(\hat{\alpha}_{2,\tau_{l}}\right)$ (and $Var\left(\hat{\alpha}_{2,\tau_{l}^{\prime}}\right)$) is the $\left(p+2\right)$th
diagonal element of the matrix $\mathbf{H}_{\tau_{l}}^{-1}\boldsymbol{\Omega}_{\tau_{l}}\mathbf{H}_{\tau_{l}}^{-1}$
(and $\left(\mathbf{H}_{\tau_{l^{\prime}}}^{-1}\boldsymbol{\Omega}_{\tau_{l^{\prime}}}\mathbf{H}_{\tau_{l^{\prime}}}^{-1}\right)$), and $Cov\left(\hat{\alpha}_{2,\tau_{l}},\hat{\alpha}_{2,\tau_{l^{\prime}}}\right)$
is the $\left(p+2\right)$th diagonal element of the matrix $\mathbf{H}_{\tau_{l}}^{-1}E\left[\mathbf{J}_{\tau_{l}}\mathbf{J}_{\tau_{l^{\prime}}}^\top\right]\mathbf{H}_{\tau_{l^{\prime}}}^{-1}$,
both scaled by $n^{-1}$.

Next we provide the result of uniform consistency: $
\left(\hat{\boldsymbol{\theta}}_{1,\tau},\hat{\boldsymbol{\theta}}_{2,\tau}\right)\stackrel{p}{\rightarrow}\left(\boldsymbol{\theta}_{1,\tau},\boldsymbol{\theta}_{2,\tau}\right)$ uniformly over $\tau\in\mathcal{T}\subset\left(0,1\right)$. We focus on a special case in which the loss function $FZ_{\tau}^{sp}(q,e,y)$ is used. Let 
\[h\left(\tau,\boldsymbol{\theta}\right):=\tau\left[ FZ_{\tau}^{sp}\left(q\left(W,\boldsymbol{\theta}_{1}\right),e\left(W,\boldsymbol{\theta}_{2}\right),Y\right)-\ln(1+\exp(Y))\right],
\]where $\boldsymbol{\theta}:=\left(\boldsymbol{\theta}_{1},\boldsymbol{\theta}_{2}\right)$. Notice that minimizing $E\left[FZ_{\tau}^{sp}\left(q\left(W,\boldsymbol{\theta}_{1}\right),e\left(W,\boldsymbol{\theta}_{2}\right),Y\right)\right]$
and $E\left[h\left(\tau,\boldsymbol{\theta}\right)\right]$
w.r.t. $\boldsymbol{\theta}$ will obtain the same minimizer, as
the former is just the latter scaled by a positive and finite constant $1/\tau$ and plus $E\left[\ln(1+\exp(Y))\right]$. Therefore using $h\left(\tau,\boldsymbol{\theta}\right)$ as a loss
function to estimate $\boldsymbol{\theta}$ will yield the same result
as using $FZ_{\tau}^{sp}\left(q\left(\boldsymbol{\theta}_{1}\right),e\left(\boldsymbol{\theta}_{2}\right),y\right)$. We will use $h\left(\tau,\boldsymbol{\theta}\right)$
as the loss function in proving the uniform consistency. The proof will rely on using the following additional assumptions and results of \citet{Newey_1991} to show that the empirical loss function uniformly converges to the true loss function over $\left(\tau,\boldsymbol{\theta}\right)\in\mathcal{T}\times\boldsymbol{\Theta}$. 
\begin{assumption}
	\item[1.] $\mathcal{T}$ is a closed sub-interval of $(0,1)$.
	\item[2.] The partial derivative functions $\nabla_{\boldsymbol{\theta}}E\left[h\left(\tau,\boldsymbol{\theta}\right)|T=c\right]$
	and $\frac{\partial}{\partial\tau}E\left[h\left(\tau,\boldsymbol{\theta}\right)|T=c\right]$ are both bounded for all $\left(\tau,\boldsymbol{\theta}\right)\in\mathcal{T}\times\boldsymbol{\Theta}$.
	\item[3.] $E\left[B_{k}\left(V\right)|T=c\right]<\infty$,
	where $B_{k}\left(V\right)\geq 0$, $k=1,\ldots,5$ are the terms defined
	in Lemma A.2 in Appendix A.4.
\end{assumption}

Assumption 5.1 is standard. Assumption 5.2 is for the equicontinuity of the true loss function $E[\bar{K}(V;v_{0},\boldsymbol{\gamma}_{0})h(\tau,\boldsymbol{\theta})]$. Assumptions 5.3 requires the first moment of the quantity $B_{k,\tau}(V)$, $k=1,\ldots,5$ (defined in Lemma A.2 in the Appendix A.4) to be finite. This is a technical condition related to the global Lipschitz continuity of the FZ loss on the parameter space. The following theorem establishes uniform consistency of our estimator over the quantile index range $\mathcal{T}$.
\begin{theorem}
	Let the FZ loss in (\ref{sample_analogue}) take the form (\ref{FZ_sp}). If Assumptions 1, 2 and 5 hold, then
	\[\sup_{\tau\in\mathcal{T}}\left\lVert(\hat{\boldsymbol{\theta}}_{1,\tau},\hat{\boldsymbol{\theta}}_{2,\tau})-\left(\boldsymbol{\theta}_{1,\tau},\boldsymbol{\theta}_{2,\tau}\right)\right\rVert\stackrel{p}{\rightarrow}0.\]
\end{theorem}
\section{Simulation}
In this section we examine the finite-sample performance of
our proposed method in the simulations. The data $\left(Y_{i},D_{i},Z_{i},X_{i}\right)$, $i=1,\ldots,n$, where $X_{i}=\left(X_{1i},X_{2i}\right)$, for the simulation are generated according to the following design: 
\begin{eqnarray*}
X_{1i},X_{2i} & \overset{\text{i.i.d}}\sim & U\left(0,1\right),\\
Z_{i}|X_{i} & \overset{\text{i.i.d}}\sim & Bern\left(\Phi\left(-1+X_{1i}+X_{2i}\right)\right),\\
D_{i} & = & Z_{i}D_{1i}+\left(1-Z_{i}\right)D_{0i},\\
Y_{i} & = & D_{i}Y_{1i}+\left(1-D_{i}\right)Y_{0i}.
\end{eqnarray*}
Here $U\left(0,1\right)$ is the uniform distribution in $[0,1]$, $Bern\left(p\right)$ is 
the Bernoulli distribution with parameter $p$, $\Phi\left(.\right)$ denotes the cumulative distribution function of the standard normal random variable, and  
\begin{eqnarray*}
D_{1i} & = & 1\left\{ \vartheta_{i}>-0.67\right\} ,\\
D_{0i} & = & 1\left\{ \vartheta_{i}>0.67\right\} ,\\
Y_{1i} & = & \left(b_{0}+b_{1}+b_{2}X_{1i}+b_{3}X_{2i}\right)\varepsilon_{i},\\
Y_{0i} & = & \left(b_{1}+b_{2}X_{1i}+b_{3}X_{2i}\right)\varepsilon_{i},\\
\left(\varepsilon_{i},\vartheta_{i}\right)|X_{i} & \overset{\text{i.i.d}}\sim & MVN\left(\mathbf{0},\Sigma_{\varepsilon,\vartheta}\right),\Sigma_{\varepsilon,\vartheta}=\left(\begin{array}{cc}
1 & \rho\\
\rho & 1
\end{array}\right),
\end{eqnarray*}
where $MVN\left(\mathbf{0},\Sigma_{\varepsilon,\vartheta}\right)$ denotes the multivariate normal distribution with mean zero and covariance matrix $\Sigma_{\varepsilon,\vartheta}$. We set $T_{i}=a$, if $D_{1i}=D_{0i}=1$; $T_{i}=c$, if $D_{1i}=1,D_{0i}=0$; and $T_{i}=ne$,
if $D_{1i}=D_{0i}=0$. Under this simulation design, the proportion of compliers is approximately 50\%, the proportions of always and never takers are both approximately 25\%, and there is no defier. The correlation coefficient $\rho$ controls for the degree of endogeneity. When $\rho\neq0$, for $T=a,ne$ (always and never takers), the treatment status $D$ is correlated with potential
outcomes $Y_{1}$ and $Y_{0}$ through $\varepsilon$ and endogeneity
arises. But for $T=c$ (compliers), the condition of unconfounded
treatment selection: $D\perp\left(Y_{1},Y_{0}\right)|\left(X_{1},X_{2}\right)$ holds here. Under this setting, 
\begin{eqnarray*}
Q_{Y|D,X,T=c}\left(\tau\right) & = & \left(b_{0}D+b_{1}+b_{2}X_{1}+b_{3}X_{2}\right)Q_{\varepsilon|X,T=c}\left(\tau\right),\\
 & = & \alpha_{1,\tau}D+\beta_{11,\tau}+\beta_{12,\tau}X_{1}+\beta_{13,\tau}X_{2},\\
CTE_{Y|D,X,T=c}\left(\tau\right) & = & \left(b_{0}D+b_{1}+b_{2}X_{1}+b_{3}X_{2}\right)CTE_{\varepsilon|X,T=c}\left(\tau\right),\\
 & = & \alpha_{2,\tau}D+\beta_{21,\tau}+\beta_{22,\tau}X_{1}+\beta_{23,\tau}X_{2}.
\end{eqnarray*}
The QTE and CTATE for compliers are $\alpha_{1,\tau}:=b_{0}Q_{\varepsilon|X,T=c}\left(\tau\right)$
and $\alpha_{2,\tau}:=b_{0}CTE_{\varepsilon|X,T=c}\left(\tau\right)$ respectively. We set the parameters $b_{0}=1$, $b_{1}=0$ and $b_{2}=b_{3}=1$. We consider two different sample sizes $n\in\{500,3000\}$, and the simulation is iterated 1000
times. We perform simulations under the cases of $\rho=0$ (no endogeneity) and $\rho=0.5$
(with endogeneity).  

We consider two alternative constructions of the projected weights in the weighted FZ loss minimization approach and assess how they affect the finite-sample performance of our proposed method. The first weight estimator is given by

\begin{equation}
\bar{K}\left(Y,D,X;\hat{v},\hat{\boldsymbol{\gamma}}\right)=1-\frac{D\left(1-\hat{v}\left(Y,D,X\right)\right)}{1-\pi\left(X,\hat{\boldsymbol{\gamma}}\right)}-\frac{\left(1-D\right)\hat{v}\left(Y,D,X\right)}{\pi\left(X,\hat{\boldsymbol{\gamma}}\right)},\label{estimation_K}
\end{equation}
where $\hat{v}\left(Y,D,X\right)$ is an estimate for $E\left[Z|Y,D,X\right]$ from a polynomial regression, and $\pi\left(X,\hat{\boldsymbol{\gamma}}\right)$ is an estimate for $P(Z=1|X)$ from a probit model with an intercept term and covariates $\left(X_{1},X_{2}\right)$. The estimate $\hat{v}\left(Y,D,X\right)$ is computed separately for the $D=1$ and $D=0$ groups, and the covariates used in the regressions include $(Y,X)$, their higher order and interaction terms. For the second weight estimator, we first calculate $K\left(D,Z,X\right)$ in (\ref{AAI_weight}) for each observation using $\pi\left(X,\hat{\boldsymbol{\gamma}}\right)$ obtained from the same probit model used in the first estimator. We then fit a polynomial regression of the calculated $K\left(D,Z,X\right)$ on
$\left(Y,D,X\right)$, their higher order and interaction terms and two additional covariates $\left(D_{0},D_{1}\right)$ which entail crucial information for classifying the types of individuals. The fitted value from the polynomial regression is used as the second projected weight estimator. By the law of iterated expectations, the identity (\ref{key_equation}) still holds with the weight function $\bar{K}(Y,D,X)$ being replaced by the weight $\bar{K}_{2}(Y,D,X,D_{1},D_{0}):=E[K(D,Z,X)|Y,D,X,D_{1},D_{0}]$, which motivates our construction of the second type of projected weight. Adopting the proof of Lemma 3.2 in \citet{AAI_2002}, we can see that $\bar{K}_{2}(Y,D,X,D_{1},D_{0})=P(T=c|Y,D,X,D_{1},D_{0})$. Thus working with the weight $\bar{K}_{2}$ amounts to estimating with information on classifying the types of individuals. We note that this weight estimator is infeasible in practice because we do not observe both $D_{0}$ and $D_{1}$ for each individual in the data. We compare the results using (\ref{estimation_K}) with those of $\bar{K}_{2}$ as the latter exploits more information and is expected to improve the performances of the resulting weighted FZ loss minimization estimators. Finally, since the weights are bounded from zero to one, if the estimated weights are greater than one or less than zero, we will shrink their values to one or zero.

Let M1 and M2 denote the weighted FZ loss minimization estimation approaches implemented using the aforementioned first and second projected weight estimators respectively. In addition, we also consider the following two benchmarks where no weighting scheme is used: (1) Estimation using all data without imposing any weight on the samples (no adjustment for endogeneity, denoted
by M3); (2) Estimation using only the sample of compliers (oracle
estimation, denoted by M4).

Figures \ref{figure4} to \ref{figure7} summarize the simulation results, including the biases, variances and mean squared
errors (MSEs) of the CTATE and QTE estimators under settings M1 to M4. The key findings are discussed as follows. For $\rho=0$, there is no concern of endogeneity. In this case, the CTATE and QTE for compliers estimated using all the data without imposing any weight scheme (M3), have the lowest variance and MSE. Imposing weights to account for endogeneity would increase estimator variability, as can be seen from our theoretical result in Theorem 2, which shows the estimated weight
contributes to variances of the proposed estimators. The estimation results using only the complier samples (M4) are associated with higher variance than those of M3. This is mainly due to the sample size effect as M4 tends to use fewer observations than M3 does in the implementation. The plots of variances for the estimated CTATE for compliers reveal
a downward trend as the quantile level rises, but those for the estimated QTE for compliers indicate that the variances are larger at the lower and higher quantile levels. For the estimation bias, the simulation results are mixed and there is no clear dominant method here. The bias, variance and MSE are all improved as the sample size increases from 500 to 3000.

For $\rho=0.5$, the endogeneity issue arises. In this case, M3 results in a much higher bias and MSE than the other three methods, although it still results in a lower variance. The high MSE of M3 is mostly due to the high bias from a lack of adjustment for endogeneity. The reasons for the lower variance of M3 are the same as those in the case of $\rho=0$: all samples are used in the estimations and there is no estimated weight. Notably, M4 (the oracle estimation case) is associated with the lowest bias, because it uses only the sample of compliers. The bias curve under M2 is flat over the quantile levels, whereas that under M1 shows declining trend. It is worth noting that M1 results in a higher MSE than M2, which suggests that using potential treatment status information 
to estimate the weight may be helpful on improving the MSE. The
performances of M1, M2 and M4 evidently improve as the sample size increased. In summary, in the case of endogeneity, although estimation without adjustment for endogeneity (M3) could still yield a lower variance, it could also result in a higher bias and thus a higher MSE. The higher variances under M1 and M2 arise from the use of the estimated weights. However, these weight schemes help to effectively mitigate the estimation bias. Overall, relative to M3, the MSE under M1 and M2 can be substantially improved under the weighting adjustment for endogeneity.
\section{Empirical Application: Effects of Enrolling in JTPA Services on Earnings}
In the following we illustrate the usefulness of our proposed method through an empirical study. We estimate the CTATE for compliers using our proposed weighted FZ loss minimization approach to evaluate the effects of enrolling in Title II programs of the Job Training Partnership Act (JTPA) in the US. We use the data of adult men and women who participated in these programs between November 1987 and September 1989. These data were previously used by \citet{AAI_2002}, who estimated the QTE for compliers on the earnings of the job training programs. We assume that the observations are i.i.d., as in \citet{AAI_2002}, for estimation. The outcome variable $Y$ is the sum of earnings in the 30 months after the random assignment. In practice, the sum of 30-month earnings is generally viewed as a continuous variable, and we assume that it remains continuously distributed given $W=(D,X)$ and conditional on the group of compliers. %Notice that continuity of the outcome variable is also required for using the quantile regression.
The treatment variable $D$ is a binary variable for enrollment in the JTPA services (1) or not (0), and the instrumental variable $Z$ is a binary variable for being offered such services (1) or not (0). The exogenous covariates include age, which is a categorical variable, as well as a set of dummy variables: black, Hispanic, high-school graduates (including GED holders), marital status, AFDC receipt (for adult women), whether the applicant worked for at least 12 weeks in the 12 months preceding the random assignment, the original recommended service strategy (classroom, OJT/JSA, other), and whether earnings data were from the second follow-up survey. These exogenous covariates are all discrete variables, and as shown in Section 3, this is allowed in our method. The total sample size is 11,204 (5,102 for adult men and 6,102 for adult women).

%a categorical variable: age\footnote{The intervals for categorizing age are: 22-25 years old, 26-29 years old, 30-35 years old, 36-44 years old and 45-54 years old.}.  

Offers of the JTPA services were randomly assigned to applicants but only approximately 60\% of those who were offered the services enrolled in the programs \citep{AAI_2002}. This may induce the problem of endogeneity in that the treatment status may be self-selected and correlated with the potential outcomes. As the offers were randomly assigned and were considered to potentially affect an applicant's intention to participate in the program, we use offer assignment as an instrumental variable. Finally, in the data, there were still individuals who received the JTPA services but did not obtain the assignment. %(with data $D_{i}=1$ and $Z_{i}=0$), and this violates the assumption of no defier.
However, as pointed out by \cite{AAI_2002}, the proportion of such violations relative to the entire samples is very low (less than 2\%); %(adult men: 1.12\% and adult women: 1.73\%),
therefore this has a negligible impact on our estimation.

We estimate (\ref{linear_q_Y}) and (\ref{linear_e_Y}) over a grid of quantile levels $\tau$ ranging from 0.1 to 0.9 with the grid size being 0.01 (81 grid points). We estimate $E[Z|Y,D,X]$ with a power series estimator separately for $D=1$ and $D=0$, and include the outcome variable $Y$ and its higher order terms as covariates. The offer assignment probability, $P(Z=1|X)$, is estimated using a probit model. These estimates are  then used to construct the weight $\tilde{K}(.)$ in (\ref{AAI_star3}) to account for endogeneity in the estimation problem. In Figure \ref{figure1}, we present the CTATE and QTE estimates after endogenous adjustment for compliers (the parameter estimates $\alpha_{1,\tau}$ of (\ref{linear_q_Y}) and $\alpha_{2,\tau}$ of (\ref{linear_e_Y})) and the corresponding 95\% pointwise confidence band (pcb), evaluated with the analytic standard errors of the estimators (see Appendix A.1), over the specified range of quantile levels. We also show the 95\% bootstrap pointwise and simultaneous confidence bands (scb's) with 300 bootstrap samples. As suggested by \cite{HL_2021}, we construct the bootstrap pcb using the percentile method \citep[][p.327]{Van_1998}, which recenters the bootstrapped parameter estimator yet does not rescale it by its standard deviation. We construct two scb’s. The first one, denoted by scb-bootstrap-ns, is also constructed using the percentile method, and the second one, denoted by scb-bootstrap-qs, standardizes the bootstrapped estimator using the rescaled bootstrap quantile spread \citep{CFM_2013} to estimate the asymptotic standard error of the parameter estimator. The procedures for constructing the scb's can be found in Appendix A.5. The results without the endogenous adjustment are shown in Appendix A.6.

%\textcolor{blue}{As suggested by \cite{HL_2021}, to avoid possibly inconsistent bootstrap standard error, we construct the bootstrap pcb using the percentile method \citep[][p.327]{Van_1998}, which recenters the bootstrapped parameter estimator yet does not rescale it using the bootstrap standard error. We construct two scb's. The first one, denoted by scb-bootstrap-ns, is also constructed with the percentile method. The second one, denoted by scb-bootstrap-qs, standardizes the bootstrapped estimator with the rescaled bootstrap quantile spread \citep{CFM_2013} as the estimate of the standard error of the parameter estimator.} 

For both adult men and women, the CTATE estimates consistently increase across quantile levels, whereas the QTE estimates exhibit some fluctuations. Upon comparing the pcb's calculated using the analytic standard errors and the bootstrap method, the differences are small. For the case of adult men, the scb's demonstrate that both the CTATE and QTE estimates generally lack statistical significance across most quantile levels. This implies a negligible stochastic dominance relationship, indicating that the benefits of participating in the JTPA program for adult men are not conclusively supported by our data. For the case of adult women, the strength of the CTATE estimates surpasses that of men, with scb's showing statistically significant positivity at certain quantile levels. Overall, the results in Figure \ref{figure1} suggest a stronger evidence for adult women that earnings from not participating in the JTPA are second order stochastically dominated by those from participating in the JTPA, and the JTPA was beneficial for risk averse female workers who complied with the assignment of the JTPA offer.

We then divide the quantile levels into eight equal-length intervals and estimate the corresponding inter-quantile average treatment effect (IQATE) for compliers. We also compare the estimated IQATE with a naive estimator: a local average of the QTE estimates for compliers (LAVG-QTE) within the same interval of quantile levels. Figure \ref{figure2} shows the estimation results of the IQATE, the corresponding pcb's (implemented with the analytic standard errors and bootstrap), and the LAVG-QTE for earnings of adult men and women at different quantile level intervals.

For adult men, only the IQATE estimates above the quantile level of 0.6 are statistically significantly positive. It is worth noting that at the intervals of quantile levels of 0.8 and 0.9, the QTE estimates are not all statistically significantly positive under the pcb's, but the IQATE estimate is. This indicates that the IQATE estimate may provide a more coherent result when we want to evaluate a policy at quantile level intervals. For adult women, the IQATE estimates are statistically significantly positive over the eight quantile level intervals. Comparing the IQATE and LAVG-QTE estimates, for both cases, they are not very different when the quantile level is above 0.5, but at the quantile levels lower than 0.5, they rather show some mild differences in the case of adult women.
\section{Conclusion}
We have introduced the conditional tail average treatment effect (CTATE), defined as the difference between CTEs of potential outcomes. The CTATE is a valuable tool for policy evaluations, as it allows for capturing the heterogeneity of treatment effects over different quantiles and is useful for detecting second order stochastic dominance and for estimating the Lorenz curve. We have developed a semiparametric method using a class of consistent loss functions proposed by \citet{FZ_2016} to estimate the CTATE for compliers under endogeneity. We also have derived asymptotic properties for our proposed estimator. Our simulation results show that the proposed method works well in the presence of treatment endogeneity. We apply our estimation approach to a policy evaluation for the JTPA program participation. We find that, after adjustment for endogeneity, for the case of adult men, both the FOSD and SOSD relationships hardly held between earning distributions of those who participated in the JTPA and of those who did not. Yet, for adult women, the SOSD of earnings for the JTPA participant over those for non-participant appeared to hold. These empirical results suggest that the JTPA could be beneficial for the risk-averse female workers who complied with the assignment of the JTPA offer.

\clearpage
\begin{figure}[!htb]
	\centering
	\mbox{
		\includegraphics[height=4.8cm,width=4.8cm]{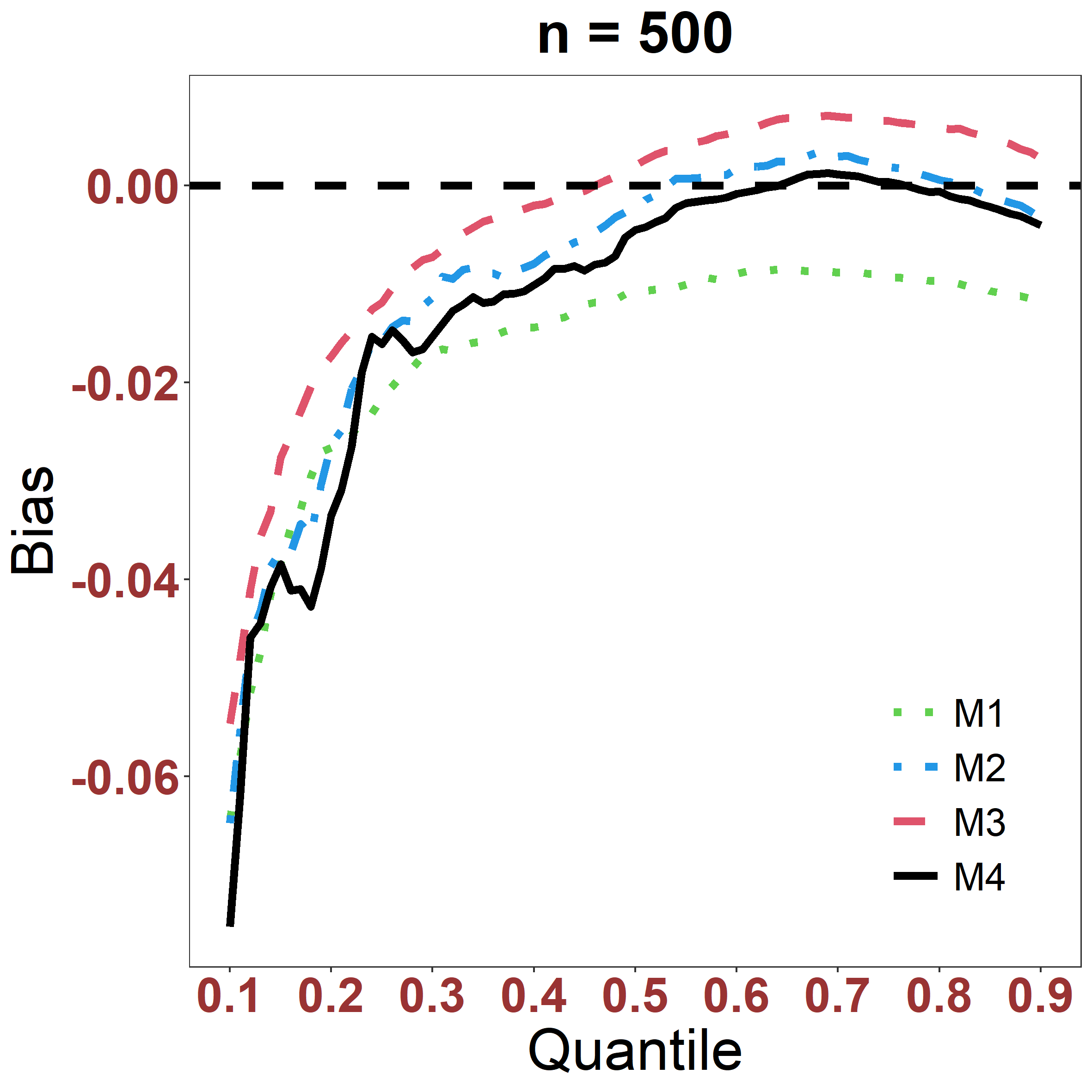}
		\includegraphics[height=4.8cm,width=4.8cm]{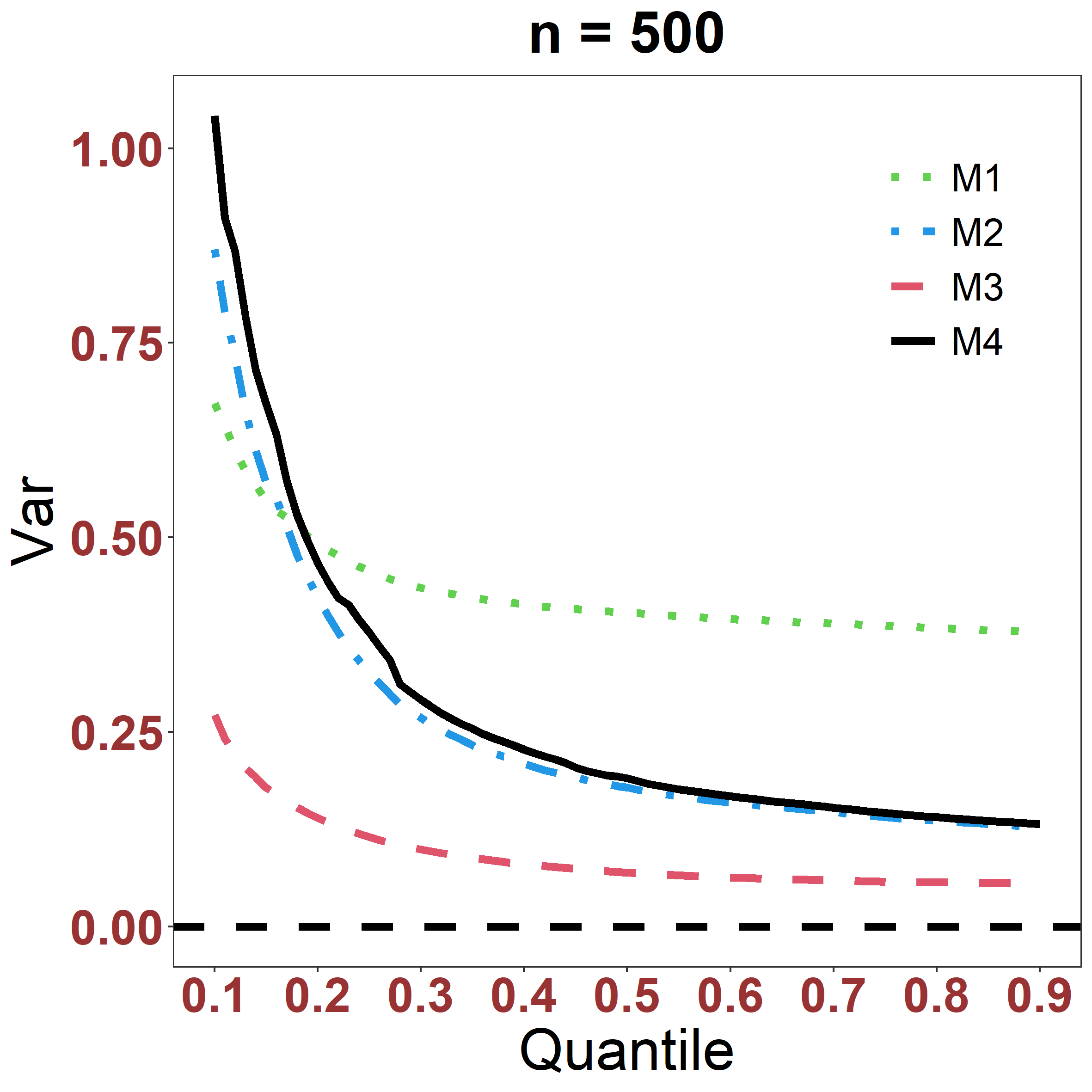}
		\includegraphics[height=4.8cm,width=4.8cm]{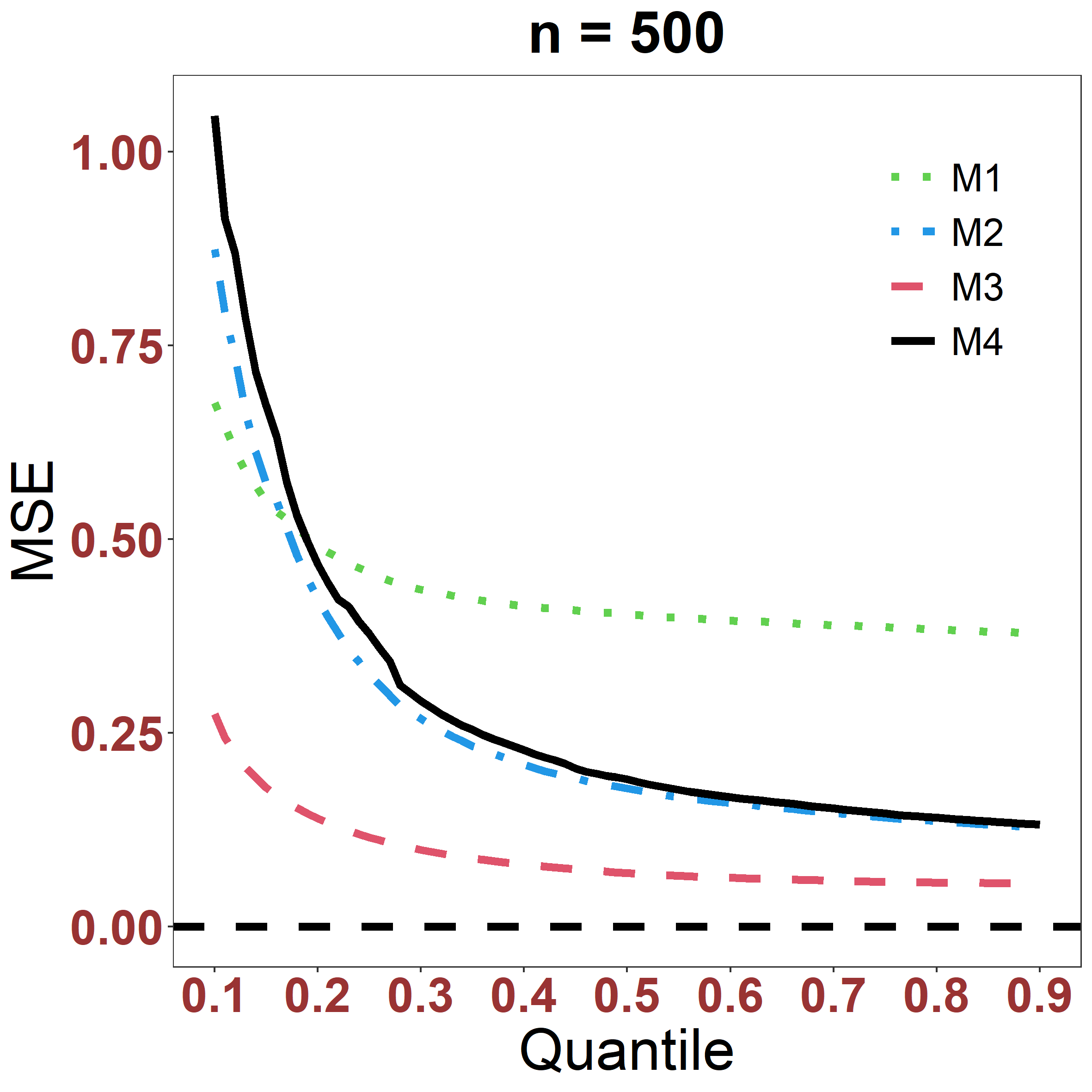}
	}
	\mbox{
	\includegraphics[height=4.8cm,width=4.8cm]{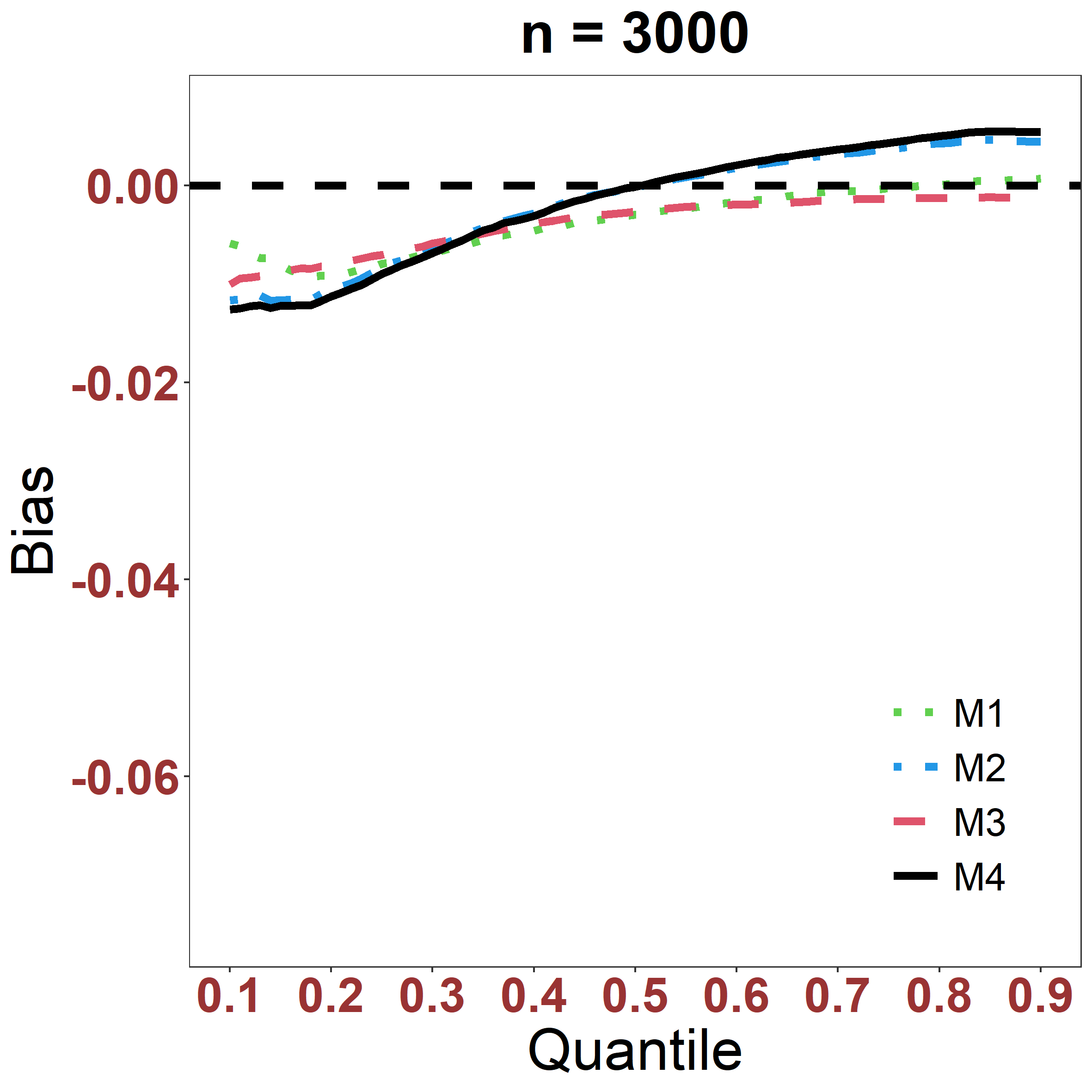}
	\includegraphics[height=4.8cm,width=4.8cm]{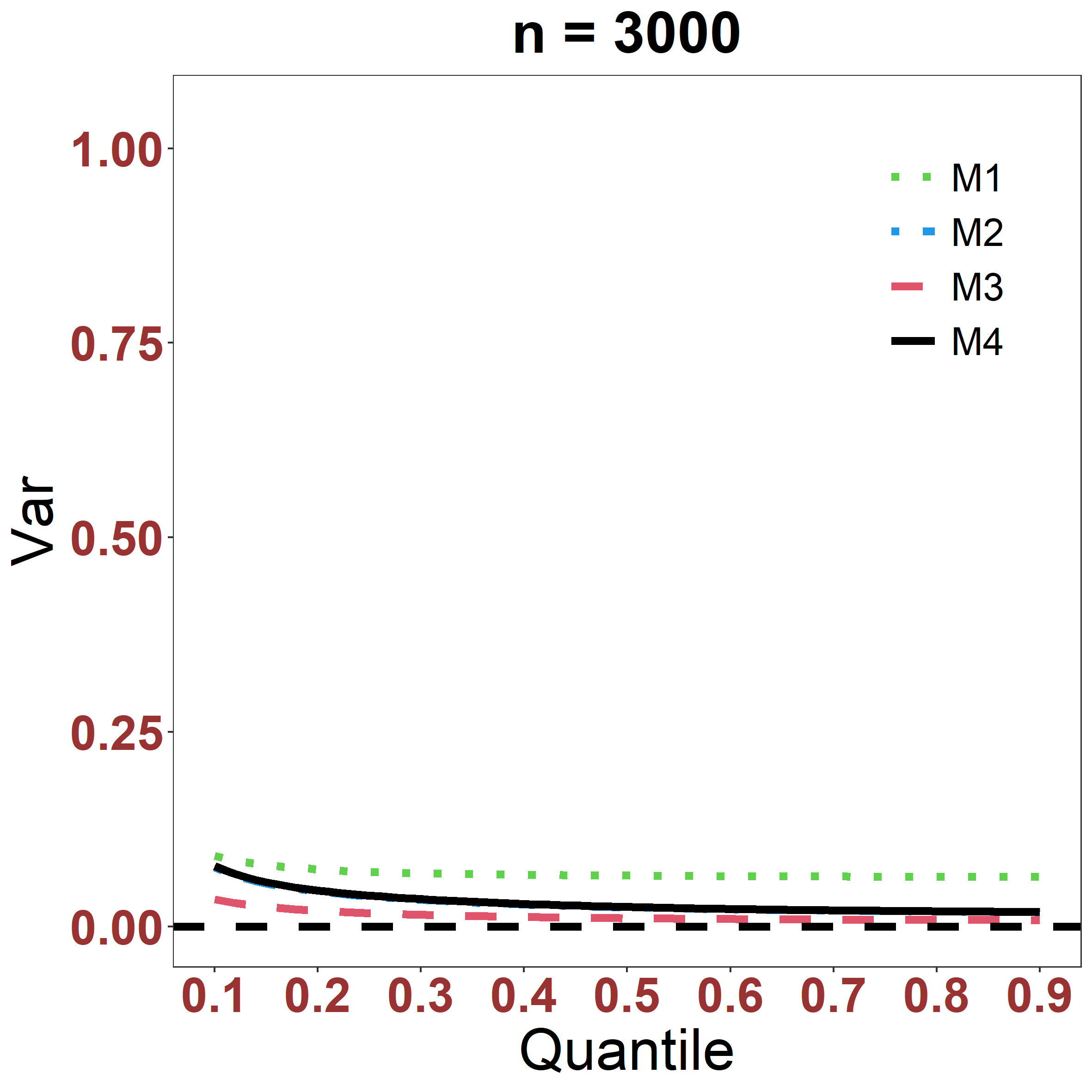}
	\includegraphics[height=4.8cm,width=4.8cm]{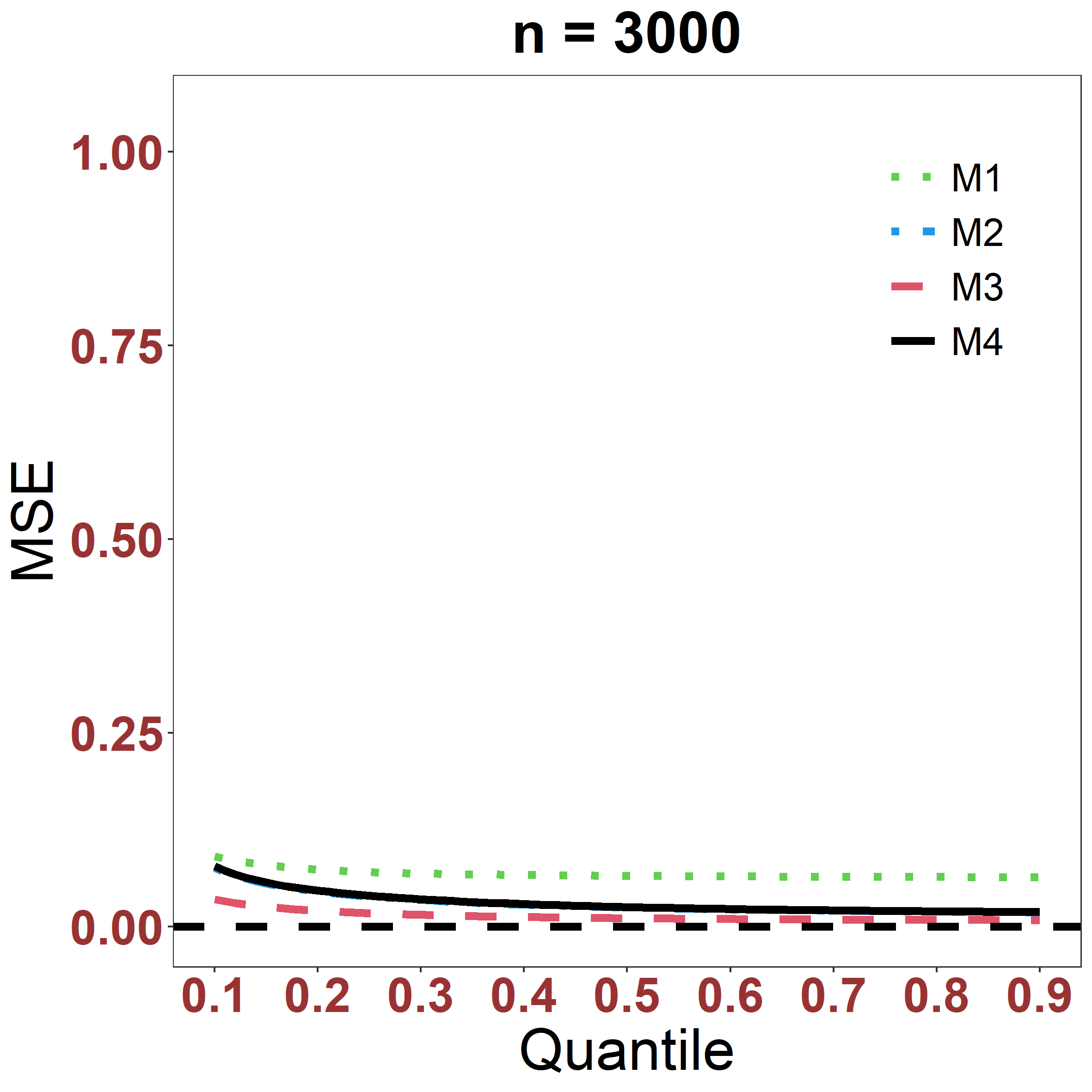}
}	
\caption{Bias, variance and MSE of estimator for estimating the CTATE for compliers when $\rho=0$ and $n=500$.}
\label{figure4}
\end{figure}
\begin{figure}[!htb]
	\centering
	\mbox{
	\includegraphics[height=4.8cm,width=4.8cm]{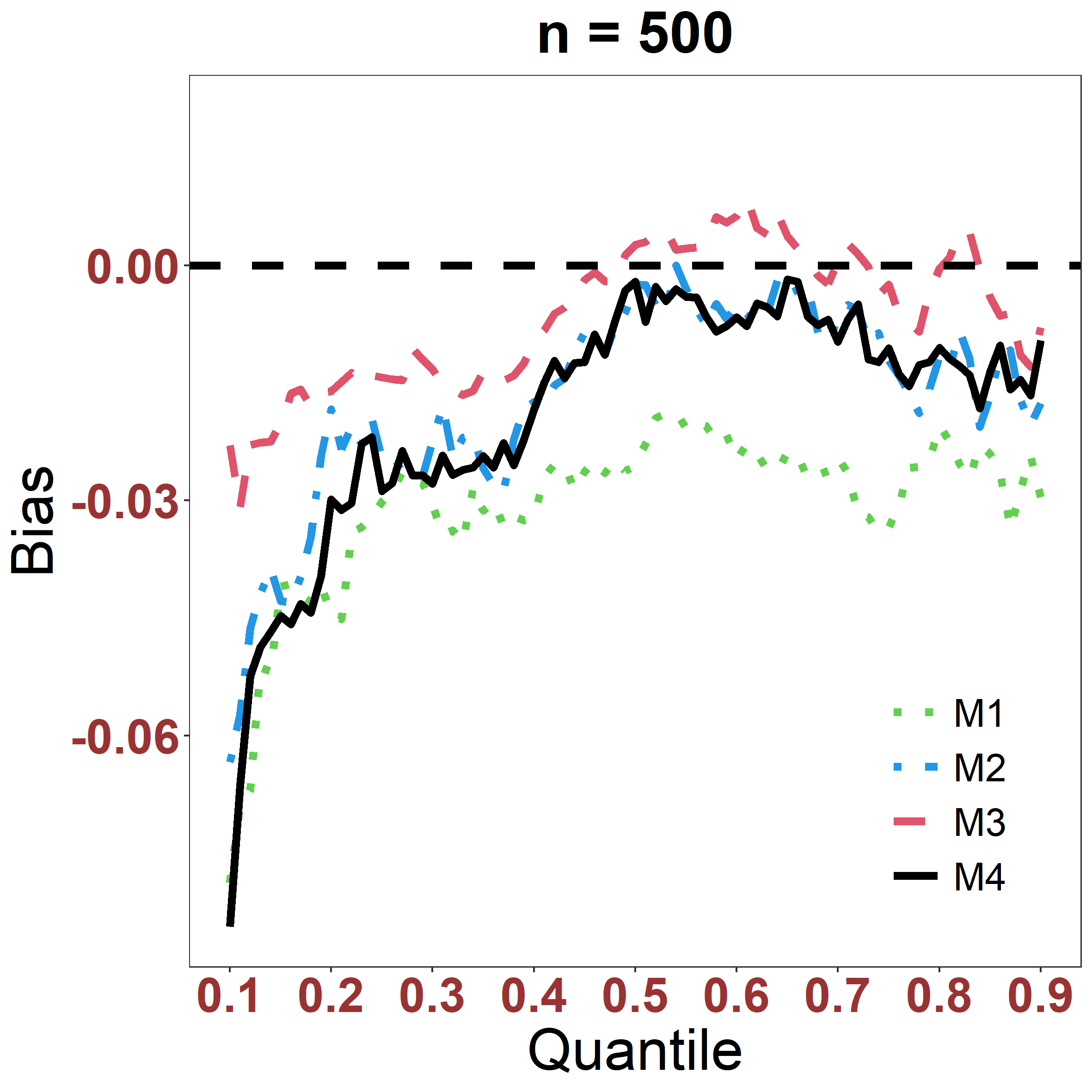}
	\includegraphics[height=4.8cm,width=4.8cm]{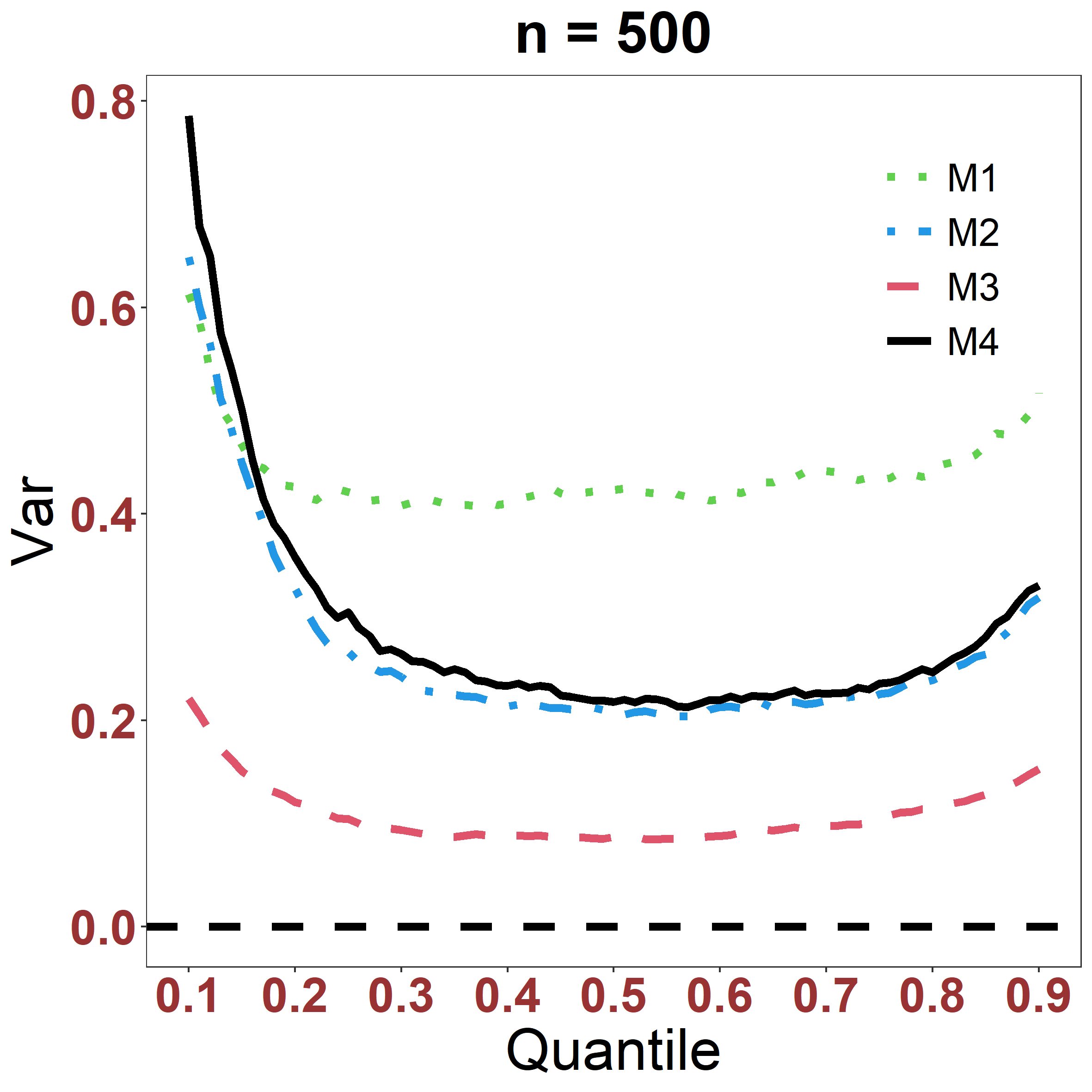}
	\includegraphics[height=4.8cm,width=4.8cm]{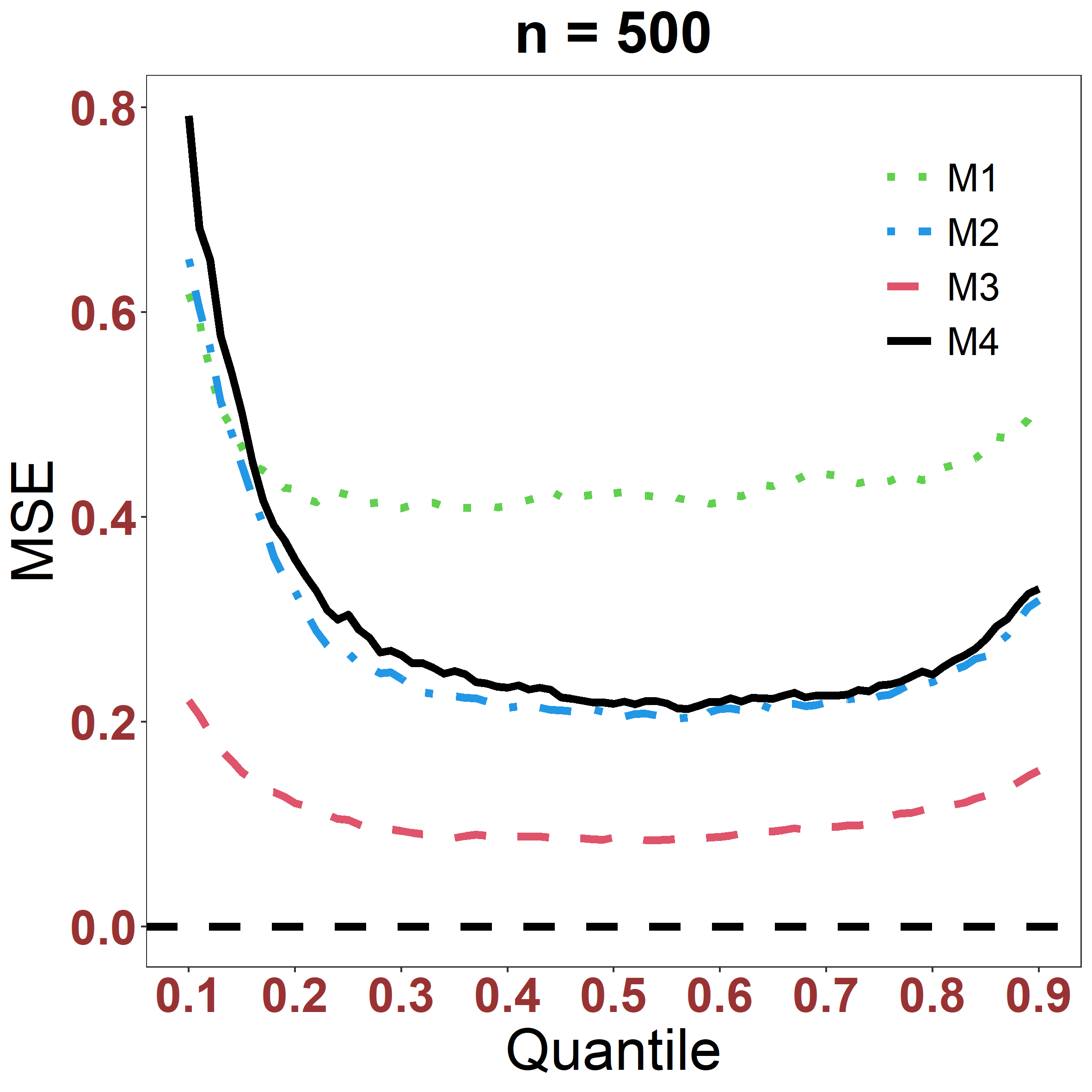}
}
	\mbox{
		\includegraphics[height=4.8cm,width=4.8cm]{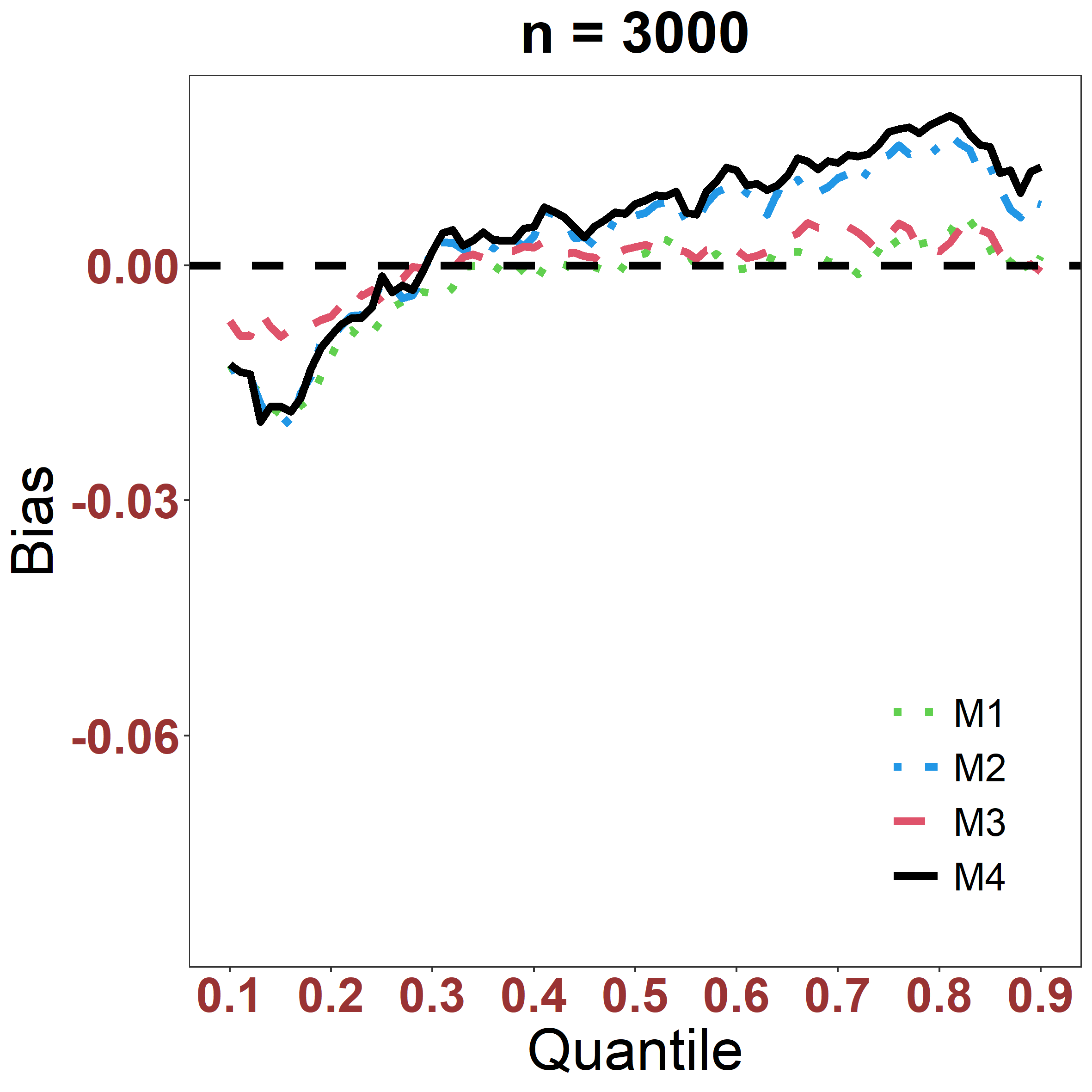}
		\includegraphics[height=4.8cm,width=4.8cm]{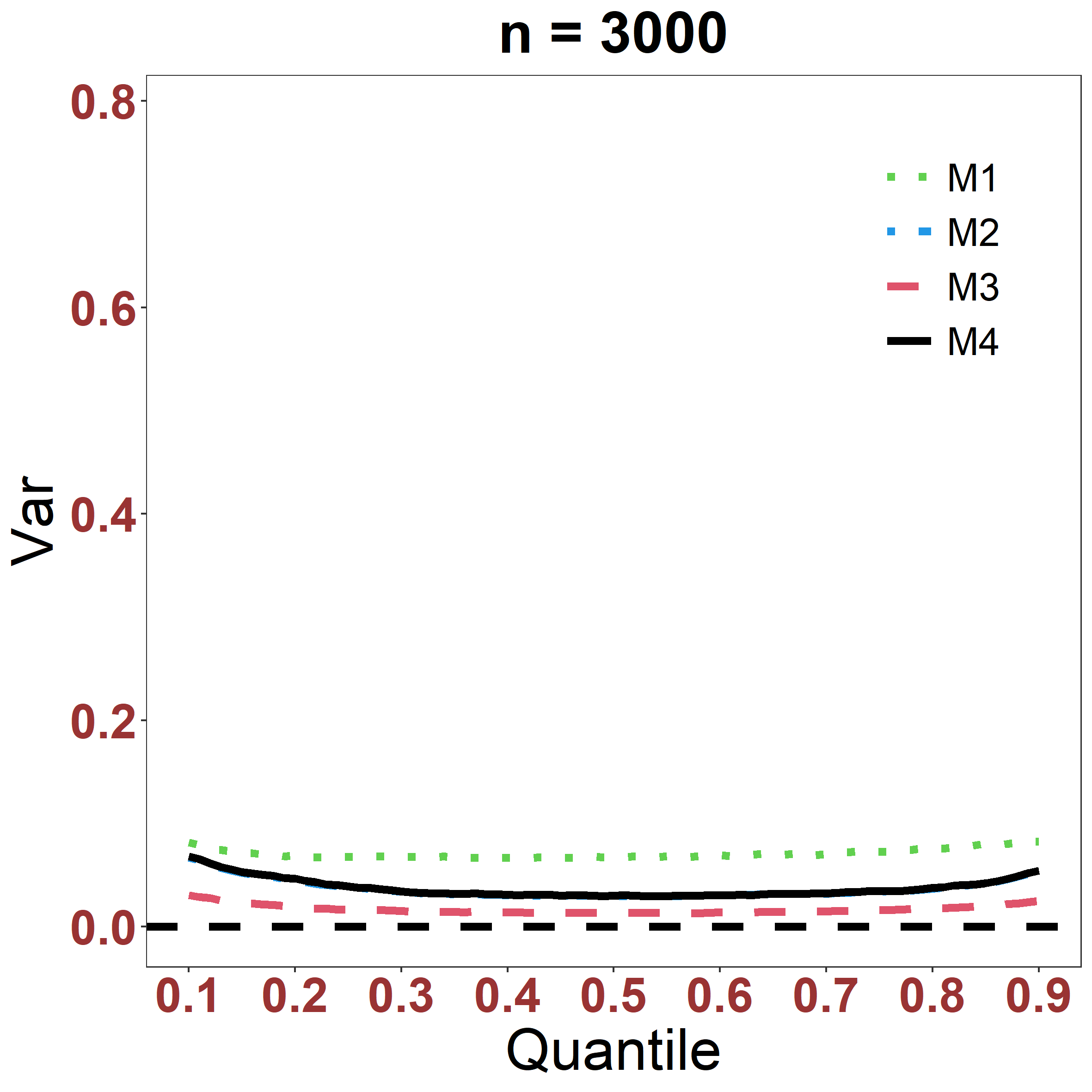}
		\includegraphics[height=4.8cm,width=4.8cm]{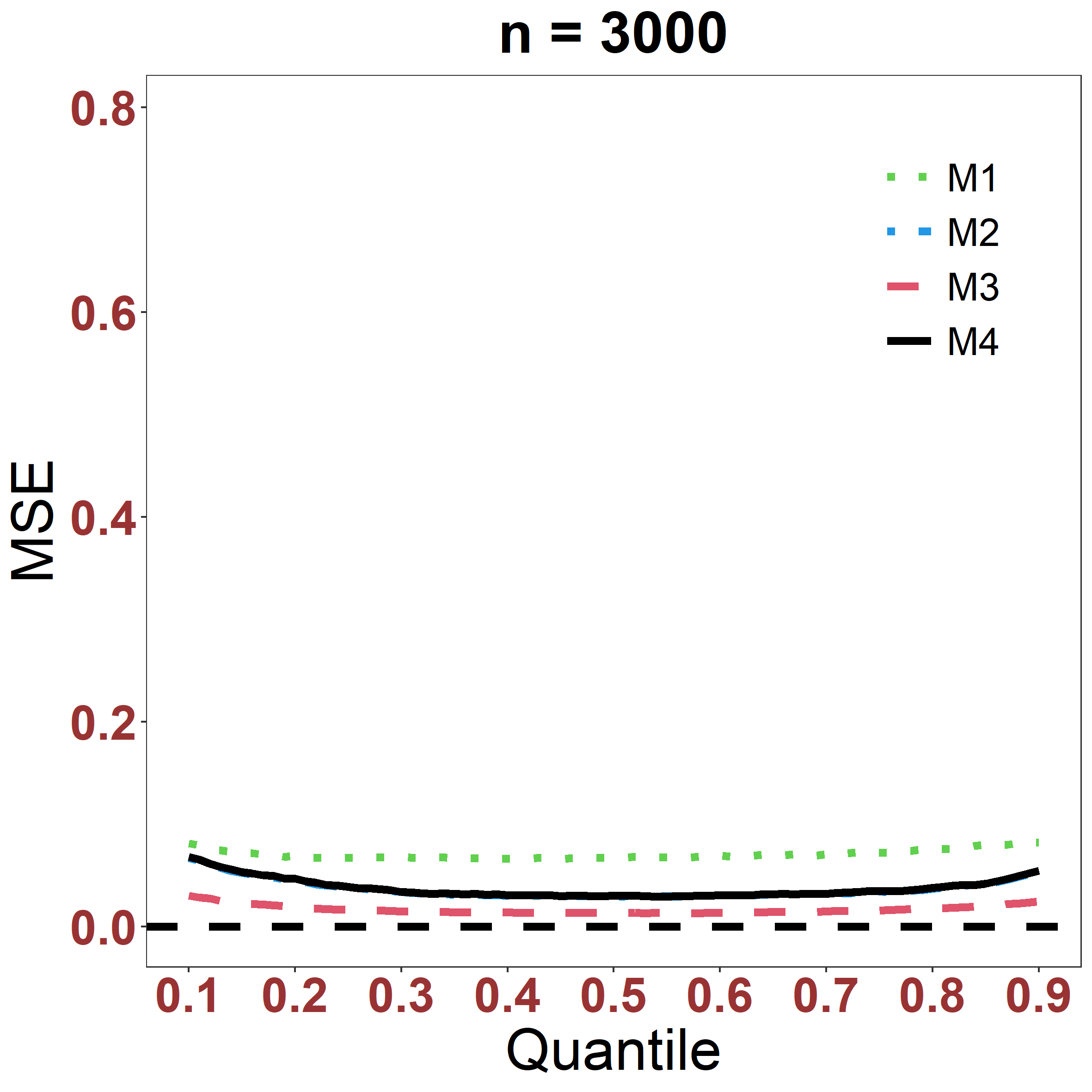}
	}	
	\caption{Bias, variance and MSE of estimator for estimating the QTE for compliers when $\rho=0$ and $n=3,000$.}
	\label{figure5}
\end{figure}
\begin{figure}[!htb]
	\centering
	\mbox{
		\includegraphics[height=4.8cm,width=4.8cm]{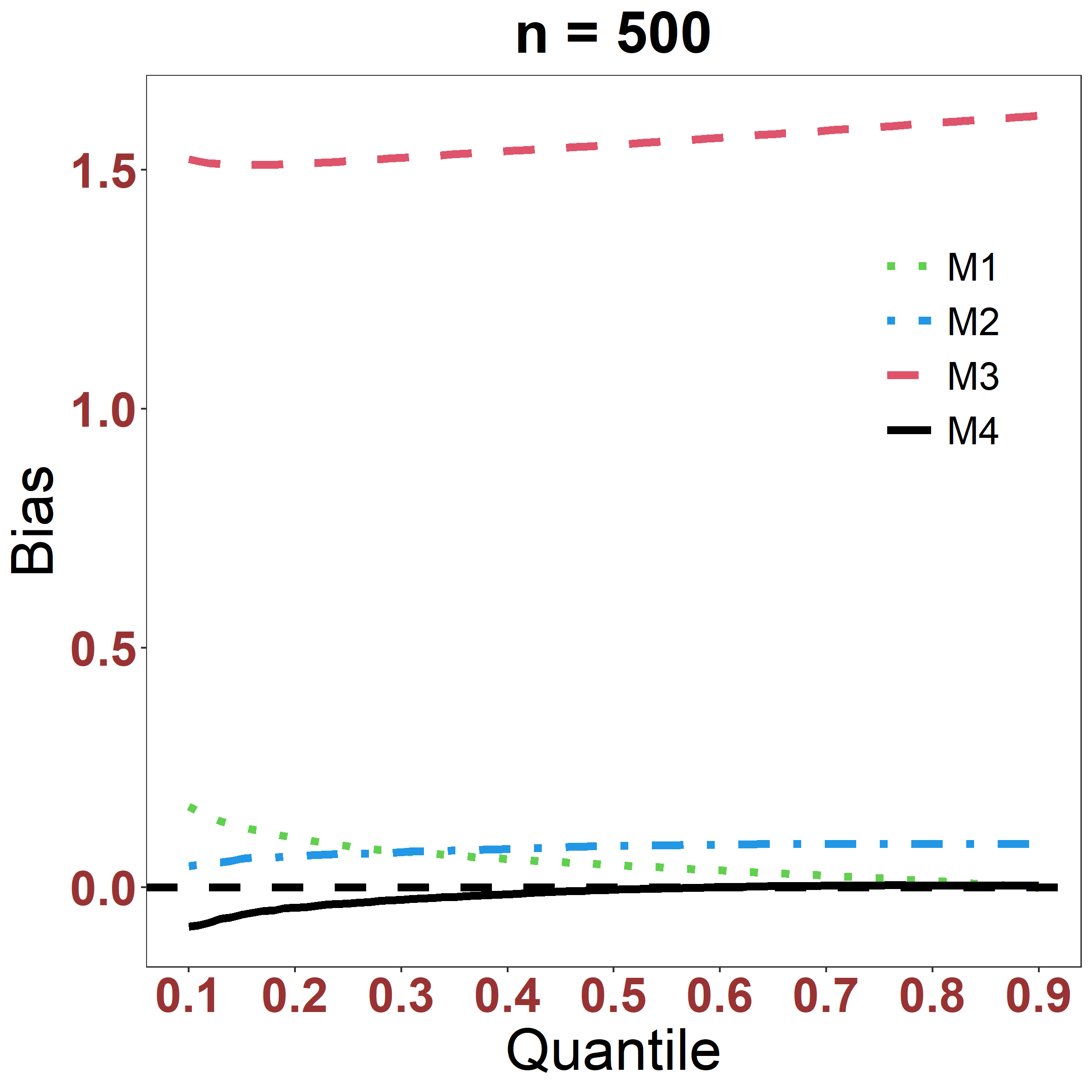}
		\includegraphics[height=4.8cm,width=4.8cm]{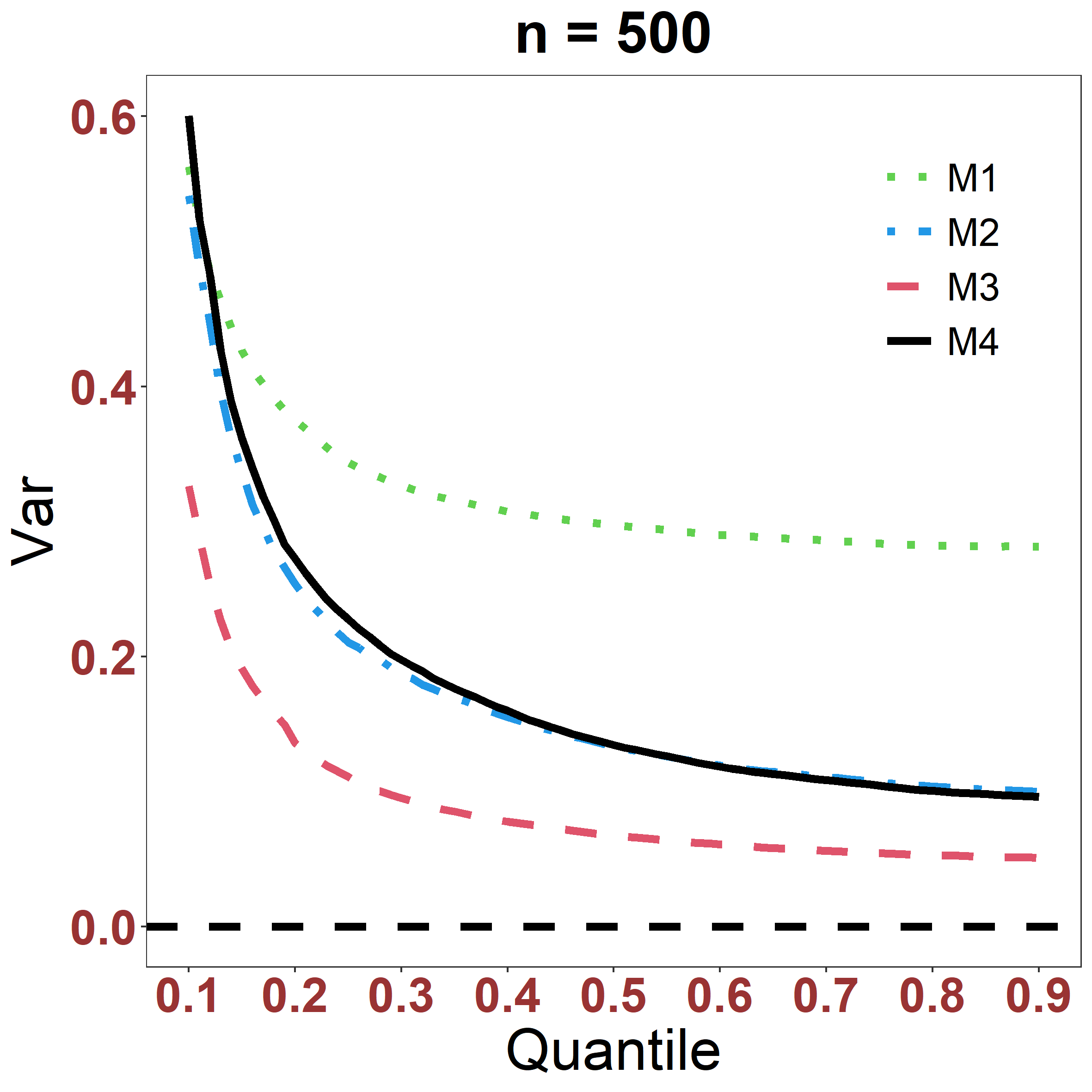}
		\includegraphics[height=4.8cm,width=4.8cm]{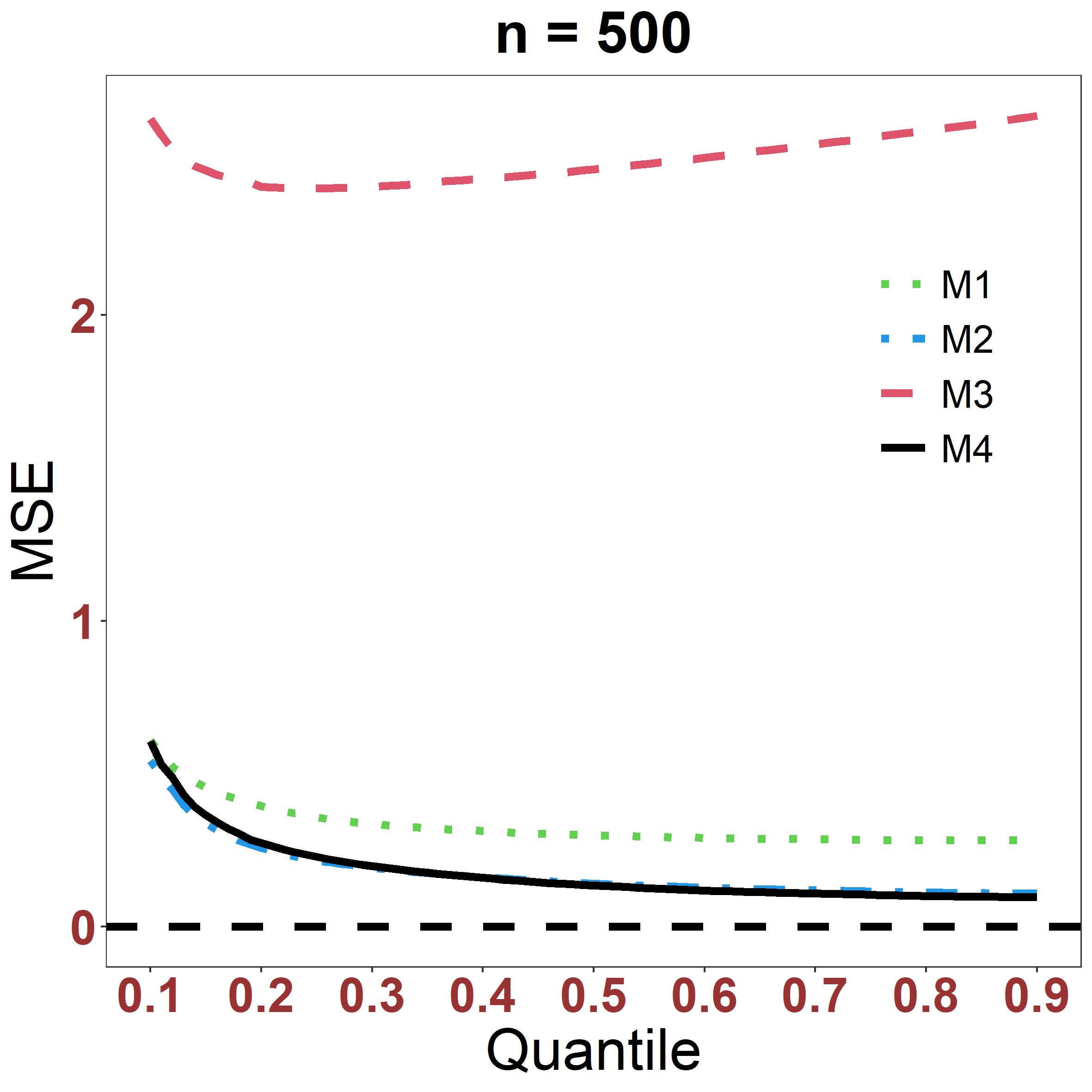}
	}
	\mbox{
		\includegraphics[height=4.8cm,width=4.8cm]{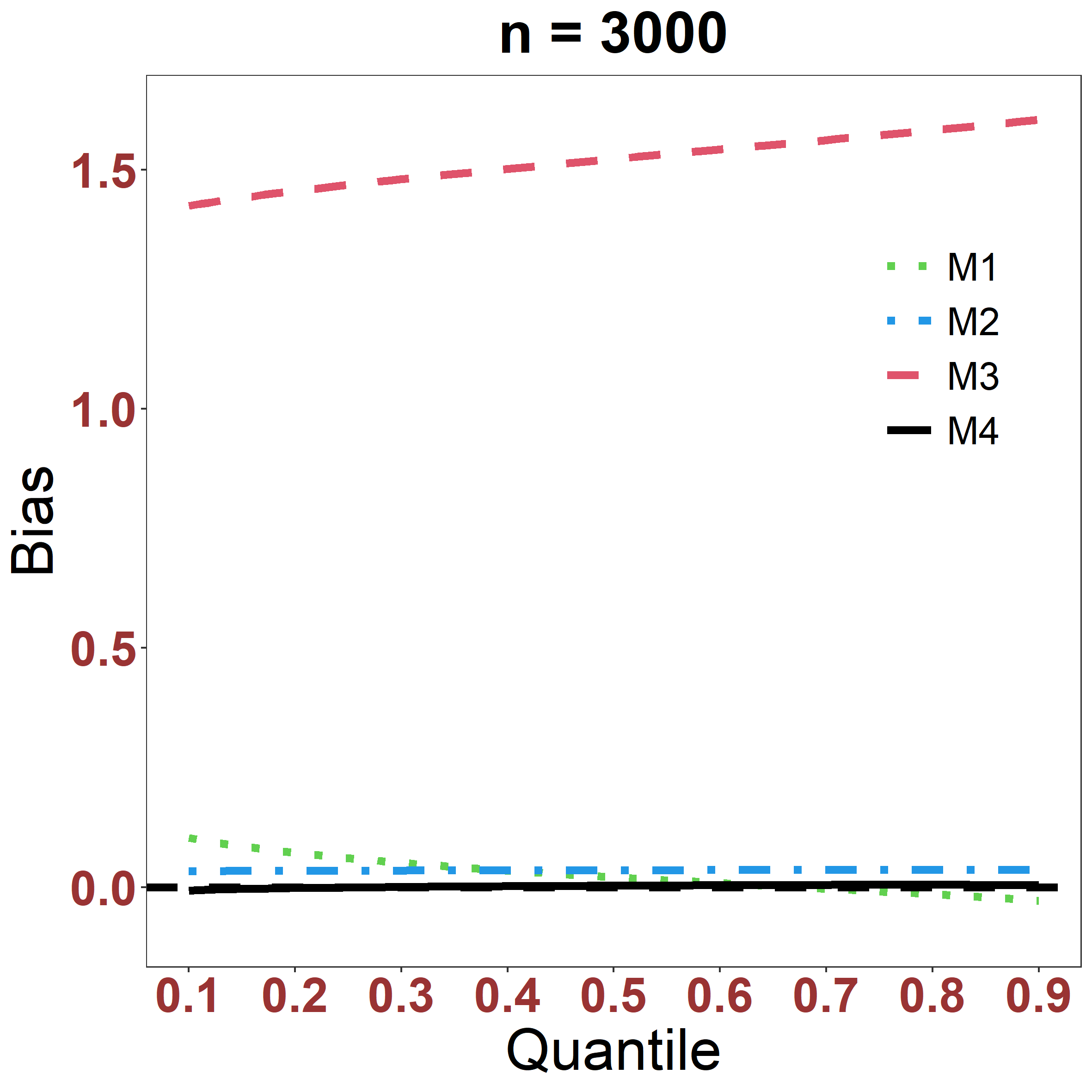}
		\includegraphics[height=4.8cm,width=4.8cm]{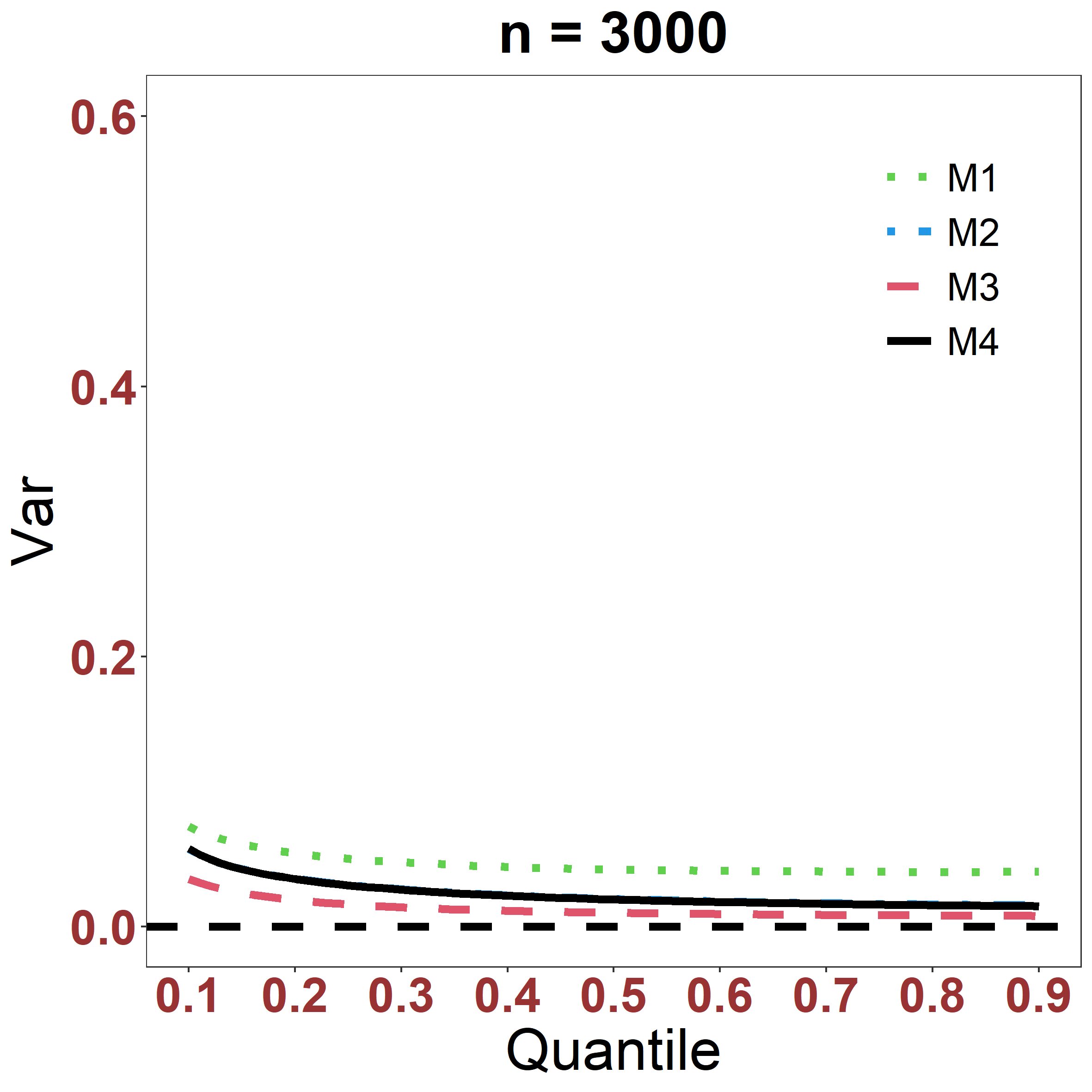}
		\includegraphics[height=4.8cm,width=4.8cm]{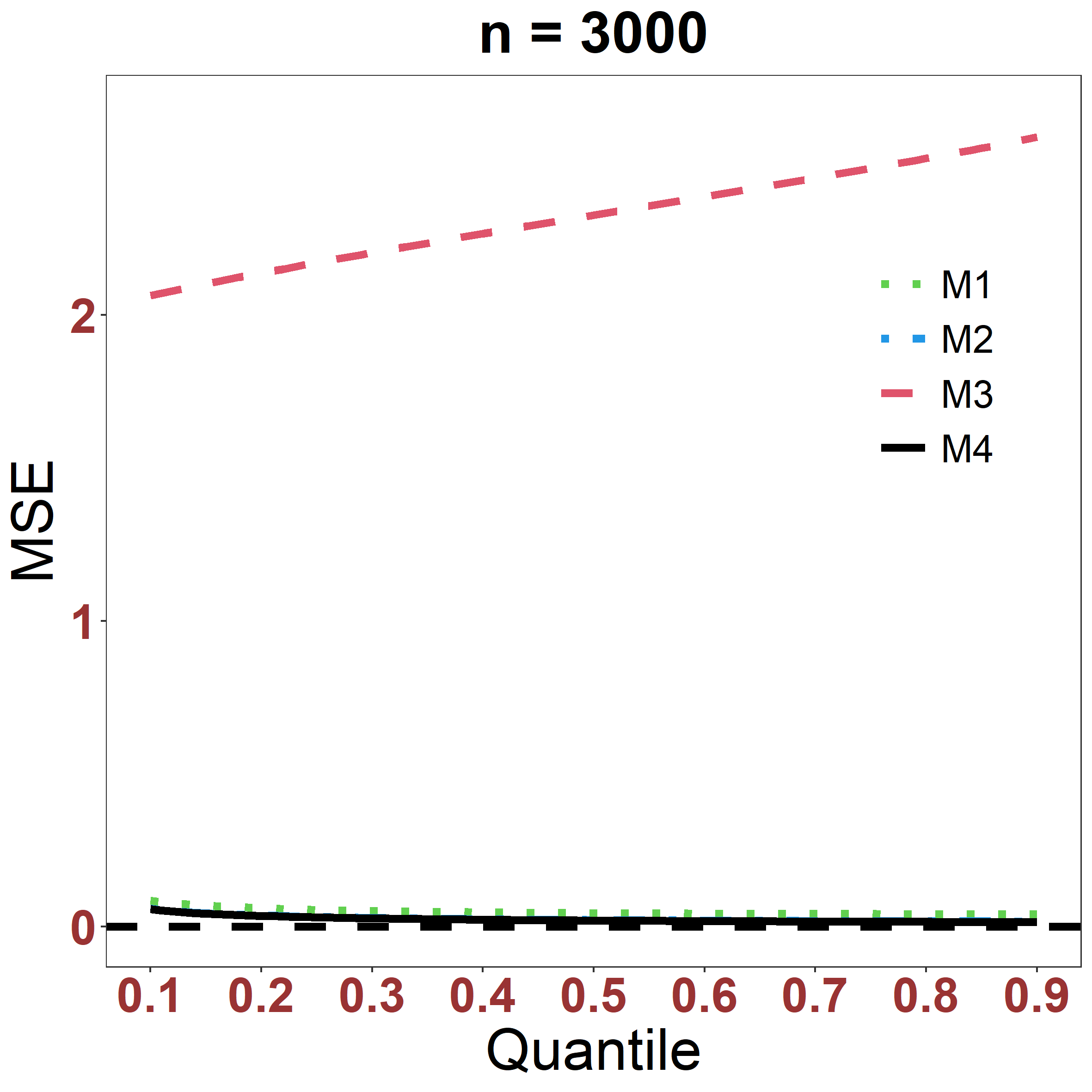}
	}	
	\caption{Bias, variance and MSE of estimator for estimating the CTATE for compliers when $\rho=0.5$ and $n=500$.}
	\label{figure6}
\end{figure}
\begin{figure}[!htb]
	\centering
	\mbox{
		\includegraphics[height=4.8cm,width=4.8cm]{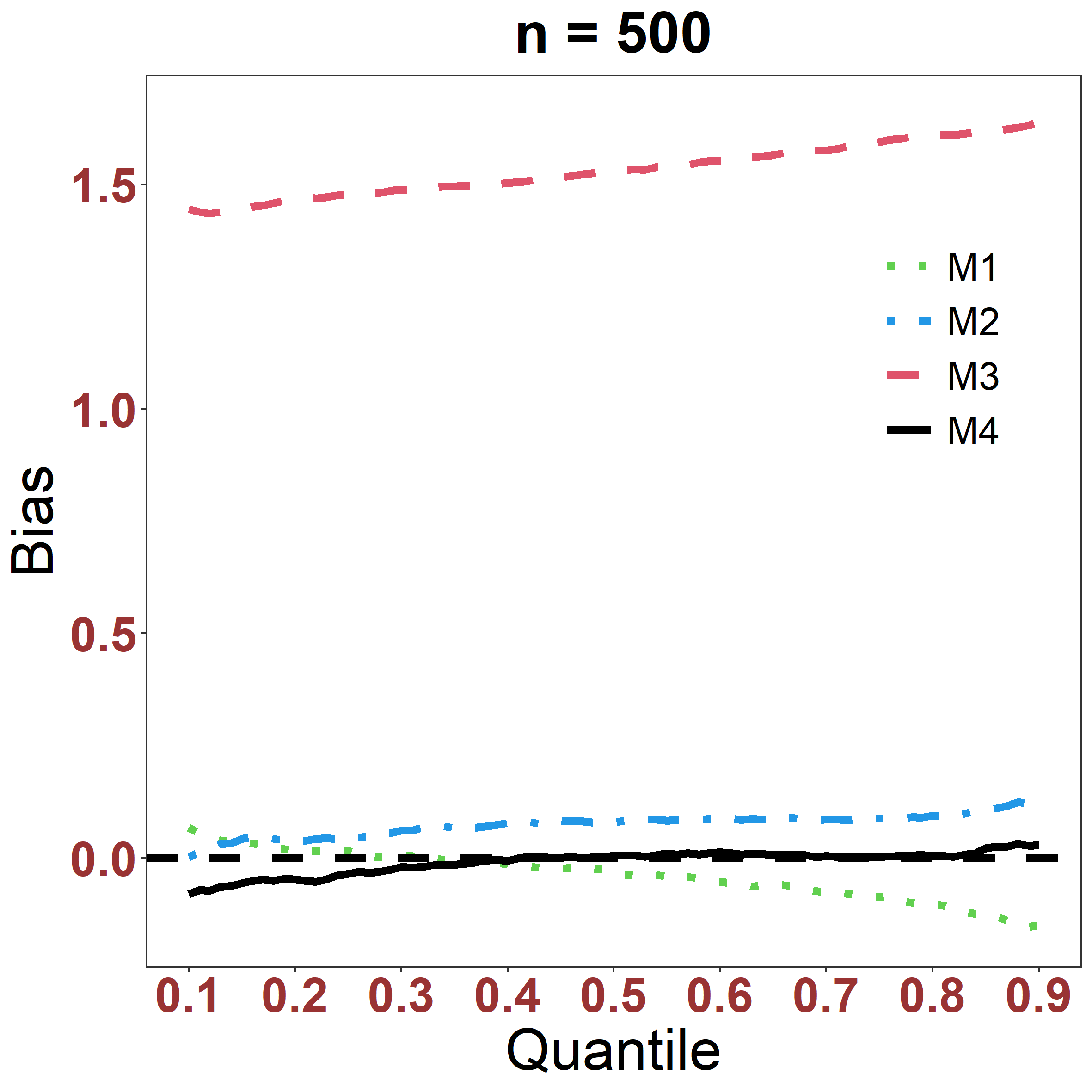}
		\includegraphics[height=4.8cm,width=4.8cm]{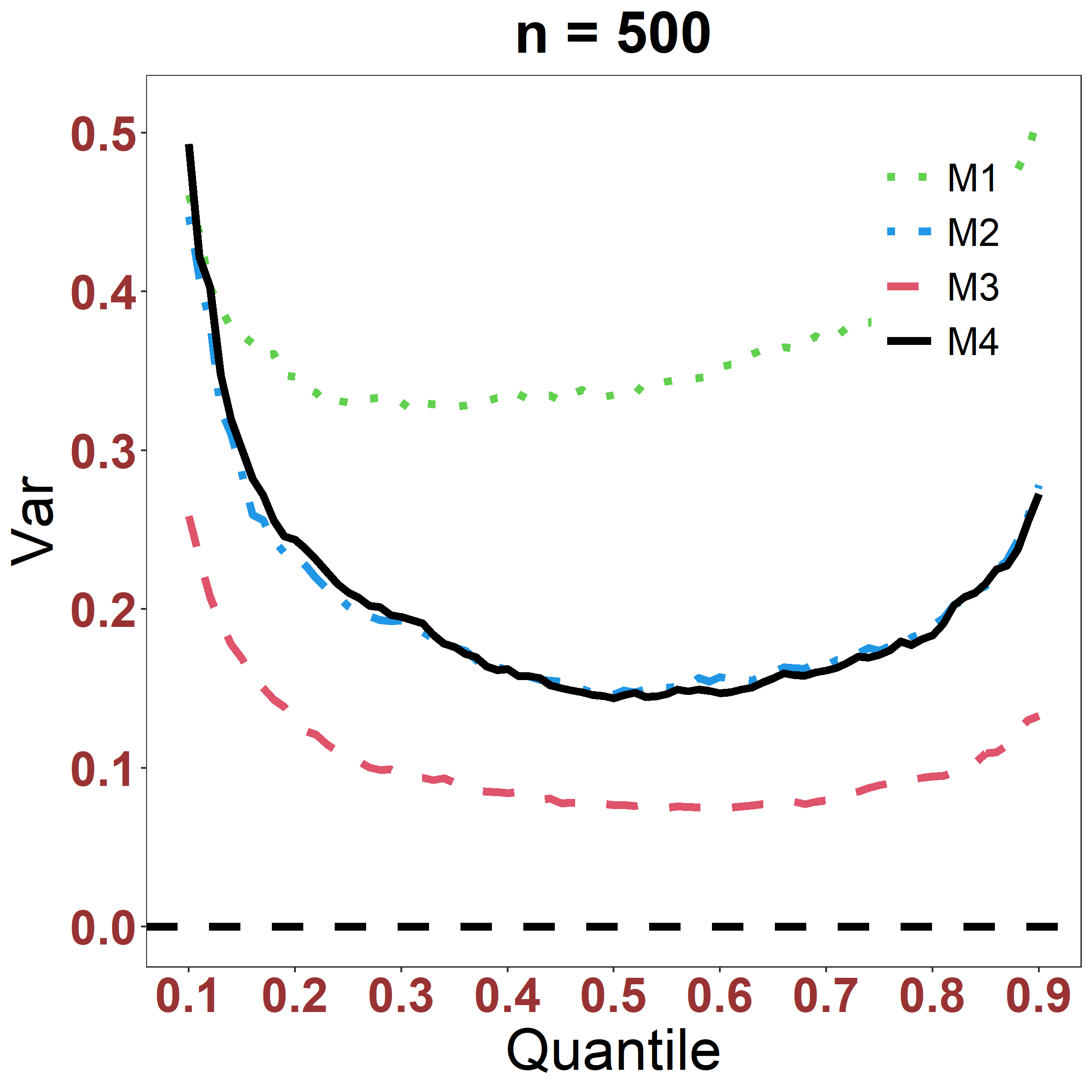}
	\includegraphics[height=4.8cm,width=4.8cm]{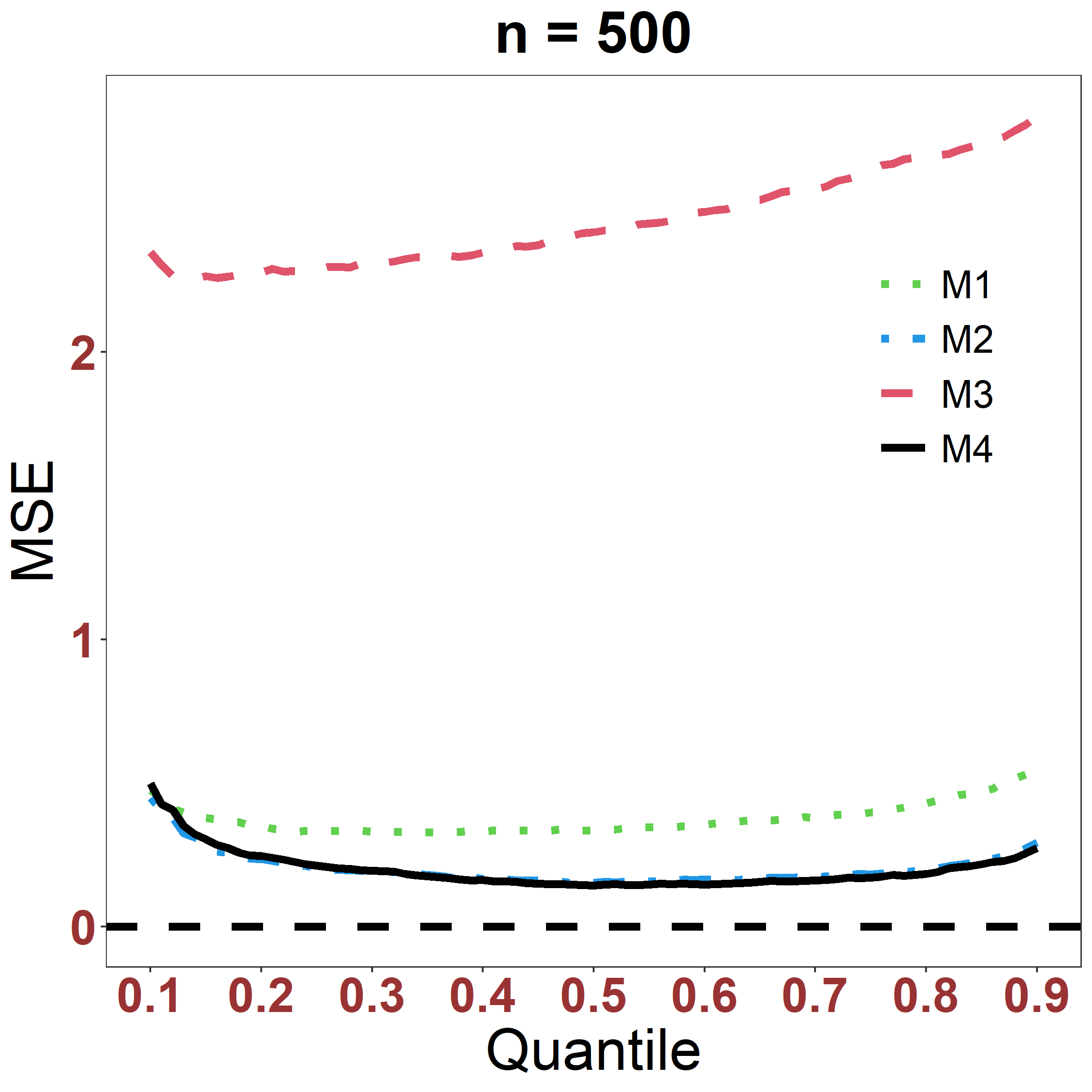}
	}
	\mbox{
		\includegraphics[height=4.8cm,width=4.8cm]{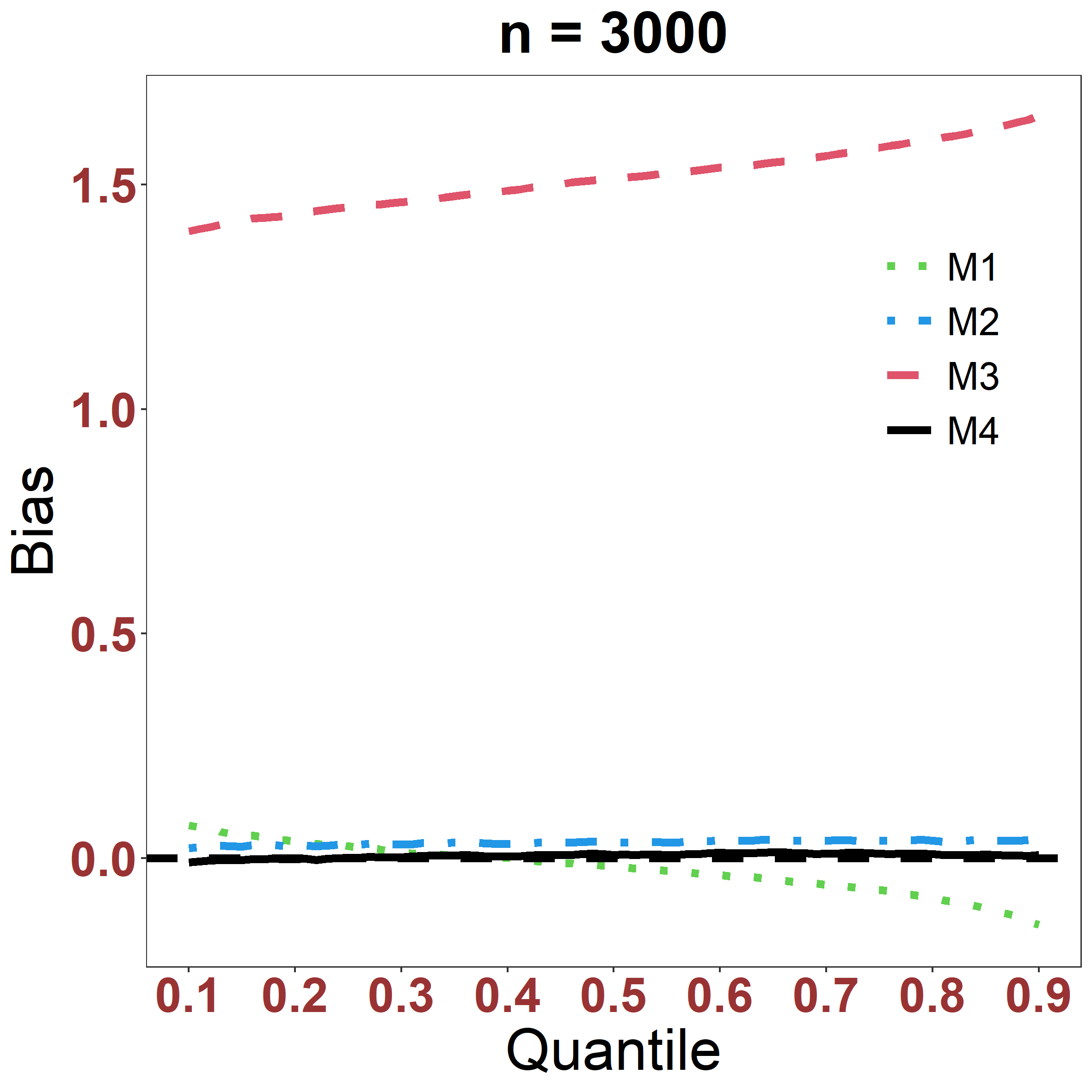}
		\includegraphics[height=4.8cm,width=4.8cm]{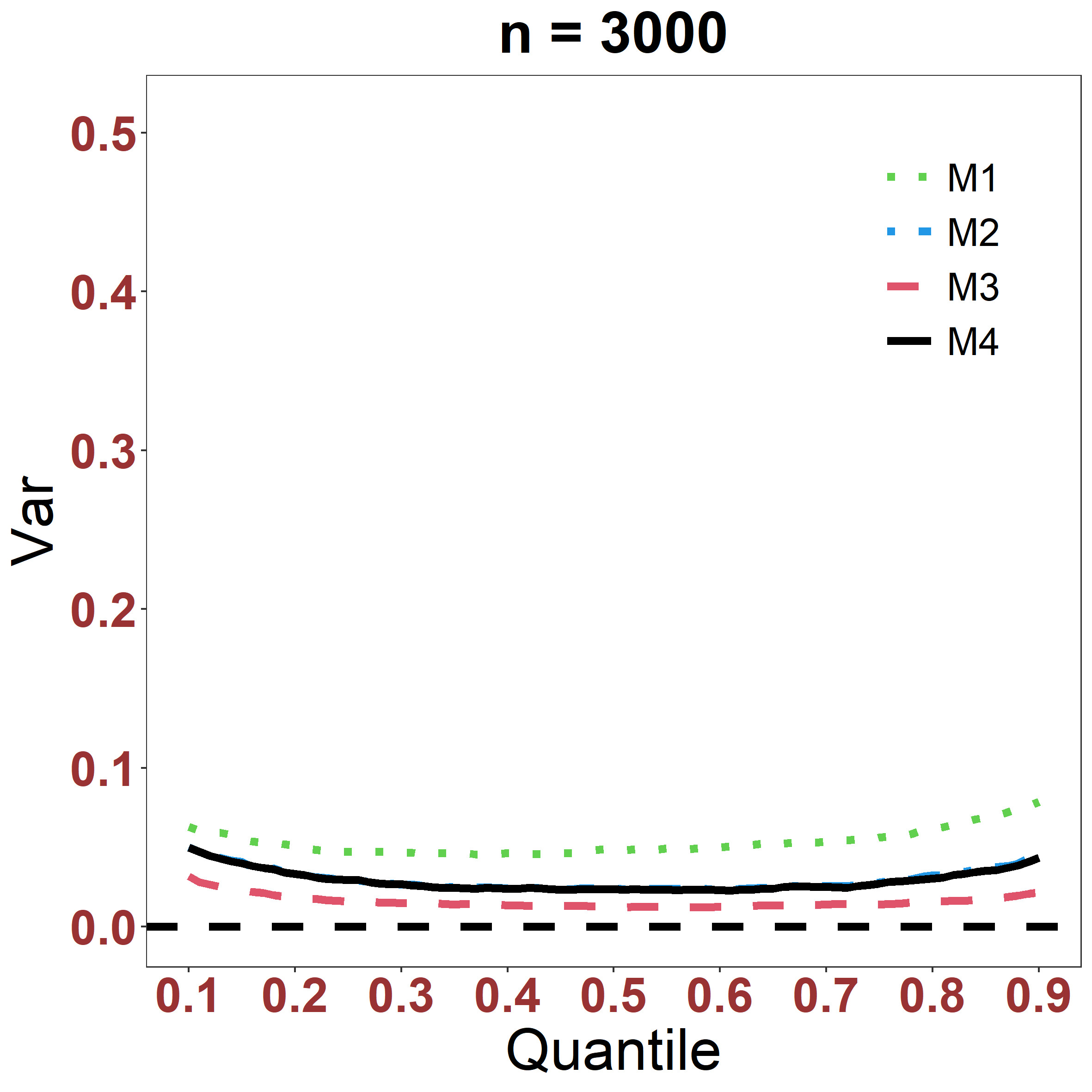}
		\includegraphics[height=4.8cm,width=4.8cm]{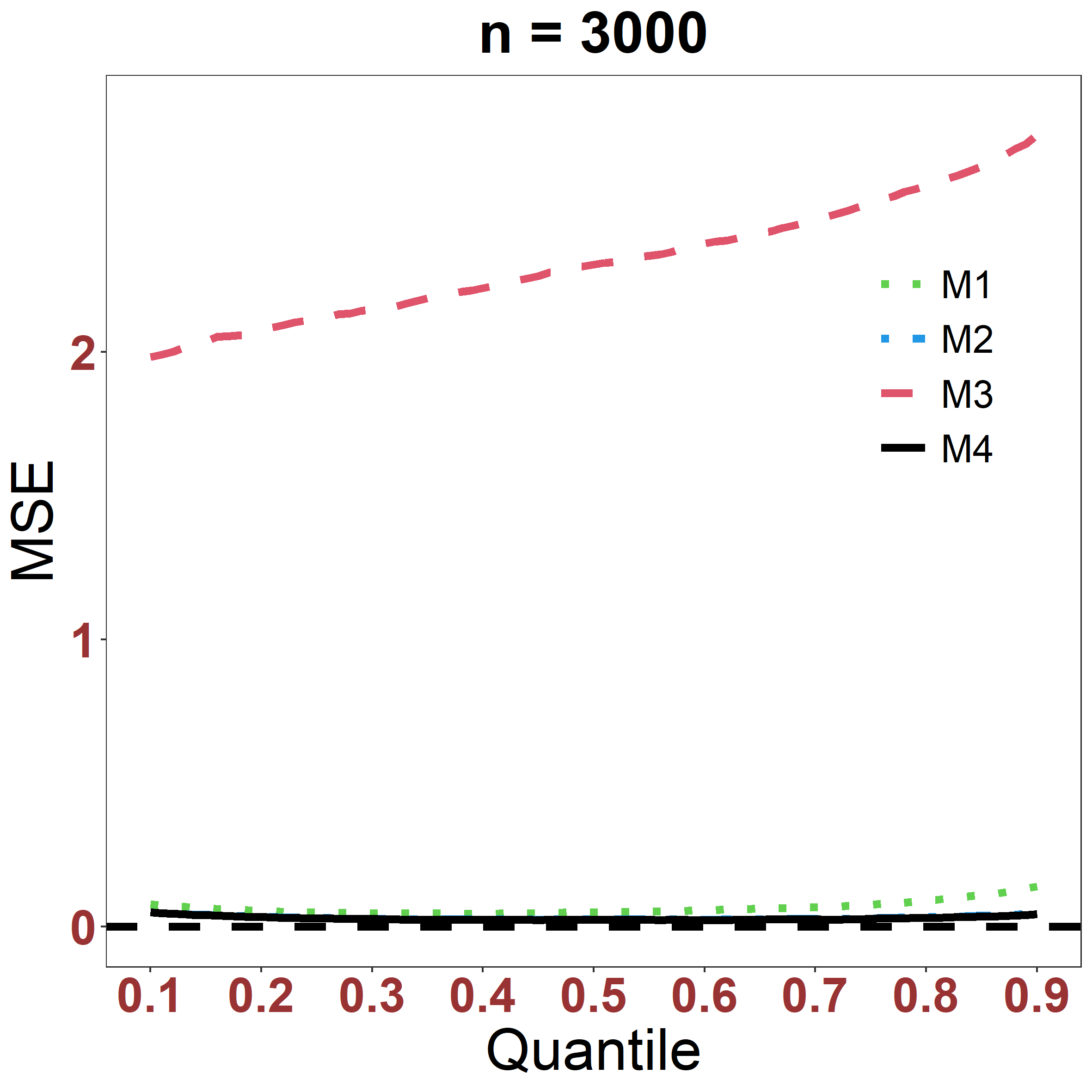}
	}	
	\caption{Bias, variance and MSE of estimator for estimating the QTE for compliers when $\rho=0.5$ and $n=3,000$.}
	\label{figure7}
\end{figure}

\begin{figure}[ht]
	\captionsetup[subfigure]{justification=centering}
\begin{subfigure}[b]{0.44\textwidth}
			\includegraphics[width=\textwidth]{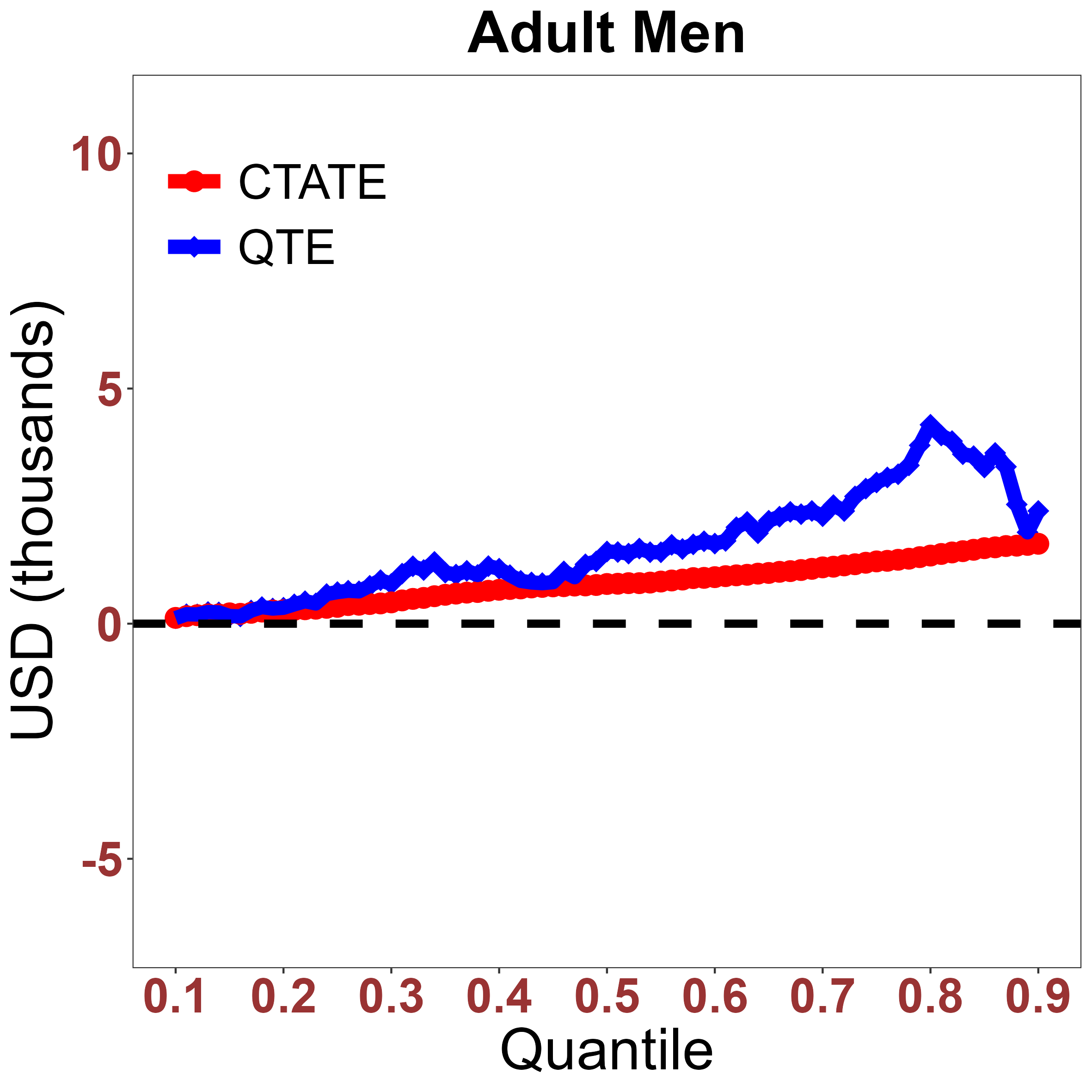}
\end{subfigure}			
 \hfill
 \begin{subfigure}[b]{0.44\textwidth}	
			\includegraphics[width=\textwidth]{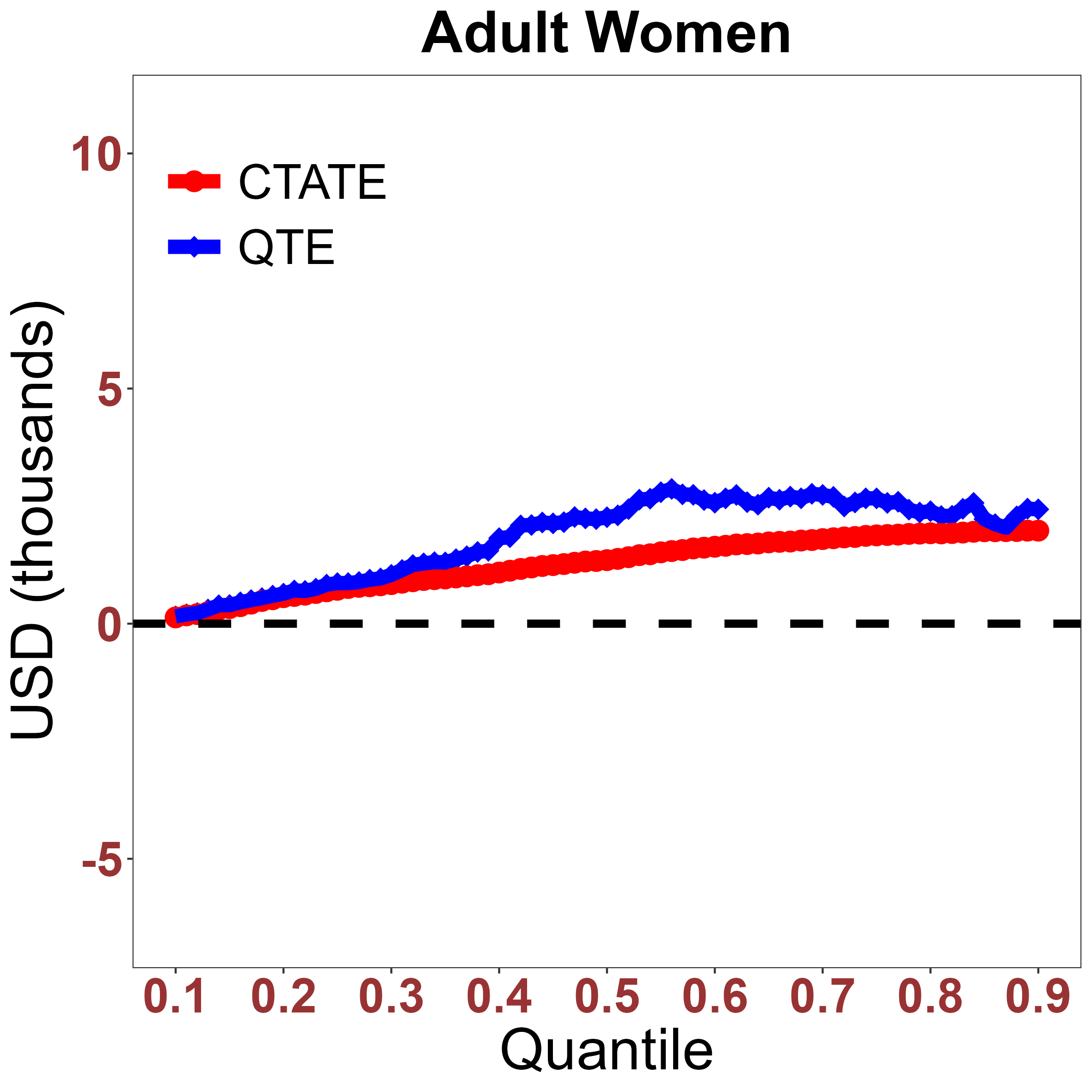}
\end{subfigure}\\			
				
\begin{subfigure}[b]{0.44\textwidth}
		\includegraphics[width=\linewidth]{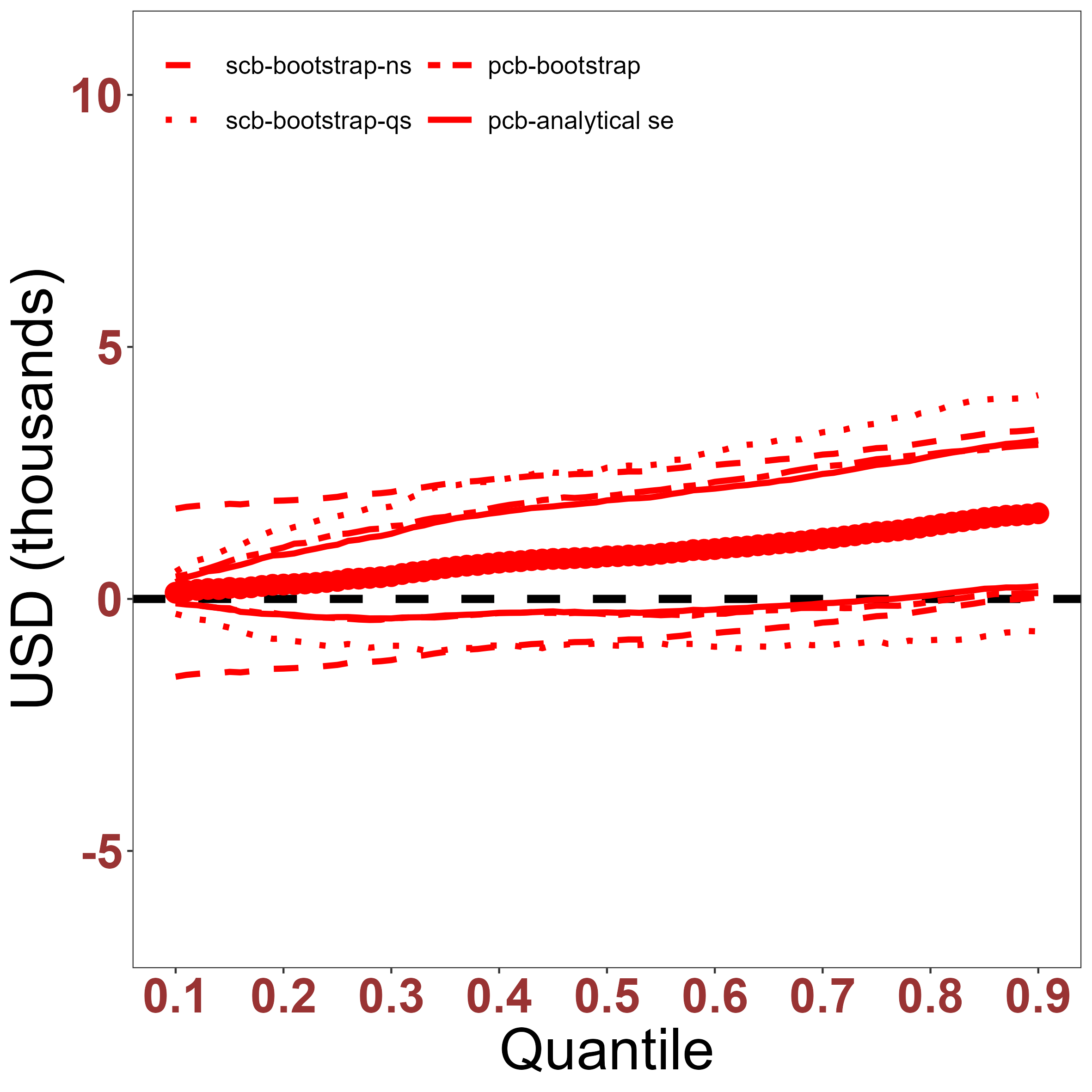}
	\end{subfigure}			
 \hfill
	\begin{subfigure}[b]{0.44\textwidth}	

		\includegraphics[width=\linewidth]{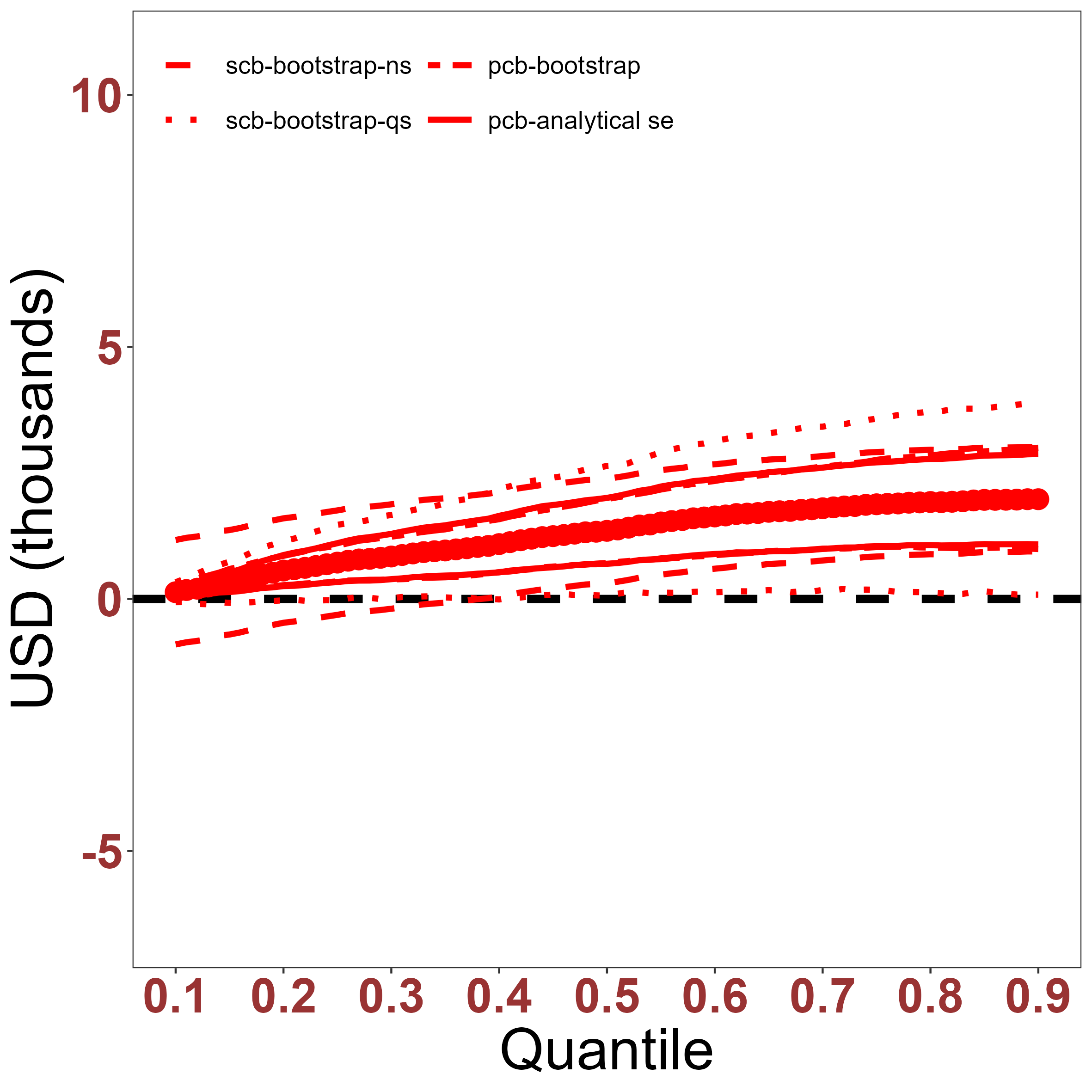}
	\end{subfigure}	\\	
	
	\begin{subfigure}[b]{0.44\textwidth}
		
		\includegraphics[width=\textwidth]{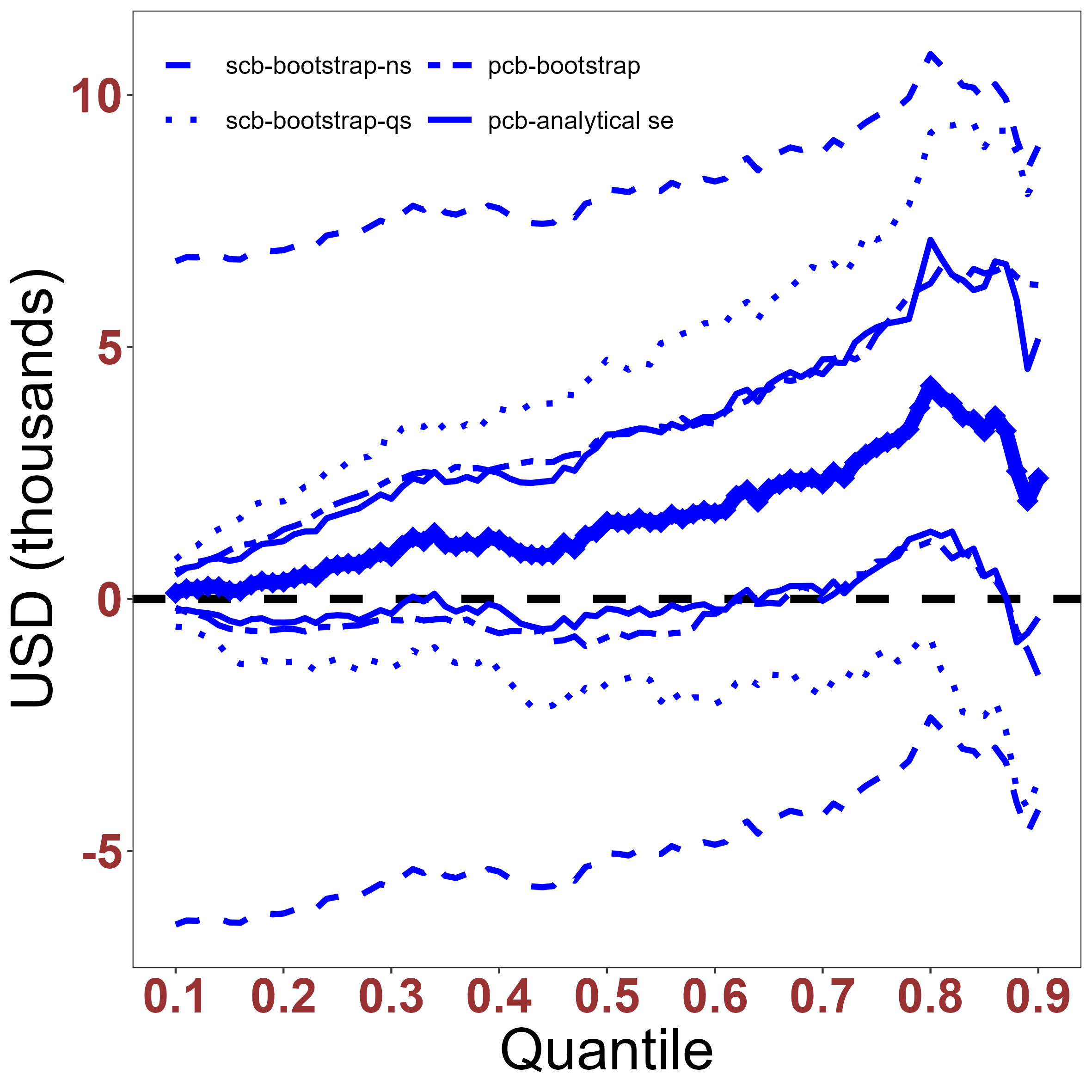}
	\end{subfigure}			
	\hfill
	\begin{subfigure}[b]{0.44\textwidth}	
		
		\includegraphics[width=\linewidth]{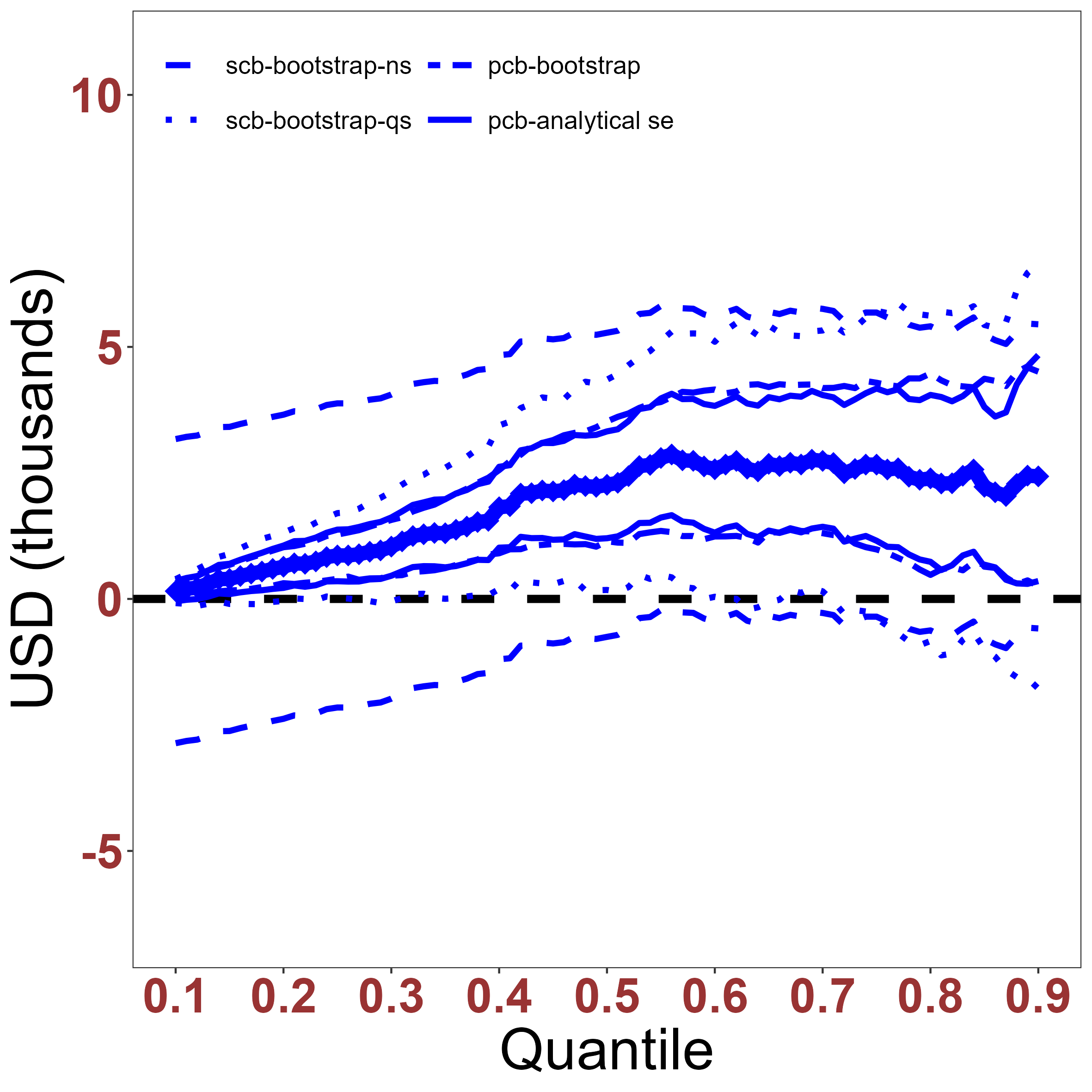}
	\end{subfigure}
		
	\caption{CTATE and QTE estimates of adult men's and women's earnings for compliers and the 95\% pointwise and simultaneous confidence bands when the treatment is endogenous. %The exogenous variables used for the estimations are the same as those used in \citet{AAI_2002}.
		}
	\label{figure1}
\end{figure}

\begin{figure}[ht]
	\captionsetup[subfigure]{justification=centering}
	\begin{subfigure}[b]{0.49\textwidth}
		\includegraphics[width=\textwidth]{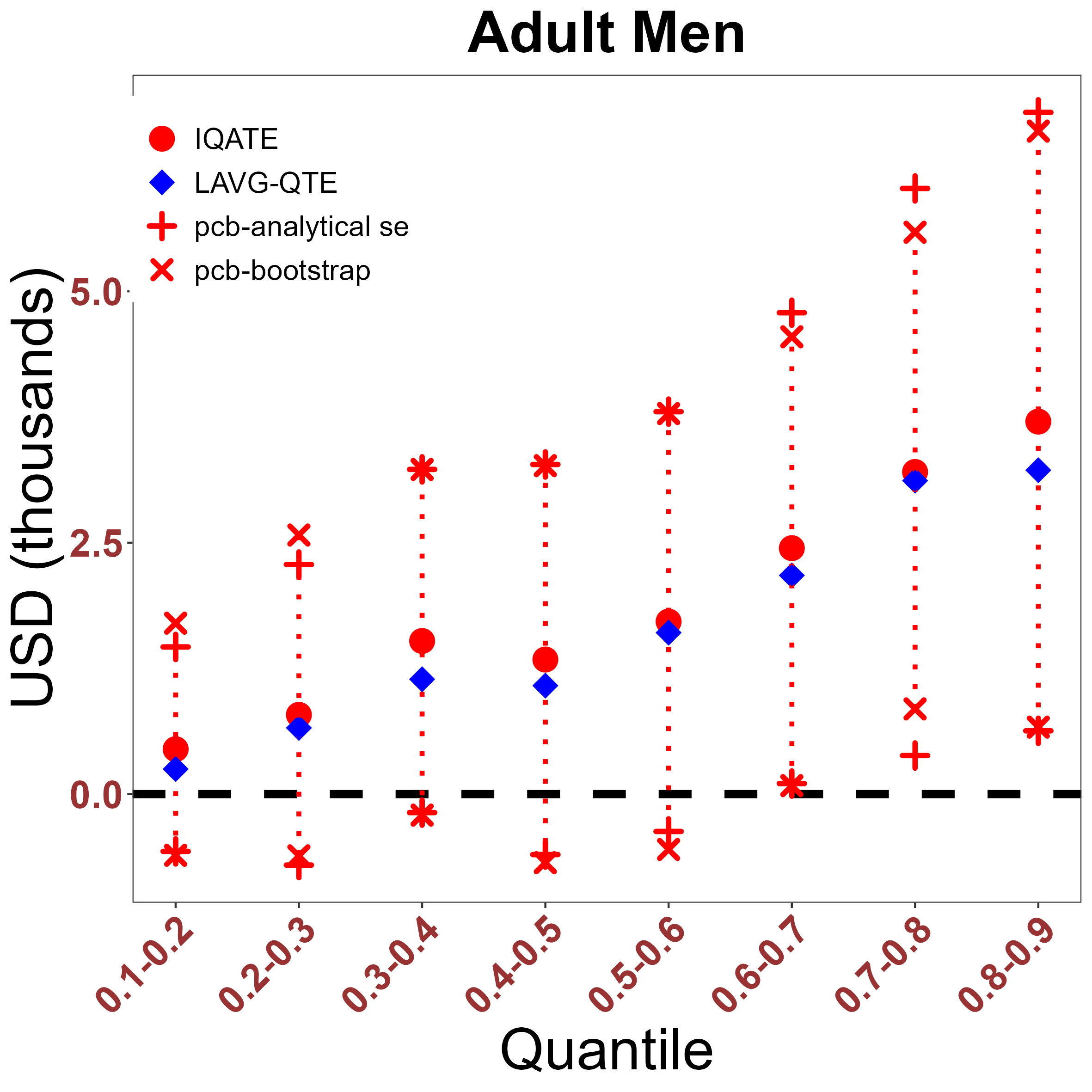}
	\end{subfigure}			
	\hfill
	\begin{subfigure}[b]{0.49\textwidth}	
		\includegraphics[width=\textwidth]{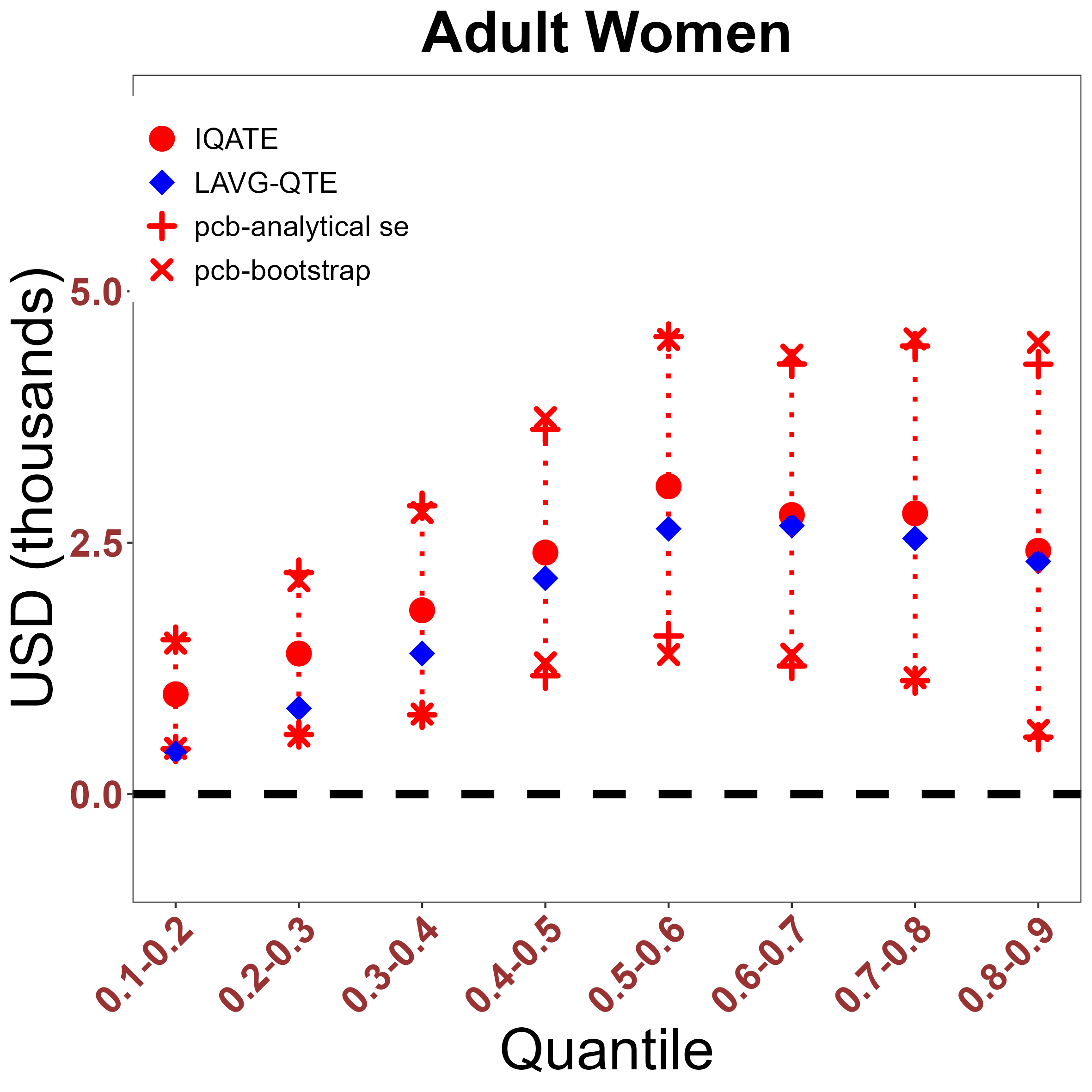}
	\end{subfigure}
	\caption{IQATE estimates of adult men's and women's earnings for compliers and the 95\% pointwise confidence bands when the treatment is endogenous.%The exogenous variables used are the same as those used in \citet{AAI_2002}.
		}
	\label{figure2}
\end{figure}
\clearpage
%\title{\textbf{Appendix of Estimation of the Local Conditional Tail Average Treatment Effect}}%\\(For On-Line Publication Only)}}
%\date{}
%\maketitle

{\large \renewcommand{\theequation}{A-\arabic{equation}} %
\setcounter{equation}{0} \appendix
}
\appendix \numberwithin{equation}{section}
\counterwithin{theorem}{section}
\counterwithin{lemma}{section}
\counterwithin{assumption}{section}
\appendix
\appendix
\section{Appendix}
\small
To ease the notation, throughout the Appendix, we will omit the random variables in the arguments for the structural functions $q(W,\boldsymbol{\theta}_{1,\tau})$ and $e(W,\boldsymbol{\theta}_{2,\tau})$ as well as the weights $K\left(W,Z;\boldsymbol{\gamma}\right)$, $\bar{K}\left(V;v,\boldsymbol{\gamma}\right)$ and $\tilde{K}\left(V;v,\boldsymbol{\gamma}\right)$ and abbreviate them respectively by $q(\boldsymbol{\theta}_{1,\tau})$, $e(\boldsymbol{\theta}_{2,\tau})$ $K\left(\boldsymbol{\gamma}\right)$, $\bar{K}\left(v,\boldsymbol{\gamma}\right)$ and $\tilde{K}\left(v,\boldsymbol{\gamma}\right)$. We also let $q_{i}(\boldsymbol{\theta}_{1,\tau}):=q(W_{i},\boldsymbol{\theta}_{1,\tau})$, $e_{i}(\boldsymbol{\theta}_{2,\tau}):=e(W_{i},\boldsymbol{\theta}_{2,\tau})$, $K_{i}\left(\boldsymbol{\gamma}\right):=K\left(W_{i},Z_{i};\boldsymbol{\gamma}\right)$, $\bar{K}_{i}\left(v,\boldsymbol{\gamma}\right):=\bar{K}\left(V_{i};v,\boldsymbol{\gamma}\right)$ and $\tilde{K}_{i}\left(v,\boldsymbol{\gamma}\right):=\tilde{K}\left(V_{i};v,\boldsymbol{\gamma}\right)$.
\subsection{Estimating the Asymptotic Covariance Matrix of the Estimators}
For estimating the asymptotic covariance matrix of the estimators, we focus on the case where the function $FZ_{\tau}^{sp}\left(q,e,y\right)$ is used and the complier quantile and conditional tail expectation models are linear in parameters. In this setting, $\nabla_{\boldsymbol{\theta}_{1}}q\left(\boldsymbol{\theta}_{1,\tau}\right)=\nabla_{\boldsymbol{\theta}_{2}}e\left(\boldsymbol{\theta}_{2,\tau}\right)=W$, where $W=\left(D,X^\top\right)^\top$. Recall that 
\[
FZ_{\tau}^{sp}\left(q,e,y\right)=\frac{\exp\left(e\right)}{1+\exp\left(e\right)}\left[e+LQ_{\tau}\left(q,y\right)\right]+\ln\left(1+\exp\left(y\right)\right)-\ln\left(1+\exp\left(e\right)\right),
\]where $LQ_{\tau}\left(q,y\right):=\tau^{-1}\max\left(q-y,0\right)-q$.
Let $G_{sp}(t):=\ln(1+\exp(t))$ denote the softplus function. Then
\[
G_{sp}^{\prime}\left(e\right)=\frac{\exp\left(e\right)}{1+\exp\left(e\right)},G_{sp}^{\prime\prime}\left(e\right)=\frac{\exp\left(e\right)}{\left[1+\exp\left(e\right)\right]^{2}}.
\]
Notice that both $G_{sp}^{\prime}\left(e\right)$ and $G_{sp}^{\prime\prime}\left(e\right)$ are bounded and continuous functions of $e$. 

Let $\mathbf{J}_{\tau}^{sp}$ denote $\mathbf{J}_{\tau}$ when $FZ_{\tau}^{sp}\left(q,e,y\right)$ is used:
\[
\mathbf{J}_{\tau}^{sp}=K\left(\boldsymbol{\gamma}_{0}\right)\nabla_{\boldsymbol{\theta}}FZ_{\tau}^{sp}\left(q\left(\boldsymbol{\theta}_{1,\tau}\right),e\left(\boldsymbol{\theta}_{2,\tau}\right),Y\right)+\mathbf{M}_{\tau}\psi\left(X\right),
\]
where
\begin{eqnarray}
	\nabla_{\boldsymbol{\theta}}FZ_{\tau}^{sp}\left(q\left(\boldsymbol{\theta}_{1,\tau}\right),e\left(\boldsymbol{\theta}_{2,\tau}\right),Y\right) & = & \left(\begin{array}{c}
		\frac{1}{\tau}\left(1\left\{ Y\leq q\left(\boldsymbol{\theta}_{1,\tau}\right)\right\} -\tau\right)G_{sp}^{\prime}\left(e\left(\boldsymbol{\theta}_{2,\tau}\right)\right)W\\
		G_{sp}^{\prime\prime}\left(e\left(\boldsymbol{\theta}_{2,\tau}\right)\right)\left[e\left(\boldsymbol{\theta}_{2,\tau}\right)+LQ_{\tau}\left(q\left(\boldsymbol{\theta}_{1,\tau}\right),Y\right)\right]W
	\end{array}\right),\label{FZ_der}
\end{eqnarray}
and $\boldsymbol{\Omega}_{\tau}^{sp}=E\left[\mathbf{J}_{\tau}^{sp}\left(\mathbf{J}_{\tau}^{sp}\right)^\top\right]$. Note that the function (\ref{FZ_der}) coincides with the derivative of $FZ_{\tau}^{sp}$ with respect to $\boldsymbol{\theta}=(\boldsymbol{\theta}_{1},\boldsymbol{\theta}_{2})$ whenever the latter exists. 
Let $\mathbf{C}_{\tau,11}^{sp}$, $\mathbf{C}_{\tau,22}^{sp}$ and
$\mathbf{H}_{\tau}^{sp}$ denote $\mathbf{C}_{\tau,11}$, $\mathbf{C}_{\tau,22}$
and $\mathbf{H}_{\tau}$ in the proof of Theorem 2 when $FZ_{\tau}^{sp}\left(q,e,y\right)$
is used:
\begin{eqnarray*}
	\mathbf{H}_{\tau}^{sp} & = & \left(\begin{array}{cc}
		\mathbf{C}_{\tau,11}^{sp*} & \mathbf{0}\\
		\mathbf{0} & \mathbf{C}_{\tau,22}^{sp*}
	\end{array}\right),
\end{eqnarray*}
where
\begin{align*}
\mathbf{C}_{\tau,11}^{sp*} & =\mathbf{C}_{\tau,11}^{sp}\times P\left(T=c\right)\\
& =E\left[\bar{K}\left(v_{0},\gamma_{0}\right)f_{Y|W,T=c}\left(q\left(\boldsymbol{\theta}_{1,\tau}\right)\right)\frac{1}{\tau}G_{sp}^{\prime}\left(e\left(\boldsymbol{\theta}_{2,\tau}\right)\right)WW^\top\right],\\
\mathbf{C}_{\tau,22}^{sp*} & =\mathbf{C}_{\tau,22}^{sp}\times P\left(T=c\right)\\
& =E\left[\bar{K}\left(v_{0},\gamma_{0}\right)G_{sp}^{\prime\prime}\left(e\left(\boldsymbol{\theta}_{2,\tau}\right)\right)WW^\top\right].
\end{align*}
The asymptotic covariance matrix of the proposed estimators is $\left(\mathbf{H}_{\tau}^{sp}\right)^{-1}\boldsymbol{\Omega}_{\tau}^{sp}\left(\mathbf{H}_{\tau}^{sp}\right)^{-1}.$ To estimate this asymptotic covariance matrix, we use the plug-in
estimator: replacing $\left(\boldsymbol{\theta}_{1,\tau},\boldsymbol{\theta}_{2,\tau},v_{0},\boldsymbol{\gamma}_{0}\right)$
in $\left(\mathbf{H}_{\tau}^{sp}\right)^{-1}\boldsymbol{\Omega}_{\tau}^{sp}\left(\mathbf{H}_{\tau}^{sp}\right)^{-1}$
with estimates $\left(\hat{\boldsymbol{\theta}}_{1,\tau},\hat{\boldsymbol{\theta}}_{2,\tau},\hat{v},\boldsymbol{\hat{\gamma}}\right)$  where $\left(\hat{\boldsymbol{\theta}}_{1,\tau},\hat{\boldsymbol{\theta}}_{2,\tau}\right)$ are obtained from the method introduced in Section 2 and $\left(\hat{v},\boldsymbol{\hat{\gamma}}\right)$ are constructed such that Assumption 4 holds. We first propose using
\[\hat{\boldsymbol{\Omega}}_{\tau}^{sp} = \frac{1}{n}\sum_{i=1}^{n}\hat{\mathbf{J}}_{i,\tau}^{sp}\hat{\mathbf{J}}_{i,\tau}^{sp\top}\]
to estimate $\boldsymbol{\Omega}_{\tau}^{sp}$, where 
\[
\hat{\mathbf{J}}_{i,\tau}^{sp}=K_{i}\left(\hat{\boldsymbol{\gamma}}\right)\nabla_{\boldsymbol{\theta}}FZ_{\tau}^{sp}\left(q_{i}\left(\hat{\boldsymbol{\theta}}_{1,\tau}\right),e_{i}\left(\hat{\boldsymbol{\theta}}_{2,\tau}\right),Y_{i}\right)+\hat{\mathbf{M}}_{\tau}\hat{\psi}\left(X_{i}\right),
\]
and 
\[
\hat{\mathbf{M}}_{\tau}=\frac{1}{n}\sum_{i=1}^{n}\nabla_{\boldsymbol{\theta}}FZ_{\tau}\left(q_{i}\left(\hat{\boldsymbol{\theta}}_{1,\tau}\right),e_{i}\left(\hat{\boldsymbol{\theta}}_{2,\tau}\right),Y_{i}\right)\nabla_{\gamma}K_{i}\left(\hat{\boldsymbol{\gamma}}\right).
\]
Here $\hat{\psi}\left(X_{i}\right)$ is an estimate of $\psi\left(X_{i}\right)$. For example,
if a probit model with a link function $\Phi\left(X_{i}^\top\boldsymbol{\gamma}\right)$
is used to estimate $\boldsymbol{\gamma}$,
\begin{equation}
\psi\left(X_{i}\right)=-\left(\frac{1}{n}\sum_{i=1}^{n}\frac{\partial s\left(\boldsymbol{\gamma};X_{i}\right)}{\partial\boldsymbol{\gamma}}\right)^{-1}s\left(\boldsymbol{\gamma};X_{i}\right),\label{influence_probit}
\end{equation}
where 
\[
s\left(\boldsymbol{\gamma};X_{i}\right)=\left(\frac{Y_{i}\phi\left(X_{i}^\top\boldsymbol{\gamma}\right)}{\Phi\left(X_{i}^\top\boldsymbol{\gamma}\right)}-\frac{\left(1-Y_{i}\right)\phi\left(X_{i}^\top\boldsymbol{\gamma}\right)}{1-\Phi\left(X_{i}^\top\boldsymbol{\gamma}\right)}\right)X_{i}
\]
and $-\left(\frac{1}{n}\sum_{i=1}^{n}\partial s\left(\boldsymbol{\gamma};X_{i}\right)/\partial\boldsymbol{\gamma}\right)^{-1}$
is the inverse of the information matrix. $\hat{\psi}\left(X_{i}\right)$
can be constructed by plugging $\hat{\boldsymbol{\gamma}}$ into $s\left(\boldsymbol{\gamma};X_{i}\right)$
and replacing $-\left(\frac{1}{n}\sum_{i=1}^{n}\partial s\left(\boldsymbol{\gamma};X_{i}\right)/\partial\boldsymbol{\gamma}\right)^{-1}$
in (\ref{influence_probit}) with the estimated covariance matrix
of $\hat{\boldsymbol{\gamma}}$ multiplied by the sample size. Let
\[
\varsigma_{\lambda_{n},i}\left(\hat{\boldsymbol{\theta}}_{1,\tau}\right)=\frac{1\left\{ \left|Y_{i}-q_{i}\left(\hat{\boldsymbol{\theta}}_{1,\tau}\right)\right|\leq\lambda_{n}\right\} }{2\lambda_{n}}
\]
where $\lambda_{n}$ is a deterministic function of $n$ and satisfies certain conditions. We then propose using 
\[
\hat{\mathbf{H}}_{\tau}^{sp}=\left(\begin{array}{cc}
\hat{\mathbf{C}}_{\tau,11}^{sp*} & \mathbf{0}\\
\mathbf{0} & \hat{\mathbf{C}}_{\tau,22}^{sp*}
\end{array}\right)
\]
to estimate $\mathbf{H}_{\tau}^{sp}$, where
\begin{eqnarray*}
	\mathbf{\hat{C}}_{\tau,11}^{sp*} & = & \frac{1}{n}\sum_{i=1}^{n}\left\{ \tilde{K}_{i}\left(\hat{v},\hat{\gamma}\right)\varsigma_{\lambda,i}\left(\hat{\boldsymbol{\theta}}_{1,\tau}\right)\frac{1}{\tau}G_{sp}^{\prime}\left(e_{i}\left(\hat{\boldsymbol{\theta}}_{2,\tau}\right)\right)W_{i}W_{i}^\top\right\}, \\
	\mathbf{\hat{C}}_{\tau,22}^{sp*} & = & \frac{1}{n}\sum_{i=1}^{n}\tilde{K}_{i}\left(\hat{v},\hat{\gamma}\right)G^{\prime\prime}\left(e_{i}\left(\hat{\boldsymbol{\theta}}_{2,\tau}\right)\right)W_{i}W_{i}^\top.
\end{eqnarray*}
The estimator $\varsigma_{\lambda_{n},i}\left(\hat{\boldsymbol{\theta}}_{1,\tau}\right)$ is the Powell sandwich estimator \citep{Powell_1984} and is used to approximate the conditional density $f_{Y|W,T=c}(q(\boldsymbol{\theta}_{1,\tau}))$. It was also used in \citet{Powell_1984}, \citet{EM_2004}, \citet{ACF_2006} and \citet{PFC_2019}. The bandwidth function $\lambda_{n}$ can be set equal to $O(n^{-1/3})$ \citep{PFC_2019} or chosen with some specific method (see discussions in \citet{Koenker_2005}). 

\begin{theorem}
    If Assumptions 1, 2 and 4 hold, $E\left[\left\Vert WW^{T}\right\Vert |T=c\right]<\infty$, and $\lambda_{n}=o\left(1\right)$
	and $\lambda_{n}^{-1}=o\left(n^{1/2}\right)$, %(iii) $\hat{\mathbf{H}}_{\tau}^{sp}$ is invertible,
	then
	\[
	\left(\hat{\mathbf{H}}_{\tau}^{sp}\right)^{-1}\hat{\boldsymbol{\Omega}}_{\tau}^{sp}\left(\hat{\mathbf{H}}_{\tau}^{sp}\right)^{-1}\overset{p.}{\rightarrow}\left(\mathbf{H}_{\tau}^{sp}\right)^{-1}\boldsymbol{\Omega}_{\tau}^{sp}\left(\mathbf{H}_{\tau}^{sp}\right)^{-1}.
	\]
\end{theorem}
Proof of Theorem A.1 can be found in Appendix A.2.

\subsection{Proofs for the Main Theoretical Results}
\begin{proof}[Proof of Theorem 1]
	The proof relies on the result of Theorem 2.1 in \citet{NM_1994}. To use this theorem, at first we need to show that \[E\left[FZ_{\tau}\left(q\left(\boldsymbol{\theta}_{1}\right),e\left(\boldsymbol{\theta}_{2}\right),Y\right)|T=c\right]\]
	is uniquely minimized at $\left(\boldsymbol{\theta}_{1,\tau},\boldsymbol{\theta}_{2,\tau}\right)$. By Assumptions 2.2 and 2.5, the loss function $FZ_{\tau}\left(q,e,y\right)$ is strictly consistent for the $\tau$-quantile and conditional tail expectation
	of a random variable (Corollary 5.5 in \citet{FZ_2016}).
	Thus $E\left[FZ_{\tau}\left(q\left(\boldsymbol{\theta}_{1}\right),e\left(\boldsymbol{\theta}_{2}\right),Y\right)|T=c\right]$
	is uniquely minimized at $\left(Q_{Y|W,T=c}\left(\tau\right),CTE_{Y|W,T=c}\left(\tau\right)\right)$, which is specified by $ \left(q\left(\boldsymbol{\theta}_{1,\tau}\right),e\left(\boldsymbol{\theta}_{2,\tau}\right)\right)$.
	%Note that
	%\[	\left(Q_{Y|W,T=c}\left(\tau\right),CTE_{Y|W,T=c}\left(\tau\right)\right)=\left(q\left(\boldsymbol{\theta}_{1,\tau}\right),e\left(\boldsymbol{\theta}_{2,\tau}\right)\right),\]
	By Assumption 2.6, it thus follows that $E\left[FZ_{\tau}\left(q\left(\boldsymbol{\theta}_{1}\right),e\left(\boldsymbol{\theta}_{2}\right),Y\right)|T=c\right]$
	is uniquely minimized at $\left(\boldsymbol{\theta}_{1,\tau},\boldsymbol{\theta}_{2,\tau}\right)$. Next if Assumption 1 holds, it can be shown that $\left(\boldsymbol{\theta}_{1,\tau},\boldsymbol{\theta}_{2,\tau}\right)$ can also be obtained from	
	\begin{equation}
		\left(\boldsymbol{\theta}_{1,\tau},\boldsymbol{\theta}_{2,\tau}\right)  =  \arg\min_{\left(\boldsymbol{\theta}_{1},\boldsymbol{\theta}_{2}\right)\in\boldsymbol{\Theta}}E\left[\bar{K}\left(v_{0},\boldsymbol{\gamma}_{0}\right)FZ_{\tau}\left(q\left(\boldsymbol{\theta}_{1}\right),e\left(\boldsymbol{\theta}_{2}\right),Y\right)\right].\label{estimation_problem}
	\end{equation}We then show that the following uniform convergence
	holds
	\begin{eqnarray}
		\sup_{\left(\boldsymbol{\theta}_{1},\boldsymbol{\theta}_{2}\right)\in\boldsymbol{\Theta}}\left| \frac{1}{n}\sum_{i=1}^{n}\tilde{K}_{i}\left(\hat{v},\hat{\boldsymbol{\gamma}}\right)FZ_{\tau}\left(q_{i}\left(\boldsymbol{\theta}_{1}\right),e_{i}\left(\boldsymbol{\theta}_{2}\right),Y_{i}\right) \right.& &\nonumber\\ 
		\left. -E\left[\bar{K}\left(v_{0},\gamma_{0}\right)FZ_{\tau}\left(q\left(\boldsymbol{\theta}_{1}\right),e\left(\boldsymbol{\theta}_{2}\right),Y\right)\right]\right| &\stackrel{p}{\rightarrow}& 0\label{uniform_convergence}
	\end{eqnarray}
	By using the triangle inequality, we have
	\begin{eqnarray}
		\sup_{\left(\boldsymbol{\theta}_{1},\boldsymbol{\theta}_{2}\right)\in\boldsymbol{\Theta}}\left| \frac{1}{n}\sum_{i=1}^{n}\tilde{K}_{i}\left(\hat{v},\hat{\boldsymbol{\gamma}}\right)FZ_{\tau}\left(q_{i}\left(\boldsymbol{\theta}_{1}\right),e_{i}\left(\boldsymbol{\theta}_{2}\right),Y_{i}\right) \right.& &\nonumber\\ 
		\left. -E\left[\bar{K}\left(v_{0},\gamma_{0}\right)FZ_{\tau}\left(q\left(\boldsymbol{\theta}_{1}\right),e\left(\boldsymbol{\theta}_{2}\right),Y\right)\right]\right| & & \nonumber\\
		\leq\sup_{\left(\boldsymbol{\theta}_{1},\boldsymbol{\theta}_{2}\right)\in\boldsymbol{\Theta}}\left| \frac{1}{n}\sum_{i=1}^{n}\tilde{K}_{i}\left(\hat{v},\hat{\boldsymbol{\gamma}}\right)FZ_{\tau}\left(q_{i}\left(\boldsymbol{\theta}_{1}\right),e_{i}\left(\boldsymbol{\theta}_{2}\right),Y_{i}\right)\right.& &\nonumber\\
		\left.-\frac{1}{n}\sum_{i=1}^{n}\bar{K}_{i}\left(v_{0},\boldsymbol{\gamma}_{0}\right)FZ_{\tau}\left(q_{i}\left(\boldsymbol{\theta}_{1}\right),e_{i}\left(\boldsymbol{\theta}_{2}\right),Y_{i}\right)\right|& & \label{uniform_convergence1}\\
		+\sup_{\left(\boldsymbol{\theta}_{1},\boldsymbol{\theta}_{2}\right)\in\boldsymbol{\Theta}}\left| \frac{1}{n}\sum_{i=1}^{n}\bar{K}_{i}\left(v_{0},\boldsymbol{\gamma}_{0}\right)FZ_{\tau}\left(q_{i}\left(\boldsymbol{\theta}_{1}\right),e_{i}\left(\boldsymbol{\theta}_{2}\right),Y_{i}\right)\right.& &\nonumber\\
		\left.-E\left[\bar{K}\left(v_{0},\boldsymbol{\gamma}_{0}\right)FZ_{\tau}\left(q\left(\boldsymbol{\theta}_{1}\right),e\left(\boldsymbol{\theta}_{2}\right),Y\right)\right]\right|.
		\label{uniform_convergence2}
	\end{eqnarray}
	%\begin{footnotesize}
	%\begin{eqnarray}
	%& &
	%\sup_{\left(\boldsymbol{\theta}_{1},\boldsymbol{\theta}_{2}\right)\in\boldsymbol{\Theta}}\left| \frac{1}{n}\sum_{i=1}^{n}\bar{K}_{i}\left(\hat{v},\hat{\boldsymbol{\gamma}}\right)FZ_{\tau}\left(q_{i}\left(\boldsymbol{\theta}_{1}\right),e_{i}\left(\boldsymbol{\theta}_{2}\right),Y_{i}\right)-E\left[\bar{K}\left(v_{0},\gamma_{0}\right)FZ_{\tau}\left(q\left(\boldsymbol{\theta}_{1}\right),e\left(\boldsymbol{\theta}_{2}\right),Y\right)\right]\right| \nonumber \\
	%&  & \leq\sup_{\left(\boldsymbol{\theta}_{1},\boldsymbol{\theta}_{2}\right)\in\boldsymbol{\Theta}}\left| \frac{1}{n}\sum_{i=1}^{n}\bar{K}_{i}\left(\hat{v},\hat{\boldsymbol{\gamma}}\right)FZ_{\tau}\left(q_{i}\left(\boldsymbol{\theta}_{1}\right),e_{i}\left(\boldsymbol{\theta}_{2}\right),Y_{i}\right)-\frac{1}{n}\sum_{i=1}^{n}\bar{K}_{i}\left(v_{0},\boldsymbol{\gamma}_{0}\right)FZ_{\tau}\left(q_{i}\left(\boldsymbol{\theta}_{1}\right),e_{i}\left(\boldsymbol{\theta}_{2}\right),Y_{i}\right)\right| \label{uniform_convergence1}\\
	%&  & +\sup_{\left(\boldsymbol{\theta}_{1},\boldsymbol{\theta}_{2}\right)\in\boldsymbol{\Theta}}\left| \frac{1}{n}\sum_{i=1}^{n}\bar{K}_{i}\left(v_{0},\boldsymbol{\gamma}_{0}\right)FZ_{\tau}\left(q_{i}\left(\boldsymbol{\theta}_{1}\right),e_{i}\left(\boldsymbol{\theta}_{2}\right),Y_{i}\right)-E\left[\bar{K}\left(v_{0},\boldsymbol{\gamma}_{0}\right)FZ_{\tau}\left(q\left(\boldsymbol{\theta}_{1}\right),e\left(\boldsymbol{\theta}_{2}\right),Y\right)\right]\right| .\label{uniform_convergence2}
	%\end{eqnarray}
	%\end{footnotesize}
	For (\ref{uniform_convergence1}), it can be shown that it is $o_{p}\left(1\right)$, since
	\begin{eqnarray*}
		%&  & \sup_{\left(\boldsymbol{\theta}_{1},\boldsymbol{\theta}_{2}\right)\in\boldsymbol{\Theta}}\left| \frac{1}{n}\sum_{i=1}^{n}\bar{K}_{i}\left(\hat{v},\hat{\boldsymbol{\gamma}}\right)FZ_{\tau}\left(q_{i}\left(\boldsymbol{\theta}_{1}\right),e_{i}\left(\boldsymbol{\theta}_{2}\right),Y_{i}\right)-\frac{1}{n}\sum_{i=1}^{n}\bar{K}_{i}\left(v_{0},\boldsymbol{\gamma}_{0}\right)FZ_{\tau}\left(q_{i}\left(\boldsymbol{\theta}_{1}\right),e_{i}\left(\boldsymbol{\theta}_{2}\right),Y_{i}\right)\right| \\
		&  & \sup_{\left(\boldsymbol{\theta}_{1},\boldsymbol{\theta}_{2}\right)\in\boldsymbol{\Theta}}\left| \frac{1}{n}\sum_{i=1}^{n}\left[\tilde{K}_{i}\left(\hat{v},\hat{\boldsymbol{\gamma}}\right)-\bar{K}_{i}\left(v_{0},\boldsymbol{\gamma}_{0}\right)\right]FZ_{\tau}\left(q_{i}\left(\boldsymbol{\theta}_{1}\right),e_{i}\left(\boldsymbol{\theta}_{2}\right),Y_{i}\right)\right| \\
		&  & \leq\max_{i\in\{1,...,n\}}\left|\tilde{K}_{i}\left(\hat{v},\hat{\boldsymbol{\gamma}}\right)-\bar{K}_{i}\left(v_{0},\gamma_{0}\right)\right|\times\sup_{\left(\boldsymbol{\theta}_{1},\boldsymbol{\theta}_{2}\right)\in\boldsymbol{\Theta}}\frac{1}{n}\sum_{i=1}^{n}\left|FZ_{\tau}\left(q_{i}\left(\boldsymbol{\theta}_{1}\right),e_{i}\left(\boldsymbol{\theta}_{2}\right),Y_{i}\right)\right| \\
		&  & \leq\max_{i\in\{1,...,n\}}\left|\bar{K}_{i}\left(\hat{v},\hat{\boldsymbol{\gamma}}\right)-\bar{K}_{i}\left(v_{0},\gamma_{0}\right)\right|\times\left[\frac{1}{n}\sum_{i=1}^{n}b\left(V_{i}\right)\right]\\
		& &=o_{p}\left(1\right),
	\end{eqnarray*}
	where the second inequality above follows from Assumptions 2.7, 2.8 and Lipschitz continuity of the mapping $t\rightarrow\min\{\max\{t,0\},1\}$ as well as the fact that $\bar{K}_{i}\left(v_{0},\gamma_{0}\right)\in[0,1]$ under Assumption 1 (see Lemma 3.2 of \citet{AAI_2002}) so that $\bar{K}_{i}\left(v_{0},\gamma_{0}\right)=\tilde{K}_{i}\left(v_{0},\gamma_{0}\right)$ for each $i\in\{1,...,n\}$.
	
	Next we show that (\ref{uniform_convergence2})
	is also $o_{p}\left(1\right)$. By Assumptions 2.4 and 2.5,  $FZ_{\tau}\left(q\left(\boldsymbol{\theta}_{1}\right),e\left(\boldsymbol{\theta}_{2}\right),Y\right)$
	is continuous in $\left(\boldsymbol{\theta}_{1},\boldsymbol{\theta}_{2}\right)\in\boldsymbol{\Theta}$.
	In addition, by Assumption 2.7,
\[E\left[\left| \bar{K}\left(v_{0},\boldsymbol{\gamma}_{0}\right)FZ_{\tau}\left(q\left(\boldsymbol{\theta}_{1}\right),e\left(\boldsymbol{\theta}_{2}\right),Y\right)\right|\right]\leq E\left[\bar{K}\left(v_{0},\boldsymbol{\gamma}\right)b\left(V\right)\right]\leq E\left[b\left(V\right)\right]<\infty.\] Using these results along with Assumptions 2.1, 2.3 and 2.4, we can apply Lemma 2.4
	of \citet{NM_1994} to deduce that (\ref{uniform_convergence2})
	is also $o_{p}\left(1\right)$. Combining all of the above results, it then follows from Theorem 2.1 of \citet{NM_1994} that
	$\left(\hat{\boldsymbol{\theta}}_{1,\tau},\hat{\boldsymbol{\theta}}_{2,\tau}\right)\stackrel{p.}{\rightarrow}\left(\boldsymbol{\theta}_{1,\tau},\boldsymbol{\theta}_{2,\tau}\right).$
\end{proof}

\begin{proof}[Proof of Lemma 1] Let $\mathcal{V}_{m}$ denote the support of the distribution of $(Y,W_{c})$ conditional on $W_{d}=m$. We first show that, for each $m\in\mathcal{W}_{d}$,  $\sup_{(y,w_{c})\in\mathcal{V}_{m}}\left|\hat{v}_{m}(y,w_{c})-v_{0,m}(y,w_{c})\right|=o_{p}\left(1\right)$, which relies on using the result of \citet{Newey_1997} for the power series
	approximation. To use this result, in addition to Assumptions 2.1,
	3.1 and 3.2, we also need to show that $Var\left(Z|V_{m}\right)$
	is bounded. This is evident since $Z$ is a binary variable.
	With these assumptions and above result, Theorem 4 of \citet{Newey_1997} implies that
	$\sup_{(y,w_{c})\in\mathcal{V}_{m}}\left|\hat{v}_{m}-v_{0,m}\right|=o_{p}\left(1\right)$ under the conditions that $r+1<s_{m}$, $\kappa_{m}\rightarrow\infty$ and $\kappa_{m}^{3}/n\rightarrow0$. By (26) %(\ref{v_hat})
	and since $m$ takes a finite number of values, we can conclude that 	$\sup_{V\in\mathcal{V}}\left|\hat{v}-v_{0}\right|=o_{p}\left(1\right).$ Now consider $\bar{K}\left(\hat{v},\hat{\boldsymbol{\gamma}}\right)$
	as an estimator for $\bar{K}\left(v_{0},\boldsymbol{\gamma}_{0}\right)$.
	To verify Assumption 2.8, notice that\begin{eqnarray}
		\sup_{V\in\mathcal{V}}\left|\bar{K}\left(V;\hat{v},\hat{\boldsymbol{\gamma}}\right)-\bar{K}\left(V;v_{0},\boldsymbol{\gamma}_{0}\right)\right| & \leq & \sup_{V\in\mathcal{V}}\left|\bar{K}\left(V;\hat{v},\hat{\boldsymbol{\gamma}}\right)-\bar{K}\left(V;v_{0},\hat{\boldsymbol{\gamma}}\right)\right|+\sup_{V\in\mathcal{V}}\left|\bar{K}\left(V;v_{0},\hat{\boldsymbol{\gamma}}\right)-\bar{K}\left(V;v_{0},\boldsymbol{\gamma}_{0}\right)\right|\nonumber \\
		& \leq & \sup_{V\in\mathcal{V}}\left|\frac{D\left(\hat{v}-v_{0}\right)}{1-\hat{\pi}}-\frac{\left(1-D\right)\left(\hat{v}-v_{0}\right)}{\hat{\pi}}\right|\nonumber \\
		&  & +\sup_{V\in\mathcal{V}}\left|D\left(1-v_{0}\right)\left(\frac{1}{1-\hat{\pi}}-\frac{1}{1-\pi_{0}}\right)\right|\nonumber \\
		&  & +\sup_{V\in\mathcal{V}}\left|\left(1-D\right)v_{0}\left(\frac{1}{\hat{\pi}}-\frac{1}{\pi_{0}}\right)\right|\nonumber \\
		& \leq & \sup_{V\in\mathcal{V}}\left|\hat{v}-v_{0}\right|\sup_{V\in\mathcal{V}}\left|\frac{D}{1-\hat{\pi}}-\frac{\left(1-D\right)}{\hat{\pi}}\right|\label{uniform_convergence_K}\\
		&  & +\sup_{V\in\mathcal{V}}\left|D\left(1-v_{0}\right)\right|\sup_{V\in\mathcal{V}}\left\Vert\frac{1}{\left(1-\bar{\pi}\right)^{2}}\nabla_{\boldsymbol{\gamma}}\bar{\pi}\right\Vert\left\Vert\hat{\boldsymbol{\gamma}}-\boldsymbol{\gamma}_{0}\right\Vert\label{uniform_convergence_K1}\\
		&  & +\sup_{V\in\mathcal{V}}\left|\left(1-D\right)v_{0}\right|\sup_{V\in\mathcal{V}}\left\Vert\frac{1}{\bar{\pi}^{2}}\nabla_{\boldsymbol{\gamma}}\bar{\pi}\right\Vert\left\Vert\hat{\boldsymbol{\gamma}}-\boldsymbol{\gamma}_{0}\right\Vert,\label{uniform_convergence_K2}
	\end{eqnarray}
	where $\hat{\pi}:=\pi(X,\hat{\boldsymbol{\gamma}})$, $\pi_{0}:=\pi(X,\boldsymbol{\gamma}_{0})$, $\bar{\pi}:=\pi\left(X,\boldsymbol{\bar{\gamma}}\right)$, and $\nabla_{\boldsymbol{\gamma}}\bar{\pi}$ is the partial derivative $\nabla_{\boldsymbol{\gamma}}\pi$ evaluated at $\boldsymbol{\gamma}=\bar{\boldsymbol{\gamma}}$, where
	$\bar{\boldsymbol{\gamma}}$ is some point between $\boldsymbol{\gamma}_{0}$
	and $\hat{\boldsymbol{\gamma}}$. Since $\sup_{V\in\mathcal{V}}\left|\hat{v}-v_{0}\right|=o_{p}\left(1\right)$
	and $\pi\left(x,\boldsymbol{\gamma}\right)$ is bounded away from 0 and 1 over the support of $X$ and for every $\boldsymbol{\gamma}$ such that  
	$\left\Vert \boldsymbol{\gamma}-\boldsymbol{\gamma}_{0}\right\Vert \leq\varepsilon$, the term (\ref{uniform_convergence_K}) is thus $o_{p}\left(1\right)$.
	Since $v_{0}=E\left[Z|V\right]$ is bounded for $V\in\mathcal{V}$
	and $D\in\left\{ 1,0\right\} $, the terms (\ref{uniform_convergence_K1}) and (\ref{uniform_convergence_K2}) are also $o_{p}\left(1\right)$ under Assumptions 3.3 and 3.4. Combining the above results, we can
	conclude that $\sup_{V\in\mathcal{V}}\left|\bar{K}\left(V;\hat{v},\hat{\boldsymbol{\gamma}}\right)-\bar{K}\left(V;v_{0},\boldsymbol{\gamma}_{0}\right)\right|=o_{p}\left(1\right)$.
\end{proof}

%\begin{proof}[Proof of corollary 1]
%	The proof is the same as that for Theorem 1 but the uniform convergence in Assumption 2.8 is constructed by using Lemma 1.
%\end{proof}

\begin{proof}[Proof of Theorem 2]The proof relies on using the
	result of Theorem 5.23 in \citet{Vaart_1998}. 
	%The condition \begin{small} 		\[
	%	\frac{1}{n}\sum_{i=1}^{n}\bar{K}_{i}\left(\hat{v},\hat{\boldsymbol{\gamma}}\right)FZ_{\tau}\left(q_{i}\left(\hat{\boldsymbol{\theta}}_{1,\tau}\right),e_{i}\left(\hat{\boldsymbol{\theta}}_{2,\tau}\right),Y_{i}\right)\leq\inf_{\left(\boldsymbol{\theta}_{1},\boldsymbol{\theta}_{2}\right)\in\boldsymbol{\Theta}}\frac{1}{n}\sum_{i=1}^{n}\bar{K}_{i}\left(\hat{v},\hat{\boldsymbol{\gamma}}\right)FZ_{\tau}\left(q_{i}\left(\boldsymbol{\theta}_{1}\right),e_{i}\left(\boldsymbol{\theta}_{2}\right),Y_{i}\right)+o_{p}\left(1\right)\]is obviously satisfied since $\left(\hat{\boldsymbol{\theta}}_{1,\tau},\hat{\boldsymbol{\theta}}_{2,\tau}\right)$ 	is a minimizer of the empirical loss. Given $\left(Y,W\right)$, 
	First, we note that measurability (w.r.t. $\left(Y,W\right)$) and almost surely differentiability (w.r.t.
	$\boldsymbol{\theta}_{1}$ and $\boldsymbol{\theta}_{2}$) of $\bar{K}\left(\hat{v},\hat{\boldsymbol{\gamma}}\right)FZ_{\tau}\left(q\left(\boldsymbol{\theta}_{1}\right),e_{1}\left(\boldsymbol{\theta}_{2}\right),Y\right)$
	can be proved by using the fact that $FZ_{\tau}\left(q\left(\boldsymbol{\theta}_{1}\right),e_{1}\left(\boldsymbol{\theta}_{2}\right),Y\right)$
	has these properties \citep{DB_2019} and $\bar{K}\left(\hat{v},\hat{\boldsymbol{\gamma}}\right)$
	is also a measurable function of $\left(Y,W\right)$. In addition, it can be shown that $FZ_{\tau}\left(q\left(\boldsymbol{\theta}_{1}\right),e_{1}\left(\boldsymbol{\theta}_{2}\right),Y\right)$
	is locally Lipschitz continuous in $\left(\boldsymbol{\theta}_{1},\boldsymbol{\theta}_{2}\right)$ around $\left(\boldsymbol{\theta}_{1,\tau},\boldsymbol{\theta}_{2,\tau}\right)$ (see Lemma 2 in Appendix A.3). %Therefore \begin{small} 
		%\[
		%\left|\bar{K}\left(\hat{v},\hat{\boldsymbol{\gamma}}\right)\left[FZ_{\tau}\left(q\left(\boldsymbol{\theta}_{1}^{1}\right),e_{1}\left(\boldsymbol{\theta}_{2}^{1}\right),Y\right)-FZ_{\tau}\left(q\left(\boldsymbol{\theta}_{1}^{2}\right),e_{1}\left(\boldsymbol{\theta}_{2}^{2}\right),Y\right)\right]\right|\leq%\leq\left|\bar{K}\left(\hat{v},\hat{\boldsymbol{\gamma}}\right)\right|A_{\tau}^{*}\left(V\right)\left|\boldsymbol{\theta}^{1}-\boldsymbol{\theta}^{2}\right|
		%
		%\left|\bar{K}\left(v_{0},\boldsymbol{\gamma}_{0}\right)+o_{p}\left(1\right)\right|A_{\tau}^{*}\left(V\right)\left|\boldsymbol{\theta}^{1}-\boldsymbol{\theta}^{2}\right|.
		%\]
	%\end{small} 
	By Assumption 4.2 and using Lemma 2 in Appendix A.3, there exists a constant $A_{\tau}^{*}\left(V\right)$ such that
	\begin{eqnarray*}
		E\left[\left|\bar{K}\left(v_{0},\boldsymbol{\gamma}_{0}\right)+o_{p}\left(1\right)\right|^{2}\left(A_{\tau}^{*}\left(V\right)\right)^{2}\right] & \cong & E\left[\left(\bar{K}\left(v_{0},\boldsymbol{\gamma}_{0}\right)\right)^{2}\left(A_{\tau}^{*}\left(V\right)\right)^{2}\right]\\
		& \leq & E\left[\bar{K}\left(v_{0},\boldsymbol{\gamma}_{0}\right)\left(A_{\tau}^{*}\left(V\right)\right)^{2}\right]\\
		& = & E\left[\left(A_{\tau}^{*}\left(V\right)\right)^{2}|T=c\right]P\left(T=c\right)\\
		& < & \infty.
	\end{eqnarray*}
	We now show that $E\left[\bar{K}\left(v_{0},\boldsymbol{\gamma}_{0}\right)FZ_{\tau}\left(q\left(\boldsymbol{\theta}_{1}\right),e\left(\boldsymbol{\theta}_{2}\right),Y\right)\right]$
	admits a second order Taylor expansion at $\left(\boldsymbol{\theta}_{1,\tau},\boldsymbol{\theta}_{2,\tau}\right)$.
	Under Assumptions 2.4 and 2.5,  $E\left[\bar{K}\left(v_{0},\boldsymbol{\gamma}_{0}\right)FZ_{\tau}\left(q\left(\boldsymbol{\theta}_{1}\right),e\left(\boldsymbol{\theta}_{2}\right),Y\right)\right]$
	is a twice differentiable function of $\boldsymbol{\theta}:=\left(\boldsymbol{\theta}_{1},\boldsymbol{\theta}_{2}\right)$ at $\left(\boldsymbol{\theta}_{1,\tau},\boldsymbol{\theta}_{2,\tau}\right)$,
	and it can be shown that 
	\begin{eqnarray*}
		\mathbf{H}_{\tau} & := & \nabla_{\boldsymbol{\theta\theta}}E\left[\bar{K}\left(v_{0},\boldsymbol{\gamma}_{0}\right)FZ_{\tau}\left(q\left(\boldsymbol{\theta}_{1,\tau}\right),e\left(\boldsymbol{\theta}_{2,\tau}\right),Y\right)\right]\\
		& = & \nabla_{\boldsymbol{\theta\theta}}E\left[FZ_{\tau}\left(q\left(\boldsymbol{\theta}_{2,\tau}\right),e\left(\boldsymbol{\theta}_{2,\tau}\right),Y\right)|T=c\right]\times P\left(T=c\right),
	\end{eqnarray*}
	where 
	\[
	\nabla_{\boldsymbol{\theta\theta}}E\left[FZ_{\tau}\left(q\left(\boldsymbol{\theta}_{1,\tau}\right),e\left(\boldsymbol{\theta}_{2,\tau}\right),Y\right)|T=c\right]=\left(\begin{array}{cc}
	\mathbf{C}_{\tau,11} & \mathbf{0}\\
	\mathbf{0} & \mathbf{C}_{\tau,22}
	\end{array}\right),
	\]
	and 
	\begin{eqnarray*}
		\mathbf{C}_{\tau,11} & = & E\left[f_{Y|W,T=c}\left(q\left(\boldsymbol{\theta}_{1,\tau}\right)\right)\left[G_{1}^{\prime}\left(q\left(\boldsymbol{\theta}_{1,\tau}\right)\right)+\frac{1}{\tau}G_{2}^{\prime}\left(e\left(\boldsymbol{\theta}_{2,\tau}\right)\right)\right]\right.\\
		&  & \times\left.\left[\nabla_{\boldsymbol{\theta}_{1}}q\left(\boldsymbol{\theta}_{1,\tau}\right)\right]^\top\left[\nabla_{\boldsymbol{\theta}_{1}}q\left(\boldsymbol{\theta}_{1,\tau}\right)\right]|T=c\right],\\
		\mathbf{C}_{\tau,22} & = & E\left[G_{2}^{\prime\prime}\left(e\left(\boldsymbol{\theta}_{2,\tau}\right)\right)\left[\nabla_{\boldsymbol{\theta}_{2}}e\left(\boldsymbol{\theta}_{2,\tau}\right)\right]^\top\left[\nabla_{\boldsymbol{\theta}_{2}}e\left(\boldsymbol{\theta}_{2,\tau}\right)\right]|T=c\right].
	\end{eqnarray*}
	With the above results, using Theorem 5.23 in \citet{Vaart_1998},
	we can have 
	\begin{eqnarray*}
		\sqrt{n}\left(\left(\hat{\boldsymbol{\theta}}_{1,\tau},\hat{\boldsymbol{\theta}}_{2,\tau}\right)-\left(\boldsymbol{\theta}_{1,\tau},\boldsymbol{\theta}_{2,\tau}\right)\right) & = & -\mathbf{H}_{\tau}^{-1}\frac{1}{\sqrt{n}}\sum_{i=1}^{n}\bar{K}_{i}\left(\hat{v},\hat{\boldsymbol{\gamma}}\right)\nabla_{\boldsymbol{\theta}}FZ_{\tau}\left(q_{i}\left(\boldsymbol{\theta}_{1,\tau}\right),e_{i}\left(\boldsymbol{\theta}_{2,\tau}\right),Y_{i}\right)+o_{p}\left(1\right)\\
		& = & -\mathbf{H}_{\tau}^{-1}\left\{ \frac{1}{\sqrt{n}}\sum_{i=1}^{n}\bar{K}_{i}\left(\hat{v},\boldsymbol{\gamma}_{0}\right)\nabla_{\boldsymbol{\theta}}FZ_{\tau}\left(q_{i}\left(\boldsymbol{\theta}_{1,\tau}\right),e_{i}\left(\boldsymbol{\theta}_{2,\tau}\right),Y_{i}\right)\right.\\
		&  & +\left.\frac{1}{n}\sum_{i=1}^{n}\left[\nabla_{\boldsymbol{\gamma}}\bar{K}_{i}\left(\hat{v},\boldsymbol{\gamma}_{0}\right)\right]^\top\sqrt{n}\left(\hat{\boldsymbol{\gamma}}-\boldsymbol{\gamma}_{0}\right)\nabla_{\boldsymbol{\theta}}FZ_{\tau}\left(q_{i}\left(\boldsymbol{\theta}_{1,\tau}\right),e_{i}\left(\boldsymbol{\theta}_{2,\tau}\right),Y_{i}\right)\right\} \\
		&  & +o_{p}\left(1\right),
	\end{eqnarray*}
	where 
	\[
	\nabla_{\boldsymbol{\gamma}}\bar{K}_{i}\left(\hat{v},\boldsymbol{\gamma}_{0}\right)=\left[\frac{\left(1-D_{i}\right)\hat{v}_{i}}{\pi_{0i}^{2}}-\frac{D_{i}\left(1-\hat{v}_{i}\right)}{\left(1-\pi_{0i}\right)^{2}}\right]\nabla_{\boldsymbol{\gamma}}\bar{\pi}_{i},
	\]
	and $\bar{\pi}_{i}=\pi\left(X_{i},\bar{\boldsymbol{\gamma}}\right)$
	and $\bar{\boldsymbol{\gamma}}$ is some point between $\hat{\boldsymbol{\gamma}}$
	and $\boldsymbol{\gamma}_{0}$, and 
	\begin{eqnarray*}
		\nabla_{\boldsymbol{\theta}}FZ_{\tau}\left(q_{i}\left(\boldsymbol{\theta}_{1,\tau}\right),e_{i}\left(\boldsymbol{\theta}_{2,\tau}\right),Y_{i}\right) & = & \left(\begin{array}{c}
			\left(1\left\{ Y_{i}\leq q\left(\boldsymbol{\theta}_{1,\tau}\right)\right\} -\tau\right)\\
			\times\left[G_{1}^{\prime}\left(q_{i}\left(\boldsymbol{\theta}_{1,\tau}\right)\right)+\frac{1}{\tau}G_{2}^{\prime}\left(e_{i}\left(\boldsymbol{\theta}_{2,\tau}\right)\right)\right]\nabla_{\boldsymbol{\theta}_{1}}q_{i}\left(\boldsymbol{\theta}_{1,\tau}\right)\\
			G_{2}^{\prime\prime}\left(e_{i}\left(\boldsymbol{\theta}_{2,\tau}\right)\right)\left[e_{i}\left(\boldsymbol{\theta}_{2,\tau}\right)\right.\\
			\left.+\frac{1}{\tau}\max\left(q_{i}\left(\boldsymbol{\theta}_{1,\tau}\right)-Y_{i},0\right)-q_{i}\left(\boldsymbol{\theta}_{1,\tau}\right)\right]\nabla_{\boldsymbol{\theta}_{2}}e_{i}\left(\boldsymbol{\theta}_{2,\tau}\right)
		\end{array}\right)
	\end{eqnarray*}
	is the Jacobian vector. By similar arguments used in proving Lemma
	A.1 in \citet{AAI_2002}, with Assumption 4.5, it can be shown that 
	\begin{eqnarray*}
		\frac{1}{\sqrt{n}}\sum_{i=1}^{n}\bar{K}_{i}\left(\hat{v},\boldsymbol{\gamma}_{0}\right)\nabla_{\boldsymbol{\theta}}FZ_{\tau}\left(q_{i}\left(\boldsymbol{\theta}_{1,\tau}\right),e_{i}\left(\boldsymbol{\theta}_{2,\tau}\right),Y_{i}\right) & = & \frac{1}{\sqrt{n}}\sum_{i=1}^{n}K_{i}\left(\boldsymbol{\gamma}_{0}\right)\nabla_{\boldsymbol{\theta}}FZ_{\tau}\left(q_{i}\left(\boldsymbol{\theta}_{1,\tau}\right),e_{i}\left(\boldsymbol{\theta}_{2,\tau}\right),Y_{i}\right)\\
		& &+o_{p}(1),\\
		\frac{1}{n}\sum_{i=1}^{n}\left[\nabla_{\boldsymbol{\gamma}}\bar{K}_{i}\left(\hat{v},\boldsymbol{\gamma}_{0}\right)\right]^\top\sqrt{n}\left(\hat{\boldsymbol{\gamma}}-\boldsymbol{\gamma}_{0}\right)\nabla_{\boldsymbol{\theta}}FZ_{\tau}\left(q_{i}\left(\boldsymbol{\theta}_{1,\tau}\right),e_{i}\left(\boldsymbol{\theta}_{2,\tau}\right),Y_{i}\right) & = & \frac{1}{n}\sum_{i=1}^{n}\left[\nabla_{\boldsymbol{\gamma}}K_{i}\left(\boldsymbol{\gamma}_{0}\right)\right]^\top\sqrt{n}\left(\hat{\boldsymbol{\gamma}}-\boldsymbol{\gamma}_{0}\right)\\
		&&\times\nabla_{\boldsymbol{\theta}}FZ_{\tau}\left(q_{i}\left(\boldsymbol{\theta}_{1,\tau}\right),e_{i}\left(\boldsymbol{\theta}_{2,\tau}\right),Y_{i}\right)+o_{p}(1),
	\end{eqnarray*}
	where
	\begin{eqnarray*}
		K_{i}\left(\boldsymbol{\gamma}_{0}\right) & = & 1-\frac{D_{i}\left(1-Z_{i}\right)}{1-\pi_{0i}}-\frac{\left(1-D_{i}\right)Z_{i}}{\pi_{0i}},\\
		\nabla_{\boldsymbol{\gamma}}K_{i}\left(\boldsymbol{\gamma}_{0}\right) & = & \left[\frac{\left(1-D_{i}\right)Z_{i}}{\pi_{0i}^{2}}-\frac{D_{i}\left(1-Z_{i}\right)}{\left(1-\pi_{0i}\right)^{2}}\right]\nabla_{\boldsymbol{\gamma}}\bar{\pi}_{i}.
	\end{eqnarray*}
	It further can be shown that
	\[
	\frac{1}{n}\sum_{i=1}^{n}\left[\nabla_{\boldsymbol{\gamma}}K_{i}\left(\boldsymbol{\gamma}_{0}\right)\right]^\top\sqrt{n}\left(\hat{\boldsymbol{\gamma}}-\boldsymbol{\gamma}_{0}\right)\nabla_{\boldsymbol{\theta}}FZ_{\tau}\left(q_{i}\left(\boldsymbol{\theta}_{1,\tau}\right),e_{i}\left(\boldsymbol{\theta}_{2,\tau}\right),Y_{i}\right)\overset{p.}{\rightarrow}\frac{1}{\sqrt{n}}\sum_{i=1}^{n}\mathbf{M}_{\tau}\psi\left(X_{i}\right),
	\]
	where 
	\[
	\mathbf{M}_{\tau}=E\left[\nabla_{\boldsymbol{\theta}}FZ_{\tau}\left(q_{i}\left(\boldsymbol{\theta}_{1,\tau}\right),e_{i}\left(\boldsymbol{\theta}_{2,\tau}\right),Y_{i}\right)\left[\nabla_{\boldsymbol{\gamma}}K_{i}\left(\boldsymbol{\gamma}_{0}\right)\right]^\top\right].
	\]
	Therefore we can conclude that
	\begin{eqnarray*}
		\sqrt{n}\left(\left(\hat{\boldsymbol{\theta}}_{1,\tau},\hat{\boldsymbol{\theta}}_{2,\tau}\right)-\left(\boldsymbol{\theta}_{1,\tau},\boldsymbol{\theta}_{2,\tau}\right)\right)& = & -\mathbf{H}_{\tau}^{-1}\left\{ \frac{1}{\sqrt{n}}\sum_{i=1}^{n}K_{i}\left(\boldsymbol{\gamma}_{0}\right)\nabla_{\boldsymbol{\theta}}FZ_{\tau}\left(q_{i}\left(\boldsymbol{\theta}_{1,\tau}\right),e_{i}\left(\boldsymbol{\theta}_{2,\tau}\right),Y_{i}\right)+\right.\\
		&  & +\left.\frac{1}{\sqrt{n}}\sum_{i=1}^{n}\mathbf{M}_{\tau}\psi\left(X_{i}\right)\right\} +o_{p}\left(1\right),
	\end{eqnarray*}
	and 
	\[
	\sqrt{n}\left(\left(\hat{\boldsymbol{\theta}}_{1,\tau},\hat{\boldsymbol{\theta}}_{2,\tau}\right)-\left(\boldsymbol{\theta}_{1,\tau},\boldsymbol{\theta}_{2,\tau}\right)\right)\stackrel{d.}{\rightarrow}N\left(\mathbf{0},\mathbf{H}_{\tau}^{-1}\boldsymbol{\Omega}_{\tau}\mathbf{H}_{\tau}^{-1}\right)
	\]
	where $\boldsymbol{\Omega}_{\tau}=E\left[\mathbf{J}_{\tau}\mathbf{J}_{\tau}^\top\right]$
	and 
	\begin{eqnarray*}
		\mathbf{J}_{\tau} & = & K\left(\boldsymbol{\gamma}_{0}\right)\nabla_{\boldsymbol{\theta}}FZ_{\tau}\left(q\left(\boldsymbol{\theta}_{1,\tau}\right),e\left(\boldsymbol{\theta}_{2,\tau}\right),Y\right)+\mathbf{M}_{\tau}\psi\left(X\right).
	\end{eqnarray*}
\end{proof} 

\begin{proof}[Proof of Theorem 3]
	At first, we show that the empirical loss function converges to the
	true loss function uniformly over $\mathcal{T}\times\boldsymbol{\Theta}$, which is a compact metric space. We use results of \citet{Newey_1991}, which requires Assumption
	5.1, and the following conditions: (i) pointwise convergence
	between the empirical loss function and the true loss function for
	each $\left(\tau,\boldsymbol{\theta}\right)\in\mathcal{T}\times\boldsymbol{\Theta}$,
	(ii) $n^{-1}\sum_{i=1}^{n}\tilde{K}_{i}\left(\hat{\nu},\hat{\boldsymbol{\gamma}}\right)h_{i}\left(\tau,\boldsymbol{\theta}\right)$
	is stochastic equicontinuous on $\left(\tau,\boldsymbol{\theta}\right)$,
	and (iii) $E\left[\bar{K}\left(v_{0},\boldsymbol{\gamma}_{0}\right)h\left(\tau,\boldsymbol{\theta}\right)\right]$
	is equicontinuous. 
	
	For (i), we have shown in Theorem 1 that if Assumptions 2.1, 2.2 and 2.4-2.8 in Theorem 1 and Assumption 5.1 (replaces Assumption 2.3) hold, for each $\tau\in\mathcal{T},$ 
	\begin{equation}
		\sup_{\boldsymbol{\theta}\in\boldsymbol{\Theta}}\left|\frac{1}{n}\sum_{i=1}^{n}\tilde{K}_{i}\left(\hat{\nu},\hat{\boldsymbol{\gamma}}\right)h_{i}\left(\tau,\boldsymbol{\theta}\right)-E\left[\bar{K}\left(v_{0},\boldsymbol{\gamma}_{0}\right)h\left(\tau,\boldsymbol{\theta}\right)\right]\right|=o_{p}\left(1\right).\label{pc}
	\end{equation}
	Thus for each $\left(\tau,\boldsymbol{\theta}\right)\in\mathcal{T}\times\boldsymbol{\Theta}$,
	the pointwise convergence of (\ref{pc}) also holds. For (ii), at first let $\boldsymbol{\theta}^{1}:=\left(\boldsymbol{\theta}_{1}^{1},\boldsymbol{\theta}_{2}^{1}\right)$
	and $\boldsymbol{\theta}^{2}:=\left(\boldsymbol{\theta}_{1}^{2},\boldsymbol{\theta}_{2}^{2}\right)$ be two vectors of parameter values. Then with Assumptions 5.1 and
	2.4, using the result of Lemma 3 in Appendix A.4, we can have
	\begin{eqnarray*}
		\left|\frac{1}{n}\sum_{i=1}^{n}\tilde{K}_{i}\left(\hat{\nu},\hat{\boldsymbol{\gamma}}\right)h_{i}\left(\tau_{1},\boldsymbol{\theta}^{1}\right)-\frac{1}{n}\sum_{i=1}^{n}\tilde{K}_{i}\left(\hat{\nu},\hat{\boldsymbol{\gamma}}\right)h_{i}\left(\tau_{2},\boldsymbol{\theta}^{2}\right)\right| \\ \leq  \frac{1}{n}\sum_{i=1}^{n}\left|\left(\bar{K}_{i}\left(v_{0},\boldsymbol{\gamma}_{0}\right)+o_{p}\left(1\right)\right)\left(h_{i}\left(\tau_{1},\boldsymbol{\theta}^{1}\right)-h_{i}\left(\tau_{2},\boldsymbol{\theta}^{2}\right)\right)\right|\\
		\leq  \frac{1}{n}\sum_{i=1}^{n}\left|\bar{K}_{i}\left(v_{0},\boldsymbol{\gamma}_{0}\right)+o_{p}\left(1\right)\right|\left|h_{i}\left(\tau_{1},\boldsymbol{\theta}^{1}\right)-h_{i}\left(\tau_{2},\boldsymbol{\theta}^{2}\right)\right|\\
		\leq \left(\left\Vert \tau_{1}-\tau_{2}\right\Vert +\left\Vert \boldsymbol{\theta}^{1}-\boldsymbol{\theta}^{2}\right\Vert \right)\left[\frac{1}{n}\sum_{i=1}^{n}\left|\bar{K}_{i}\left(v_{0},\boldsymbol{\gamma}_{0}\right)+o_{p}\left(1\right)\right|B^{*}\left(V_{i}\right)\right],
	\end{eqnarray*}
	where $B^{*}\left(V_{i}\right)$ is the constant defined in Lemma
	3. Notice that $\left|\bar{K}_{i}\left(v_{0},\boldsymbol{\gamma}_{0}\right)+o_{p}\left(1\right)\right|B^{*}\left(V_{i}\right)\geq0$
	for every $i=1,\ldots,n$, and $B^{*}\left(V_{i}\right)$ is determined
	by a sum of $\ln2$ and constants $B_{k}\left(V\right)$, $k=1,\ldots,5$.
	With Assumption 5.3, it can be shown that
	\begin{eqnarray*}
		\frac{1}{n}\sum_{i=1}^{n}\left|\bar{K}_{i}\left(v_{0},\boldsymbol{\gamma}_{0}\right)+o_{p}\left(1\right)\right|B^{*}\left(V_{i}\right) & \leq & \frac{1}{n}\sum_{i=1}^{n}\left(\left|\bar{K}_{i}\left(v_{0},\boldsymbol{\gamma}_{0}\right)\right|+\left|o_{p}\left(1\right)\right|\right)B^{*}\left(V_{i}\right)\\
		& = & \frac{1}{n}\sum_{i=1}^{n}\left|\bar{K}_{i}\left(v_{0},\boldsymbol{\gamma}_{0}\right)\right|B^{*}\left(V_{i}\right)+\left(\frac{1}{n}\sum_{i=1}^{n}\left|o_{p}\left(1\right)\right|B^{*}\left(V_{i}\right)\right)\\
		& \stackrel{p.}{\rightarrow} & E\left[B^{*}\left(V_{i}\right)|T=c\right]P\left(T=c\right)\\
		& < & \infty.
	\end{eqnarray*}
	Therefore we can conclude that $n^{-1}\sum_{i=1}^{n}\left|\bar{K}_{i}\left(v_{0},\boldsymbol{\gamma}_{0}\right)+o_{p}\left(1\right)\right|B^{*}\left(V_{i}\right)$
	is $O_{p}\left(1\right)$ and the empirical loss $n^{-1}\sum_{i=1}^{n}\tilde{K}_{i}\left(\hat{\nu},\hat{\boldsymbol{\gamma}}\right)h_{i}\left(\tau,\boldsymbol{\theta}\right)$
	is stochastic equicontinuous on $\left(\tau,\boldsymbol{\theta}\right)\in\mathcal{T}\times\boldsymbol{\Theta}$.
	Finally, for (iii), with Assumption 5.2, $E\left[\bar{K}\left(v_{0},\boldsymbol{\gamma}_{0}\right)h\left(\tau,\boldsymbol{\theta}\right)\right]$
	is equicontinuous since its first order derivatives w.r.t. $\left(\tau,\boldsymbol{\theta}\right)\in\mathcal{T}\times\boldsymbol{\Theta}$
	is bound. Combining the above results, with Corollary 2.1 of \citet{Newey_1991},
	we can conclude that
	\[
	\sup_{\left(\tau,\boldsymbol{\theta}\right)\in\mathcal{T}\times\boldsymbol{\Theta}}\left|\frac{1}{n}\sum_{i=1}^{n}\tilde{K}_{i}\left(\hat{\nu},\hat{\boldsymbol{\gamma}}\right)h_{i}\left(\tau,\boldsymbol{\theta}\right)-E\left[\bar{K}\left(v_{0},\boldsymbol{\gamma}_{0}\right)h\left(\tau,\boldsymbol{\theta}\right)\right]\right|=o_{p}\left(1\right).
	\]
	Then uniformly over $\tau\in\mathcal{T}$, for any $\varepsilon>0$,
	with probability approaching one, we have 
	\begin{eqnarray}
		\frac{1}{n}\sum_{i=1}^{n}\tilde{K}_{i}\left(\hat{\nu},\hat{\boldsymbol{\gamma}}\right)h_{i}\left(\tau,\hat{\boldsymbol{\theta}}_{\tau}\right) & < & \frac{1}{n}\sum_{i=1}^{n}\tilde{K}_{i}\left(\hat{\nu},\hat{\boldsymbol{\gamma}}\right)h_{i}\left(\tau,\boldsymbol{\theta}_{\tau}\right)+\frac{\varepsilon}{3},\label{uc}\\
		E\left[\bar{K}\left(v_{0},\boldsymbol{\gamma}_{0}\right)h\left(\tau,\hat{\boldsymbol{\theta}}_{\tau}\right)\right] & < & \frac{1}{n}\sum_{i=1}^{n}\tilde{K}_{i}\left(\hat{\nu},\hat{\boldsymbol{\gamma}}\right)h_{i}\left(\tau,\hat{\boldsymbol{\theta}}_{\tau}\right)+\frac{\varepsilon}{3},\label{uc1}\\
		\frac{1}{n}\sum_{i=1}^{n}\tilde{K}_{i}\left(\hat{\nu},\hat{\boldsymbol{\gamma}}\right)h_{i}\left(\tau,\boldsymbol{\theta}_{\tau}\right) & < & E\left[\bar{K}\left(v_{0},\boldsymbol{\gamma}_{0}\right)h\left(\tau,\boldsymbol{\theta}_{\tau}\right)\right]+\frac{\varepsilon}{3}.\label{uc2}
	\end{eqnarray}
	The inequality (\ref{uc}) holds because $\hat{\boldsymbol{\theta}}_{\tau}$
	minimizes $n^{-1}\sum_{i=1}^{n}\bar{K}_{i}\left(\hat{\nu},\hat{\boldsymbol{\gamma}}\right)h_{i}\left(\tau,\boldsymbol{\theta}\right)$
	and the inequalities (\ref{uc1}) and (\ref{uc2}) hold because of
	(\ref{pc}). With (\ref{uc}) , (\ref{uc1}) and (\ref{uc2}), we
	can conclude that 
	\[
	E\left[\bar{K}\left(v_{0},\boldsymbol{\gamma}_{0}\right)h\left(\tau,\hat{\boldsymbol{\theta}}_{\tau}\right)\right]<E\left[\bar{K}\left(v_{0},\boldsymbol{\gamma}_{0}\right)h\left(\tau,\boldsymbol{\theta}_{\tau}\right)\right]+\varepsilon
	\]
	holds uniformly over $\tau\in\mathcal{T}$ with probability approaching
	one. Let $\mathcal{M}$ be any open set in $\boldsymbol{\Theta}$
	containing $\boldsymbol{\theta}_{\tau}:=\left(\boldsymbol{\theta}_{1,\tau},\boldsymbol{\theta}_{2,\tau}\right)$. Notice that $\mathcal{M}^{c}\cap\boldsymbol{\Theta}$ is compact. Also $E\left[\bar{K}\left(v_{0},\boldsymbol{\gamma}_{0}\right)h\left(\tau,\boldsymbol{\theta}\right)\right]$
	is continuous in $\boldsymbol{\theta}$ by Assumption 2.4 and is uniquely
	minimized at $\boldsymbol{\theta}_{\tau}$ by Assumptions 2.5 and
	2.6. Therefore
	\[
	\inf_{\boldsymbol{\theta}\in\mathcal{M}^{c}\cap\boldsymbol{\Theta}}E\left[\bar{K}\left(v_{0},\boldsymbol{\gamma}_{0}\right)h\left(\tau,\boldsymbol{\theta}\right)\right]=E\left[\bar{K}\left(v_{0},\boldsymbol{\gamma}_{0}\right)h\left(\tau,\boldsymbol{\theta}^{*}\right)\right]>E\left[\bar{K}\left(v_{0},\boldsymbol{\gamma}_{0}\right)h\left(\tau,\boldsymbol{\theta}_{\tau}\right)\right],
	\]
	for some $\boldsymbol{\theta}^{*}\in\mathcal{M}^{c}\cap\boldsymbol{\Theta}$.
	Then setting \[\varepsilon=\inf_{\boldsymbol{\theta}\in\mathcal{M}^{c}\cap\boldsymbol{\Theta}}E\left[\bar{K}\left(v_{0},\boldsymbol{\gamma}_{0}\right)h\left(\tau,\boldsymbol{\theta}\right)\right]-E\left[\bar{K}\left(v_{0},\boldsymbol{\gamma}_{0}\right)h\left(\tau,\boldsymbol{\theta}_{\tau}\right)\right],\]
	we can have 
	\[
	E\left[\bar{K}\left(v_{0},\boldsymbol{\gamma}_{0}\right)h\left(\tau,\hat{\boldsymbol{\theta}}_{\tau}\right)\right]<\inf_{\boldsymbol{\theta}\in\mathcal{M}^{c}\cap\boldsymbol{\Theta}}E\left[\bar{K}\left(v_{0},\boldsymbol{\gamma}_{0}\right)h\left(\tau,\boldsymbol{\theta}\right)\right]
	\]
	holds uniformly over $\tau\in\mathcal{T}$ with probability approaching
	one. Thus $\hat{\boldsymbol{\theta}}_{\tau}\in\mathcal{M}$ uniformly
	over $\tau\in\mathcal{T}$ with probability approaching one.
\end{proof}

\begin{proof}[Proof of Theorem A.1]
	It is straightforward to show that $n^{-1}\sum_{i=1}^{n}\hat{\mathbf{J}}_{i,\tau}^{sp}\hat{\mathbf{J}}_{i,\tau}^{sp\top}\overset{p.}{\rightarrow}\boldsymbol{\Omega}_{\tau}^{sp}$
	when the assumptions hold. We thus focus on proving $\hat{\mathbf{H}}_{\tau}^{sp}\overset{p.}{\rightarrow}\mathbf{H}_{\tau}^{sp}$.
	We show that $\mathbf{\hat{C}}_{\tau,11}^{sp*}\overset{p.}{\rightarrow}\mathbf{C}_{\tau,11}^{sp*}$
	and similar arguments can be used to prove $\mathbf{\hat{C}}_{\tau,22}^{sp*}\overset{p.}{\rightarrow}\mathbf{C}_{\tau,22}^{sp*}$.
	To see this, notice that
	\[
	\left\Vert \mathbf{\hat{C}}_{\tau,11}^{sp*}-\mathbf{C}_{\tau,11}^{sp*}\right\Vert \leq I_{1}+I_{2}+I_{3}+I_{4}+I_{5},
	\]
	where
	\begin{eqnarray*}
		I_{1} & = & \left\Vert \mathbf{\hat{C}}_{\tau,11}^{sp*}-\frac{1}{n}\sum_{i=1}^{n}\left\{ \bar{K}_{i}\left(v_{0},\gamma_{0}\right)\varsigma_{\lambda,i}\left(\hat{\boldsymbol{\theta}}{}_{1,\tau}\right)\frac{1}{\tau}G_{sp}^{\prime}\left(e\left(\hat{\boldsymbol{\theta}}_{2,\tau}\right)\right)W_{i}W_{i}^\top\right\} \right\Vert ,\\
		I_{2} & = & \left\Vert \frac{1}{n}\sum_{i=1}^{n}\left\{ \bar{K}_{i}\left(v_{0},\gamma_{0}\right)\varsigma_{\lambda_{n},i}\left(\hat{\boldsymbol{\theta}}{}_{1,\tau}\right)\frac{1}{\tau}G_{sp}^{\prime}\left(e_{i}\left(\hat{\boldsymbol{\theta}}_{2,\tau}\right)\right)W_{i}W_{i}^\top\right\} \right.\\
		&  & \left.-\frac{1}{n}\sum_{i=1}^{n}\left\{ \bar{K}_{i}\left(v_{0},\gamma_{0}\right)\varsigma_{\lambda_{n},i}\left(\boldsymbol{\theta}{}_{1,\tau}\right)\frac{1}{\tau}G_{sp}^{\prime}\left(e_{i}\left(\hat{\boldsymbol{\theta}}_{2,\tau}\right)\right)W_{i}W_{i}^\top\right\} \right\Vert ,\\
		I_{3} & = & \left\Vert \frac{1}{n}\sum_{i=1}^{n}\left\{ \bar{K}_{i}\left(v_{0},\gamma_{0}\right)\varsigma_{\lambda_{n},i}\left(\boldsymbol{\theta}{}_{1,\tau}\right)\frac{1}{\tau}G_{sp}^{\prime}\left(e_{i}\left(\hat{\boldsymbol{\theta}}_{2,\tau}\right)\right)W_{i}W_{i}^\top\right\} \right.\\
		&  & \left.-\frac{1}{n}\sum_{i=1}^{n}\left\{ \bar{K}_{i}\left(v_{0},\gamma_{0}\right)\varsigma_{\lambda_{n},i}\left(\boldsymbol{\theta}{}_{1,\tau}\right)\frac{1}{\tau}G_{sp}^{\prime}\left(e_{i}\left(\boldsymbol{\theta}_{2,\tau}\right)\right)W_{i}W_{i}^\top\right\} \right\Vert ,\\
		I_{4} & = & \left\Vert \frac{1}{n}\sum_{i=1}^{n}\left\{ \bar{K}_{i}\left(v_{0},\gamma_{0}\right)\varsigma_{\lambda_{n},i}\left(\boldsymbol{\theta}{}_{1,\tau}\right)\frac{1}{\tau}G_{sp}^{\prime}\left(e_{i}\left(\boldsymbol{\theta}_{2,\tau}\right)\right)W_{i}W_{i}^\top\right\} \right.\\
		&  & -\left.E\left[\bar{K}\left(v_{0},\gamma_{0}\right)\varsigma_{\lambda_{n}}\left(\boldsymbol{\theta}_{1,\tau}\right)\frac{1}{\tau}G_{sp}^{\prime}\left(e\left(\boldsymbol{\theta}_{2,\tau}\right)\right)WW^\top\right]\right\Vert,\\
		I_{5} & = & \left\Vert E\left[\bar{K}\left(v_{0},\gamma_{0}\right)\varsigma_{\lambda_{n}}\left(\boldsymbol{\theta}_{1,\tau}\right)\frac{1}{\tau}G_{sp}^{\prime}\left(e\left(\boldsymbol{\theta}_{2,\tau}\right)\right)WW^\top\right]-\mathbf{C}_{\tau,11}^{sp*}\right\Vert .
	\end{eqnarray*}
	For $I_{1}$, it can be shown that 
	\begin{eqnarray*}
		I_{1} & \leq & \sup\left|\tilde{K}_{i}\left(\hat{v},\hat{\gamma}\right)-\bar{K}_{i}\left(v_{0},\gamma_{0}\right)\right|\\
		&  & \times\frac{1}{n}\sum_{i=1}^{n}\left|\varsigma_{\lambda,i}\left(\hat{\boldsymbol{\theta}}{}_{1,\tau}\right)\frac{1}{\tau}G_{sp}^{\prime}\left(e\left(\hat{\boldsymbol{\theta}}_{2,\tau}\right)\right)\times W_{i}W_{i}^\top\right|\\
		& \leq & \sup\left|\bar{K}_{i}\left(\hat{v},\hat{\gamma}\right)-\bar{K}_{i}\left(v_{0},\gamma_{0}\right)\right|\times\\
		&  & \frac{1}{n}\sum_{i=1}^{n}\left|\varsigma_{\lambda,i}\left(\hat{\boldsymbol{\theta}}{}_{1,\tau}\right)\frac{1}{\tau}G_{sp}^{\prime}\left(e\left(\hat{\boldsymbol{\theta}}_{2,\tau}\right)\right)\times W_{i}W_{i}^\top\right|\\
		& = & o_{p}\left(1\right)
	\end{eqnarray*}
	Next, convergence between $\varsigma_{\lambda_{n},i}\left(\hat{\boldsymbol{\theta}}{}_{1,\tau}\right)$
	and $\varsigma_{\lambda_{n},i}\left(\boldsymbol{\theta}{}_{1,\tau}\right)$
	can be constructed by using the fact that 
	\[\frac{1}{\lambda_{n}}\left\Vert \hat{\boldsymbol{\theta}}_{1,\tau}-\boldsymbol{\theta}{}_{1,\tau}\right\Vert =o_{p}\left(1\right),
	\]
	and similar arguments in \citet{Powell_1984} and \citet{EM_2004}.
	Together with the facts that $\frac{1}{\tau}G_{sp}^{\prime}\left(e\left(\hat{\boldsymbol{\theta}}_{2,\tau}\right)\right)$,
	$\varsigma_{\lambda_{n},i}\left(.\right)$ and $E\left[\left\Vert WW^\top\right\Vert |T=c\right]$
	are bounded and $G_{sp}^{\prime}\left(e_{i}\left(\boldsymbol{\theta}_{2}\right)\right)$
	is continuous on the parameters $\boldsymbol{\theta}_{2}$, it can
	be shown that $I_{2}=o_{p}\left(1\right)$ and $I_{3}=o_{p}\left(1\right)$.
	For $I_{4}$, since the samples are i.i.d. and $\bar{K}\left(v_{0},\gamma_{0}\right)$,
	$\varsigma_{\lambda_{n},i}\left(\boldsymbol{\theta}{}_{1,\tau}\right)$,
	$\frac{1}{\tau}G_{sp}^{\prime}\left(e_{i}\left(\boldsymbol{\theta}_{2,\tau}\right)\right)$
	and $E\left[\left\Vert W_{i}W_{i}^\top\right\Vert |T=c\right]$ are
	bounded, it is easy to see that $I_{4}=o_{p}\left(1\right)$. Let $Y-q\left(\boldsymbol{\theta}_{1,\tau}\right)=\varepsilon_{\tau}^{q}$
	and then $f_{Y|W,T=c}\left(q\left(\boldsymbol{\theta}_{1,\tau}\right)\right)=f_{\varepsilon_{\tau}^{q}|W,T=c}\left(0\right)$. %since $q\left(\boldsymbol{\theta}_{1,\tau}\right):=q\left(W,\boldsymbol{\theta}_{1,\tau}\right)$. 
	For $I_{5}$, at first notice that as $\lambda_{n}=o\left(1\right)$,
	\begin{eqnarray*}
		\left|E\left[\varsigma_{\lambda_{n}}\left(\boldsymbol{\theta}_{1,\tau}\right)|W,T=c\right]-f_{\varepsilon_{\tau}^{q}|W,T=c}\left(0\right)\right| & = & \left|E\left[\frac{1\left\{ \left|Y_{i}-q_{i}\left(\boldsymbol{\theta}_{1,\tau}\right)\right|\leq\lambda_{n}\right\} }{2\lambda_{n}}.|W,T=co\right]-f_{\varepsilon_{\tau}^{q}|W,T=co}\left(0\right)\right|\\
		& = & \left|\frac{1}{2\lambda_{n}}\int_{-\lambda_{n}}^{\lambda_{n}}f_{\varepsilon_{\tau}^{q}|W,T=co}\left(a\right)da-f_{\varepsilon_{\tau}^{q}|W,T=co}\left(0\right)\right|\\
		& \leq & \left|\frac{1}{2\lambda_{n}}\times2\lambda_{n}f_{\varepsilon_{\tau}^{q}|W,T=co}\left(a^{*}\right)-f_{\varepsilon_{\tau}^{q}|W,T=co}\left(0\right)\right|,\\
		& = & o_{p}(1),
	\end{eqnarray*}
	where $f_{\varepsilon_{\tau}^{q}|W,T=co}\left(a^{*}\right)=\max_{a\in\left[-\lambda_{n},\lambda_{n}\right]}f_{\varepsilon_{\tau}^{q}|W,T=co}\left(a\right)$. Thus 
	\[
	\lim_{\lambda_{n}\rightarrow0}E\left[\varsigma_{\lambda_{n}}\left(\boldsymbol{\theta}_{1,\tau}\right)|W,T=c\right]=f_{Y|W,T=c}\left(q\left(\boldsymbol{\theta}_{1,\tau}\right)\right).
	\]Furthermore,
	\begin{eqnarray*}
		\lim_{\lambda_{n}\rightarrow0}E\left[\bar{K}\left(v_{0},\gamma_{0}\right)\varsigma_{\lambda_{n}}\left(\boldsymbol{\theta}_{1,\tau}\right)\frac{1}{\tau}G_{sp}^{\prime}\left(e\left(\boldsymbol{\theta}_{2,\tau}\right)\right)WW^\top\right] & = & \lim_{\lambda_{n}\rightarrow0}E\left[\varsigma_{\lambda_{n}}\left(\boldsymbol{\theta}_{1,\tau}\right)\frac{1}{\tau}G_{sp}^{\prime}\left(e\left(\boldsymbol{\theta}_{2,\tau}\right)\right)WW^\top|T=c\right]\times P\left(T=c\right)\\
		& = & E\left[\lim_{\lambda_{n}\rightarrow0}E\left[\varsigma_{\lambda_{n}}\left(\boldsymbol{\theta}_{1,\tau}\right)|W,T=c\right]\frac{1}{\tau}G_{sp}^{\prime}\left(e\left(\boldsymbol{\theta}_{2,\tau}\right)\right)\right.\\
		& &\left.\times WW^\top|T=c\right]\times P\left(T=c\right)\\
		& = & E\left[f_{Y|W,T=c}\left(q\left(\boldsymbol{\theta}_{1,\tau}\right)\right)\frac{1}{\tau}G_{sp}^{\prime}\left(e\left(\boldsymbol{\theta}_{2,\tau}\right)\right)\right.\\
		&   &\left.\times WW^\top|T=c\right]\times P\left(T=c\right)\\
		& = & E\left[\bar{K}\left(v_{0},\gamma_{0}\right)f_{Y|W,T=co}\left(q\left(\boldsymbol{\theta}_{1,\tau}\right)\right)\frac{1}{\tau}G_{sp}^{\prime}\left(e\left(\boldsymbol{\theta}_{2,\tau}\right)\right)WW^\top\right]\\
		& = & C_{\tau,11}^{sp*}.
	\end{eqnarray*}
	Thus $I_{5}=o\left(1\right)$. Combining the above results, we can
	conclude that $\mathbf{\hat{C}}_{\tau,11}^{sp*}\overset{p.}{\rightarrow}\mathbf{C}_{\tau,11}^{sp*}$.
\end{proof}

\subsection{Local Lipschitz continuity.} 
Let $\boldsymbol{\vartheta}_{\boldsymbol{\theta}_{1}}\left(\tau\right):=\left\{ \boldsymbol{\theta}_{1}:\left\Vert \boldsymbol{\theta}_{1}-\boldsymbol{\theta}_{1,\tau}\right\Vert <\varepsilon_{1}\right\} $
and $\boldsymbol{\vartheta}_{\boldsymbol{\theta}_{2}}\left(\tau\right):=\left\{ \boldsymbol{\theta}_{2}:\left\Vert \boldsymbol{\theta}_{2}-\boldsymbol{\theta}_{2,\tau}\right\Vert <\varepsilon_{2}\right\} $
be neighborhoods of $\boldsymbol{\theta}_{1,\tau}$ and
$\boldsymbol{\theta}_{2,\tau}$, and $\boldsymbol{\vartheta}_{\boldsymbol{\theta}}\left(\tau\right)=\boldsymbol{\vartheta}_{\boldsymbol{\theta}_{1}}\left(\tau\right)\cap\boldsymbol{\vartheta}_{\boldsymbol{\theta}_{2}}\left(\tau\right)$.
Let $\boldsymbol{\theta}_{1}^{1}$ and $\boldsymbol{\theta}_{1}^{2}$
be two different vectors of parameter values in $\boldsymbol{\vartheta}_{\boldsymbol{\theta}_{1}}\left(\tau\right)$
and $\boldsymbol{\theta}_{2}^{1}$ and $\boldsymbol{\theta}_{2}^{2}$
be two different vectors of parameter values in $\boldsymbol{\vartheta}_{\boldsymbol{\theta}_{2}}\left(\tau\right)$.
Let $\boldsymbol{\theta}^{1}:=\left(\boldsymbol{\theta}_{1}^{1},\boldsymbol{\theta}_{2}^{1}\right)$
and $\boldsymbol{\theta}^{2}:=\left(\boldsymbol{\theta}_{1}^{2},\boldsymbol{\theta}_{2}^{2}\right)$. %and $\boldsymbol{\theta}^{1}$ and $\boldsymbol{\theta}^{2}$ are two parameter vectors in $\boldsymbol{\vartheta}_{\boldsymbol{\theta}}\left(\tau\right)$.
Let
\begin{eqnarray*}
	A_{1,\tau}\left(V\right) & := & \sup_{\boldsymbol{\theta}_{1}\in\boldsymbol{\vartheta}_{\boldsymbol{\theta}_{1}}\left(\tau\right)}\left\Vert G_{1}^{\prime}\left(q\left(\boldsymbol{\theta}_{1}\right)\right)\nabla_{\boldsymbol{\theta}_{1}}q\left(\boldsymbol{\theta}_{1}\right)\right\Vert,\\
	A_{2,\tau}\left(V\right) & := & \sup_{\boldsymbol{\theta}\in\boldsymbol{\vartheta}_{\boldsymbol{\theta}}\left(\tau\right)}\left\Vert G_{2}^{\prime\prime}\left(e\left(\boldsymbol{\theta}_{2}\right)\right)q\left(\boldsymbol{\theta}_{1}\right)\nabla_{\boldsymbol{\theta}_{2}}e\left(\boldsymbol{\theta}_{2}\right)\right\Vert,\\
	A_{3,\tau}\left(V\right) & := & \sup_{\boldsymbol{\theta}\in\boldsymbol{\vartheta}_{\boldsymbol{\theta}}\left(\tau\right)}\left\Vert G_{2}^{\prime}\left(e\left(\boldsymbol{\theta}_{2}\right)\right)\nabla_{\boldsymbol{\theta}_{1}}q\left(\boldsymbol{\theta}_{1}\right)\right\Vert,\\
	A_{4,\tau}\left(V\right) & := & \sup_{\boldsymbol{\theta}_{2}\in\boldsymbol{\vartheta}_{\boldsymbol{\theta}_{2}}\left(\tau\right)}\left\Vert YG_{2}^{\prime\prime}\left(e\left(\boldsymbol{\theta}_{2}\right)\right)\nabla_{\boldsymbol{\theta}_{2}}e\left(\boldsymbol{\theta}_{2}\right)\right\Vert,\\
	A_{5,\tau}\left(V\right) & := & \sup_{\boldsymbol{\theta}\in\boldsymbol{\vartheta}_{\boldsymbol{\theta}}\left(\tau\right)}\left\Vert G_{2}^{\prime\prime}\left(e\left(\boldsymbol{\theta}_{2}\right)\right)\left(e\left(\boldsymbol{\theta}_{2}\right)-q\left(\boldsymbol{\theta}_{1}\right)\right)\nabla_{\boldsymbol{\theta}_{2}}e\left(\boldsymbol{\theta}_{2}\right)\right\Vert,
\end{eqnarray*}
where $V=\left(Y,W\right)$.

\begin{lemma}[Local Lipschitz continuity] 
	If Assumptions 2.4 and 2.5 hold, for all $\left(\boldsymbol{\theta}_{1}^{1},\boldsymbol{\theta}_{1}^{2}\right)\in\boldsymbol{\vartheta}_{\boldsymbol{\theta}_{1}}\left(\tau\right)$
	and $\left(\boldsymbol{\theta}_{2}^{1},\boldsymbol{\theta}_{2}^{2}\right)\in\boldsymbol{\vartheta}_{\boldsymbol{\theta}_{2}}\left(\tau\right)$, given $V=\left(Y,W\right)$, there exists a constant $A_{\tau}^{*}\left(V\right)$ such that
	\[
	\left|FZ_{\tau}\left(q\left(\boldsymbol{\theta}_{1}^{1}\right),e_{1}\left(\boldsymbol{\theta}_{2}^{1}\right),Y\right)-FZ_{\tau}\left(q\left(\boldsymbol{\theta}_{1}^{2}\right),e_{1}\left(\boldsymbol{\theta}_{2}^{2}\right),Y\right)\right|\leq A_{\tau}^{*}\left(V\right)\left\Vert \boldsymbol{\theta}^{1}-\boldsymbol{\theta}^{2}\right\Vert.
	\]
	%where $A_{\tau}^{*}\left(V\right)$ is the Lipschitz constant and $V=\left(Y,W\right)$.
\end{lemma}
\begin{proof}
	To ease the notation, we let $q_{1}:=q\left(\boldsymbol{\theta}_{1}^{1}\right)$,
	$q_{2}:=q\left(\boldsymbol{\theta}_{1}^{2}\right)$, $e_{1}:=e\left(\boldsymbol{\theta}_{2}^{1}\right)$,
	and $e_{2}:=q\left(\boldsymbol{\theta}_{2}^{2}\right)$. We rewrite $FZ_{\tau}\left(q,e,y\right)$ as
	\begin{eqnarray*}
		FZ_{\tau}\left(q,e,y\right) & = & 0.5\left[\left|G_{1}\left(q\right)-G_{1}\left(y\right)\right|+G_{1}\left(q\right)-G_{1}\left(y\right)\right]-\tau G_{1}\left(q\right)\\
		&  & +\frac{1}{\tau}G_{2}^{\prime}\left(e\right)\left\{ 0.5\left[\left|q-y\right|+q-y\right]+\tau\left(e-q\right)\right\} \\
		&  & -G_{2}\left(e\right)+\eta\left(y\right).
	\end{eqnarray*}
	Then 
	\begin{eqnarray*}
		FZ_{\tau}\left(q_{1},e_{1},Y\right)-FZ_{\tau}\left(q_{2},e_{2},Y\right) & = & 0.5\left[\left|G_{1}\left(q_{1}\right)-G_{1}\left(Y\right)\right|+G_{1}\left(q_{1}\right)-\left|G_{1}\left(q_{2}\right)-G_{1}\left(Y\right)\right|-G_{1}\left(q_{2}\right)\right]\nonumber \\
		&  & -\tau\left(G_{1}\left(q_{1}\right)-G_{1}\left(q_{2}\right)\right)\nonumber \\
		&  & +\frac{1}{2\tau}\left(G_{2}^{\prime}\left(e_{1}\right) \left[\left|q_{1}-Y\right|+q_{1}\right]-G_{2}^{\prime}\left(e_{2}\right) \left[\left|q_{2}-Y\right|+q_{2}\right] \right)\nonumber \\
		&  & -\frac{Y}{2\tau}\left[G_{2}^{\prime}\left(e_{1}\right)-G_{2}^{\prime}\left(e_{2}\right)\right]\nonumber \\
		&  & +\left[G_{2}^{\prime}\left(e_{1}\right)\left(e_{1}-q_{1}\right)-G_{2}^{\prime}\left(e_{2}\right)\left(e_{2}-q_{2}\right)\right.\nonumber\\& &
		\left.-\left(G_{2}\left(e_{1}\right)-G_{2}\left(e_{2}\right)\right)\right]\label{difference}
	\end{eqnarray*}
	We next prove that the absolute difference $\left|FZ_{\tau}\left(q_{1},e_{1},Y\right)-FZ_{\tau}\left(q_{2},e_{2},Y\right)\right|$
	is bounded by some functions. Without loss of generality, we assume
	$q_{1}\geq q_{2}$. The results for $q_{1}<q_{2}$ can be proved in
	a similar manner. We divide the proof into three cases.\\
	Cases 1. Suppose $q_{1}\geq q_{2}>Y$. Since $G_{1}\left(.\right)$
	is increasing, 
	\begin{eqnarray*}
		FZ_{\tau}\left(q_{1},e_{1},Y\right)-FZ_{\tau}\left(q_{2},e_{2},Y\right) & = & \left(1-\tau\right)\left[G_{1}\left(q_{1}\right)-G_{1}\left(q_{2}\right)\right]\\
		&  & +\frac{1}{\tau}\left[G_{2}^{\prime}\left(e_{1}\right)q_{1}-G_{2}^{\prime}\left(e_{2}\right)q_{2}\right]\\
		&  & -\frac{Y}{\tau}\left[G_{2}^{\prime}\left(e_{1}\right)-G_{2}^{\prime}\left(e_{2}\right)\right]\\
		&  & +\left[G_{2}^{\prime}\left(e_{1}\right)\left(e_{1}-q_{1}\right)-G_{2}^{\prime}\left(e_{2}\right)\left(e_{2}-q_{2}\right)\right.\\& &
		\left.-\left(G_{2}\left(e_{1}\right)-G_{2}\left(e_{2}\right)\right)\right].
	\end{eqnarray*}And
	\begin{eqnarray}
		\left|FZ_{\tau}\left(q_{1},e_{1},Y\right)-FZ_{\tau}\left(q_{2},e_{2},Y\right)\right| & \leq & \left|G_{1}\left(q_{1}\right)-G_{1}\left(q_{2}\right)\right|\nonumber \\
		&  & +\frac{1}{\tau}\left|G_{2}^{\prime}\left(e_{1}\right)q_{1}-G_{2}^{\prime}\left(e_{2}\right)q_{2}\right|\nonumber \\
		&  & +\frac{1}{\tau}\left|Y\left[G_{2}^{\prime}\left(e_{1}\right)-G_{2}^{\prime}\left(e_{2}\right)\right]\right|\label{inequality} \\
		&  & +\left|G_{2}^{\prime}\left(e_{1}\right)\left(e_{1}-q_{1}\right)-G_{2}^{\prime}\left(e_{2}\right)\left(e_{2}-q_{2}\right)\right.\nonumber\\& &
		\left.-\left(G_{2}\left(e_{1}\right)-G_{2}\left(e_{2}\right)\right)\right|,\nonumber
	\end{eqnarray}
	Case 2. Suppose $q_{1}\geq Y\geq q_{2}$. Since $G_{1}\left(.\right)$
	is increasing, 
	\begin{eqnarray*}
		FZ_{\tau}\left(q_{1},e_{1},Y\right)-FZ_{\tau}\left(q_{2},e_{2},Y\right) & = & G_{1}\left(q_{1}\right)-G_{1}\left(Y\right)\\
		&  & -\tau\left(G_{1}\left(q_{1}\right)-G_{1}\left(q_{2}\right)\right)\\
		&  & +\frac{1}{\tau}\left[G_{2}^{\prime}\left(e_{1}\right)q_{1}-G_{2}^{\prime}\left(e_{1}\right)Y\right]\\
		&  & +\left[G_{2}^{\prime}\left(e_{1}\right)\left(e_{1}-q_{1}\right)-G_{2}^{\prime}\left(e_{2}\right)\left(e_{2}-q_{2}\right)\right.\\& &
		\left.-\left(G_{2}\left(e_{1}\right)-G_{2}\left(e_{2}\right)\right)\right].
	\end{eqnarray*}
	By $G_{2}^{\prime}\left(e_{1}\right)\geq0$, $G_{2}^{\prime}\left(e_{1}\right)Y\geq G_{2}^{\prime}\left(e_{1}\right)q_{2}$
	and $G_{2}^{\prime}\left(e_{1}\right)q_{1}-G_{2}^{\prime}\left(e_{1}\right)Y\leq G_{2}^{\prime}\left(e_{1}\right)q_{1}-G_{2}^{\prime}\left(e_{1}\right)q_{2}$.
	Again since $G_{1}\left(q_{1}\right)$, $G_{1}\left(Y\right)\geq G_{1}\left(q_{2}\right)$
	and $G_{1}\left(q_{1}\right)-G_{1}\left(Y\right)\leq G_{1}\left(q_{1}\right)-G_{1}\left(q_{2}\right)$,
	\begin{eqnarray}
		\left|FZ_{\tau}\left(q_{1},e_{1},Y\right)-FZ_{\tau}\left(q_{2},e_{2},Y\right)\right| & \leq & \left|G_{1}\left(q_{1}\right)-G_{1}\left(q_{2}\right)\right|\nonumber \\
		&  & +\frac{1}{\tau}\left|G_{2}^{\prime}\left(e_{1}\right)q_{1}-G_{2}^{\prime}\left(e_{2}\right)q_{2}\right|\label{inequality1} \\
		&  & +\left|G_{2}^{\prime}\left(e_{1}\right)\left(e_{1}-q_{1}\right)-G_{2}^{\prime}\left(e_{2}\right)\left(e_{2}-q_{2}\right)\right.\nonumber\\& &
		\left.-\left(G_{2}\left(e_{1}\right)-G_{2}\left(e_{2}\right)\right)\right|.\nonumber
	\end{eqnarray}
	Case 3. Suppose $Y>q_{1}\geq q_{2}$. Since $G_{1}\left(.\right)$ is increasing, 
	\begin{eqnarray*}
		FZ_{\tau}\left(q_{1},e_{1},Y\right)-FZ_{\tau}\left(q_{2},e_{2},Y\right) & = & -\tau\left(G_{1}\left(q_{1}\right)-G_{1}\left(q_{2}\right)\right)\\
		&  & +\left[G_{2}^{\prime}\left(e_{1}\right)\left(e_{1}-q_{1}\right)-G_{2}^{\prime}\left(e_{2}\right)\left(e_{2}-q_{2}\right)\right.\\& &
		\left.-\left(G_{2}\left(e_{1}\right)-G_{2}\left(e_{2}\right)\right)\right].
	\end{eqnarray*}
	And 
	\begin{eqnarray}
		\left|FZ_{\tau}\left(q_{1},e_{1},Y\right)-FZ_{\tau}\left(q_{2},e_{2},Y\right)\right| & \leq & \tau\left|G_{1}\left(q_{1}\right)-G_{1}\left(q_{2}\right)\right| \nonumber\\
		&  & +\left|G_{2}^{\prime}\left(e_{1}\right)\left(e_{1}-q_{1}\right)-G_{2}^{\prime}\left(e_{2}\right)\left(e_{2}-q_{2}\right)\right.\label{inequality2}\\& &
		\left.-\left(G_{2}\left(e_{1}\right)-G_{2}\left(e_{2}\right)\right)\right|.\nonumber
	\end{eqnarray}
	For case 1, if Assumptions 2.4 and 2.5 hold, given $\left(Y,W\right)$
	and $\tau\in\left(0,1\right)$, we can have the following results
	for the terms in the right hand side of (\ref{inequality}):
	\[
	\left|G_{1}\left(q_{1}\right)-G_{1}\left(q_{2}\right)\right|\leq A_{1,\tau}\left(V\right)\left\Vert \boldsymbol{\theta}_{1}^{1}-\boldsymbol{\theta}_{1}^{2}\right\Vert \leq A_{1,\tau}\left(V\right)\left\Vert \boldsymbol{\theta}^{1}-\boldsymbol{\theta}^{2}\right\Vert ,
	\]
	\[
	\frac{1}{\tau}\left|G_{2}^{\prime}\left(e_{1}\right)q_{1}-G_{2}^{\prime}\left(e_{2}\right)q_{2}\right| \leq \frac{1}{\tau}\left[A_{2,\tau}\left(V\right)+A_{3,\tau}\left(V\right)\right]\left\Vert \boldsymbol{\theta}^{1}-\boldsymbol{\theta}^{2}\right\Vert ,
	\]
	\[
	\frac{1}{\tau}\left|Y\left[G_{2}^{\prime}\left(e_{1}\right)-G_{2}^{\prime}\left(e_{2}\right)\right]\right| \leq A_{4,\tau}\left(V\right)\left\Vert \boldsymbol{\theta}_{2}^{1}-\boldsymbol{\theta}_{2}^{2}\right\Vert \leq A_{4,\tau}\left(V\right)\left\Vert \boldsymbol{\theta}^{1}-\boldsymbol{\theta}^{2}\right\Vert ,
	\]
	\[
	\left|G_{2}^{\prime}\left(e_{1}\right)\left(e_{1}-q_{1}\right)-G_{2}^{\prime}\left(e_{2}\right)\left(e_{2}-q_{2}\right)-\left(G_{2}\left(e_{1}\right)-G_{2}\left(e_{2}\right)\right)\right|\leq 
	\left[A_{2,\tau}\left(V\right)+A_{5,\tau}\left(V\right)\right]\left\Vert \boldsymbol{\theta}^{1}-\boldsymbol{\theta}^{2}\right\Vert. \]
	Summing the above inequalities together, we can conclude that given
	$\left(Y,W\right)$, $FZ_{\tau}\left(q\left(\boldsymbol{\theta}_{1}\right),e_{1}\left(\boldsymbol{\theta}_{2}\right),Y\right)$ is locally Lipschitz continuous in $\boldsymbol{\theta}$
	around $\boldsymbol{\theta}_{\tau}$. Using the above results for case 1, we can also prove that if Assumptions
	2.4 and 2.5 hold, $FZ_{\tau}\left(q\left(\boldsymbol{\theta}_{1}\right),e_{1}\left(\boldsymbol{\theta}_{2}\right),Y\right)$
	is also locally Lipschitz continuous in $\boldsymbol{\theta}$ around $\boldsymbol{\theta}_{\tau}$
	for cases 2 and 3.
\end{proof}

\subsection{Global Lipschitz continuity}
Recall that
\begin{eqnarray*}
	FZ_{\tau}^{sp}\left(q\left(\boldsymbol{\theta}_{1}\right),e\left(\boldsymbol{\theta}_{2}\right),y\right) & = & \frac{\exp\left(e\left(\boldsymbol{\theta}_{2}\right)\right)}{1+\exp\left(e\left(\boldsymbol{\theta}_{2}\right)\right)}\left[e\left(\boldsymbol{\theta}_{2}\right)+\frac{1}{\tau}\max\left(q\left(\boldsymbol{\theta}_{1}\right)-y,0\right)-q\left(\boldsymbol{\theta}_{1}\right)\right]\\
	&  & +\ln\left(1+\exp\left(y\right)\right)-\ln\left(1+\exp\left(e\left(\boldsymbol{\theta}_{2}\right)\right)\right),
\end{eqnarray*}
and \[h\left(\tau,\boldsymbol{\theta}\right):=\tau\left[ FZ_{\tau}^{sp}\left(q\left(\boldsymbol{\theta}_{1}\right),e\left(\boldsymbol{\theta}_{2}\right),Y\right)-\ln\left(1+\exp\left(Y\right)\right)\right].\] Let $\mathcal{T}:=\left(0,1\right)$ and
\begin{eqnarray*}
	B_{1}\left(V\right) & := & \sup_{\boldsymbol{\theta}_{1}\in\boldsymbol{\Theta}_{1}}\left|q\left(\boldsymbol{\theta}_{1}\right)\right|,\\
	B_{2}\left(V\right) & := & \sup_{\boldsymbol{\theta}\in\boldsymbol{\Theta}}\left\Vert q\left(\boldsymbol{\theta}_{1}\right)\nabla_{\boldsymbol{\theta}_{2}}e\left(\boldsymbol{\theta}_{2}\right)\right\Vert,\\
	B_{3}\left(V\right) & := & \sup_{\boldsymbol{\theta}\in\boldsymbol{\Theta}_{1}}\left\Vert\nabla_{\boldsymbol{\theta}_{1}}q\left(\boldsymbol{\theta}_{1}\right)\right\Vert,\\
	B_{4}\left(V\right) & := & \sup_{\boldsymbol{\theta}_{2}\in\boldsymbol{\Theta}_{2}}\left\Vert Y\nabla_{\boldsymbol{\theta}_{2}}e\left(\boldsymbol{\theta}_{2}\right)\right\Vert,\\
	B_{5}\left(V\right) & := & \sup_{\boldsymbol{\theta}\in\boldsymbol{\Theta}}\left\Vert\left(e\left(\boldsymbol{\theta}_{2}\right)-q\left(\boldsymbol{\theta}_{1}\right)\right)\nabla_{\boldsymbol{\theta}_{2}}e\left(\boldsymbol{\theta}_{2}\right)\right\Vert.
\end{eqnarray*}
Let $\boldsymbol{\theta}^{1}:=\left(\boldsymbol{\theta}_{1}^{1},\boldsymbol{\theta}_{2}^{1}\right)$
and $\boldsymbol{\theta}^{2}:=\left(\boldsymbol{\theta}_{1}^{2},\boldsymbol{\theta}_{2}^{2}\right)$ be two vectors of parameter values.
\begin{lemma}[Global Lipschitz continuity]
	If Assumption 2.4 holds, for
	all $\left(\tau_{1},\tau_{2}\right)\in\mathcal{T}$, $\left(\boldsymbol{\theta}_{1}^{1},\boldsymbol{\theta}_{1}^{2}\right)\in\boldsymbol{\Theta}_{1}$
	and $\left(\boldsymbol{\theta}_{2}^{1},\boldsymbol{\theta}_{2}^{2}\right)\in\boldsymbol{\Theta}_{2}$,
	given $V=\left(Y,W\right)$, there exists a constant $B_{\tau}^{*}\left(V\right)$
	such that
	\[
	\left|h\left(\tau_{1},\boldsymbol{\theta}^{1}\right)-h\left(\tau_{2},\boldsymbol{\theta}^{2}\right)\right|\leq B_{\tau}^{*}\left(V\right)\left(\left\Vert \tau_{1}-\tau_{2}\right\Vert +\left\Vert \boldsymbol{\theta}^{1}-\boldsymbol{\theta}^{2}\right\Vert \right).
	\]
\end{lemma}

\begin{proof}
	We use similar notations and strategy as in proving Lemma 2. At first
	note that $h\left(\tau,\boldsymbol{\theta}\right)$ can be rewritten as
	\begin{eqnarray*}
		h\left(\tau,\boldsymbol{\theta}\right) & = & \frac{\exp\left(e\left(\boldsymbol{\theta}_{2}\right)\right)}{1+\exp\left(e\left(\boldsymbol{\theta}_{2}\right)\right)}\left\{ 0.5\left[\left|q\left(\boldsymbol{\theta}_{1}\right)-Y\right|+q\left(\boldsymbol{\theta}_{1}\right)-Y\right]+\tau\left(e\left(\boldsymbol{\theta}_{2}\right)-q\left(\boldsymbol{\theta}_{1}\right)\right)\right\} \\
		&  & -\tau\ln\left(1+\exp\left(e\left(\boldsymbol{\theta}_{2}\right)\right)\right).
	\end{eqnarray*}
	Then 
	\begin{eqnarray*}
		h\left(\tau_{1},\boldsymbol{\theta}^{1}\right)-h\left(\tau_{2},\boldsymbol{\theta}^{2}\right) & = & \frac{0.5\exp\left(e_{1}\right)}{1+\exp\left(e_{1}\right)}\left[\left|q_{1}-Y\right|+q_{1}-Y\right]-\frac{0.5\exp\left(e_{2}\right)}{1+\exp\left(e_{2}\right)}\left[\left|q_{2}-Y\right|+q_{2}-Y\right]\\
		&  & +\left.\left[\frac{\tau_{1}\exp\left(e_{1}\right)}{1+\exp\left(e_{1}\right)}\left(e_{1}-q_{1}\right)-\frac{\tau_{2}\exp\left(e_{2}\right)}{1+\exp\left(e_{2}\right)}\left(e_{2}-q_{2}\right)\right.\right.\\
		&  & \left.-\left(\tau_{1}\ln\left(1+\exp\left(e_{1}\right)\right)-\tau_{2}\ln\left(1+\exp\left(e_{2}\right)\right)\right)\right]
	\end{eqnarray*}
	We assume $q_{1}\geq q_{2}$ and $\tau_{1}\geq\tau_{2}$. The results
	for $q_{1}<q_{2}$ and $\tau_{1}<\tau_{2}$ can be proved in a similar
	manner. We separate the proof into three cases.\\
	Cases 1. Suppose $q_{1}\geq q_{2}>Y$. It can be shown that
	\begin{eqnarray*}
		\left|h\left(\tau_{1},\boldsymbol{\theta}^{1}\right)-h\left(\tau_{2},\boldsymbol{\theta}^{2}\right)\right| & \leq & \left|\frac{\exp\left(e_{1}\right)}{1+\exp\left(e_{1}\right)}q_{1}-\frac{\exp\left(e_{2}\right)}{1+\exp\left(e_{2}\right)}q_{2}\right|\\
		&  & +\left|Y\frac{\exp\left(e_{1}\right)}{1+\exp\left(e_{1}\right)}-Y\frac{\exp\left(e_{2}\right)}{1+\exp\left(e_{2}\right)}\right|\\
		&  & +\left|\frac{\tau_{1}\exp\left(e_{1}\right)}{1+\exp\left(e_{1}\right)}\left(e_{1}-q_{1}\right)-\frac{\tau_{2}\exp\left(e_{2}\right)}{1+\exp\left(e_{2}\right)}\left(e_{2}-q_{2}\right)\right.\\
		&  & \left.-\left(\tau_{1}\ln\left(1+\exp\left(e_{1}\right)\right)-\tau_{2}\ln\left(1+\exp\left(e_{2}\right)\right)\right)\right|
	\end{eqnarray*}
	Case 2. Suppose $q_{1}\geq Y\geq q_{2}$. It can be shown that
	\begin{eqnarray*}
		\left|h\left(\tau_{1},\boldsymbol{\theta}^{1}\right)-h\left(\tau_{2},\boldsymbol{\theta}^{2}\right)\right| & \leq & \left|\frac{\exp\left(e_{1}\right)}{1+\exp\left(e_{1}\right)}q_{1}-\frac{\exp\left(e_{2}\right)}{1+\exp\left(e_{2}\right)}q_{2}\right|\\
		&  & +\left|\frac{\tau_{1}\exp\left(e_{1}\right)}{1+\exp\left(e_{1}\right)}\left(e_{1}-q_{1}\right)-\frac{\tau_{2}\exp\left(e_{2}\right)}{1+\exp\left(e_{2}\right)}\left(e_{2}-q_{2}\right)\right.\\
		&  & \left.-\left(\tau_{1}\ln\left(1+\exp\left(e_{1}\right)\right)-\tau_{2}\ln\left(1+\exp\left(e_{2}\right)\right)\right)\right|
	\end{eqnarray*}
	Case 3. Suppose $Y>q_{1}\geq q_{2}$. It can be shown that
	\begin{eqnarray*}
		\left|h\left(\tau_{1},\boldsymbol{\theta}^{1}\right)-h\left(\tau_{2},\boldsymbol{\theta}^{2}\right)\right| & \leq & +\left|\frac{\tau_{1}\exp\left(e_{1}\right)}{1+\exp\left(e_{1}\right)}\left(e_{1}-q_{1}\right)-\frac{\tau_{2}\exp\left(e_{2}\right)}{1+\exp\left(e_{2}\right)}\left(e_{2}-q_{2}\right)\right.\\
		&  & \left.-\left(\tau_{1}\ln\left(1+\exp\left(e_{1}\right)\right)-\tau_{2}\ln\left(1+\exp\left(e_{2}\right)\right)\right)\right|
	\end{eqnarray*}
	Let $\left(\bar{\tau},\bar{q},\bar{e}\right)$ be some middle point
	between $\left(\tau_{1},q_{1},e_{1}\right)$ and $\left(\tau_{2},q_{2},e_{2}\right)$.
	Using the mean value theorem, it is straightforward to show that
	\begin{eqnarray}
		\left|\frac{\exp\left(e_{1}\right)}{1+\exp\left(e_{1}\right)}q_{1}-\frac{\exp\left(e_{2}\right)}{1+\exp\left(e_{2}\right)}q_{2}\right| & \leq & \left|\frac{\exp\left(\bar{e}\right)}{1+\exp\left(\bar{e}\right)}\right|\left|q_{1}-q_{2}\right|\nonumber \\
		&  & +\left|\frac{\exp\left(\bar{e}\right)}{1+\exp\left(\bar{e}\right)}\left(1-\frac{\exp\left(\bar{e}\right)}{1+\exp\left(\bar{e}\right)}\right)\right|\left|\bar{q}\left(e_{1}-e_{2}\right)\right|\nonumber \\
		& \leq & \left|q_{1}-q_{2}\right|+\left|\bar{q}\left(e_{1}-e_{2}\right)\right|,\label{inequality5}
	\end{eqnarray}
	\begin{eqnarray}
		\left|Y\frac{\exp\left(e_{1}\right)}{1+\exp\left(e_{1}\right)}-Y\frac{\exp\left(e_{2}\right)}{1+\exp\left(e_{2}\right)}\right| & \leq & \left|\frac{\exp\left(\bar{e}\right)}{1+\exp\left(\bar{e}\right)}\left(1-\frac{\exp\left(\bar{e}\right)}{1+\exp\left(\bar{e}\right)}\right)\right|\left|Y\left(e_{1}-e_{2}\right)\right|\nonumber \\
		& \leq & \left|Y\left(e_{1}-e_{2}\right)\right|,\label{inequality6}
	\end{eqnarray}
	and
	\begin{eqnarray}
		\left|\frac{\tau_{1}\exp\left(e_{1}\right)}{1+\exp\left(e_{1}\right)}\left(e_{1}-q_{1}\right)-\frac{\tau_{2}\exp\left(e_{2}\right)}{1+\exp\left(e_{2}\right)}\left(e_{2}-q_{2}\right)\right.\nonumber\\
		\left.-\left(\tau_{1}\ln\left(1+\exp\left(e_{1}\right)\right)-\tau_{2}\ln\left(1+\exp\left(e_{2}\right)\right)\right)\right| & \leq & \left|\frac{\exp\left(\bar{e}\right)}{1+\exp\left(\bar{e}\right)}\bar{e}-\ln\left(1+\exp\left(\bar{e}\right)\right)\right|\left|\tau_{1}-\tau_{2}\right|\nonumber\\
		&  & +\left|\frac{\exp\left(\bar{e}\right)\bar{q}}{1+\exp\left(\bar{e}\right)}\right|\left|\tau_{1}-\tau_{2}\right|+\left|\frac{\bar{\tau}\exp\left(\bar{e}\right)}{1+\exp\left(\bar{e}\right)}\right|\left|q_{1}-q_{2}\right|\nonumber\\
		&  & +\left|\frac{\bar{\tau}\exp\left(\bar{e}\right)}{1+\exp\left(\bar{e}\right)}\left(1-\frac{\exp\left(\bar{e}\right)}{1+\exp\left(\bar{e}\right)}\right)\right|\nonumber \\
		&  & \times\left|\left(\bar{e}-\bar{q}\right)\left(e_{1}-e_{2}\right)\right|\nonumber\\
		& \leq & \left(\ln2+\left|\bar{q}\right|\right)\left|\tau_{1}-\tau_{2}\right|+\left|q_{1}-q_{2}\right|\label{inequality7}\\
		&  & +\left|\left(\bar{e}-\bar{q}\right)\left(e_{1}-e_{2}\right)\right|\nonumber
	\end{eqnarray}
	To see why (\ref{inequality7}) holds, let $\varpi\left(x\right)=\exp\left(x\right)\left(1+\exp\left(x\right)\right)^{-1}x-\ln\left(1+\exp\left(x\right)\right)$.
	It can be shown that $\varpi\left(x\right)$ is monotonically decreasing
	for $x<0$ and monotonically increasing for $x\geq0$ and $\varpi\left(0\right)=-\ln2$,
	$\lim_{x\rightarrow\infty}\varpi\left(x\right)=0$, and $\lim_{x\rightarrow-\infty}\varpi\left(x\right)=0$.
	%key equation is 
	%\[
	%x-\ln(1+exp(x)) = \ln(exp(x)/(1+epx(x)))
	%\]
	Therefore we can conclude that 
	\[
	\left|\varpi\left(x\right)\right|\leq\ln2
	\]
	for all $x\in\mathbb{R}$. If Assumption 2.4 holds, it can be shown for (\ref{inequality5}), 
	\begin{eqnarray*}
		\left|q_{1}-q_{2}\right|+\left|\bar{q}\left(e_{1}-e_{2}\right)\right| & \leq & \left|\nabla_{\boldsymbol{\theta}_{1}}q\left(\bar{\boldsymbol{\theta}}_{1}\right)\right|\left\Vert\boldsymbol{\theta}_{1}^{1}-\boldsymbol{\theta}_{1}^{2}\right\Vert+\left|q\left(\boldsymbol{\theta}_{1}\right)\nabla_{\boldsymbol{\theta}_{2}}e\left(\boldsymbol{\theta}_{2}\right)\right|\left\Vert\boldsymbol{\theta}_{2}^{1}-\boldsymbol{\theta}_{2}^{2}\right\Vert\\
		& \leq & B_{3}\left(V\right)\left\Vert\boldsymbol{\theta}_{1}^{1}-\boldsymbol{\theta}_{1}^{2}\right\Vert+B_{2}\left(V\right)\left\Vert\boldsymbol{\theta}_{2}^{1}-\boldsymbol{\theta}_{2}^{2}\right\Vert,
	\end{eqnarray*}
	for (\ref{inequality6}),
	\begin{eqnarray*}
		\left|Y\left(e_{1}-e_{2}\right)\right| & \leq & B_{4}\left(V\right)\left\Vert\boldsymbol{\theta}_{2}^{1}-\boldsymbol{\theta}_{2}^{2}\right\Vert,
	\end{eqnarray*}
	and for (\ref{inequality7}), 
	\begin{eqnarray*}
		\left(\ln2+\left|\bar{q}\right|\right)\left|\tau_{1}-\tau_{2}\right|+\left|q_{1}-q_{2}\right|+\left|\left(\bar{e}-\bar{q}\right)\left(e_{1}-e_{2}\right)\right| & \leq & \left(\ln2+B_{1}\left(V\right)\right)\left|\tau_{1}-\tau_{2}\right|+B_{3}\left(V\right)\left\Vert\boldsymbol{\theta}_{1}^{1}-\boldsymbol{\theta}_{1}^{2}\right\Vert\\
		&  & +B_{5}\left(V\right)\left\Vert\boldsymbol{\theta}_{2}^{1}-\boldsymbol{\theta}_{2}^{2}\right\Vert.
	\end{eqnarray*}
	Summing the above inequalities together, given $V=\left(Y,W\right)$,
	we can conclude that \[\left|h\left(\tau_{1},\boldsymbol{\theta}^{1}\right)-h\left(\tau_{2},\boldsymbol{\theta}^{2}\right)\right|\leq B^{*}\left(V\right)\left(\left\Vert \tau_{1}-\tau_{2}\right\Vert +\left\Vert \boldsymbol{\theta}^{1}-\boldsymbol{\theta}^{2}\right\Vert \right),\]
	where $B^{*}\left(V\right)$ is a constant and is determined by a
	sum of $\ln2$ and constants $B_{i}\left(V\right)$, $i=1,\ldots,5$. 
\end{proof}

\subsection{Constructing the simultaneous confidence bands (scb) with the bootstrap}
We use the models of linear in parameters of (6) and (7) %(\ref{linear_q_Y}) and (\ref{linear_e_Y})
as an example to illustrate the bootstrap procedures for constructing the simultaneous confidence bands. Let $\boldsymbol{\theta}_{\tau}=\left(\alpha_{1,\tau},\boldsymbol{\beta}_{1,\tau}^\top,\alpha_{2,\tau},\boldsymbol{\beta}_{2,\tau}^\top\right)^\top$ be a vector for parameters at the $\tau$-quantile. The bootstrap procedures are summarized as follows.
%\clearpage
{\begin{framed}
\begin{enumerate}
	\item Draw a bootstrap sample of size $n$: $W_{1}^{*},\ldots,W_{n}^{*}$
	with replacement, where
	\[
	W_{i}^{*}=\left(Y_{i}^{*},D_{i}^{*},X_{i}^{*\top},Z_{i}^{*}\right)^\top,
	\]
	$i=1,\ldots,n$, is a vector for the $i$th random draw sample. 
	\item Re-estimate the weight $\bar{K}$ with the bootstrap sample.
	Let $\hat{\tilde{K}}_{i}^{*}$, $i=1,\ldots,n$ denote the bootstrap
	estimated truncated weight evaluated with $\left(Y_{i}^{*},D_{i}^{*},X_{i}^{*}\right)$. 
	\item With the bootstrap estimated truncated weight $\hat{\bar{K}}_{i}^{*}$,
	obtain the bootstrap estimated parameters \[\hat{\boldsymbol{\theta}}_{\tau,b}^{*}=\left(\hat{\alpha}_{1,\tau,b}^{*},\hat{\boldsymbol{\beta}}_{1,\tau,b}^{*\top},\hat{\alpha}_{2,\tau,b}^{*},\hat{\boldsymbol{\beta}}_{2,\tau,b}^{*\top}\right)^\top\]
	for $\tau\in\left(0,1\right)$.
	\item Repeat procedures 1 to 3 $B$ times to obtain $B$ bootstrap estimated
	parameters $\hat{\boldsymbol{\theta}}_{\tau,1}^{*},\ldots,\hat{\boldsymbol{\theta}}_{\tau,B}^{*}$.
\end{enumerate}
\end{framed}
}
%We conduct the above bootstrap procedures for $\tau\in (0,1)$. 
With the bootstrap estimated parameters, we use $\hat{\alpha}_{2,\tau}$ as an example to illustrate the procedures to construct the scb as follows.
{\begin{framed}
\begin{enumerate}
	\item Calculate the rescaled bootstrap quantile spread \citep{CFM_2013}
	\begin{equation}
		\hat{s}_{\hat{\alpha}_{2,\tau}}^{*}=\frac{|Q_{\hat{\alpha}_{2,\tau,b}^{*}}(g_1)-Q_{\hat{\alpha}_{2,\tau,b}^{*}}(g_2)|}{|\Phi^{-1}(g_1)-\Phi^{-1}(g_2)|},%\label{b_ser1}
	\end{equation} as an estimate of the bootstrap standard deviation of $\hat{\alpha}_{2,\tau}$, with the $B$
	bootstrap estimates $\left(\hat{\alpha}_{2,\tau,1}^{*},\ldots,\hat{\alpha}_{2,\tau,B}^{*}\right)$, where  $Q_{\hat{\alpha}_{2,\tau,b}^{*}}(g)$ is the $g$th quantile of $\sqrt{n}\left(\hat{\alpha}_{2,\tau,b}^{*}-\hat{\alpha}_{2,\tau}\right)$ conditional on the data, and $\Phi^{-1}(g)$ is the $g$th quantile of $N(0,1)$, and $0<g_2<g_1<1$.
	\item Calculate the bootstrap $t$ statistic of $\hat{\alpha}_{2,\tau}$:
	\[\hat{t}_{\hat{\alpha}_{2,\tau},b}^{*}=\frac{\sqrt{n}\left(\hat{\alpha}_{2,\tau,b}^{*}-\hat{\alpha}_{2,\tau}\right)}{\hat{s}^{*}_{\hat{\alpha}_{2,\tau}}},\] $b=1,\ldots,B$. 
	\item For each $b=1,\ldots,B$, calculate the maximal absolute bootstrap
	$t$ statistic for $\tau\in\left(0,1\right)$: 
	\[\hat{t}_{\hat{\alpha}_{2},b}^{*}=\sup_{\tau\in\left(0,1\right)}\left|\hat{t}_{\hat{\alpha}_{2,\tau},b}\right|.\]
	\item Let $\hat{t}_{\hat{\alpha}_{2}}^{*(1-g)}$ be the $\left(1-g\right)$th
	sample quantile of $\hat{t}_{\hat{\alpha}_{2},1}^{*},\ldots,\hat{t}_{\hat{\alpha}_{2},B}^{*}$.
	The $1-g$ scb of $\hat{\alpha}_{2,\tau}$ is 
	\begin{equation}
		\left[\hat{\alpha}_{2,\tau}-\hat{t}_{\hat{\alpha}_{2}}^{*(1-g)}\frac{\hat{s}_{\hat{\alpha}_{2,\tau}}^{*}}{\sqrt{n}},\hat{\alpha}_{2,\tau}+\hat{t}_{\hat{\alpha}_{2}}^{*(1-g)}\frac{\hat{s}_{\hat{\alpha}_{2,\tau}}^{*}}{\sqrt{n}}\right].\label{scb}
	\end{equation}
    \item As for constructing the non-standardized version of scb \citep{HL_2021}, we can just let $\hat{s}_{\hat{\alpha}_{2,\tau}}^{*}=1$ in $\hat{t}_{\hat{\alpha}_{2,\tau},b}^{*}$ and (\ref{scb}) and following the same procedures. The $1-g$ scb of $\hat{\alpha}_{2,\tau}$ now becomes
    \begin{equation}
    	\left[\hat{\alpha}_{2,\tau}-\frac{\hat{d}_{\hat{\alpha}_{2}}^{*(1-g)}}{\sqrt{n}},\hat{\alpha}_{2,\tau}+\frac{\hat{d}_{\hat{\alpha}_{2}}^{*(1-g)}}{\sqrt{n}}\right],\label{scb1}
    \end{equation}where $\hat{d}_{\hat{\alpha}_{2}}^{*(1-g)}$ is the ($1-g$)th sample quantile of 
    $\hat{d}_{\hat{\alpha}_{2},b}^{*}=\sup_{\tau\in\left(0,1\right)}\left|\sqrt{n}\left(\hat{\alpha}_{2,\tau,b}^{*}-\hat{\alpha}_{2,\tau}\right)\right|$, $b=1,\ldots,B$.
\end{enumerate}
\end{framed}
}

\subsection{Empirical results when $D$ is assumed to be exogenous}
Figure \ref{figure1} shows the estimation results of the CTATE and QTE for the case of no adjustment for endogeneity, which is equivalent to setting the estimated weight $\tilde{K}_{i}=1$. Again, the CTATE estimates increase monotonically with the quantile levels, while the QTE estimates demonstrate some fluctuations. Compared with the case of adjustment for endogeneity, the estimates shown here overall have higher values. The pcb's evaluated with the analytic standard errors in Appendix A.1 and bootstrap also look very similar. For both adult men and women, the scb's reveal that the CTATE estimates are statistically significantly positive over quantile levels ranging from 0.25 to 0.9. However, depending on the scb, the results for positivity of the QTE estimates are somehow different, especially for adult men. For adult women, FOSD is supported by the evidence shown here. As for adult men, the evidence for SOSD of participating in the JTPA is much stronger than that for FOSD. Figure \ref{figure2} shows the estimation results of the IQATE, the corresponding pcb's (implemented with the analytic standard errors and the bootstrap) and the LAVG-QTE for the earnings of adult men and women at different quantile level intervals. For both cases of adult men and women, the IQATE estimates are statistically significantly positive over the eight quantile level intervals. Comparing the IQATE estimates and LAVG-QTE, they are not very different when the quantile level is above 0.5, but at quantile levels lower than 0.5, they rather show some mild differences.

\subsection{The Lorenz curve effect}
Conditional tail expectations of potential outcomes $CTE_{Y_{d}|X}\left(\tau\right)$ are also relevant for deriving the Lorenz curves of the corresponding distributions. Suppose that the potential outcome $Y_{d}$ has a non-zero mean conditional on $X$. The Lorenz curve, an frequently used measure for the degree of income or wealth inequality, is defined as a ratio of the partial mean to the overall mean of $Y_{d}$ :\[
LO\left(\tau,d\right) := \frac{\tau CTE_{Y_{d}|X}(\tau)}{E\left[Y_{d}|X\right]} = \frac{\tau CTE_{Y_{d}|X}(\tau)}{\lim_{\tau\rightarrow 1} \tau CTE_{Y_{d}|X}(\tau)}.
\] The Lorenz effect \citep{CFM_2013} is then defined as \[LO\left(\tau,1\right)-LO\left(\tau,0\right),\]which is useful for measuring how the degree of inequality of an outcome of interest in a population changes across two treatment regimes. The Lorenz effect for compliers can be calculated empirically using the estimated parameters of the following model:
\begin{equation}
	CTE_{Y|D,X,T=c}(\tau )=\alpha_{2,\tau}D+X^\top\boldsymbol{\beta}_{2,\tau}
	\nonumber
\end{equation}using the proposed method in this paper.

\subsection{Discussion of assumptions in the context of our empirical application}
In this section, we discuss the regularity assumptions made for the asymptotic results of our study in the context of the empirical application in Section 5. Recall that in our empirical application, the outcome variable $Y$ is the individual's 30-month earnings. The treatment variable $D$ is a binary indicator of enrollment in JTPA services. The instrumental variable $Z$ is a binary variable for being offered such services. The exogenous covariates $X$ are all discrete variables that take finite values. Throughout our empirical study, we estimate the parameters in specifications (6) and (7) to estimate the QTE and CTATE for a group of compliers. %Namely, $q_{i}(\boldsymbol{\theta}_{1})=q(W_{i},\boldsymbol{\theta}_{1})= \alpha _{1}D_{i}+X_{i}^\top\boldsymbol{\beta}_{1}$ and $e_{i}(\boldsymbol{\theta}_{2})=e(W_{i},\boldsymbol{\theta}_{1})=\alpha _{2}D_{i}+X_{i}^\top\boldsymbol{\beta}_{2}$ are regression functions of (6) and (7) evaluated at parameter values $\boldsymbol{\theta}_{1}=(\alpha_{1},\boldsymbol{\beta}_{1}^\top)^\top$ and $\boldsymbol{\theta}_{2}=(\alpha_{2},\boldsymbol{\beta}_{2}^\top)^\top$. 
We set $G_{1}\left(t\right)=0$, $G_{2}\left(t\right)=\eta\left(t\right)=\ln\left(1+\exp\left(t\right)\right)$ in (14) and thus use the function $FZ_{\tau}^{sp}\left(q,e,y\right)$, which is defined in (16), as the specification for the FZ loss in the estimation.   

Assumption 2.1 for i.i.d. observations is standard in cross-sectional empirical studies. This assumption is also maintained in \citet{AAI_2002} and \citet{CH_2008} for evaluating the JTPA programs using the same dataset. For Assumption 2.2, the earnings variable is generally viewed as a continuous random variable. Here, we assume that it remains continuously distributed given $W=(D,X)$ and conditional on the group of compliers. Namely, Assumption 2.2 is assumed to hold in the data. The compactness of the parameter space in Assumption 2.3 is often required for proving consistency as the estimated loss function is non-smooth and non-convex. This assumption holds automatically in practical implementation of our algorithm in Section 2.4, because the domains in the optimization problems of that algorithm are always constrained by the precision of real numbers admitted in the numerical computation. Assumption 2.4 holds immediately under the linear specifications (6) and (7). Assumption 2.5, which guarantees the uniqueness of the solution $(Q_{Y}(\tau),CTE_{Y}(\tau))$ to the minimization problem (13), is satisfied in our empirical application using the FZ loss specification $FZ_{\tau}^{sp}\left(q,e,y\right)$. Assumption 2.6 is a rank condition for the regressors in (6) and (7). It holds when $Var(W|T=c)$ is nonsingular. Assumption 2.7 is a dominance condition, which ensures that the population objective function in (27) is continuous in $(\boldsymbol{\theta}_{1},\boldsymbol{\theta}_{2})$. As the parameter space $\boldsymbol{\Theta}$ is compact and $W$ takes only finite values, this assumption also holds for the FZ loss $FZ_{\tau}^{sp}\left(q,e,y\right)$. Assumption 2.8 concerns uniform consistency of the estimated weight $\bar{K}(V;\hat{\boldsymbol{\gamma}},\hat{v})$, which relies on the consistency of the estimators of $\boldsymbol{\gamma}$ and $v$. %In our empirical application, we estimate the JTPA offer assignment probability $P(Z=1|X)$ using the probit model. 

In our empirical application, these exogenous covariates are all discrete variables, and as shown in Section 3, this is allowed in our method. Finally, in the data, there were still individuals who received the JTPA services but did not obtain the assignment. However, as pointed out by \cite{AAI_2002}, the proportion of such violations relative to the entire sample is very low (less than 2\%); %(adult men: 1.12\% and adult women: 1.73\%),
therefore this has a negligible impact on our estimation.

We next discuss the key assumptions underpinning the main asymptotic results of our proposed method, and compare them with those established by \citet{CH_2006}. We aim to highlight the similarities that are pivotal to the robustness and reliability of the two approaches. One fundamental assumption common to both our method and that of \citet{CH_2006} is the requirement for data to be independently and identically distributed. This assumption is standard in econometric analysis. In both methods, there is an underlying condition that the parameter space should be compact. This is essential for proving the consistency of our estimators and those of \citet{CH_2006}, as the estimating objective functions in both approaches are non-convex. %Our assumption also includes the continuity of certain functions with respect to parameters. % the objective function in \citet{CH_2006} is not continuous in parameters.
Additionally, the absolute continuity of the conditional distribution of $Y$ should hold. Chernozhukov and Hansen's approach also aligns with this requirement, thereby underscoring its importance in econometric analysis. However, for boundedness of the density function of $Y$ conditional on $W$, our proposed method and Chernozhukov and Hansen's approach have different requirements. The former requires that this conditional density should be bounded away from zero, while the latter requires that it should be bounded from above. Finally, both our method and that of \citet{CH_2006} are flexible in the forms of regression equations. This flexibility allows for a more nuanced and adaptable approach to modeling.

In conclusion, while our proposed method introduces innovative elements, it also shares fundamental assumptions with the well-regarded work of \citet{CH_2006}, thereby ensuring a solid theoretical foundation and practical applicability.

\subsection{Comparisons with the integrated-QTE based estimator on simulations and empirical application}
In this section we compare the performance of our proposed method with that of a trimmed version of the integrated-QTE based estimator for estimating the CTATE for compliers. Recall the integrated-QTE based estimator introduced in Section 2.5.
\begin{equation}
	\widehat{\text{IntQ}}^{\star}(\tau) = \frac{1}{\tau}\int_{0}^{\tau}\hat{\alpha}_{1,u}du
	\label{IntQ}
\end{equation} for estimating the CTATE at quantile level $\tau\in (0,1)$, where $\hat{\alpha}_{1,u}$ is an estimate of $\alpha_{1,u}$, the QTE at quantile level $u$ with some method. Let $\text{IntQ}^{\star}(\tau) = \tau^{-1}\int_{0}^{\tau}\alpha_{1,u}du$ denote the CTATE expressed in the integral form of the QTE. We have discussed pros and cons of (\ref{IntQ}) in Section 2.5.

In practice, we use a trimmed version of (\ref{IntQ}) to circumvent issues with estimation at extreme quantiles, such as
\begin{equation}
	\widehat{\text{IntQ}}(\tau)=\frac{1}{\tau}\int_{\underline{\tau}}^{\tau}\hat{\alpha}_{1,u}du,
	\label{trimmed_IntQ}
\end{equation}where $\underline{\tau}$ is a small constant and $0<\underline{\tau}<\tau$. It is easy to see that (\ref{trimmed_IntQ}) will incur a truncation bias for the estimation of the CTATE ($\text{IntQ}^{\star}(\tau)$) at quantile level $\tau$, but one may hope that this bias will not be too large when setting the trimming constant $\underline{\tau}$ to be a very small value. However, as we show in the following simulations, even setting $\underline{\tau}$ as small as 0.005, the resulting bias is still substantial. 

Let $\text{IntQ}(\tau)=\tau^{-1}\int_{\underline{\tau}}^{\tau}\alpha_{1,u}du$ and $\text{\underline{IntQ}}(\tau)=\tau^{-1}\int_{0}^{\underline{\tau}}\alpha_{1,u}du$. Then $\text{IntQ}^{\star}(\tau)=\text{IntQ}(\tau)+\text{\underline{IntQ}}(\tau)$. Thus, for estimating the CTATE, the bias of (\ref{trimmed_IntQ}) can be written as
\begin{equation}
	E\left[\widehat{\text{IntQ}}(\tau)-\text{IntQ}^{\star}(\tau)\right] = E\left[\widehat{\text{IntQ}}(\tau)-\text{IntQ}(\tau)\right] -\text{\underline{IntQ}}(\tau),
	\label{bias_IntQ}
\end{equation}which indicates that the bias of (\ref{trimmed_IntQ}) stems from that of $\widehat{\text{IntQ}}(\tau)$ on estimating the population value $\text{IntQ}(\tau)$ and a remainder term $-\text{\underline{IntQ}}(\tau)$. Ideally if the term $E\left[\widehat{\text{IntQ}}(\tau)-\text{IntQ}(\tau)\right]$is negligible, a small $-\text{\underline{IntQ}}(\tau)$ will make the whole bias not too severe. However, $-\text{\underline{IntQ}}(\tau)$ could sometimes be substantial, even when $\underline{\tau}$ is set to be small. For example, in the benchmark case when $\alpha_{1,u}$ is the $u$th quantile of a standard normal random variable, at $\tau=0.1$, setting $\underline{\tau} = 0.01 (0.005)$ will make $-\text{\underline{IntQ}}(\tau)$ roughly equal to 15.2\% (8.2\%) of $|\text{IntQ}^{\star}(\tau)|$.

In the following we conduct simulations using the same data generation process as in Section 4, to compare the performances of the proposed estimator using the FZ loss ($\hat{\alpha}_{2,\tau}$ from solving (17) in the main text) and the trimmed integrated-QTE based estimator $\widehat{\text{IntQ}}(\tau)$ on estimating the CTATE for compliers. We use the weighted quantile regression (WQR) proposed by \cite{AAI_2002} for estimating $\alpha_{1,u}$. We also use the same estimated weights to implement the WQR and our FZ loss based CTATE estimators. To practically compute (\ref{trimmed_IntQ}), we use its discrete analog (denoted by $\tilde{\alpha}_{2,\tau}$):
\begin{equation}
	\tilde{\alpha}_{2,\tau} = \frac{1}{\tau}\sum_{j=1}^{k}\hat{\alpha}_{1,u_{j}}\Delta u_{j}=\frac{\sum_{j=1}^{k}\Delta u_{j}\hat{\alpha}_{1,u_{j}}}{\sum_{j=1}^{k}\Delta u_{j}},
\end{equation}where $\Delta u_{j} = u_{j}-u_{j-1}$ is the grid point space for a grid of points $(u_{0} ,u_{1},u_{2},\ldots,u_{k})$ with $u_{0}:=0$, $u_{1}:= \underline{\tau}$ and $u_{k}:=\tau$. %To avoid numerically unstable For WQR estimations at the extreme lower quantile,
We then set $\underline{\tau}=0.01$ for sample size $n = 500$ and 3,000 (denoted by IntQ), and $\underline{\tau}=0.005$ for $n = 3,000$ only (denoted by IntQ'). The grid point space $\Delta u_{j}$ for numerical integration is 0.005. Notice that with these settings ($\underline{\tau}=0.01$ and 0.005), in our data generation process, the remainder term $-\text{\underline{IntQ}}(\tau)$ is positive. We show the biases, variances and MSEs of these estimators over quantile levels ranging from 0.1 to 0.9, under different settings for the weight estimations (M1 to M4, see Section 4) in Figures \ref{figure3} to \ref{figure6}. 

Comparing to the CTATE estimator using the FZ loss, $\hat{\alpha}_{2,\tau}$, it can be seen that the discrete analog of the trimmed integrated-QTE based estimator $\tilde{\alpha}_{2,\tau}$ (denoted by IntQ and IntQ' in the figures) generally has a substantial positive estimation bias over all the quantile levels. Such a positive bias could be largely due to the positive remainder term $-\underline{\text{IntQ}}(\tau)$ in our data generation process. 
This bias declines as the quantile level increases, reflecting that omitting the estimation of $-\underline{\text{IntQ}}(\tau)$ has less impact on estimating the CTATE as the quantile level increases. In our cases, including $\hat{\alpha}_{1,u}$ at the extreme lower quantile level $u=0.005$ (IntQ') can mitigate the positive bias, but it also seems to result in a slightly higher variance. When $n=500$, $\tilde{\alpha}_{2,\tau}$ overall has a smaller variance than $\hat{\alpha}_{2,\tau}$. %and this is possibly because of the fact that $\tilde{\alpha}_{2,\tau}$ can be viewed as a weighted average of $\hat{\alpha}_{1,u_{j}}$ for $u_{j}\leq \tau$ with the weight $\Delta u_{j}/\sum_{j=1}^{k}\Delta u_{j}$. %and/or (2) $\hat{\alpha}_{1,u}$ from the WQR may be less variant.
However, as the sample size increases, the difference in variances largely disappears. In some cases, due to a lower variance, $\tilde{\alpha}_{2,\tau}$ can also result in a lower MSE than $\hat{\alpha}_{2,\tau}$, but similar to the variance case, the difference of the MSEs also becomes negligible as the sample size increases.

In Figure \ref{figure7}, we compare estimates of the CTATE for compliers from using the discrete analog of the trimmed integrated-QTE based estimator ($\tilde{\alpha}_{2,\tau}$) and those from using the FZ loss ($\hat{\alpha}_{2,\tau}$) in our empirical application with the JTPA data. Here we set $\underline{\tau}=0.05$, as we find $\hat{\alpha}_{1,u_{j}}$ values are zero for $u_{j}<0.05$ in the WQR estimations. From the figure, the CTATE estimates from using $\tilde{\alpha}_{2,\tau}$ has a lower value than those from using the FZ loss over all the quantile levels. This suggests that $\tilde{\alpha}_{2,\tau}$ seems to result in a downward bias here.

\subsection{Choices on $G_{1}(t)$ and $G_{2}(t)$}
In Section 2.4 of the main text, we briefly discuss our specifications on $G_{1}(t)=0$ and $G_{2}(t)=\eta(t)=\ln(1+\exp(t))$ (softplus function). Here we give a more detailed analysis. In the FZ loss used in our paper, we specify $G_{1}(t)=0$, $G_{2}(t)=\eta(t)=\ln(1+\exp(t))$ (softplus function), which satisfies the conditions required for delivering a consistent FZ loss function for both quantile and conditional tail expectation (CTE) of a random variable listed in Corollary 5.5 in \cite{FZ_2016} (or Section 2.3 in this paper). Unlike the squared loss for the mean or check loss for the quantile, there is no natural choice of specification for the FZ loss for both quantile and CTE. However, some guidelines for the choices have been provided in previous literature \citep{FZ_2016,DB_2019}, and our specifications follow these guidelines. % There are several reasons we adopt such specifications,
We discuss them as follows. 

One suggested guideline is that the specifications should make the FZ loss positively homogeneous of some order $b\in \mathbb{R}$:
\[
FZ_{\tau}(q(c\boldsymbol{\theta}_{1}),e(c\boldsymbol{\theta}_{2}),cY)=c^{b}FZ_{\tau}(q(\boldsymbol{\theta}_{1}),e(\boldsymbol{\theta}_{2}),Y)
\]for all $c>0$. This property is important since the ordering of the losses should be unaffected by changing the unit of data measurement. For the FZ loss function, previous studies have found that $G_{1}(t)=0$ is a good candidate for ensuring that the FZ loss satisfies the positively homogeneous property, and that $G_{1}(t)=0$ is also the most frequently used specification in practice \citep{FZ_2016,DB_2019,PFC_2019}. Notice that in the FZ loss in (14), setting $G_{1}(q)=G_{1}(y)=0$ does not completely remove the term $1\{y\leq q\}$, an important component in the check loss for estimating the quantile function. In fact, the check loss is implicitly kept in $LQ_{\tau}(q,y)$ in (14) and therefore estimating the quantile function is still viable through the specified FZ loss ($FZ_{\tau}^{sp}$).  

Another guideline for the specification of the FZ loss concerns moment conditions required for deriving the relevant asymptotic results (e.g., assumptions for Lemma A.1 and A.2) and estimating the asymptotic covariance matrix of the estimators, in which the choice of $G_{2}^{\prime}(e)$ and $G_{2}^{\prime\prime}(e)$ should satisfy these required conditions in a parsimonious and least restrictive fashion. %\citep{DB_2019}.
This leads us to the choice of $G_{2}(e)=\ln(1+\exp(e))$. Then both $G_{2}^{\prime}(e)=\exp(e)/(1+\exp(e))$ and $G_{2}^{\prime\prime}(e) = G_{2}^{\prime}(e)(1-G_{2}^{\prime}(e))$ are bounded for all $e\in \mathbb{R}$ and the required moment conditions can be easily satisfied. 

In practice, there is another important reason for setting $G_{2}(e)=\ln(1+\exp(e))$: it does not restrict the sign of the estimated CTE, %(or the domain of $G_{2}(e)$),
which is crucial for applications in microeconometrics. In the literature on estimating structural models for forecasting the expected shortfall (ES, the CTE of an asset's return) using the FZ loss, one common choice of $G_{2}(e)$ is $-\ln(-e)$, which satisfies the requirement for positively homogeneous FZ loss \citep{PFC_2019, Taylor_2019, MT_2020, CYY_2022}. However, the domain of $G_{2}(e)=-\ln(-e)$ is restricted to the negative real line ($e\in \mathbb{R}^{-}$). This is reasonable for the ES, since the ES generally is negative. However, this is a serious limitation on applications in microeconometrics. %this severely limits the usage of the FZ loss. 
In addition, when $G_{2}(e)=-\ln(-e)$, $G_{2}^{\prime}(e)=-e^{-1}$ and $G_{2}^{\prime}(e)=e^{-2}$ are not bounded for all $e\in \mathbb{R}^{-}$. Therefore in this situation, to derive the relevant asymptotic results, our assumptions may not be applicable and new assumptions on the properties of $G_{2}^{\prime}(e)$ and $G_{2}^{\prime\prime}(e)$ may be needed. %under our assumptions.%which might make it  to derive the relevant asymptotic results under our assumptions.

For $\eta(.)$, because it is a function for outcome $Y$ only and does not involve any parameter for the estimation, setting it in an arbitrary (but bounded) manner will not have any effect on the estimates of the parameters. However, as shown in \cite{DB_2019}, by setting $\eta(y)=\tau G_{1}(y)+G_{2}(y)$, the FZ loss can be guaranteed to be a non-negative function. In our case, as we set $G_{1}(y)=0$, $\eta(y)=G_{2}(y)=\ln(1+\exp(y))$ meets this requirement. %for nonnegative FZ loss.
The property of non-negativity for a loss function is useful for evaluating model fit. For example, as the FZ loss is nonnegative, a pseudo-$R^{2}$ for evaluating model fit may be defined as \citep{DB_2019}:
\[
R^{2}_{\tau}=1-\frac{\sum_{i=1}^{n}FZ_{\tau}(q_{i}(\hat{\boldsymbol{\theta}}_{1,\tau}),e_{i}(\hat{\boldsymbol{\theta}}_{2,\tau}),Y_{i})}{\sum_{i=1}^{n}FZ_{\tau}(q_{i}(\tilde{\boldsymbol{\theta}}_{1,\tau}),e_{i}(\tilde{\boldsymbol{\theta}}_{2,\tau}),Y_{i})},
\]where $(\hat{\boldsymbol{\theta}}_{1,\tau},\hat{\boldsymbol{\theta}}_{2,\tau})$ are the parameter estimates of the full model and $(\tilde{\boldsymbol{\theta}}_{1,\tau},\tilde{\boldsymbol{\theta}}_{2,\tau})$ are the parameter estimates of the model with an intercept term only. 

In summary, our specifications for $G_{1}(t)$, $G_{2}(t)$ and $\eta(t)$ satisfy the requirements for delivering a valid consistent FZ loss function for both the quantile and CTE of a random variable. From a practical perspective, our specifications alleviate the impact of changing the units of data measurement on parameter estimates, and do not restrict the sign of the estimated CTE. Due to bounded $G_{2}^{\prime}(t)$ and $G_{2}^{\prime}(t)$ in our specification, the relevant asymptotic results can also be established. Finally, our specified FZ loss is nonnegative, which is helpful for building measures to evaluate model fit.   
\subsubsection{Sensitivity of alternative $G_{1}(t)$ and $G_{2}(t)$}
There exist other specifications for $G_{1}(t)$ and $G_{2}(t)$. For example, $G_{1}(t)=t$ (identity function) and $G_{2}(t)=\exp(t)$ (exponential function) were used in \cite{FZ_2016} and \cite{DB_2019}. These alternative specifications satisfy the requirements in Corollary 5.5 in \cite{FZ_2016}, but do not completely satisfy other aforementioned conditions for making the FZ loss behave well: $G_{1}(t)=t$ cannot ensure that the FZ loss satisfies positive homogeneity \citep{DB_2019}, and $G_{2}(t)=\exp(t)$ is not a bounded function. However, we find that these alternative specifications have a small impact on the resulting estimates in our empirical application (at least they appear visually similar to those obtained using $G_{1}(t)=0$ and $G_{2}(t)=\ln(1+\exp(t))$). 

For implementation, our proposed iterative scheme in Section 2.4 can be applied to these alternative specifications, with a suitable modification on the weight $\hat{\xi}^{(k)}_{i}$ in the weighted quantile regression estimation. The computational complexity and burden of the iterative scheme under the alternative specifications are approximately the same as before. 

Table \ref{table0} lists the combinations of the specifications for $G_{1}(t)$ and $G_{2}(t)$ for different FZ losses, and their corresponding $\hat{\xi}^{(k)}_{i}$. These includes the proposed $FZ_{\tau}^{sp}$, and the other three alternatives: $FZ_{\tau}^{id-sp}$, $FZ_{\tau}^{exp}$ and $FZ_{\tau}^{id-exp}$. In Tables \ref{table1} to \ref{table3}, we report the maximum, mean and minimum values of the differences between the estimates from using the FZ loss with the alternative specifications %on $G_{1}(.)$ and $G_{2}(.)$,
and those from using $FZ_{\tau}^{sp}$, over quantile levels from 0.1 to 0.9. It can be seen that setting $G_{2}(t)=\exp(t)$ in the FZ loss seems to result in a larger deviation. However, compared with the estimates, the magnitude of the deviation is relatively small. In Figures \ref{figure8} to \ref{figure10} we graphically compare the resulting estimates. Overall, the estimates generated from using the alternative FZ losses are visually very similar to those generated from using $FZ_{\tau}^{sp}$, and we can conclude that the estimates are not sensitive to different specifications on the FZ losses.

\subsection{Comparisons with the weighted quantile regression in \cite{AAI_2002} on simulations and empirical application}
In this section we compare the estimator of QTE from using the FZ loss with the QTE estimator in \cite{AAI_2002} through simulations and the empirical applications of the JTPA. We use the same data generation process as in Section 4, and show the biases, variances and MSEs of the QTE estimators using the weighted quantile regression (WQR) in \cite{AAI_2002} (denoted by AAI) and that using the FZ loss (denoted by FZ), under different settings for the weight estimations (M1 to M4, see Section 4). Figures \ref{figure11} to \ref{figure14} present the simulation results when $\rho = 0$ (no endogeneity) and $\rho = 0.5$ (endogeneity), each with sample sizes for the simulation $n = 500$ and 3,000. From the figures, when the sample size is large ($n = 3,000$), the simulation results for the WQR estimator are very similar to those for the estimator using the FZ loss (as discussed in Section 4). But when the sample size is small ($n = 500$), the two estimators show some differences. 

As the figures show, the QTE estimator using the WQR tends to have a lower variance (shown in the middle row of each figure) than that using the FZ loss when the sample size is small (500). This is expected, since in the FZ loss estimation, additional parameters (those of the CTE function) need to be estimated. Consequently the variance of each parameter estimation tends to be higher, given that the sample size is fixed. However, this discrepancy in variances vanishes as the sample increases, as shown in the cases of $n = 3,000$.

Except for the case when $\rho = 0.5$ and no endogeneity adjustment (M3), the magnitude of biases generated by the two estimators is relatively much smaller than that of the corresponding variances. The variance dominates in the MSE calculation in these cases, and the WQR estimator therefore has a lower MSE (shown in the bottom row of each figure) than the estimator using the FZ loss. When $\rho = 0$ and the sample size is small, the estimator using the FZ loss seems to have a more downward bias than the WQR estimator, but the discrepancy of biases becomes less obvious as the sample size becomes large. As for $\rho = 0.5$, there is an endogeneity issue, and the same as in Section 4, the case of no endogeneity adjustment (M3) results in a much higher bias and MSE than the other three methods (either using the weight to adjust the endogeneity or using only samples of compliers for the estimation), although it still results in a lower variance (all samples are used in the estimations and there is no estimated weight). The high MSE of M3 is mostly due to the high bias from a lack of adjustment for endogeneity.

In summary, from the simulations, we find that the WQR estimator in \cite{AAI_2002} generates lower variance and MSE than the estimator using the FZ loss to estimate the conditional QTE for compliers in a small sample size. This is because we need to estimate the additional parameter of the CTE function when using the FZ loss. However, in the large sample, the differences in variance and MSE becomes negligible for the two estimators. 

We then use the WQR of \cite{AAI_2002} to re-estimate the conditional QTE for compliers with JTPA data. Figure \ref{figure15} shows graphical comparisons and Table \ref{table4} shows the minimum, mean and maximum values of the differences between these WQR estimates (denoted by AAI) and those obtained using the FZ loss (denoted by FZ), over quantile levels ranging from 0.1 to 0.9. The weight $\tilde{K}_{i}$ is the same for both estimations. The results are quite similar. Notice that we use the same functional form (linear in parameters) as in \cite{AAI_2002} for modelling the conditional quantile of the potential outcome for compliers, and both estimators are asymptotically equivalent (consistent for estimating $\boldsymbol{\theta}_{1}$). In addition, as suggested in Section 2.4, $\hat{\boldsymbol{\theta}}^{(1)}_{1}$ used for the first iteration step in the proposed iterative scheme (for estimation using the FZ loss), was obtained from the WQR estimates of \cite{AAI_2002}. %(with the same weight $\tilde{K}_{i}$).
Finally, large sample sizes (5,102 for adult men and 6,102 for adult women) are also helpful on stabilizing the performances of the estimations. These are the possible reasons why the differences between the two estimates are very small. 

We emphasize that in this paper, our main focus is on estimating the CTATE for compliers, which can not be achieved by using only the WQR of \cite{AAI_2002}. However, this does not imply that estimating the conditional quantile function is unimportant. A valid estimation of the conditional quantile function perhaps is the most crucial for a valid estimation of the CTE function using the FZ loss, since the two are jointly estimated, and their estimation qualities are closely related. %As estimates of the conditional QTE for compliers from using the WQR of \cite{AAI_2002} and using the FZ loss are very similar, this provides evidence to support validity of our estimation of the CTATE for compliers. 
\clearpage
\begin{table}
	\centering
	\caption{Specifications on $G_{1}(t)$ and $G_{2}(t)$ for different FZ losses.}
	\begin{tabular}{ccccc}
		\hline 
		& $FZ_{\tau}^{sp}$ & $FZ_{\tau}^{id-sp}$ & $FZ_{\tau}^{exp}$ & $FZ_{\tau}^{id-exp}$\tabularnewline
		\hline 
		$G_{1}(t)$ & 0 & $t$ & 0 & $t$\tabularnewline
		$G_{2}(t)$ & $\ln\left(1+\exp\left(t\right)\right)$ & $\ln\left(1+\exp\left(t\right)\right)$ & $\exp\left(t\right)$ & $\exp\left(t\right)$\tabularnewline
		$G_{2}^{\prime}(t)$ & $\frac{\exp\left(t\right)}{1+\exp\left(t\right)}$ & $\frac{\exp\left(t\right)}{1+\exp\left(t\right)}$ & $\exp\left(t\right)$ & $\exp\left(t\right)$\tabularnewline
		$\tilde{\xi}_{i}^{(k)}$ & $\tilde{K}_{i}\frac{\exp[e_{i}(\hat{\boldsymbol{\theta}}_{2}^{(k)})]}{1+\exp[e_{i}(\hat{\boldsymbol{\theta}}_{2}^{(k)})]}$ & $\left(1+\frac{1}{\tau}\right)\tilde{K}_{i}\frac{\exp[e_{i}(\hat{\boldsymbol{\theta}}_{2}^{(k)})]}{1+\exp[e_{i}(\hat{\boldsymbol{\theta}}_{2}^{(k)})]}$ & $\tilde{K}_{i}\exp[e_{i}(\hat{\boldsymbol{\theta}}_{2}^{(k)})]$ & $\left(1+\frac{1}{\tau}\right)\tilde{K}_{i}\exp[e_{i}(\hat{\boldsymbol{\theta}}_{2}^{(k)})]$\tabularnewline
		\hline 
	\end{tabular}
	\label{table0}
\end{table}
%FZ\_id vs. FZ, Units: USD (thousands dollars)
\begin{table}
	\centering
	\caption{
		Minimum, mean and maximum values of differences between the CTATE and QTE estimates for compliers when using $FZ_{\tau}^{id-sp}$ and $FZ_{\tau}^{sp}$, over quantile levels ranging from 0.1 to 0.9. The unit is thousands USD.}
	%Maximum, mean and minimum of differences between estimates of the CTATE and QTE over quantile levels 0.1 to 0.9 from using $FZ_{\tau}^{id-sp}$ and $FZ_{\tau}^{sp}$, Units: USD (thousands dollars)}
\begin{tabular}{cccccccc}
	\hline 
	&  & Men &  &  &  & Women & \tabularnewline
	\cline{2-4} \cline{3-4} \cline{4-4} \cline{6-8} \cline{7-8} \cline{8-8} 
	& Min. & Mean & Max. &  & Min. & Mean & Max.\tabularnewline
	\hline 
	Endo. Treat. &  &  &  &  &  &  & \tabularnewline
	CTATE & -0.0508 & -0.0020 & 0.0168 &  & -0.0157 & -5e-04 & 0.0015\tabularnewline
	QTE & -0.0713 & -0.0026 & 0.0092 &  & -0.0118 & 0.0000 & 0.0131\tabularnewline
	No. Endo. Treat. &  &  &  &  &  &  & \tabularnewline
	CTATE & -0.0460 & -8e-04 & 0.0424 &  & -0.0165 & -3e-04 & 0.0060\tabularnewline
	QTE & -0.0098 & 1e-04 & 0.0137 &  & -0.0115 & 2e-04 & 0.0177\tabularnewline
	\hline 
\end{tabular}
\label{table1}
\end{table}

%FZ\_exp vs. FZ, Units: USD (thousands dollars)

\begin{table}
\centering
\caption{
	Minimum, mean and maximum values of differences between the CTATE and QTE estimates for compliers when using $FZ_{\tau}^{exp}$ and $FZ_{\tau}^{sp}$, over quantile levels ranging from 0.1 to 0.9. The unit is thousands USD.
}

%Maximum, mean and minimum of differences between estimates of the CTATE and QTE over quantile levels 0.1 to 0.9 from using $FZ_{\tau}^{exp}$ and $FZ_{\tau}^{sp}$ , Units: USD (thousands dollars)}
\begin{tabular}{cccccccc}
\hline 
&  & Men &  &  &  & Women & \tabularnewline
\cline{2-4} \cline{3-4} \cline{4-4} \cline{6-8} \cline{7-8} \cline{8-8} 
& Min. & Mean & Max. &  & Min. & Mean & Max.\tabularnewline
\hline 
Endo. Treat. &  &  &  &  &  &  & \tabularnewline
CTATE & -0.0098 & 0.0152 & 0.0535 &  & -0.0065 & 0.0020 & 0.0315\tabularnewline
QTE & -0.0062 & 0.0056 & 0.1196 &  & -0.0014 & 0.0047 & 0.0140\tabularnewline
No. Endo. Treat. &  &  &  &  &  &  & \tabularnewline
CTATE & -4e-04 & 0.0290 & 0.0537 &  & -5e-04 & 0.0049 & 0.0209\tabularnewline
QTE & -0.0564 & 0.0000 & 0.0740 &  & -0.0213 & 0.0000 & 0.0191\tabularnewline
\hline 
\end{tabular}
\label{table2}
\end{table}
%FZ\_id\_exp vs. FZ, Units: USD (thousands dollars)

\begin{table}
\centering
\caption{
Minimum, mean and maximum values of differences between the CTATE and QTE estimates for compliers when using $FZ_{\tau}^{id-exp}$ and $FZ_{\tau}^{sp}$, over quantile levels 0.1 to 0.9. The unit is thousands USD.
}

%Maximum, mean and minimum of differences between estimates of the CTATE and QTE over quantile levels 0.1 to 0.9 from using $FZ_{\tau}^{id-exp}$ and $FZ_{\tau}^{sp}$, Units: USD (thousands dollars)}
\begin{tabular}{cccccccc}
\hline 
&  & Men &  &  &  & Women & \tabularnewline
\cline{2-4} \cline{3-4} \cline{4-4} \cline{6-8} \cline{7-8} \cline{8-8} 
& Min. & Mean & Max. &  & Min. & Mean & Max.\tabularnewline
\hline 
Endo. Treat. &  &  &  &  &  &  & \tabularnewline
CTATE & -0.0108 & 0.0167 & 0.0585 &  & -0.0014 & 0.0044 & 0.0155\tabularnewline
QTE & -0.0150 & 0.0022 & 0.0699 &  & -0.0065 & 3e-04 & 0.0087\tabularnewline
No. Endo. Treat. &  &  &  &  &  &  & \tabularnewline
CTATE & -5e-04 & 0.0303 & 0.0565 &  & -5e-04 & 0.0046 & 0.0209\tabularnewline
QTE & -0.0076 & 4e-04 & 0.0320 &  & -0.0018 & 3e-04 & 0.0123\tabularnewline
\hline 
\end{tabular}
\label{table3}
\end{table}

\begin{table}
	\centering
	\caption{Minimum, mean and maximum values of differences between the QTE estimates for compliers when using the weighted quantile regression (WQR) of \cite{AAI_2002} and when using the FZ loss, over quantile levels 0.1 to 0.9. The unit is thousands USD.}
	\begin{tabular}{cccc}
		\hline 
		& Min. & Mean & Max.\tabularnewline
		\hline 
		Adult Men &  &  & \tabularnewline
		Endo Treat. & -0.0720 & -0.0061 & 0.0247\tabularnewline
		No Endo. Treat. & -0.0161 & -3.0e-04 & 0.0158\tabularnewline
		&  &  & \tabularnewline
		Adult Women &  &  & \tabularnewline
		Endo Treat. & -0.0246 & 1.0e-04 & 0.0240\tabularnewline
		No Endo. Treat. & -0.0115 & -1.0e-04 & 0.0177\tabularnewline
		\hline 
	\end{tabular}
	\label{table4}
\end{table}
\begin{figure}[ht]
	\captionsetup[subfigure]{justification=centering}
\begin{subfigure}[b]{0.44\textwidth}
			\includegraphics[width=\textwidth]{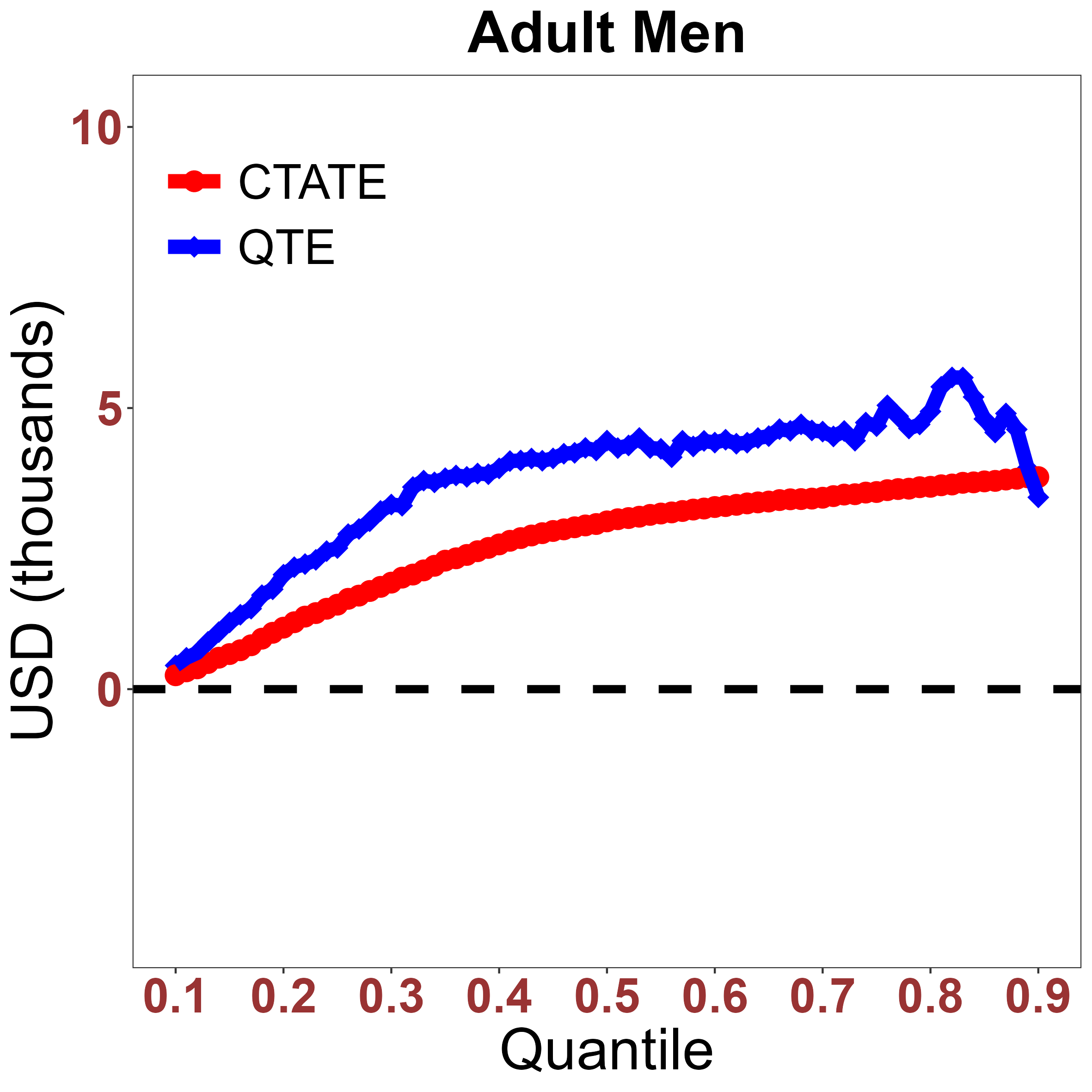}
\end{subfigure}			
 \hfill
 \begin{subfigure}[b]{0.44\textwidth}	
			\includegraphics[width=\textwidth]{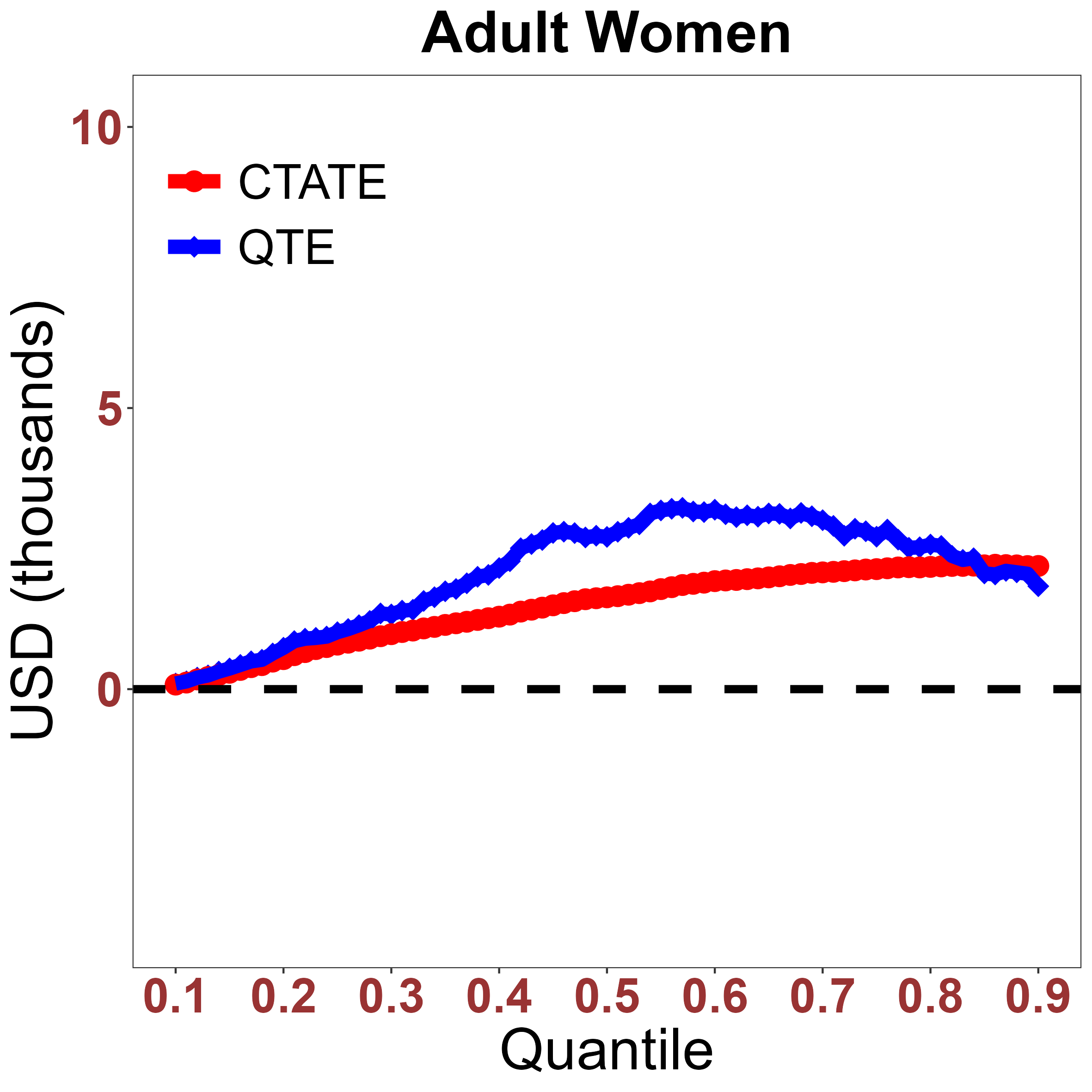}
\end{subfigure}\\			
				
\begin{subfigure}[b]{0.44\textwidth}
		\includegraphics[width=\linewidth]{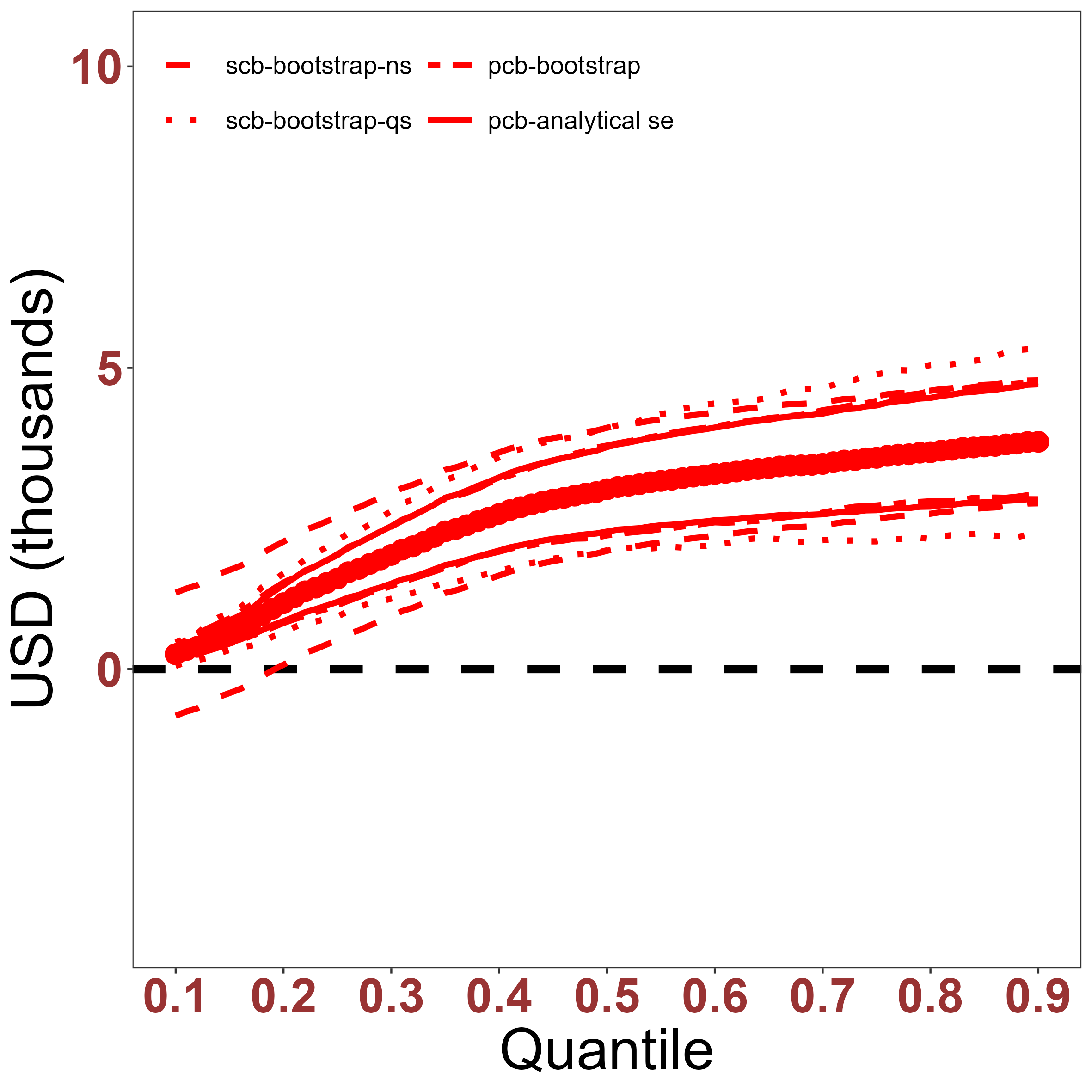}
	\end{subfigure}			
 \hfill
	\begin{subfigure}[b]{0.44\textwidth}	

		\includegraphics[width=\linewidth]{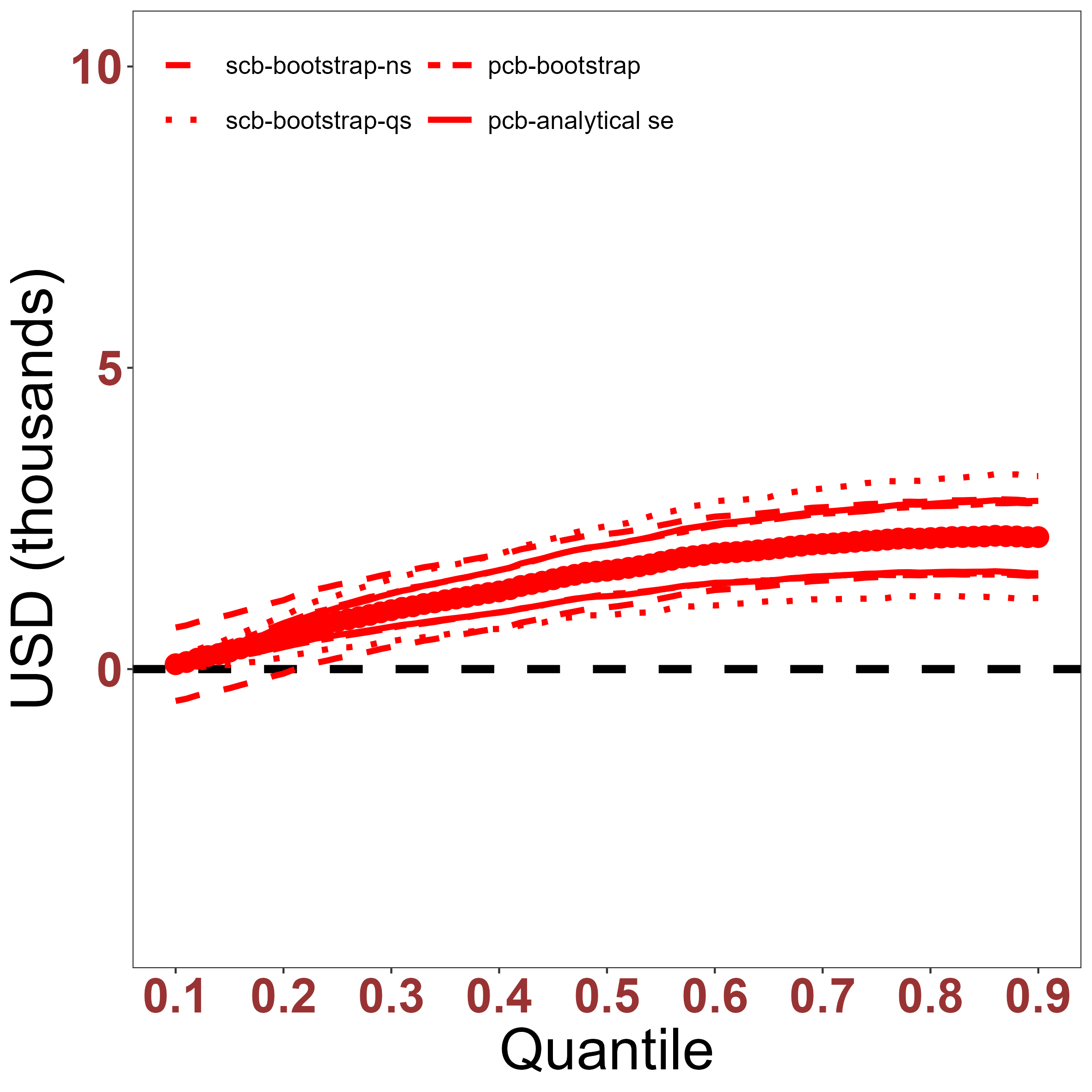}
	\end{subfigure}	\\	
	
	\begin{subfigure}[b]{0.44\textwidth}
		
		\includegraphics[width=\textwidth]{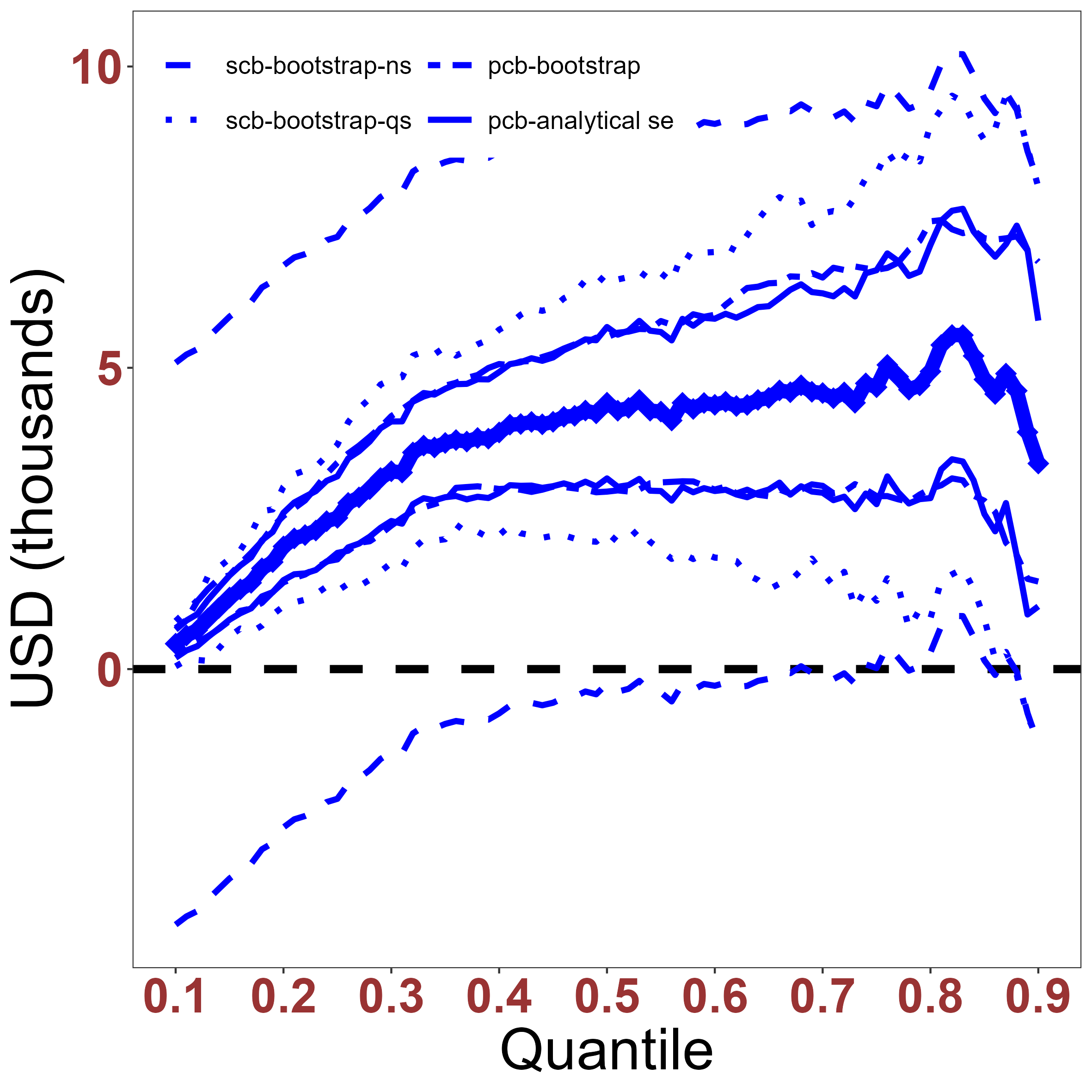}
	\end{subfigure}			
	\hfill
	\begin{subfigure}[b]{0.44\textwidth}	
		
		\includegraphics[width=\linewidth]{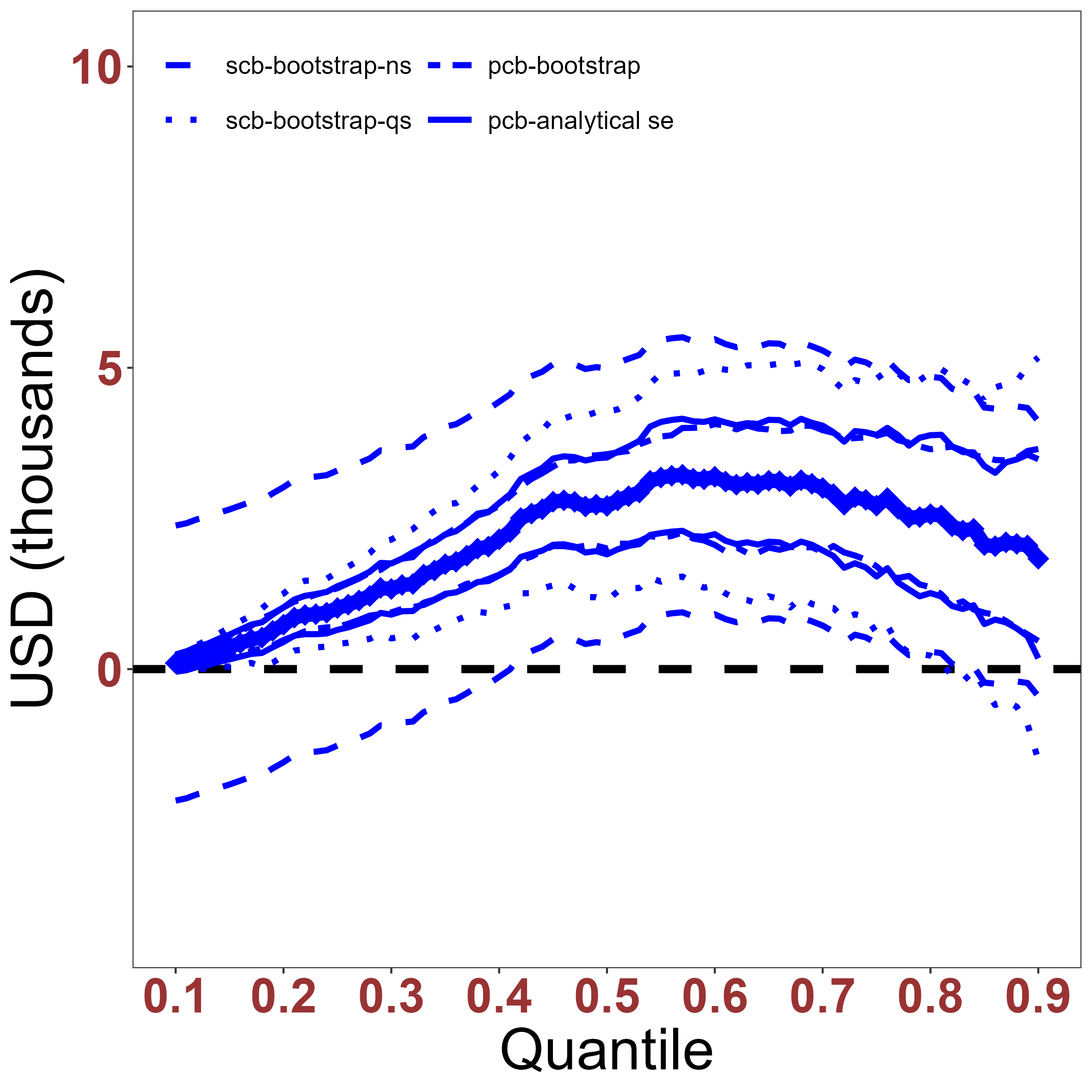}
	\end{subfigure}
		
	\caption{CTATE and QTE estimates of adult men's and women's earnings for compliers and the 95\% pointwise and simultaneous confidence bands when the treatment is exogenous. %The exogenous variables used for the estimations are the same as those used in \citet{AAI_2002}.
		}
	\label{figure1}
\end{figure}

\begin{figure}[ht]
	\captionsetup[subfigure]{justification=centering}
	\begin{subfigure}[b]{0.49\textwidth}
		\includegraphics[width=\textwidth]{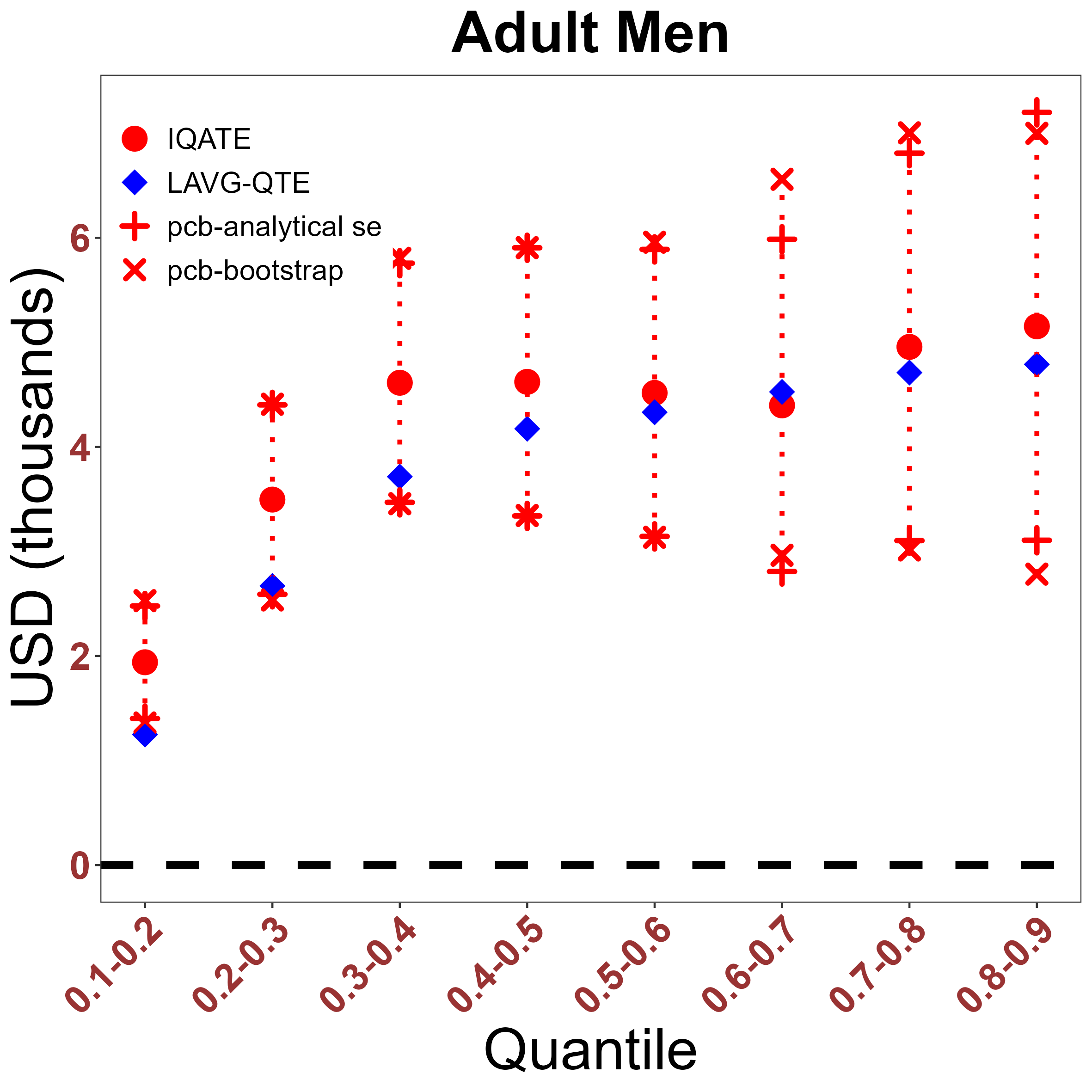}
	\end{subfigure}			
	\hfill
	\begin{subfigure}[b]{0.49\textwidth}	
		\includegraphics[width=\textwidth]{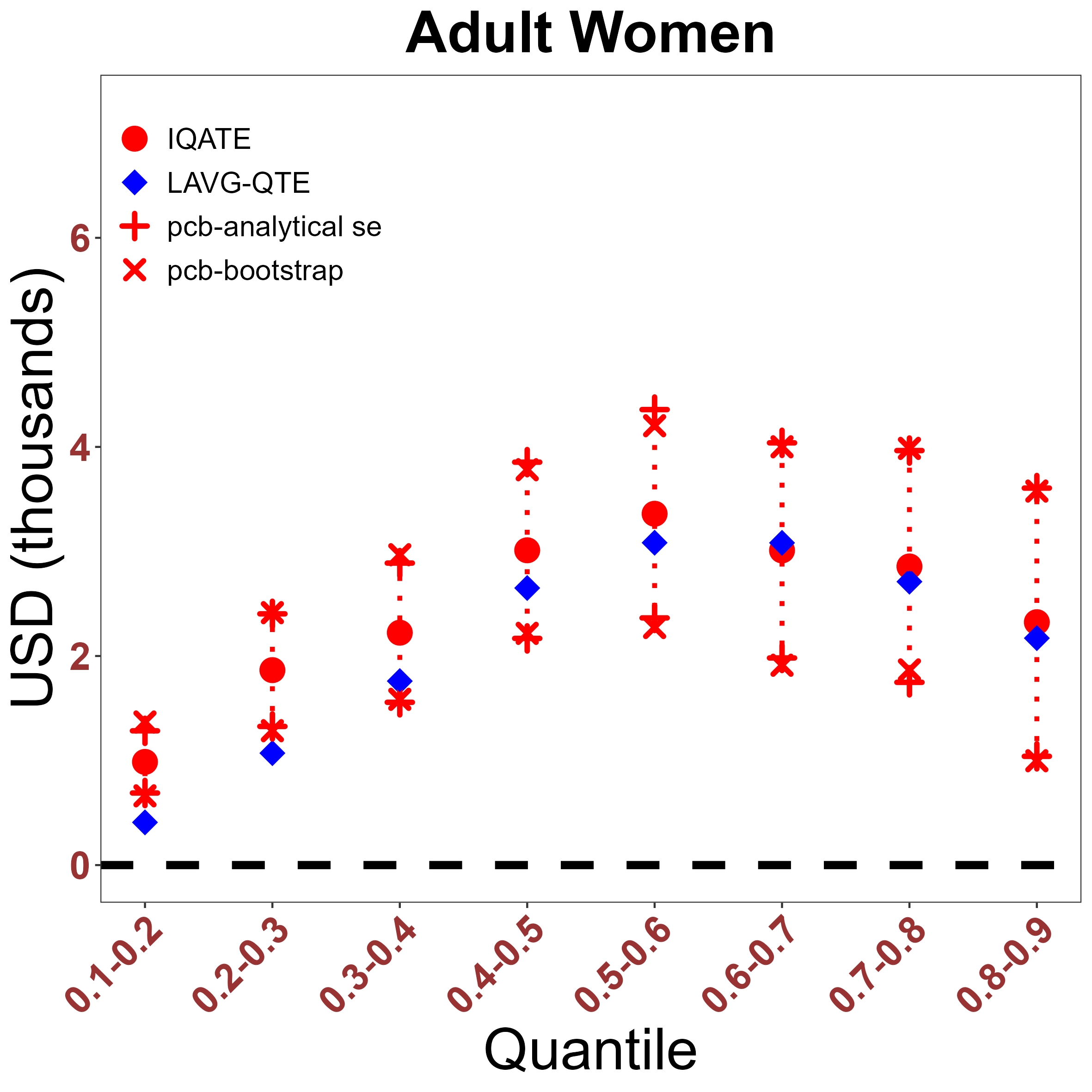}
	\end{subfigure}
	\caption{IQATE estimates of adult men's and women's earnings for compliers and the 95\% pointwise confidence bands when the treatment is exogenous.%The exogenous variables used are the same as those used in \citet{AAI_2002}.
		}
	\label{figure2}
\end{figure}
\clearpage
\begin{figure}[!htb]
	%\begin{subfigure}
	\centering
	\mbox{
		\includegraphics[width = 3.9cm, height = 3.2cm]{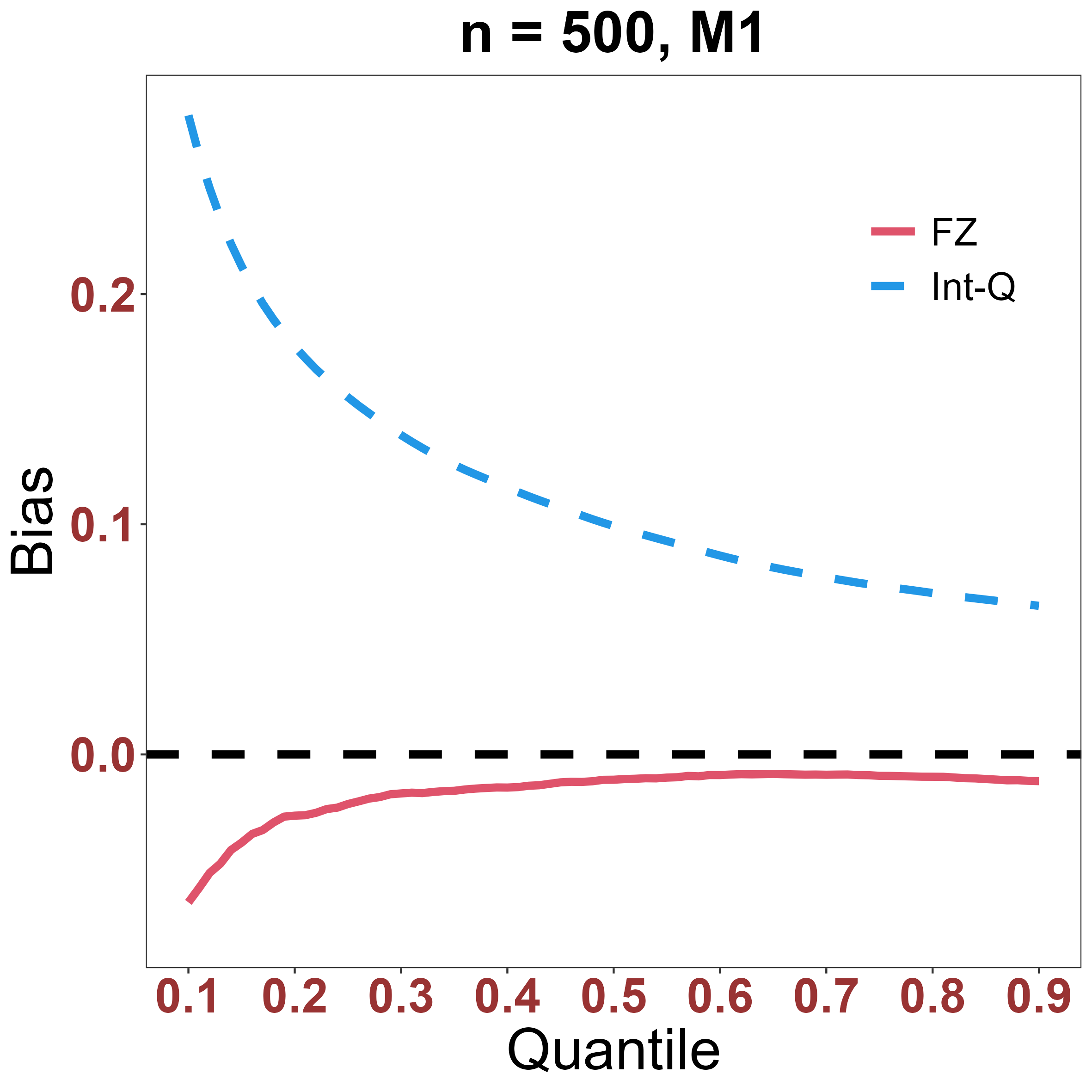}
		\includegraphics[width = 3.9cm, height = 3.2cm]{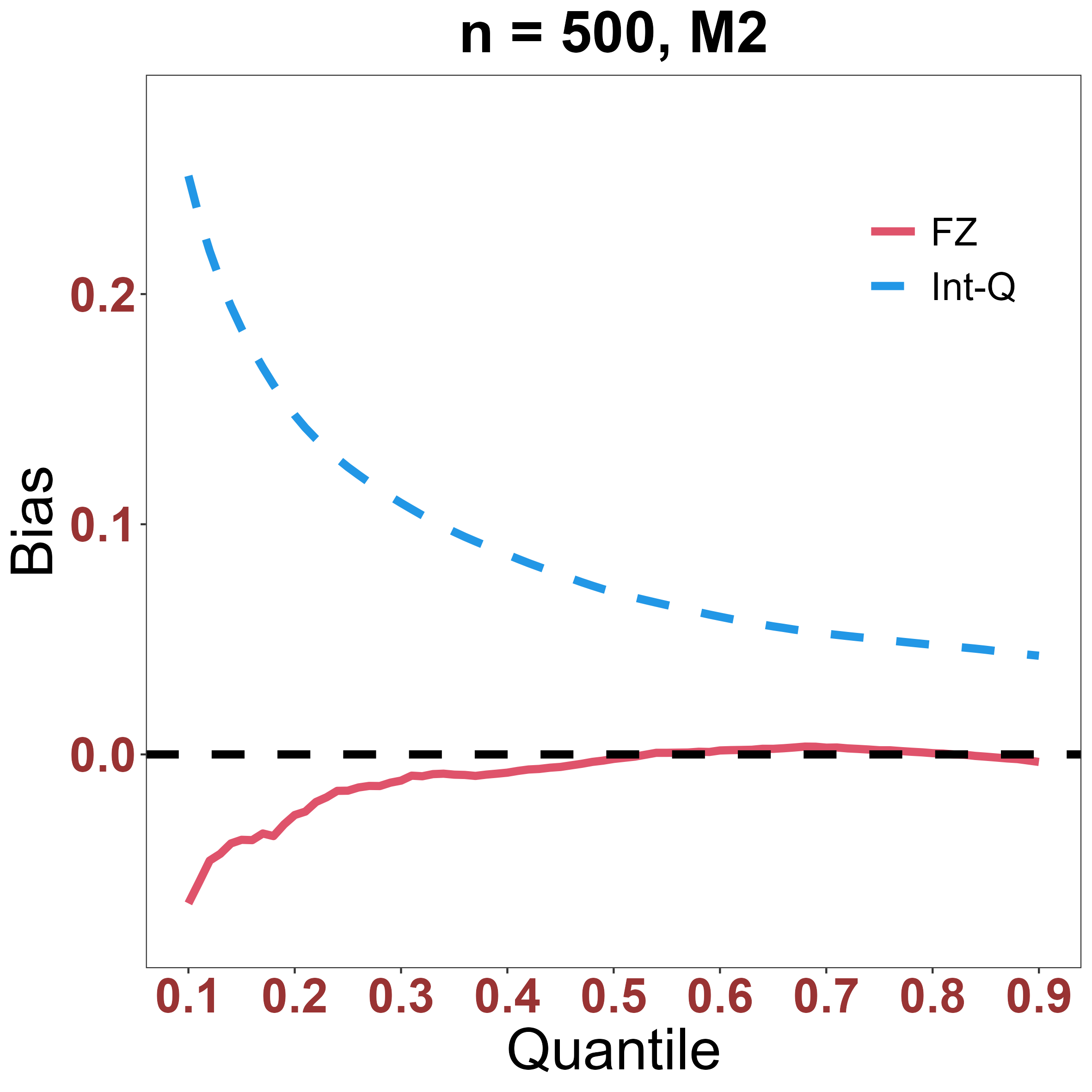}
		\includegraphics[width = 3.9cm, height = 3.2cm]{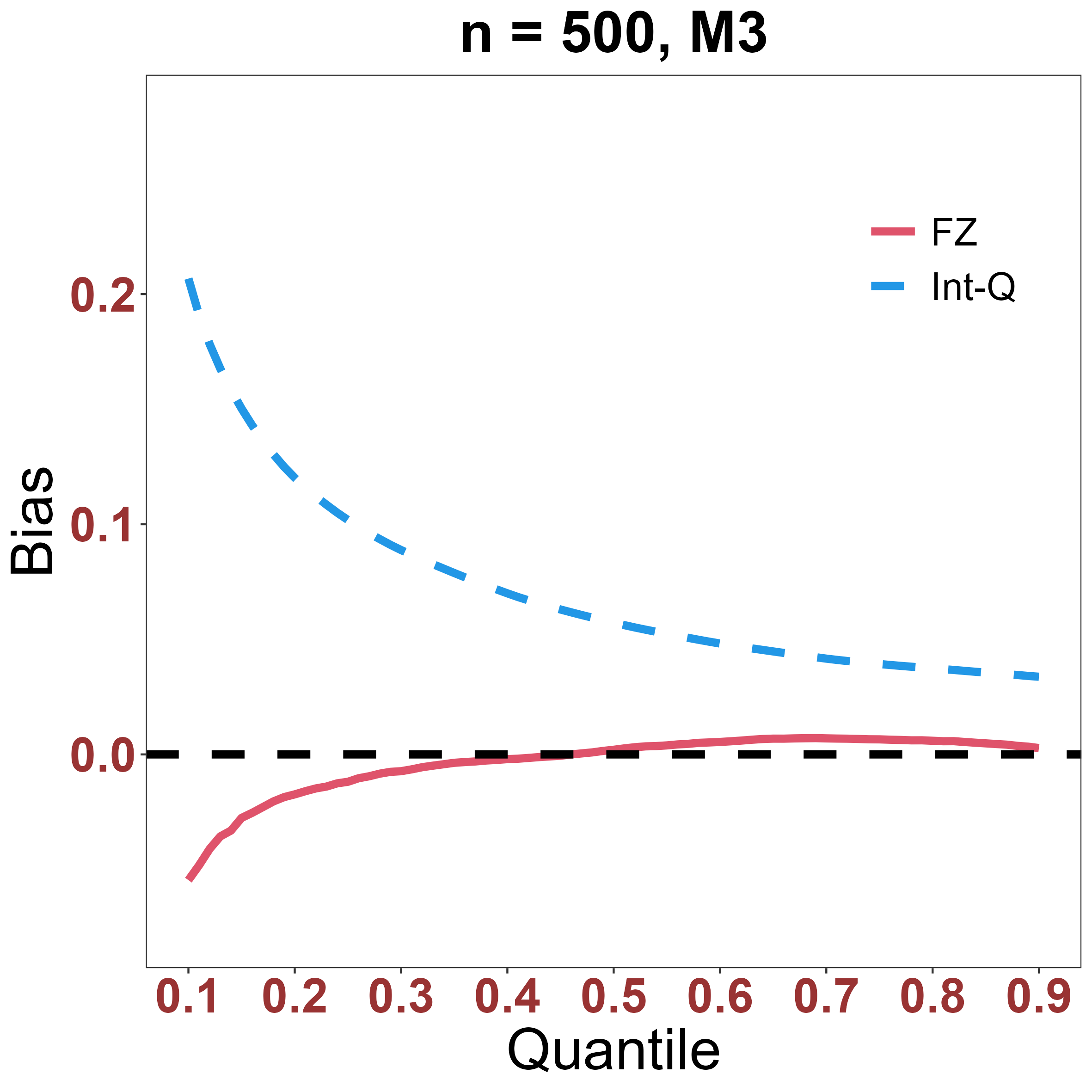}	
		\includegraphics[width = 3.9cm, height = 3.2cm]{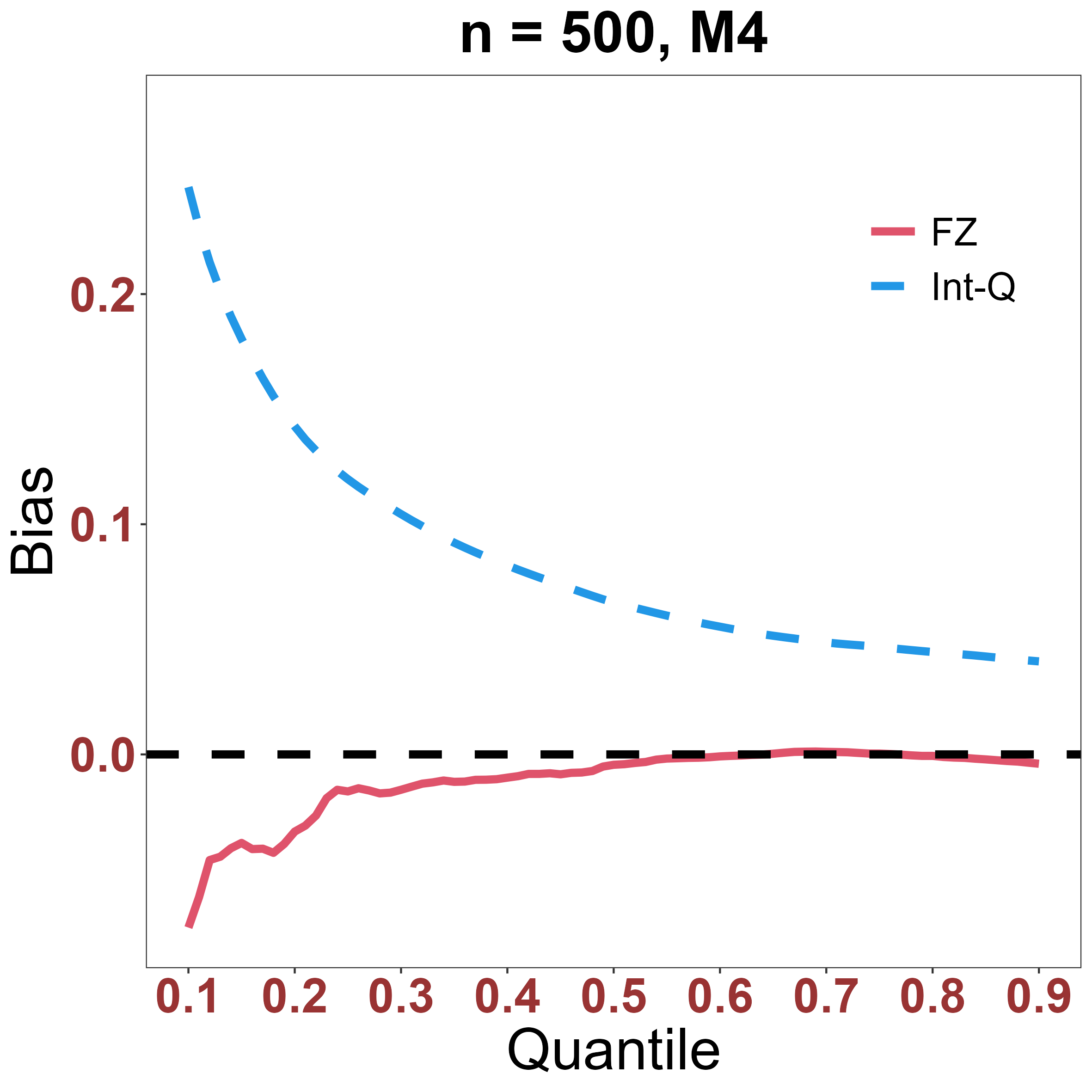}
	}	
	\mbox{	\includegraphics[width = 3.9cm, height = 3.2cm]{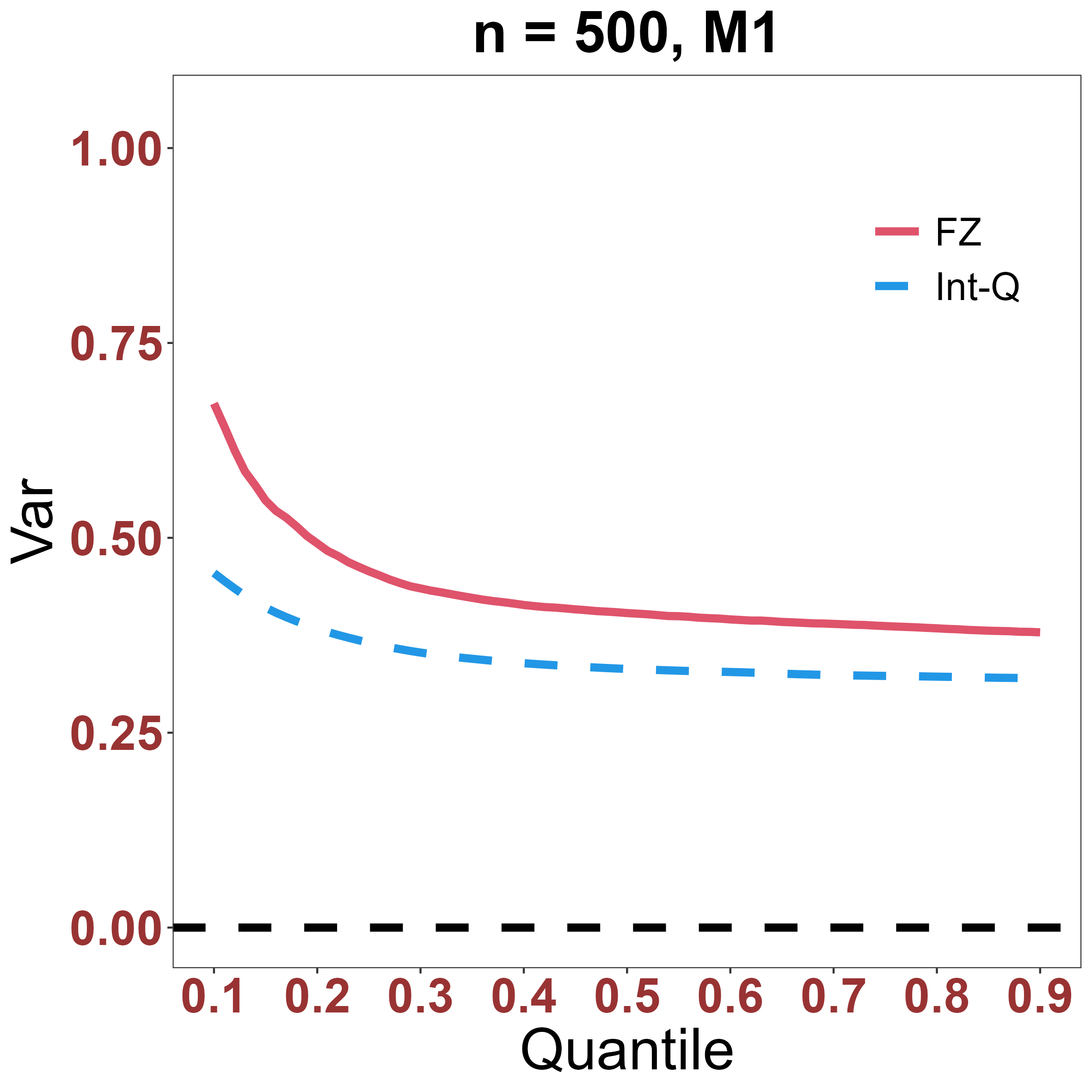}
		\includegraphics[width = 3.9cm, height = 3.2cm]{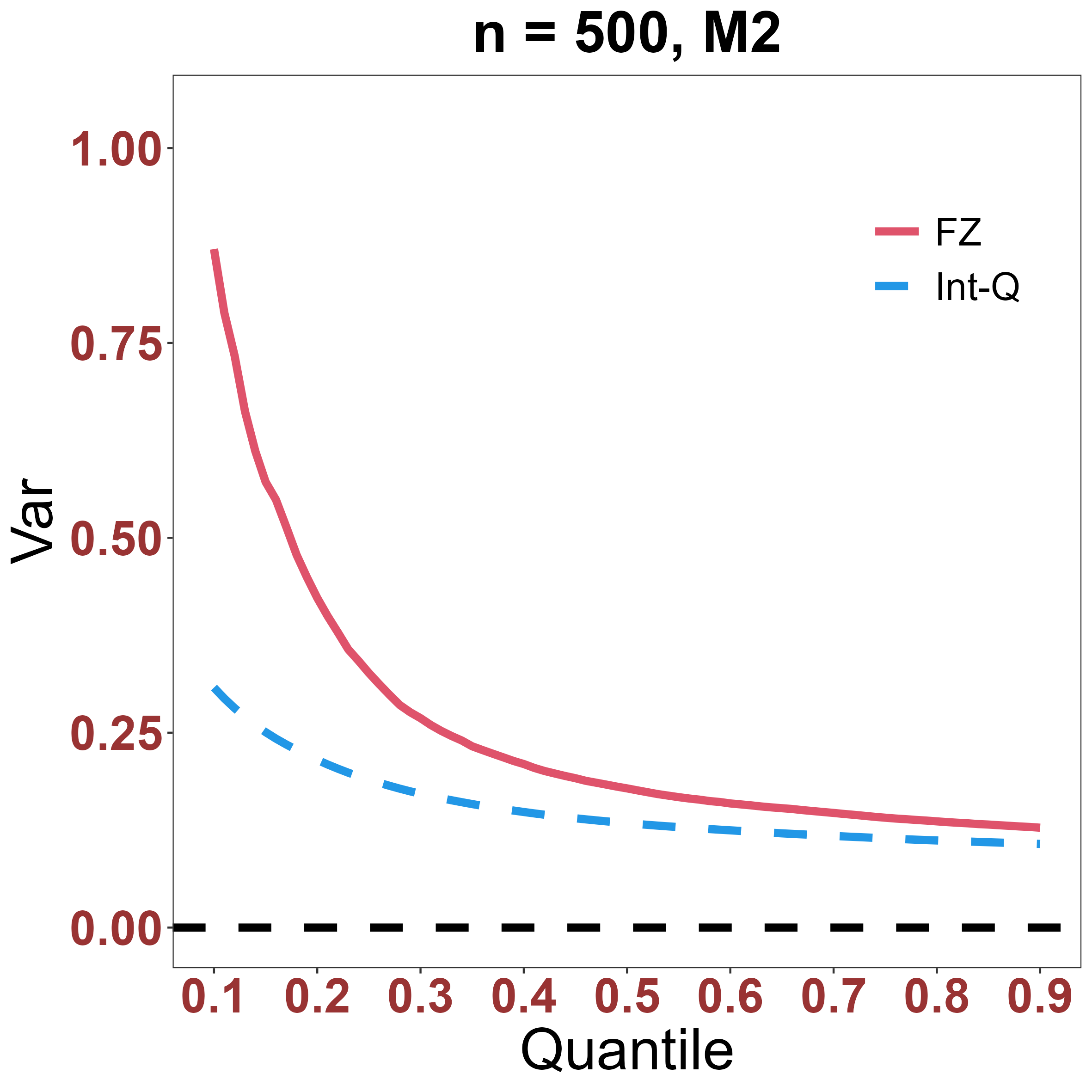}
		\includegraphics[width = 3.9cm, height = 3.2cm]{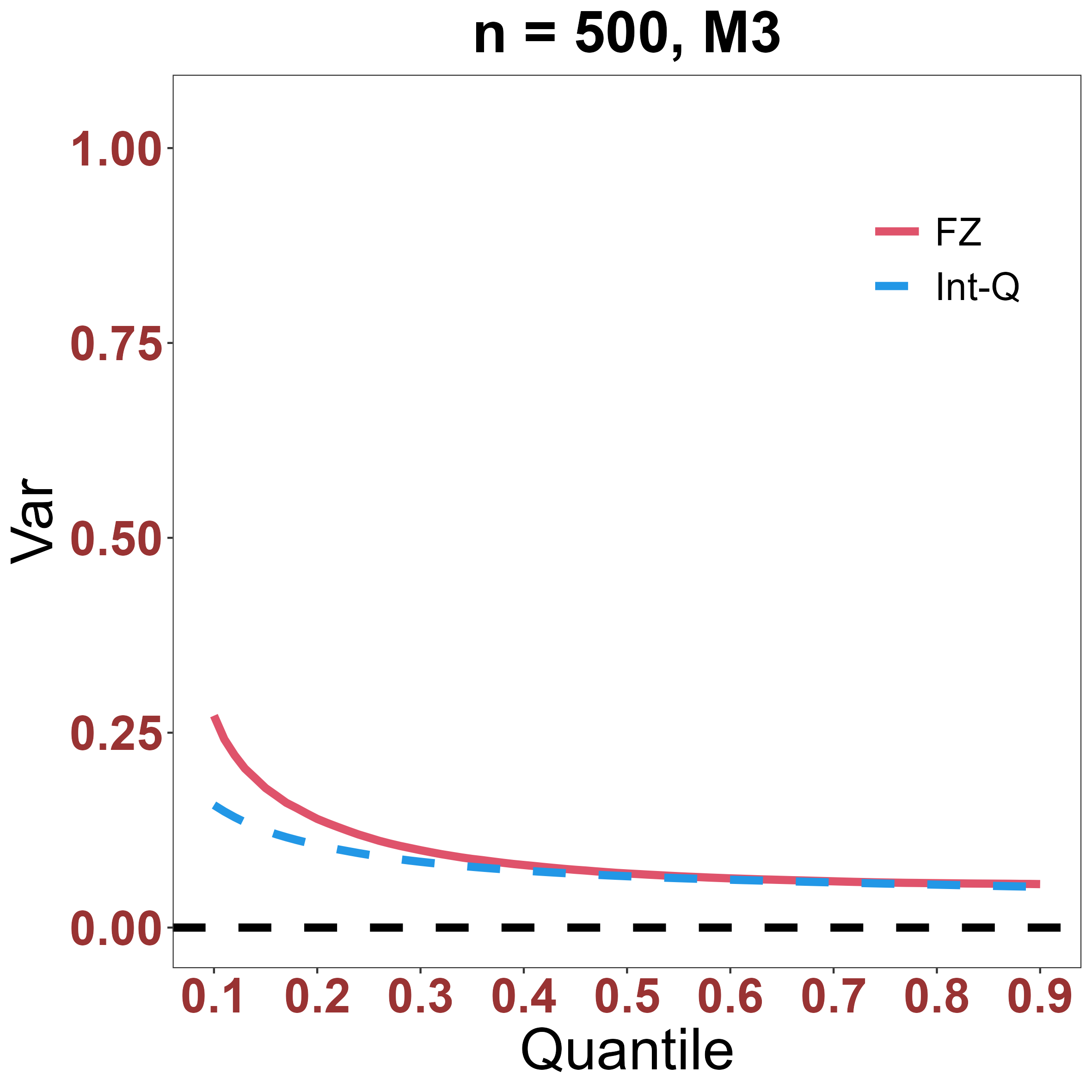}	
		\includegraphics[width = 3.9cm, height = 3.2cm]{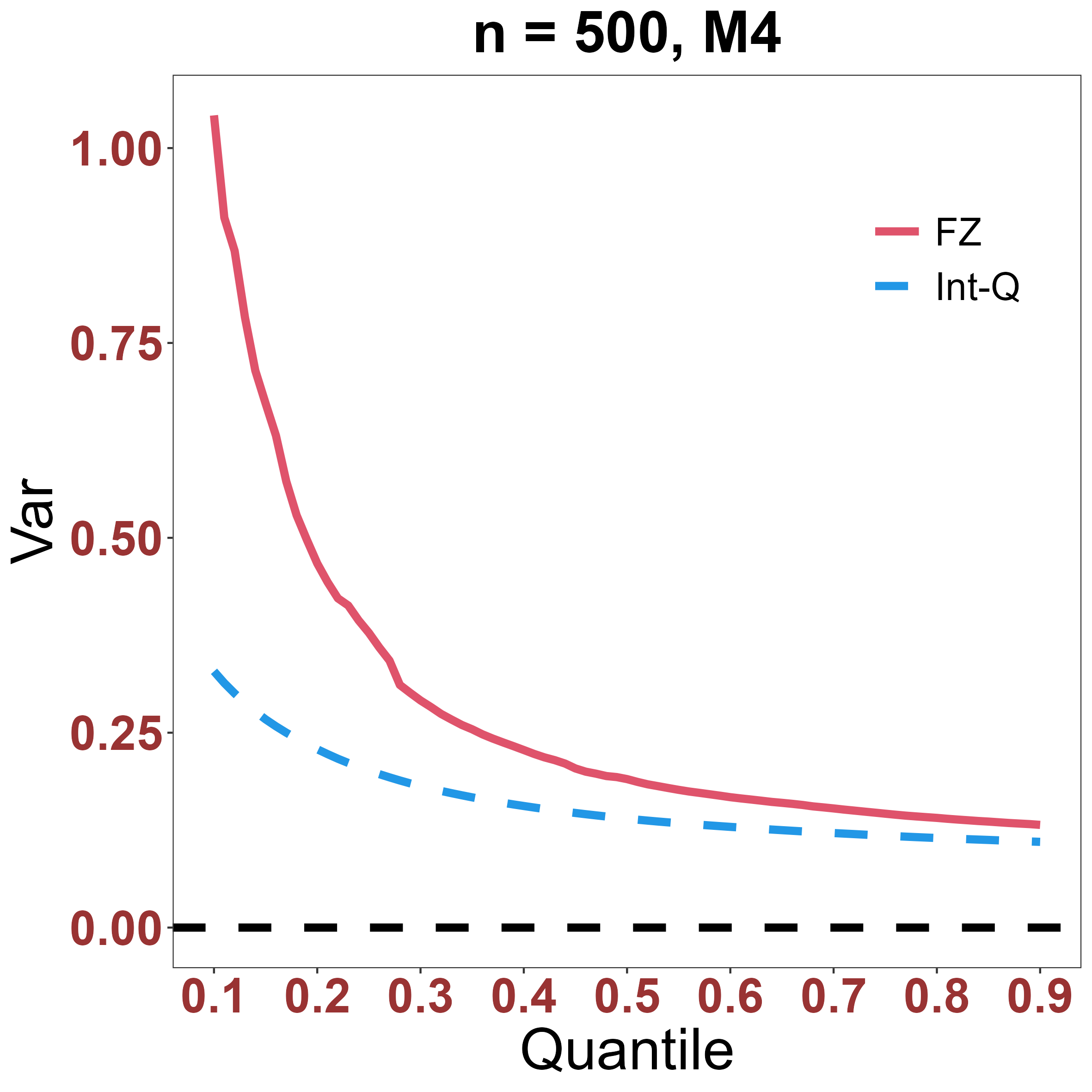}
	}	
	\mbox{	\includegraphics[width = 3.9cm, height = 3.2cm]{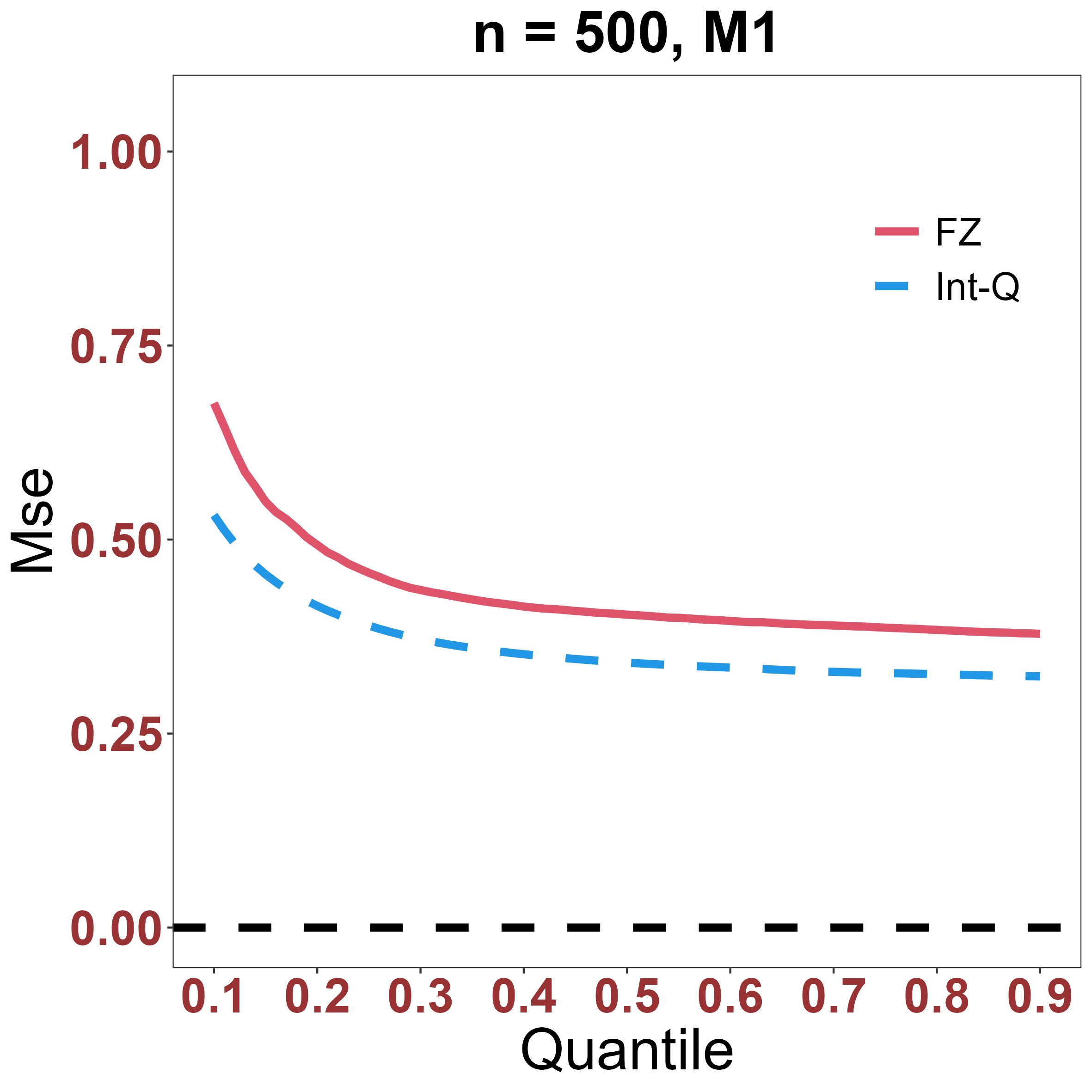}
		\includegraphics[width = 3.9cm, height = 3.2cm]{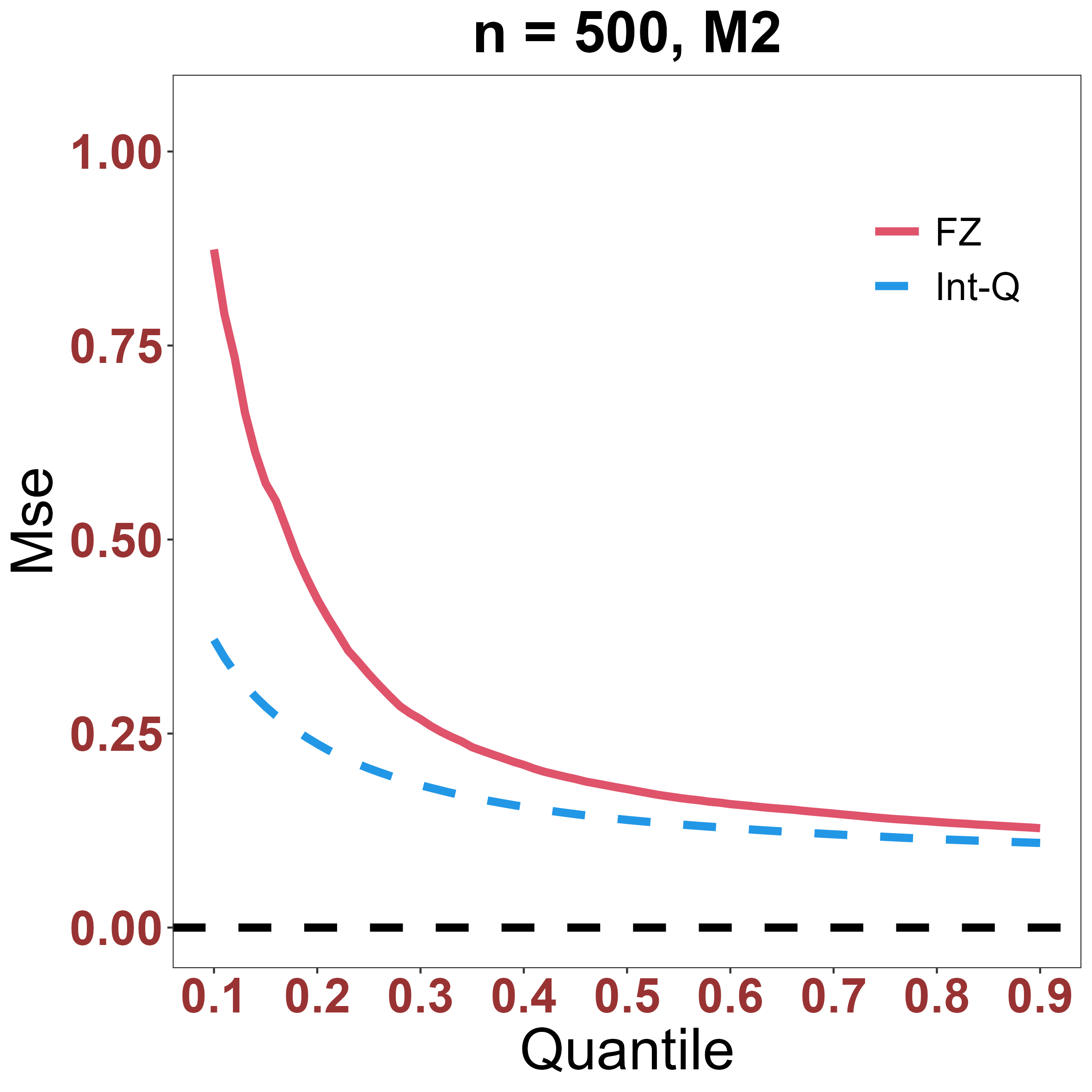}
		\includegraphics[width = 3.9cm, height = 3.2cm]{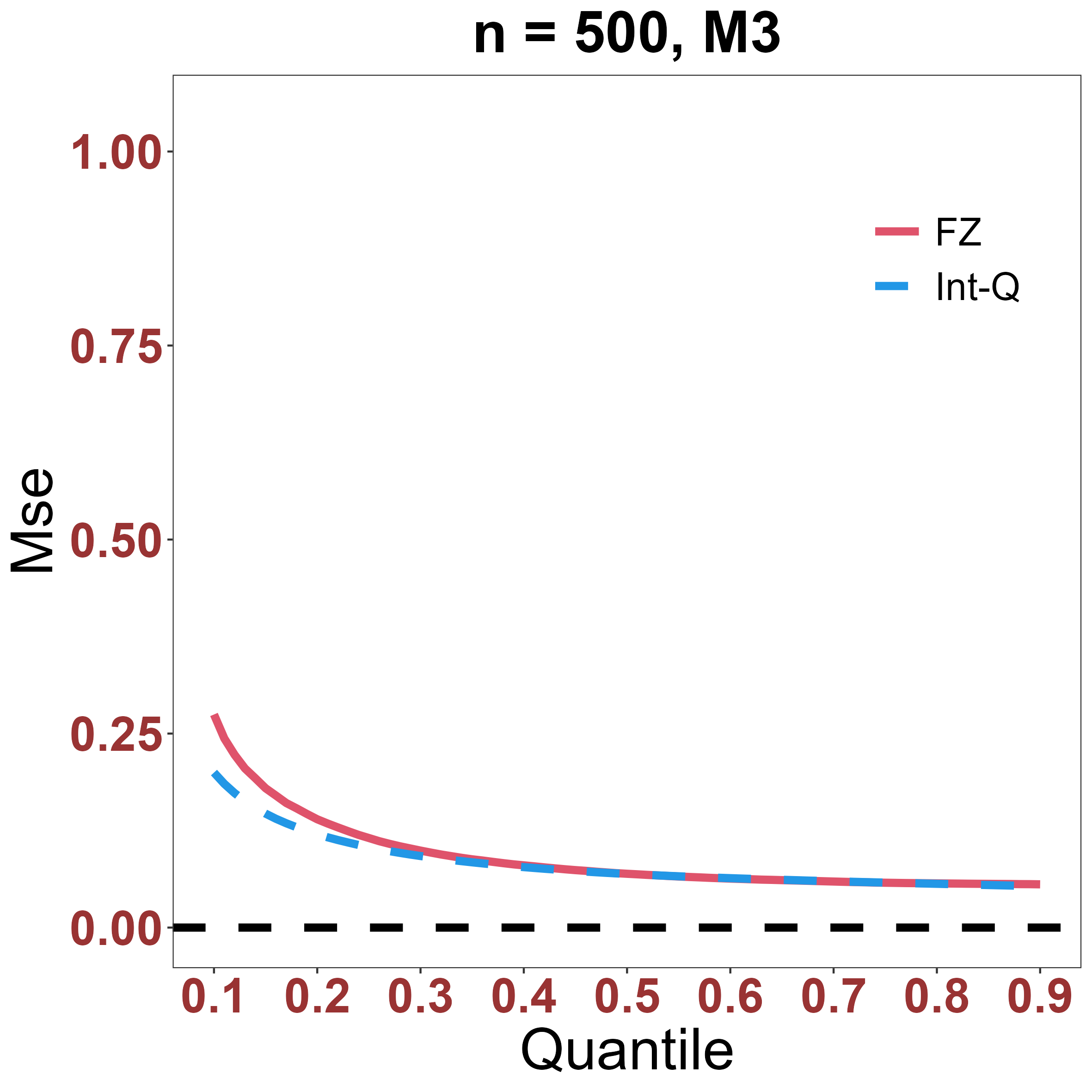}	
		\includegraphics[width = 3.9cm, height = 3.2cm]{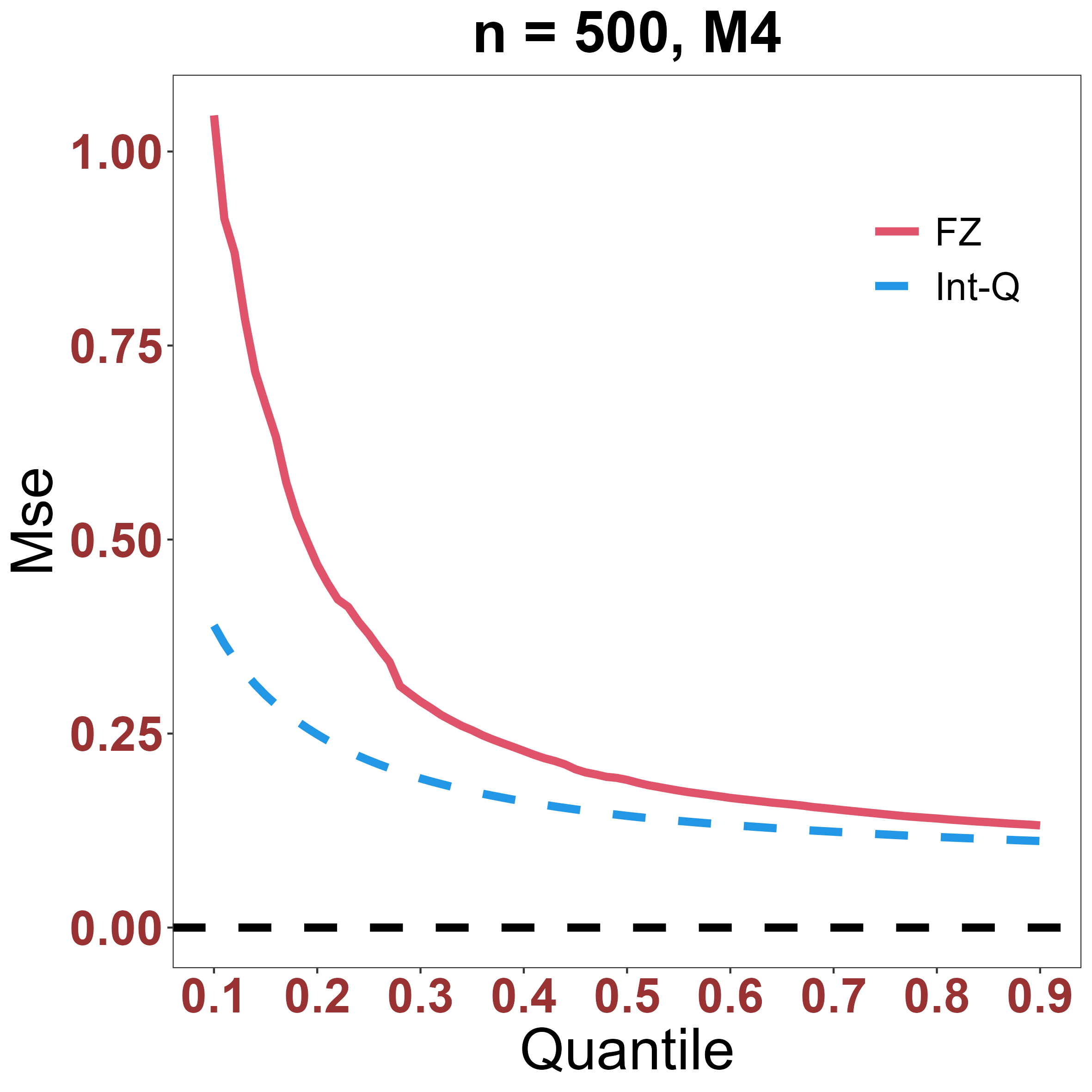}	
	}
	%\end{subfigure}
	\caption{Bias, variance and MSE of the CTATE estimator using the FZ loss (FZ) and that using the trimmed integrated-QTE based estimators (IntQ and IntQ') when $\rho=0$ and $n=500$.}
	\label{figure3}
\end{figure}

\begin{figure}[!htb]
	%\begin{subfigure}
	\centering
	\mbox{
		\includegraphics[width = 3.9cm, height = 3.2cm]{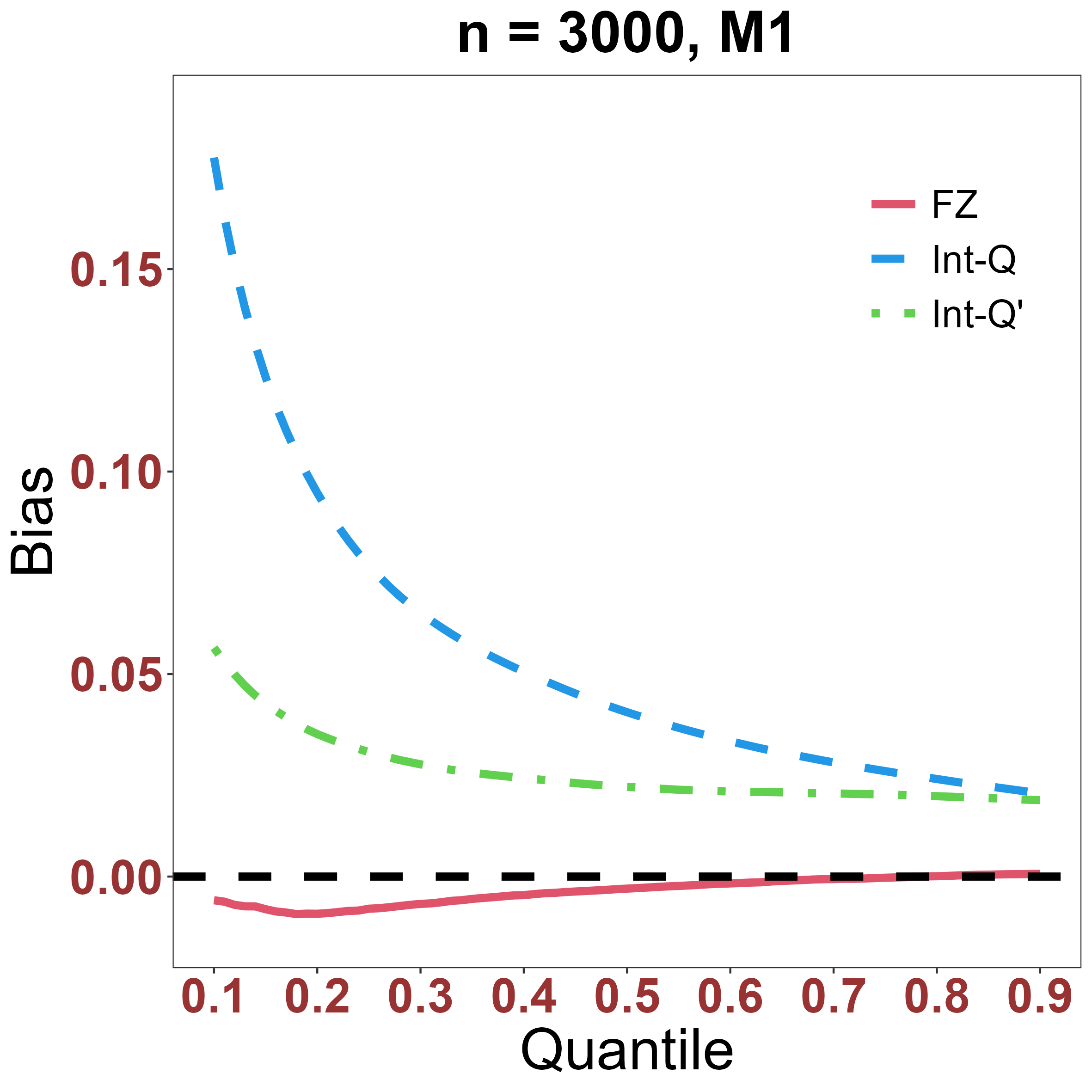}
		\includegraphics[width = 3.9cm, height = 3.2cm]{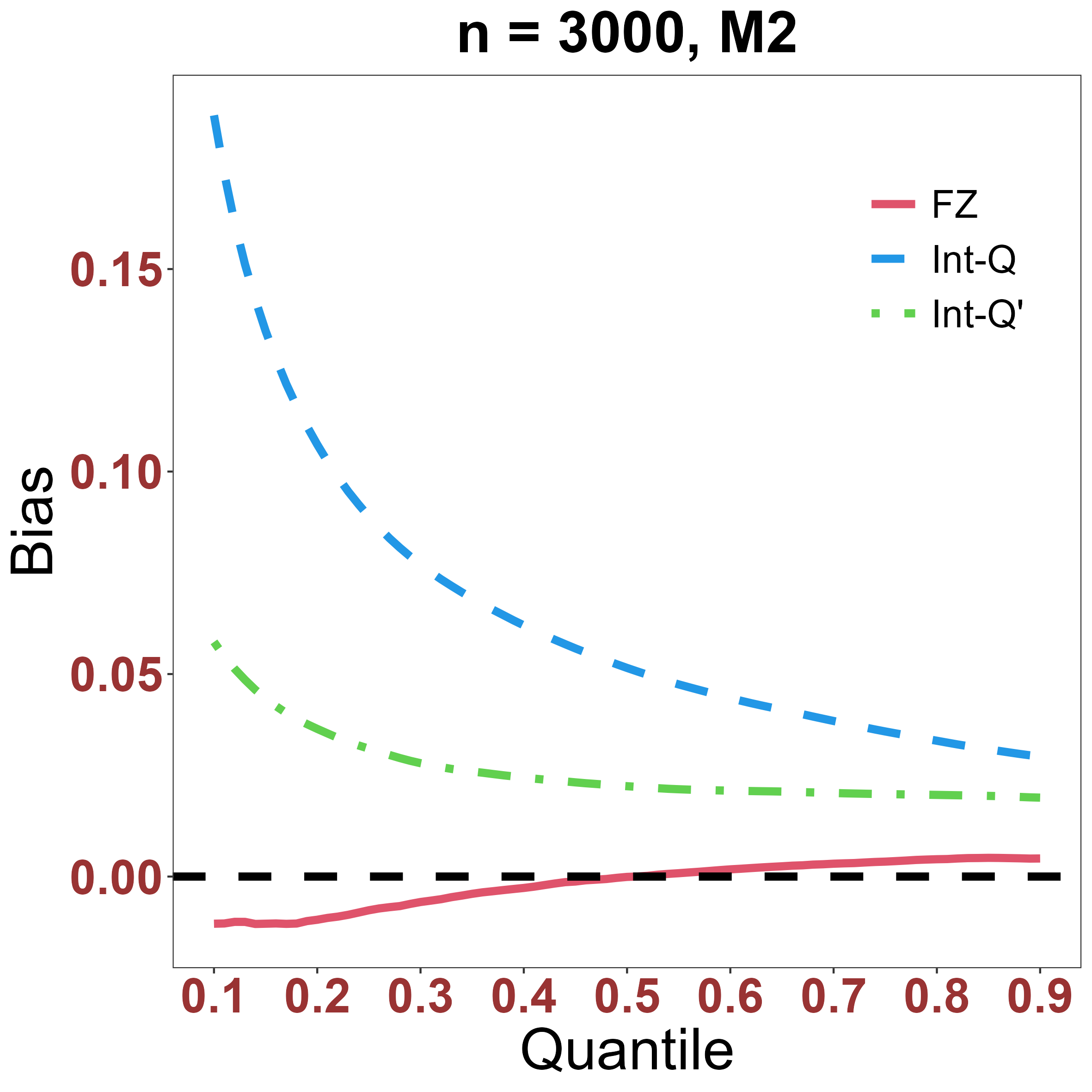}
		\includegraphics[width = 3.9cm, height = 3.2cm]{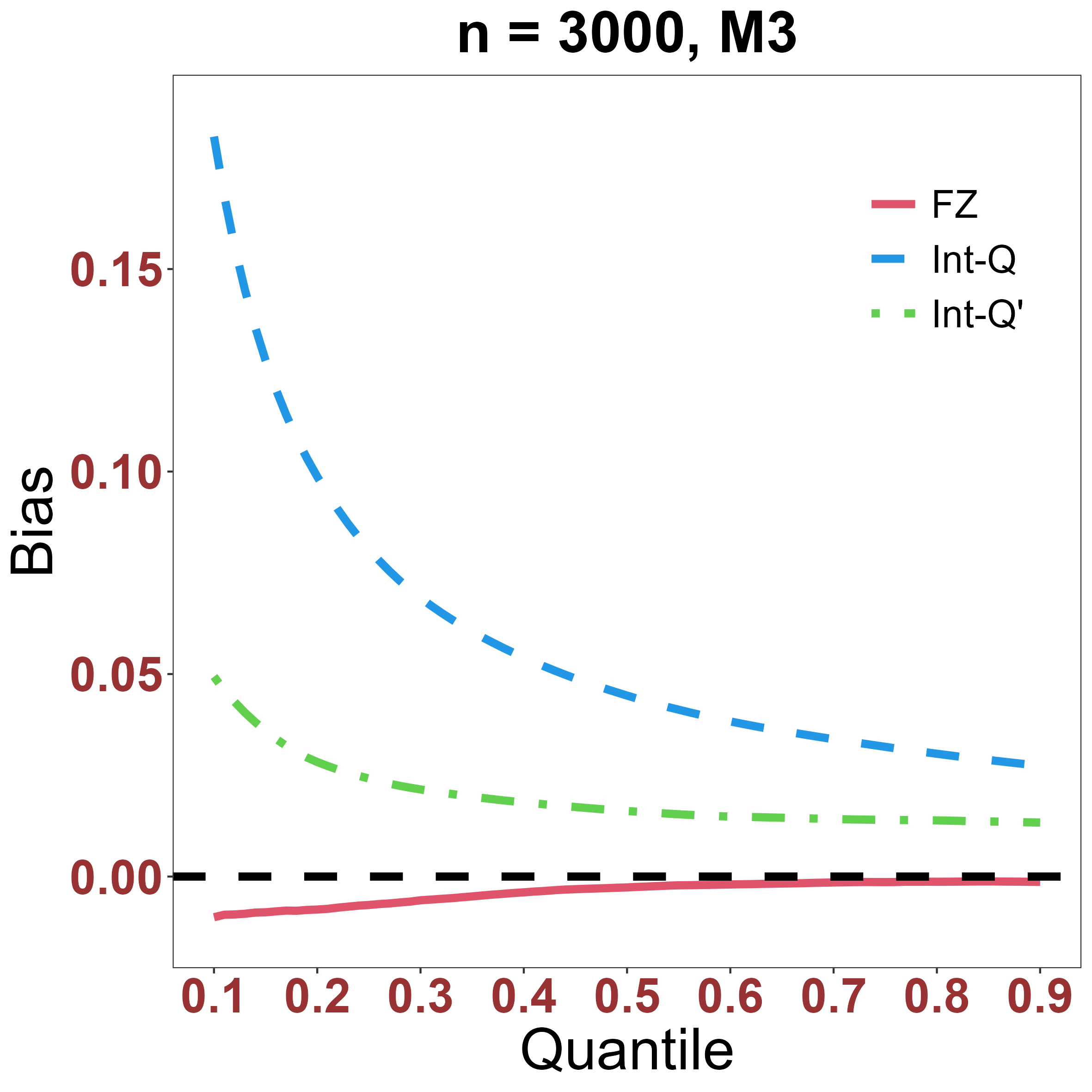}	
		\includegraphics[width = 3.9cm, height = 3.2cm]{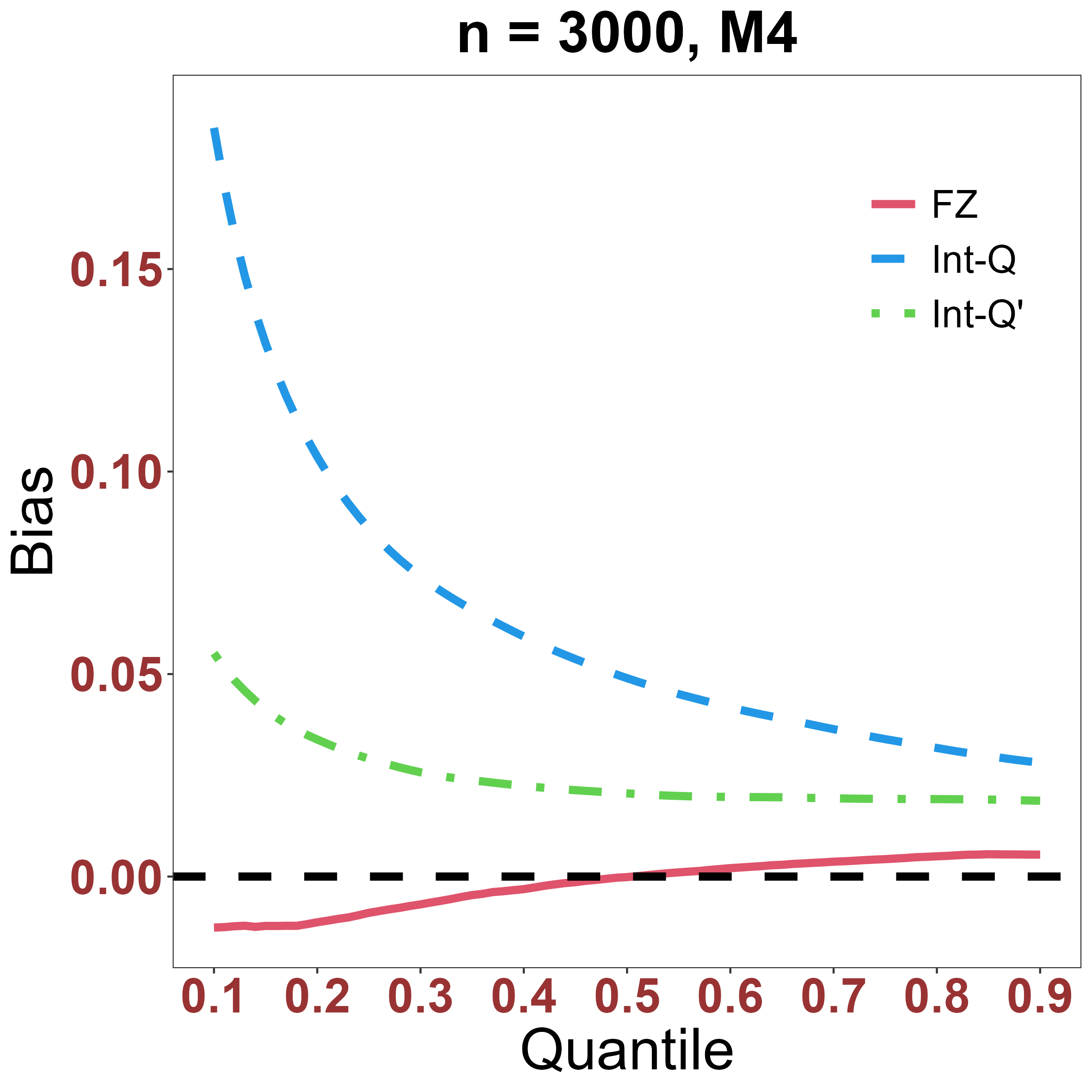}
	}	
	\mbox{
		\includegraphics[width = 3.9cm, height = 3.2cm]{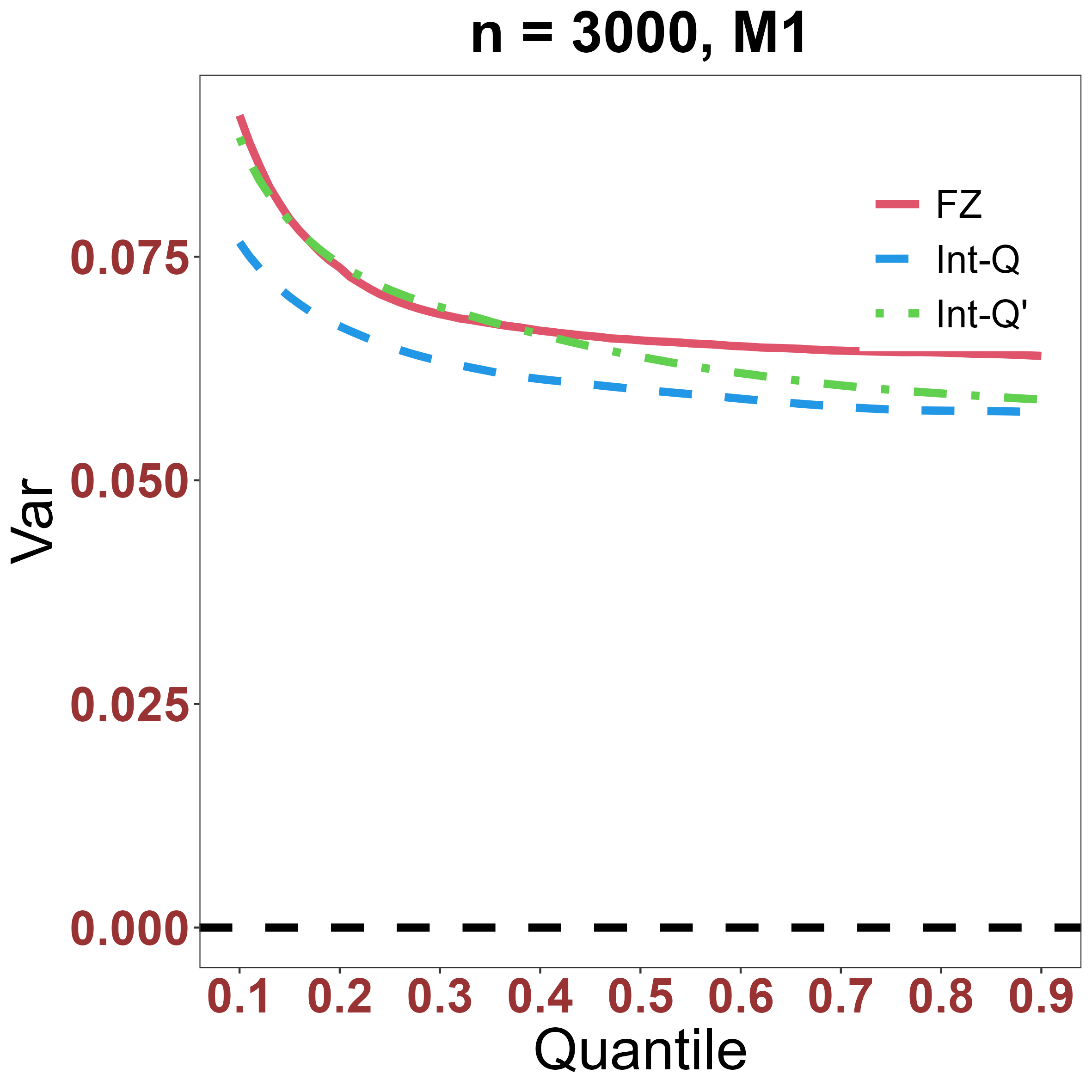}
		\includegraphics[width = 3.9cm, height = 3.2cm]{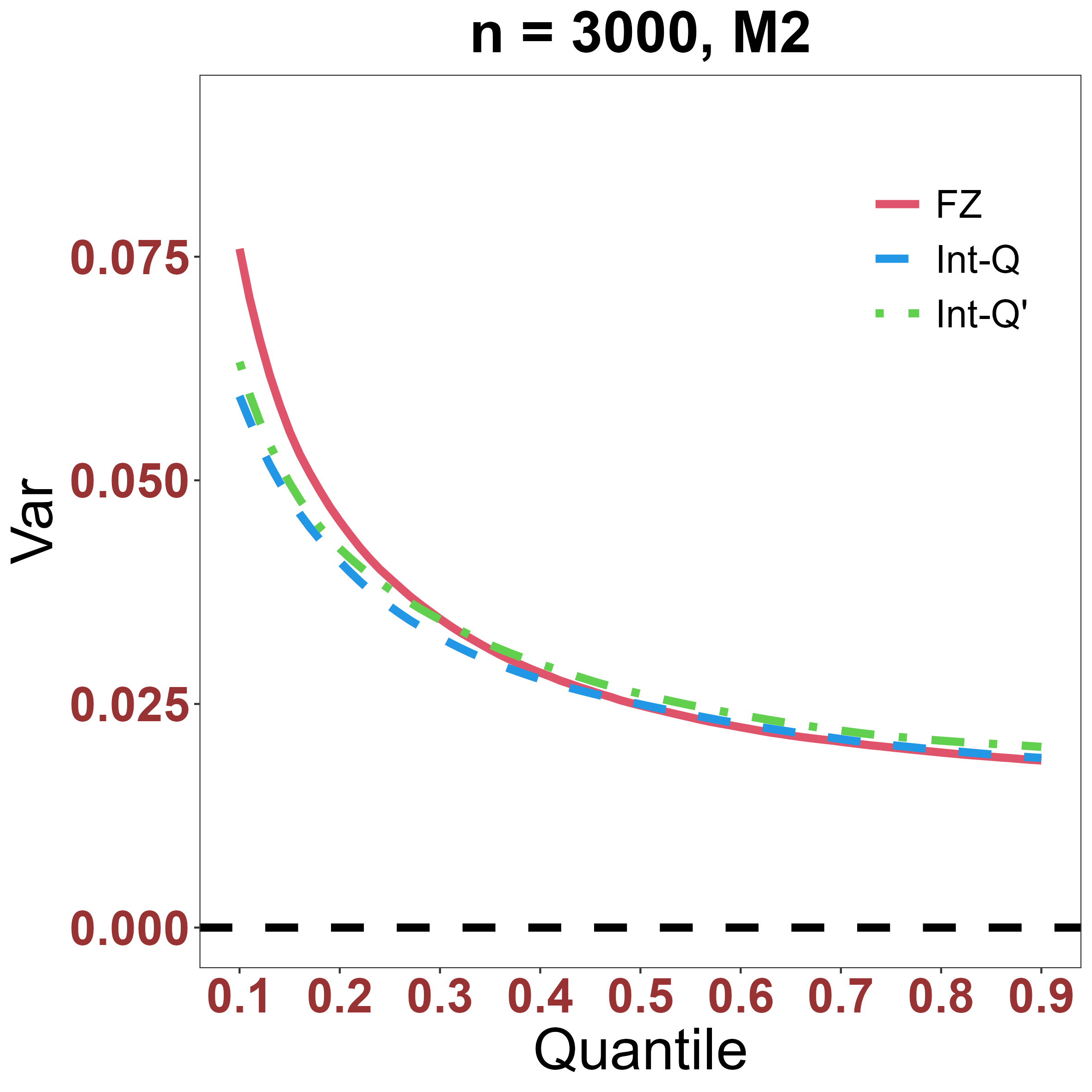}
		\includegraphics[width = 3.9cm, height = 3.2cm]{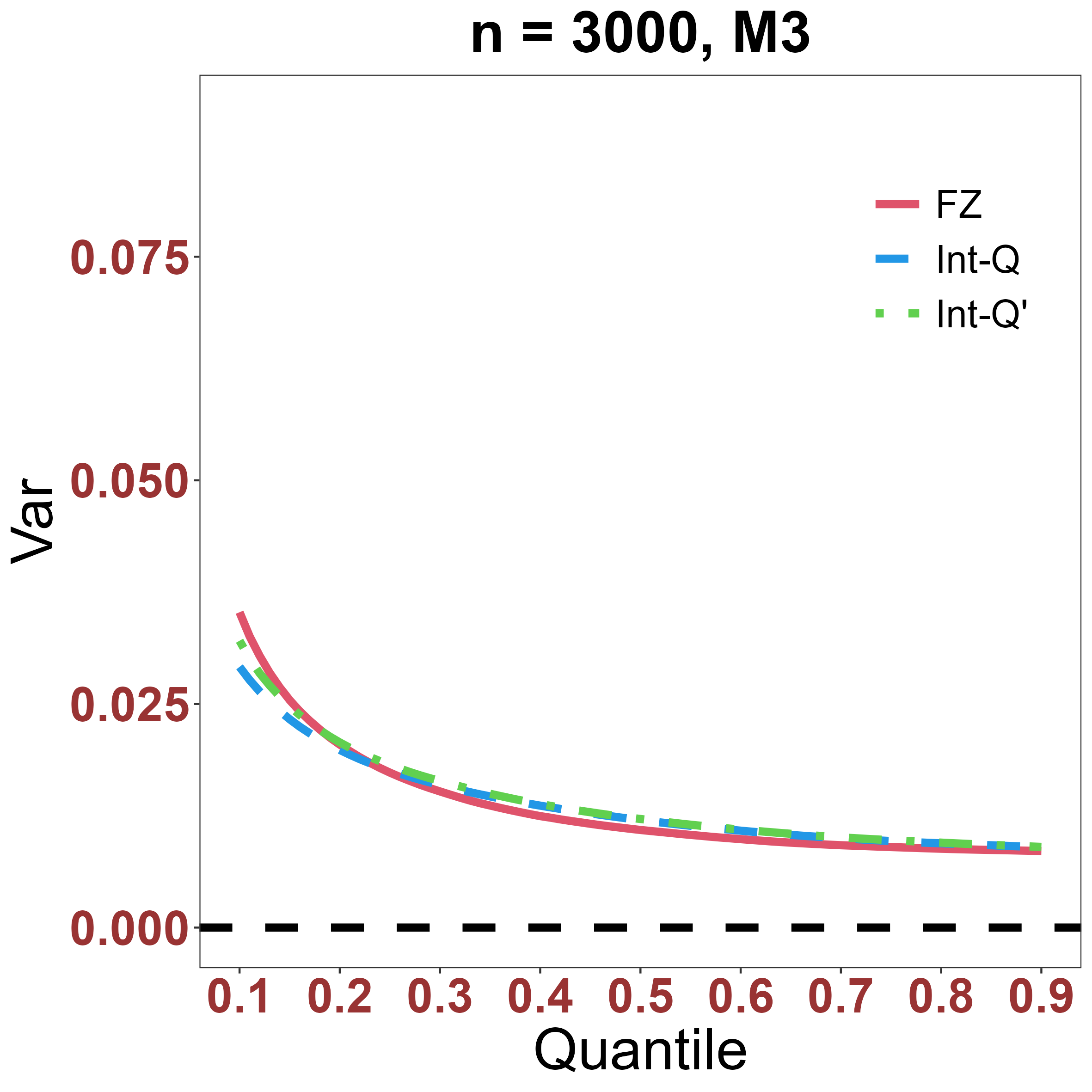}	
		\includegraphics[width = 3.9cm, height = 3.2cm]{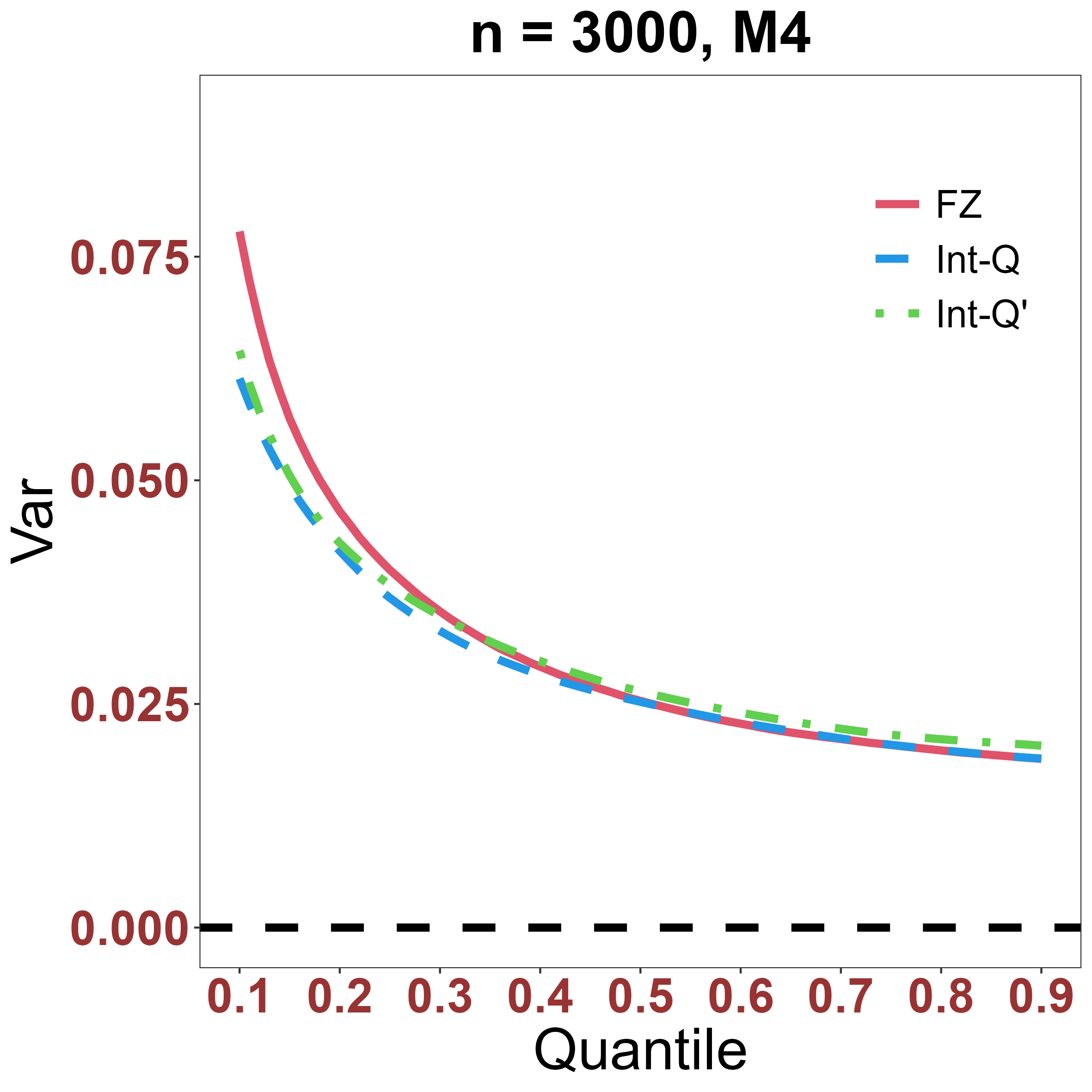}
	}	
	\mbox{
		\includegraphics[width = 3.9cm, height = 3.2cm]{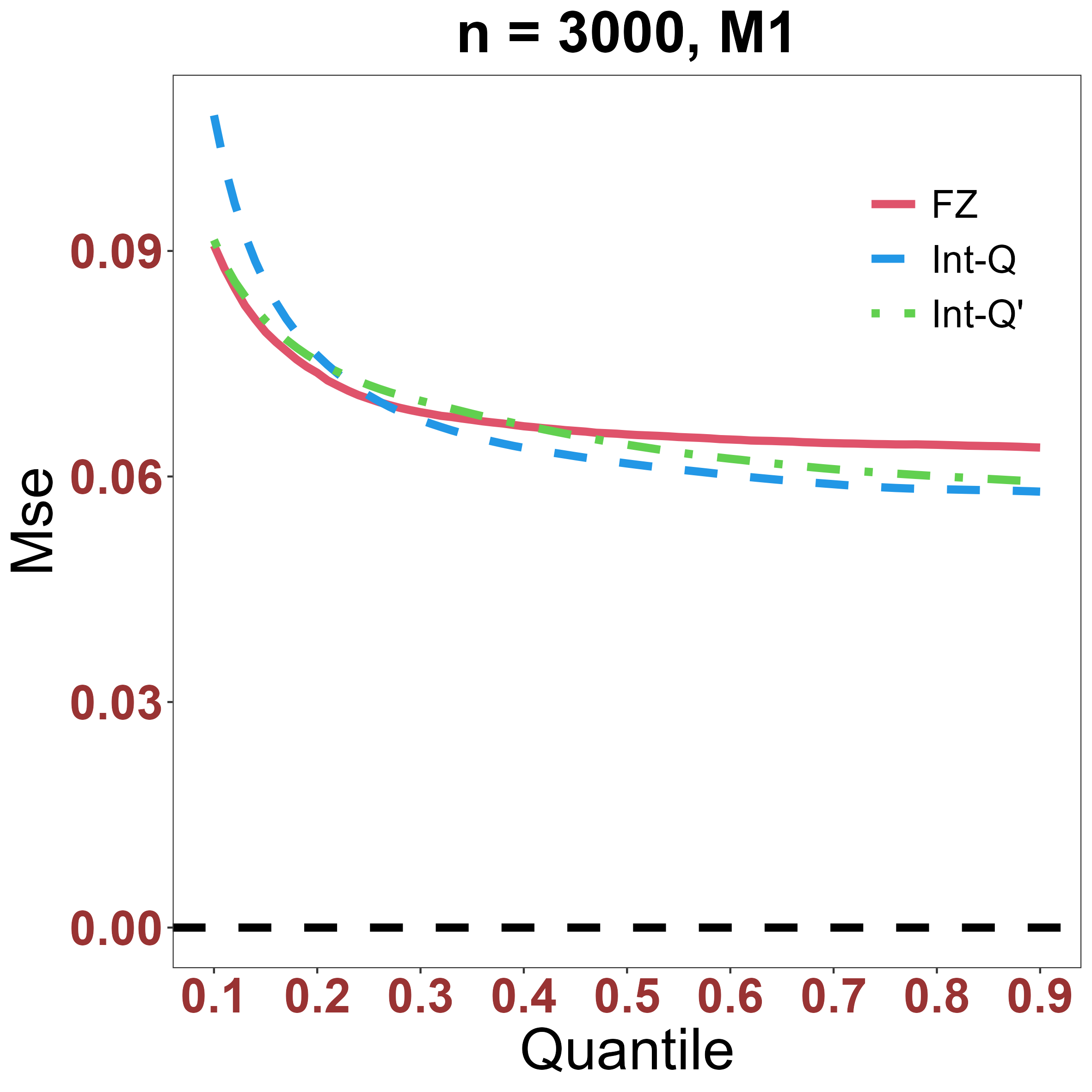}
		\includegraphics[width = 3.9cm, height = 3.2cm]{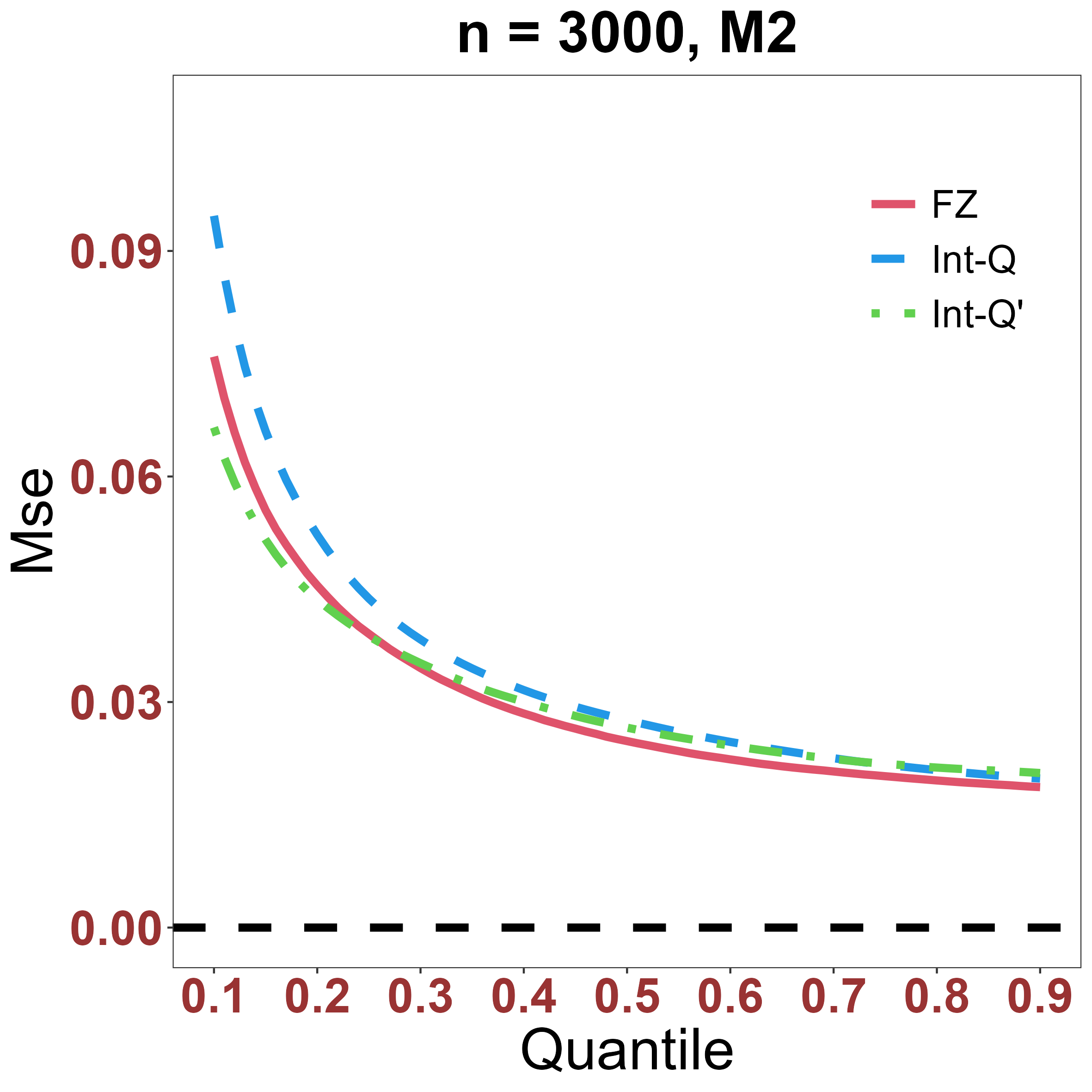}
		\includegraphics[width = 3.9cm, height = 3.2cm]{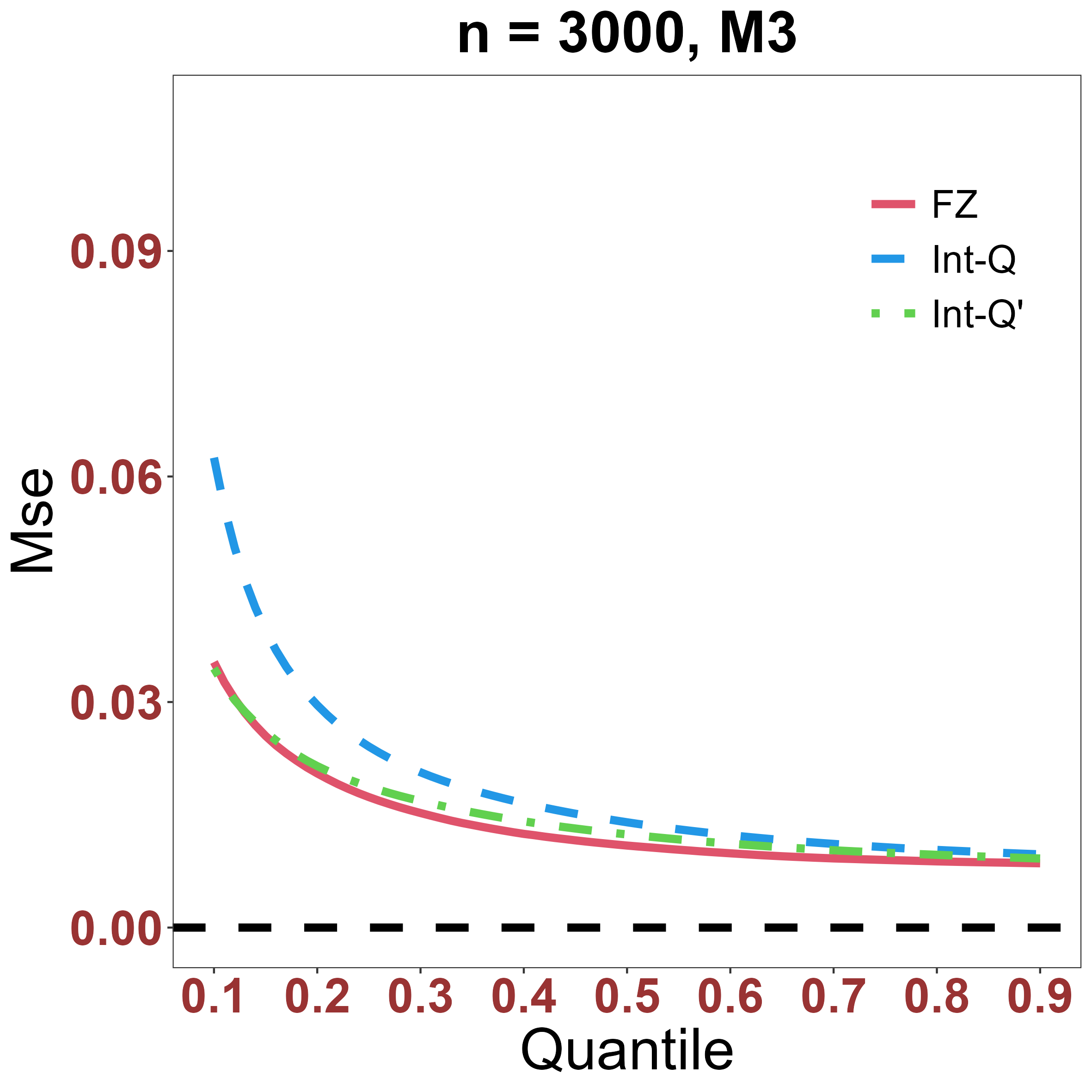}	
		\includegraphics[width = 3.9cm, height = 3.2cm]{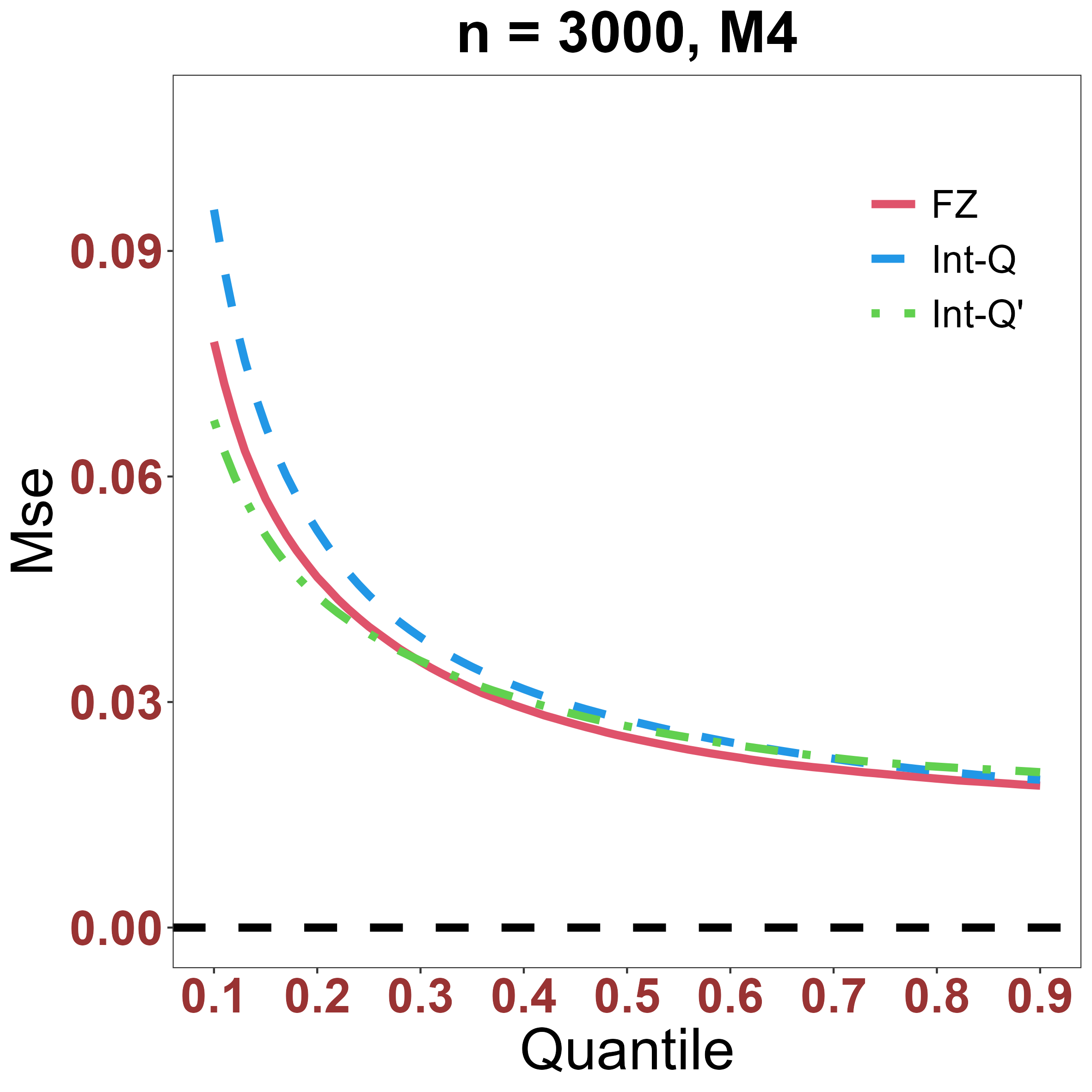}
	}	
	%\end{subfigure}
	\caption{Bias, variance and MSE of the CTATE estimator using the FZ loss (FZ) and that using the trimmed integrated-QTE based estimators (IntQ and IntQ') when $\rho=0$ and $n=3,000$.}
	\label{figure4}
\end{figure}

\begin{figure}[!htb]
	%\begin{subfigure}
	\centering
	\mbox{
		\includegraphics[width = 3.9cm, height = 3.2cm]{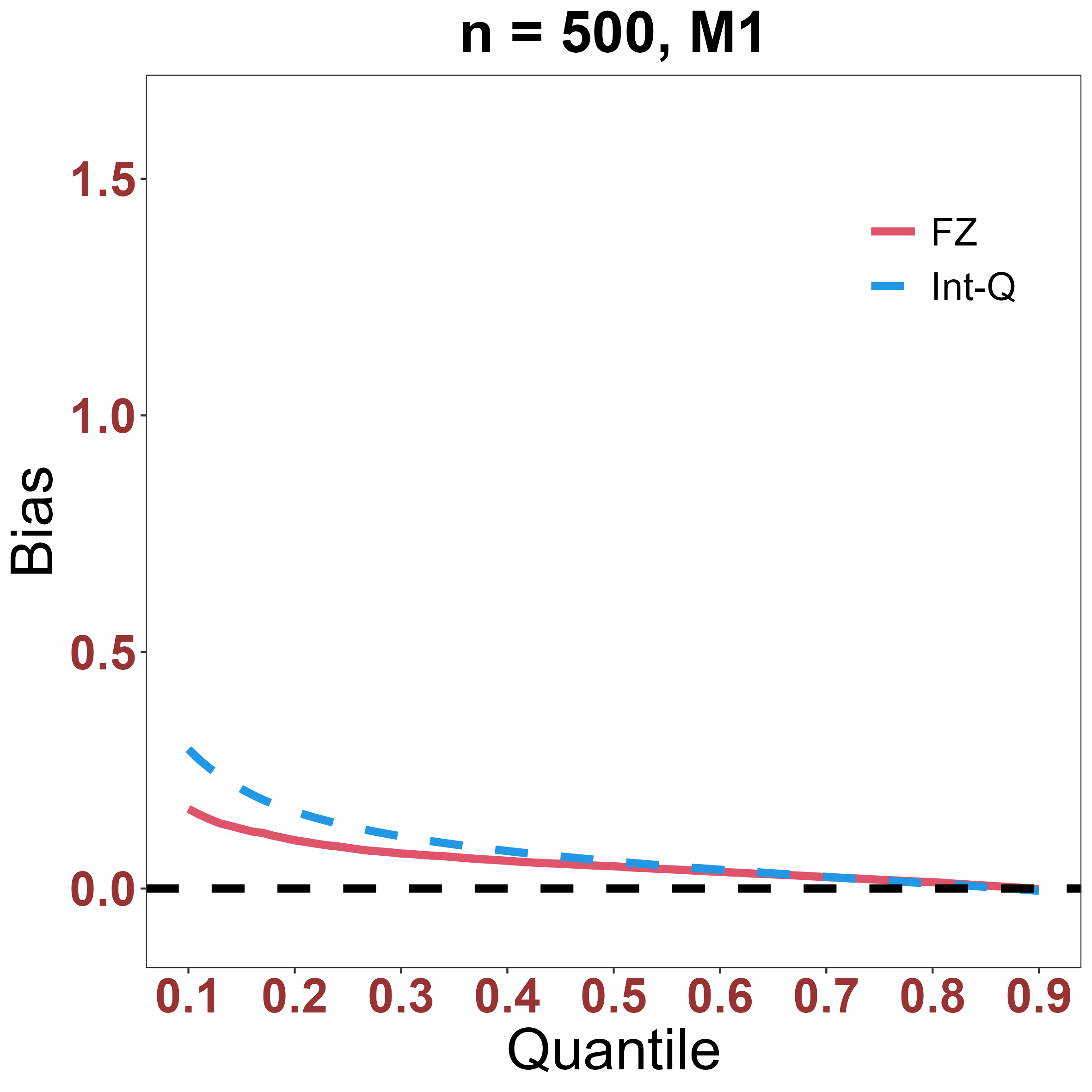}
		\includegraphics[width = 3.9cm, height = 3.2cm]{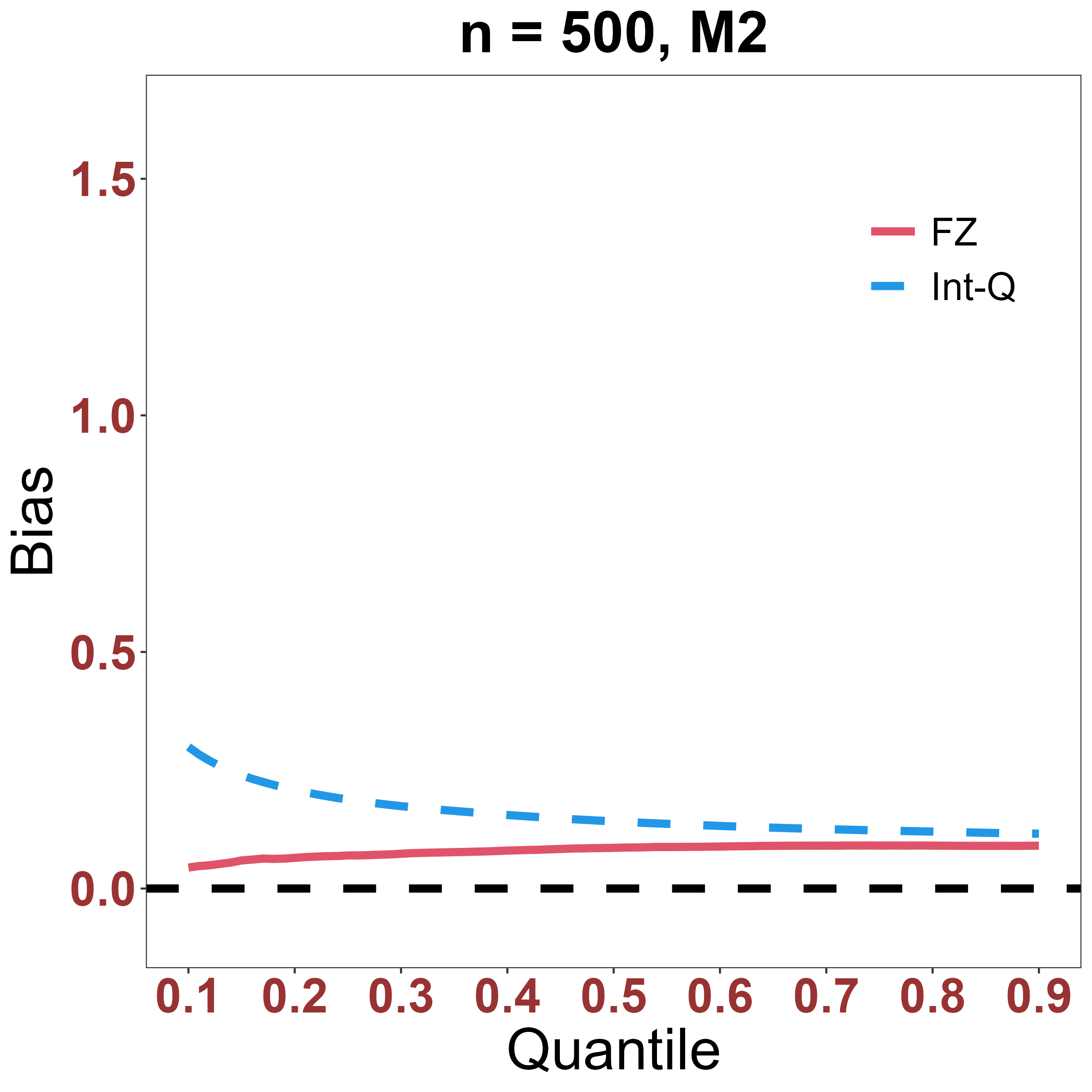}
		\includegraphics[width = 3.9cm, height = 3.2cm]{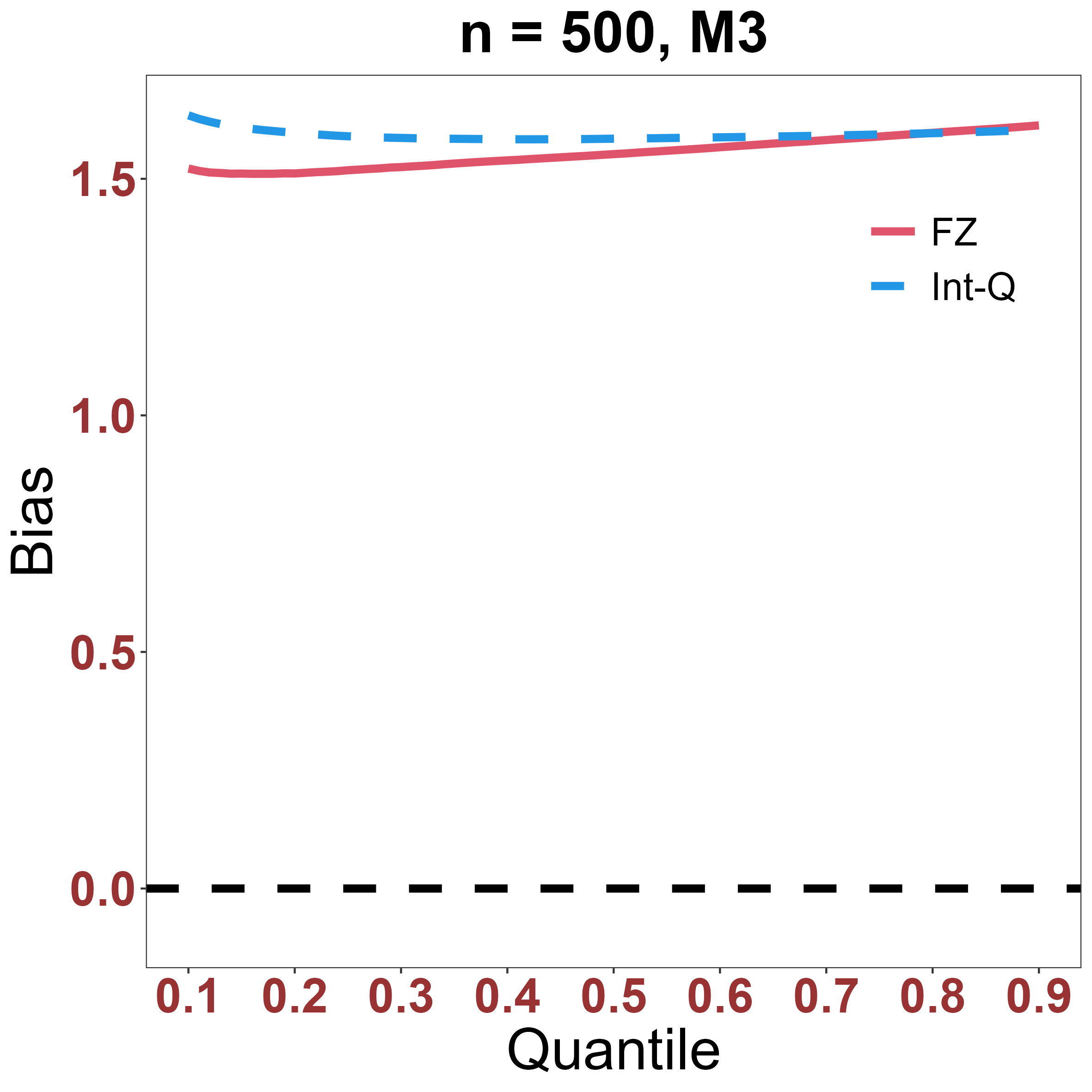}	
		\includegraphics[width = 3.9cm, height = 3.2cm]{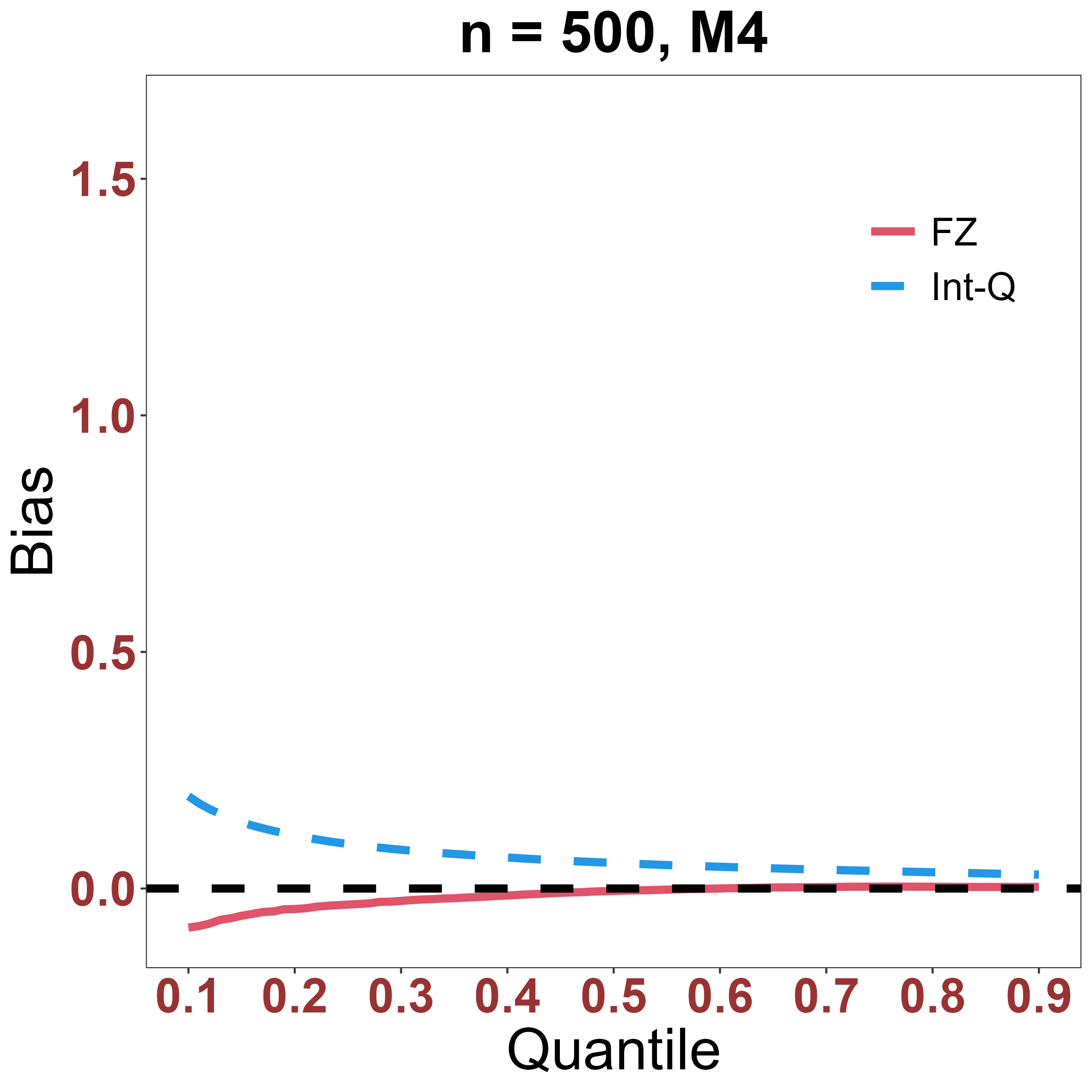}
	}	
	\mbox{	\includegraphics[width = 3.9cm, height = 3.2cm]{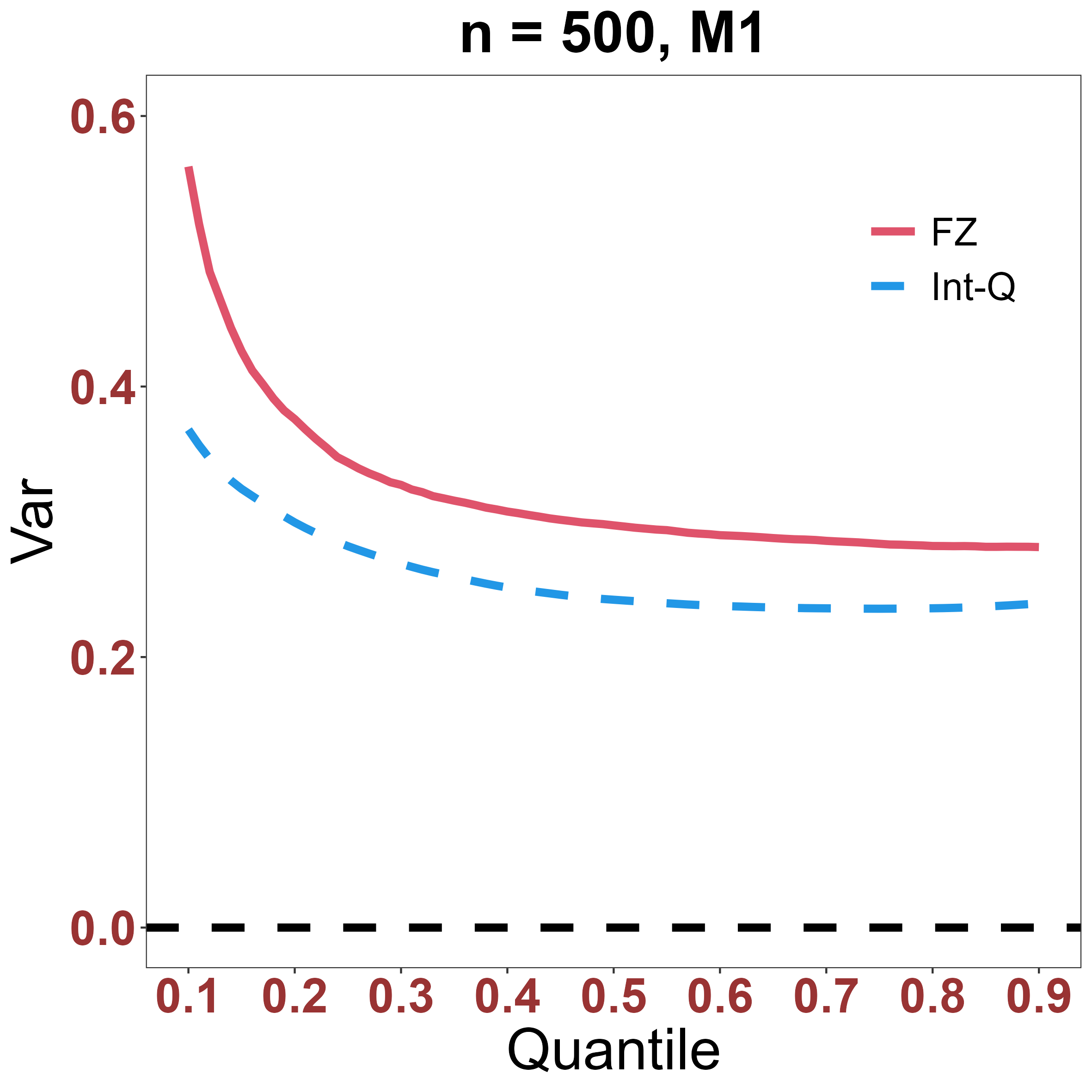}
		\includegraphics[width = 3.9cm, height = 3.2cm]{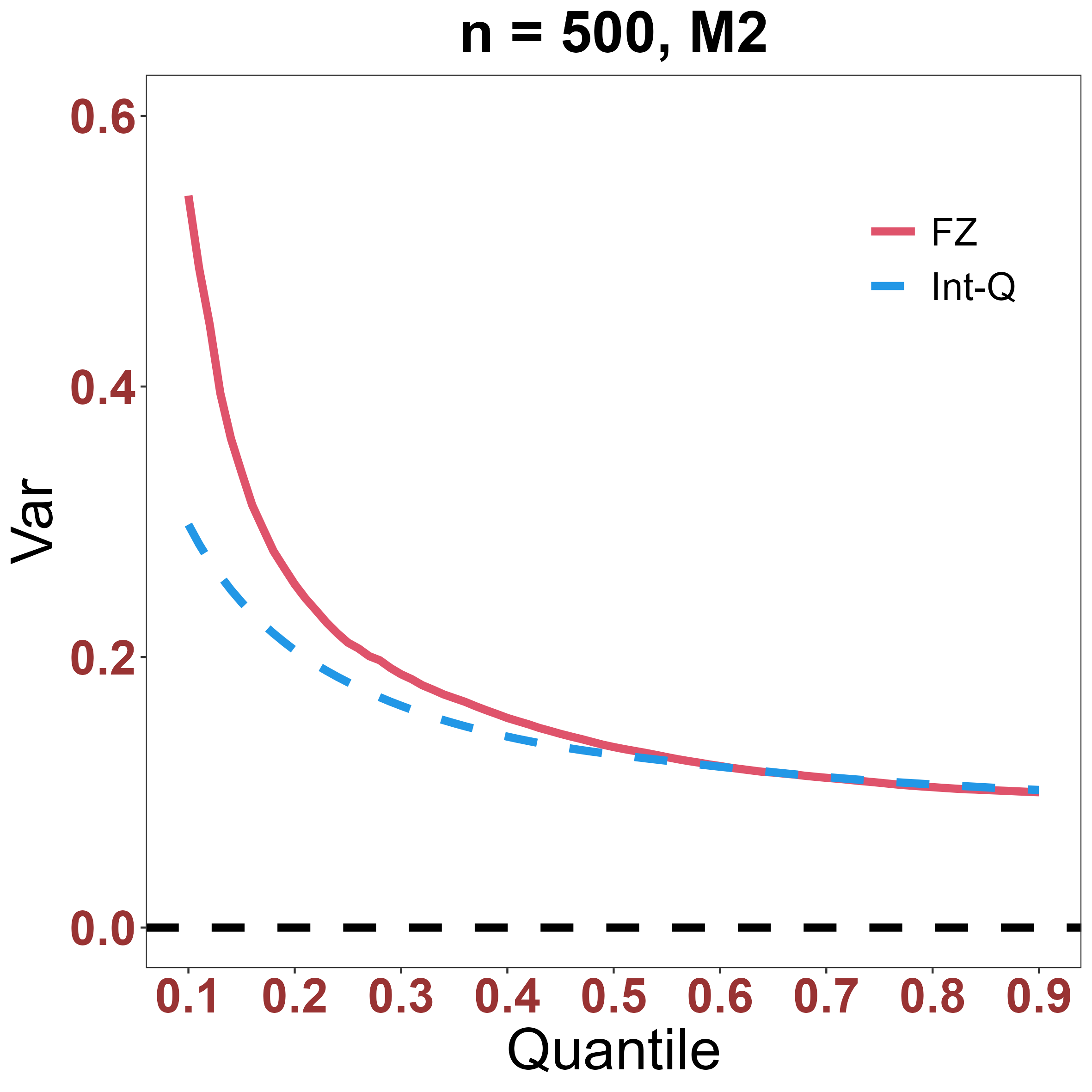}
		\includegraphics[width = 3.9cm, height = 3.2cm]{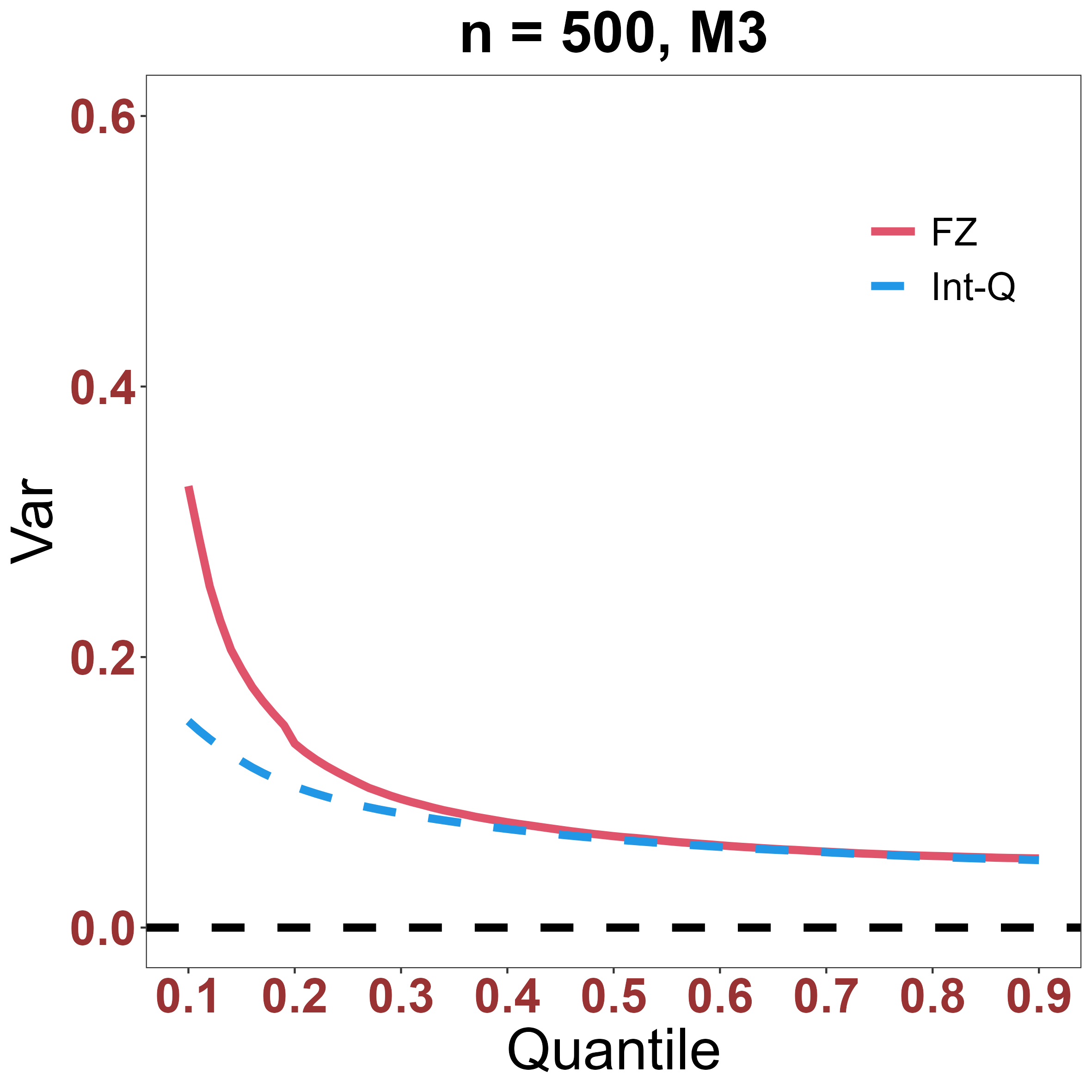}	
		\includegraphics[width = 3.9cm, height = 3.2cm]{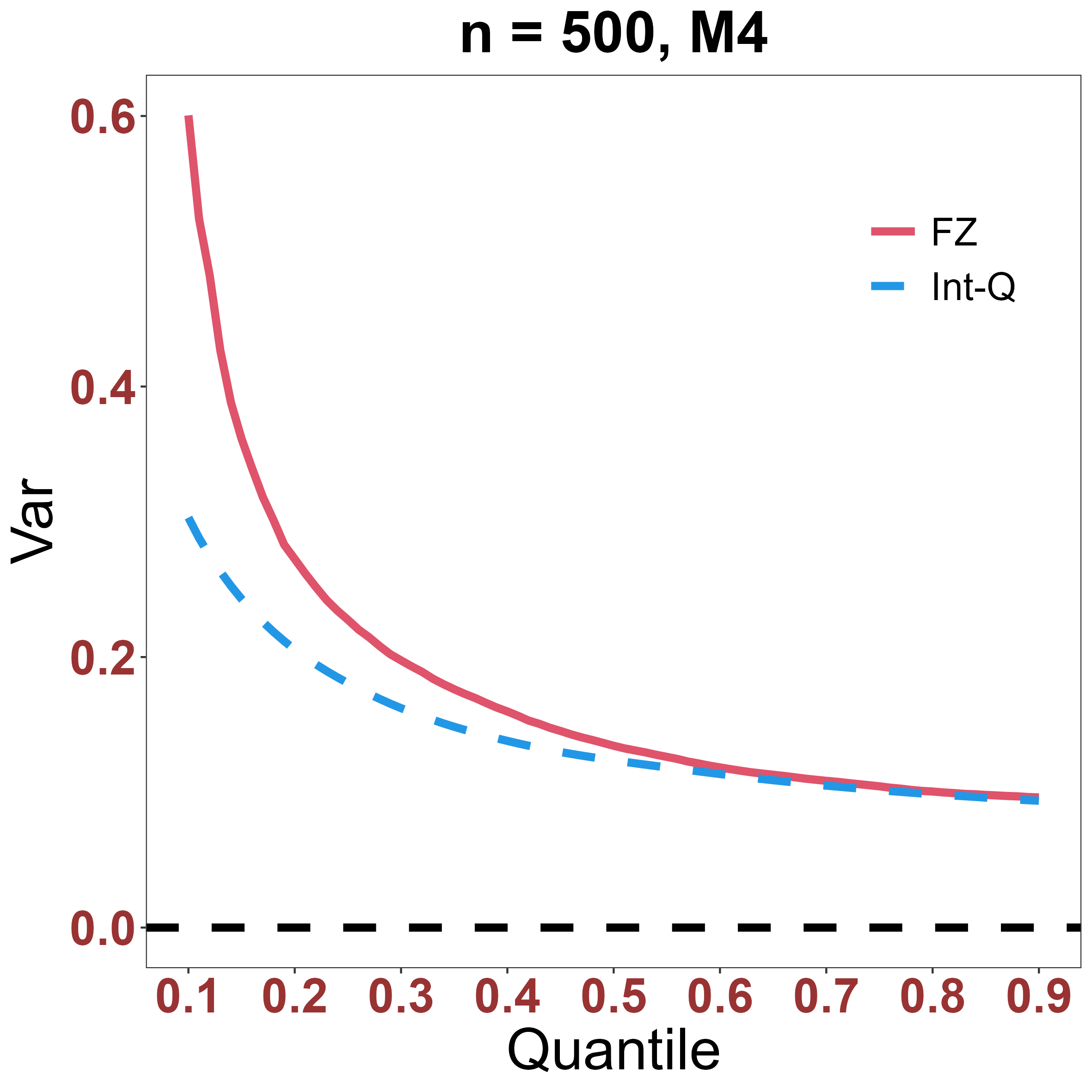}
	}	
	\mbox{	\includegraphics[width = 3.9cm, height = 3.2cm]{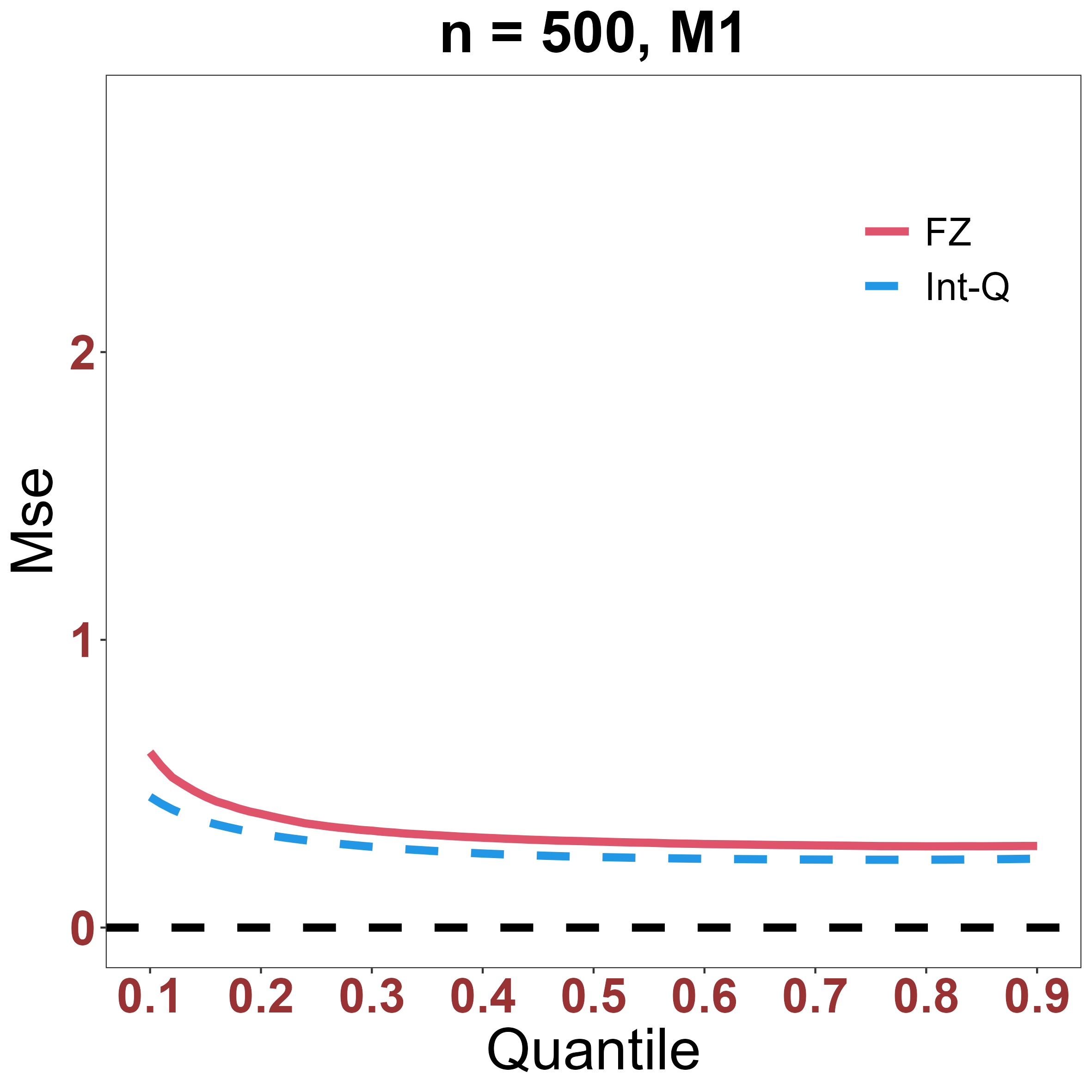}
		\includegraphics[width = 3.9cm, height = 3.2cm]{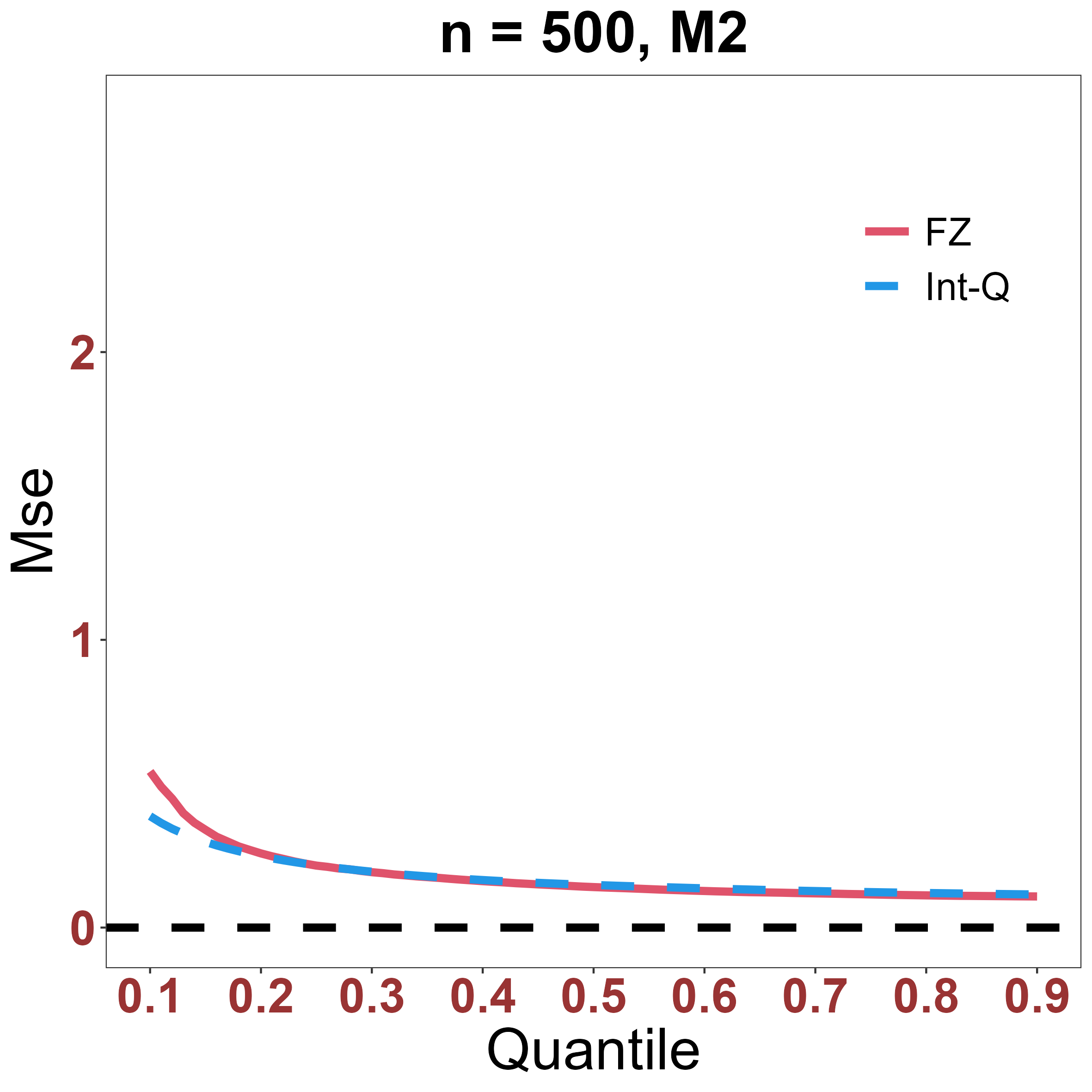}
		\includegraphics[width = 3.9cm, height = 3.2cm]{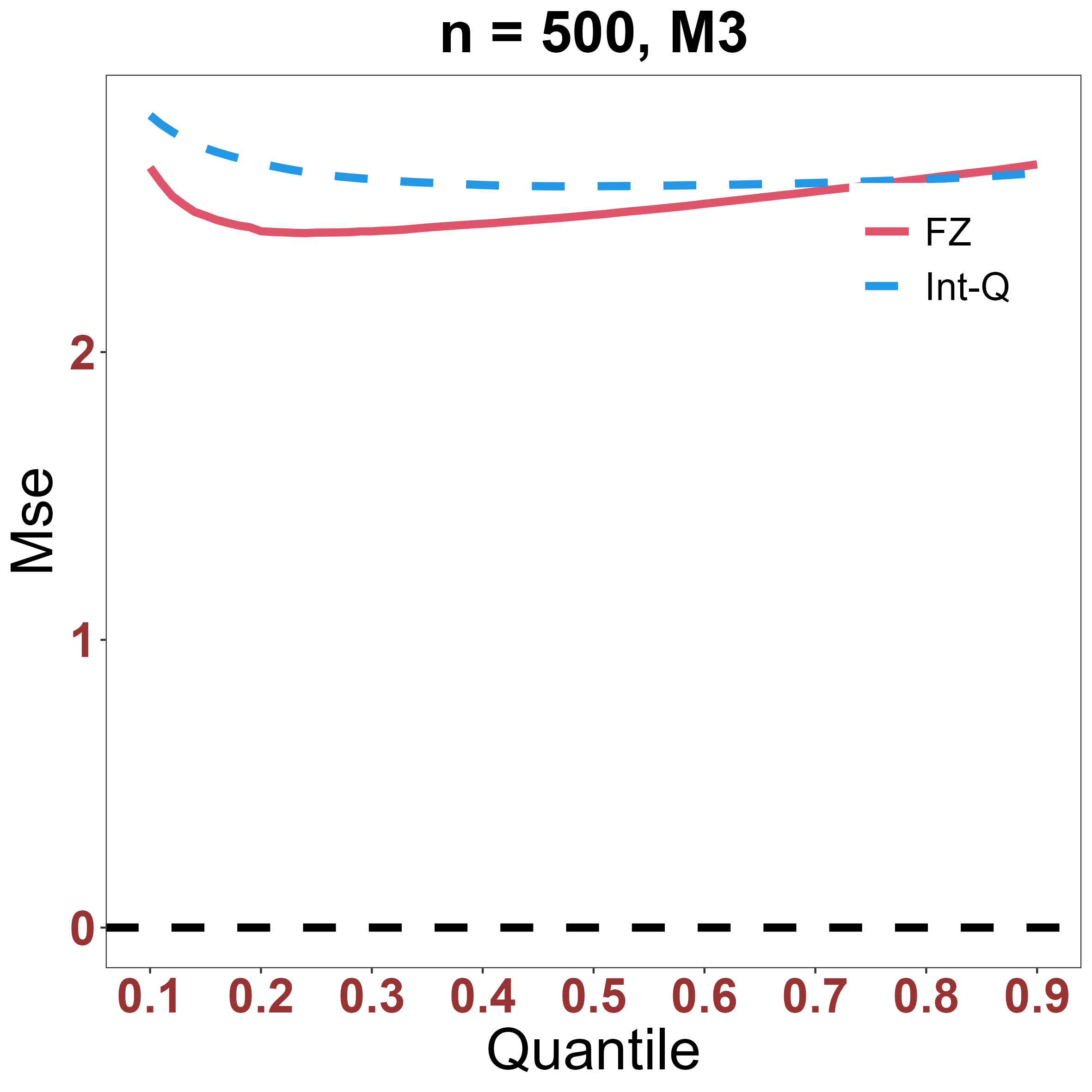}	
		\includegraphics[width = 3.9cm, height = 3.2cm]{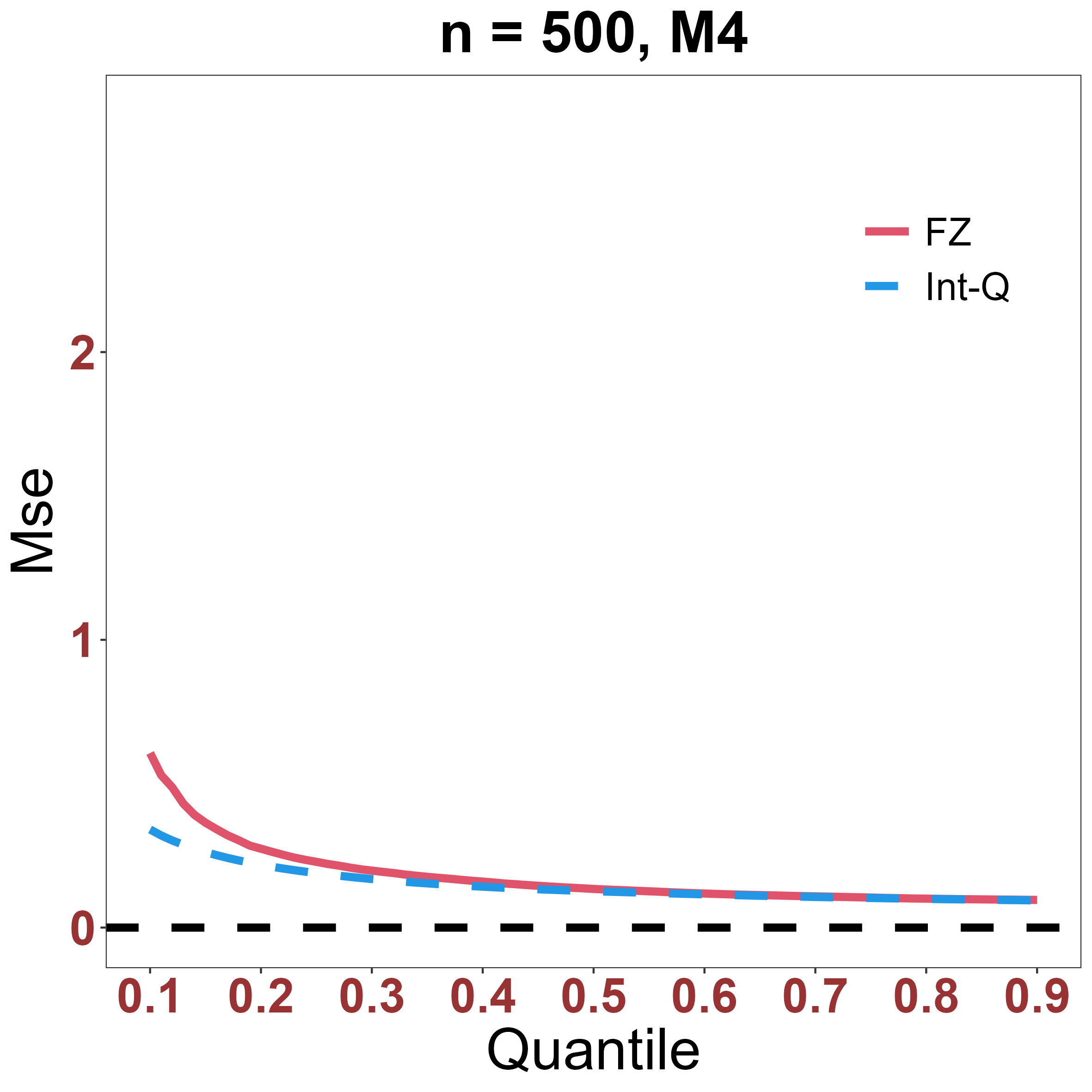}	
	}
	%\end{subfigure}
	\caption{Bias, variance and MSE of the CTATE estimator using the FZ loss (FZ) and that using the trimmed integrated-QTE based estimators (IntQ and IntQ') when $\rho=0.5$ and $n=500$.}
	\label{figure5}
\end{figure}

\begin{figure}[!htb]
	%\begin{subfigure}
	\centering
	\mbox{
		\includegraphics[width = 3.9cm, height = 3.2cm]{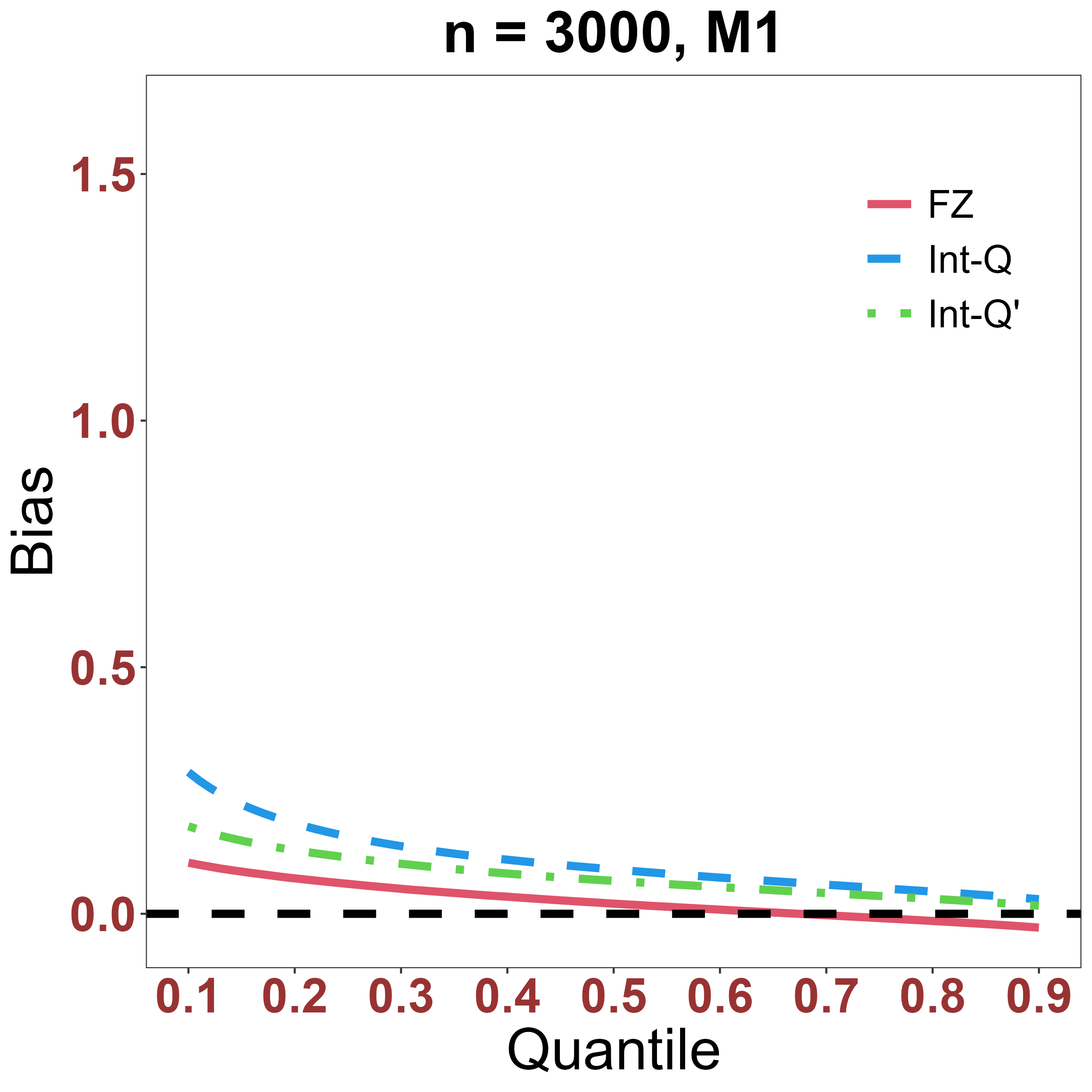}
		\includegraphics[width = 3.9cm, height = 3.2cm]{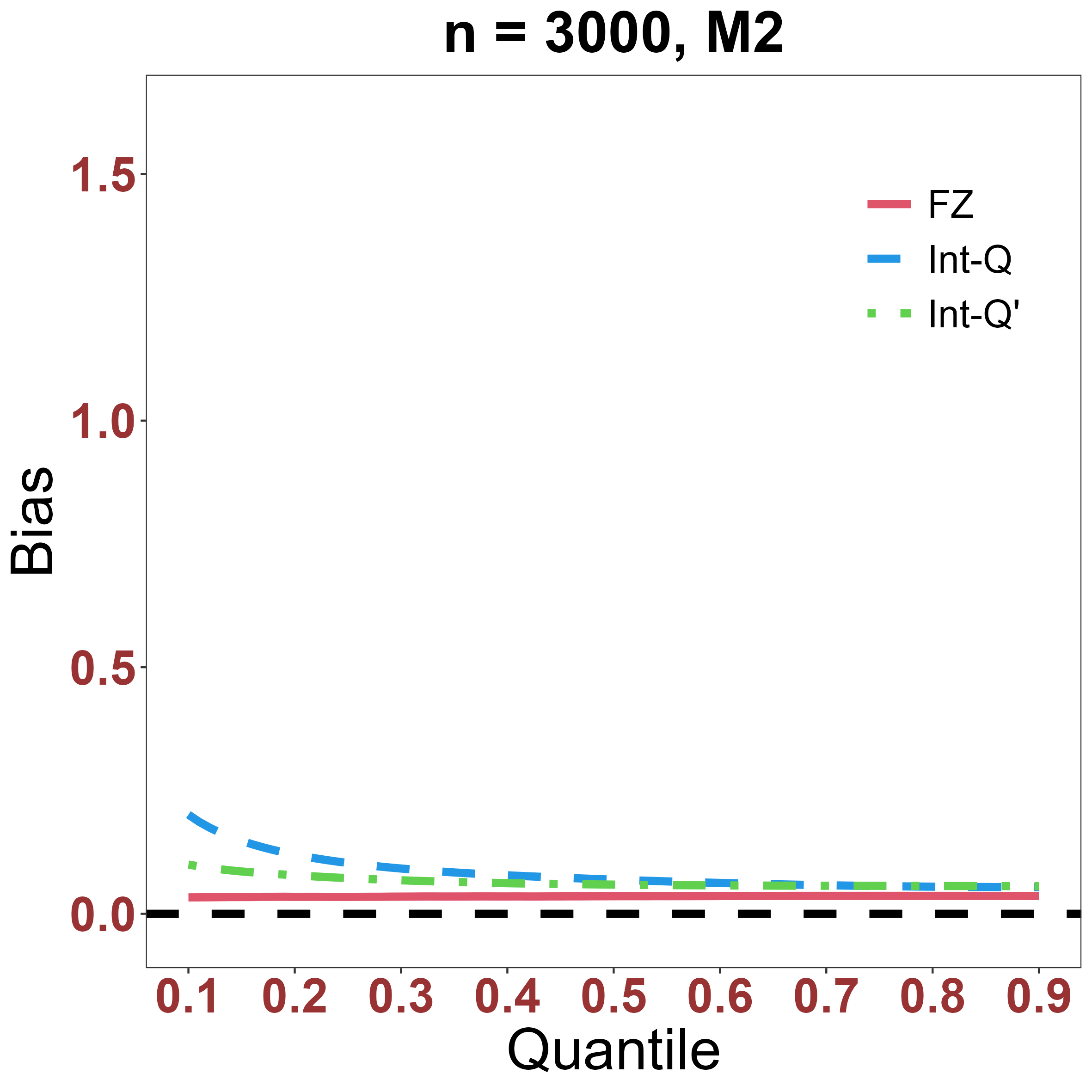}
		\includegraphics[width = 3.9cm, height = 3.2cm]{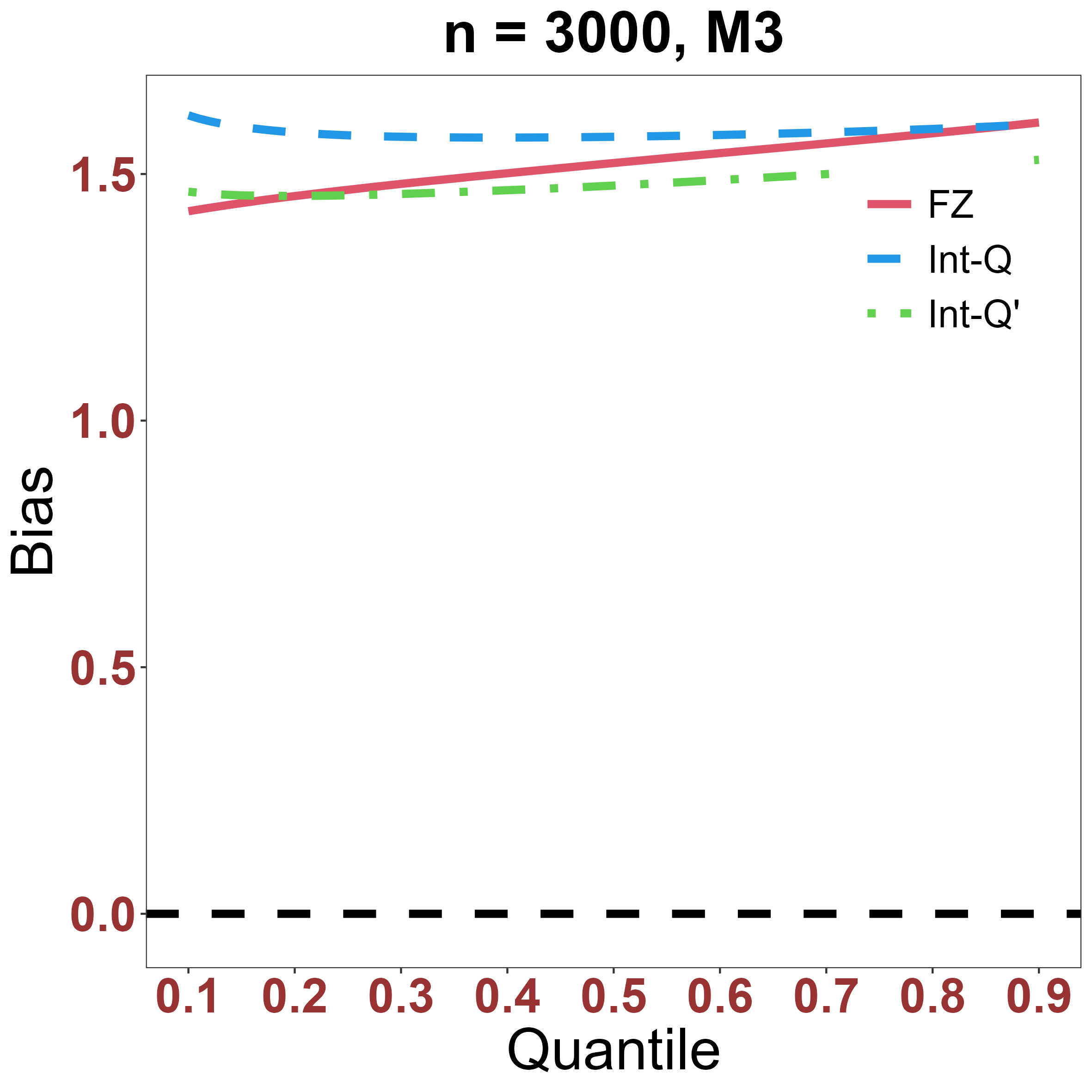}	
		\includegraphics[width = 3.9cm, height = 3.2cm]{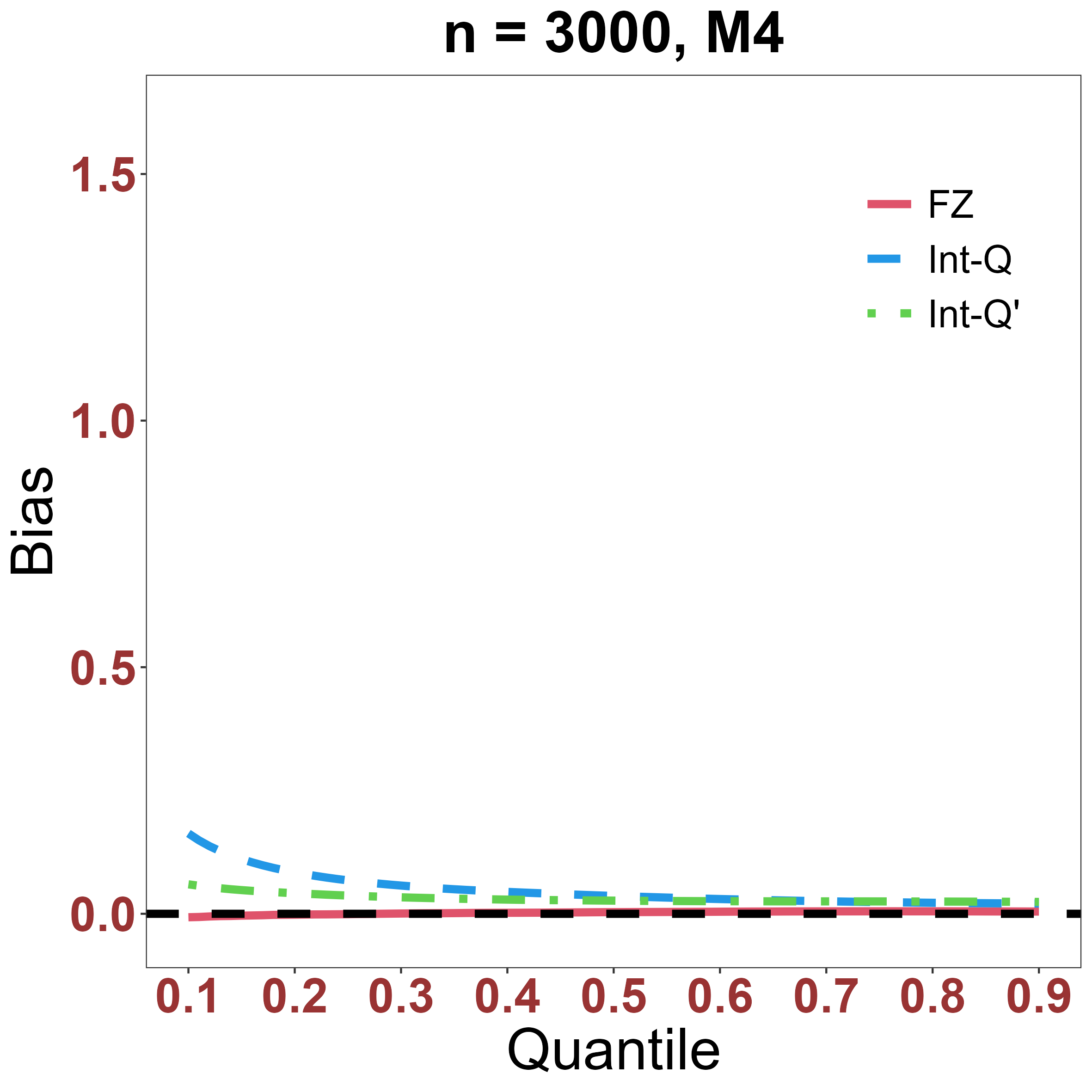}
	}	
	\mbox{
		\includegraphics[width = 3.9cm, height = 3.2cm]{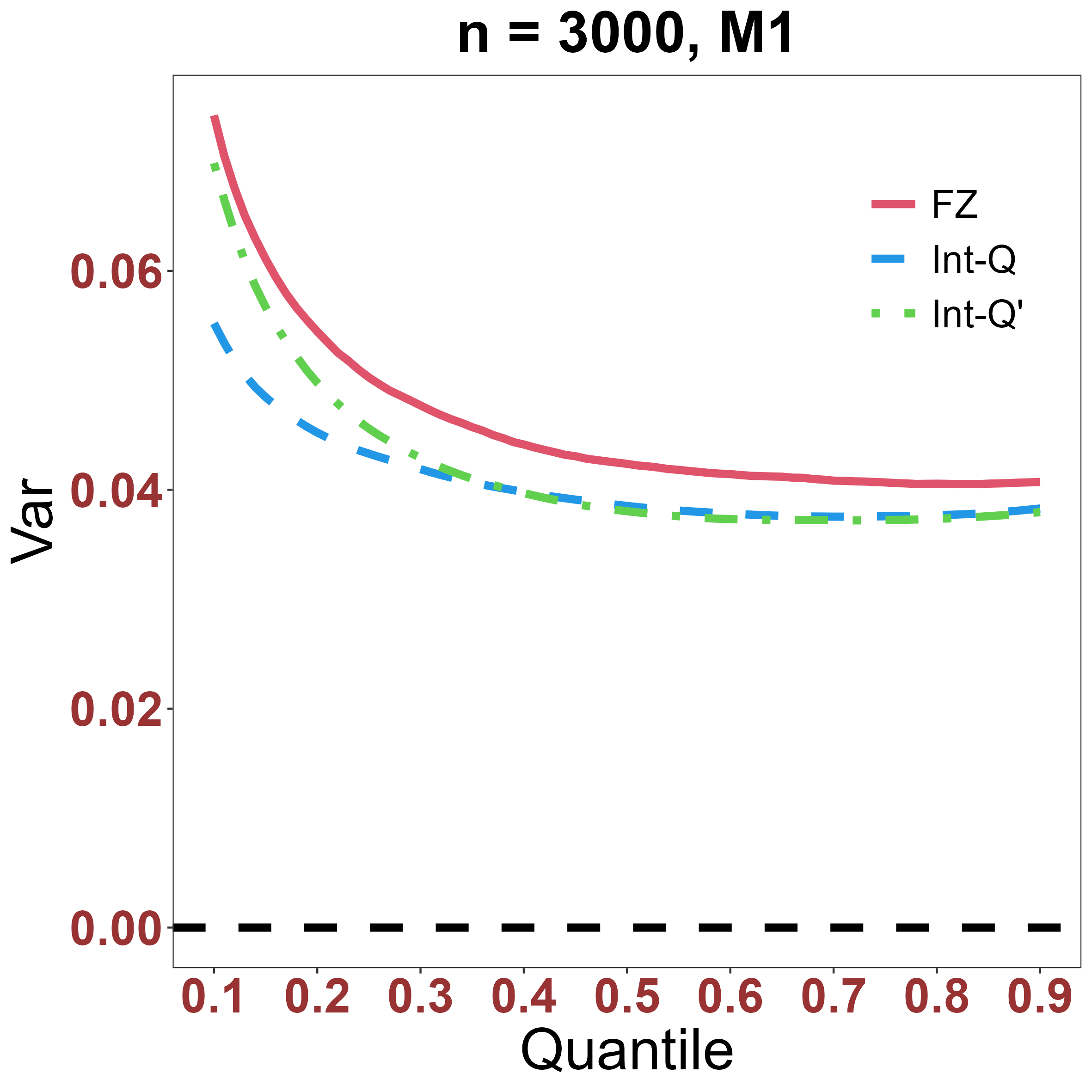}
		\includegraphics[width = 3.9cm, height = 3.2cm]{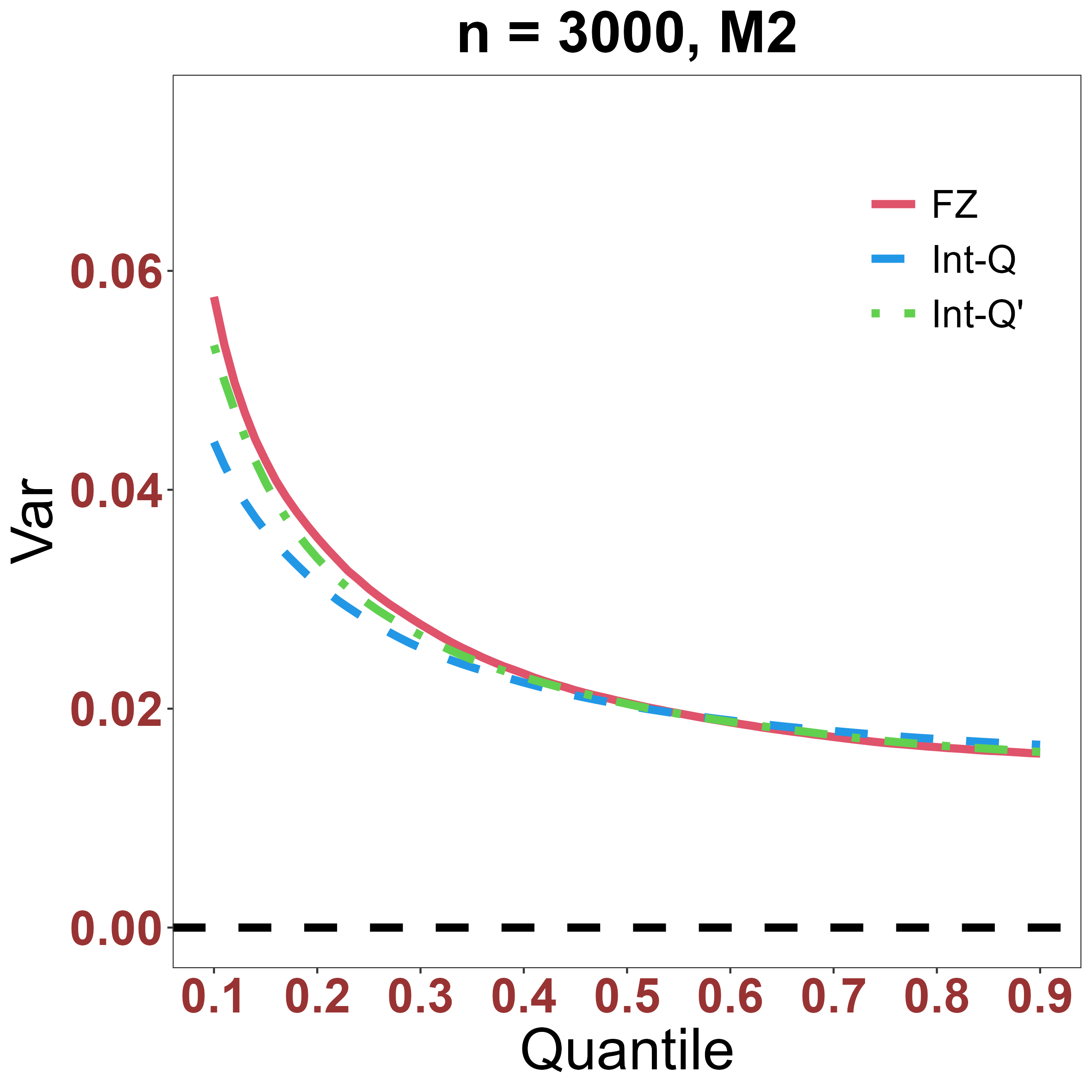}
		\includegraphics[width = 3.9cm, height = 3.2cm]{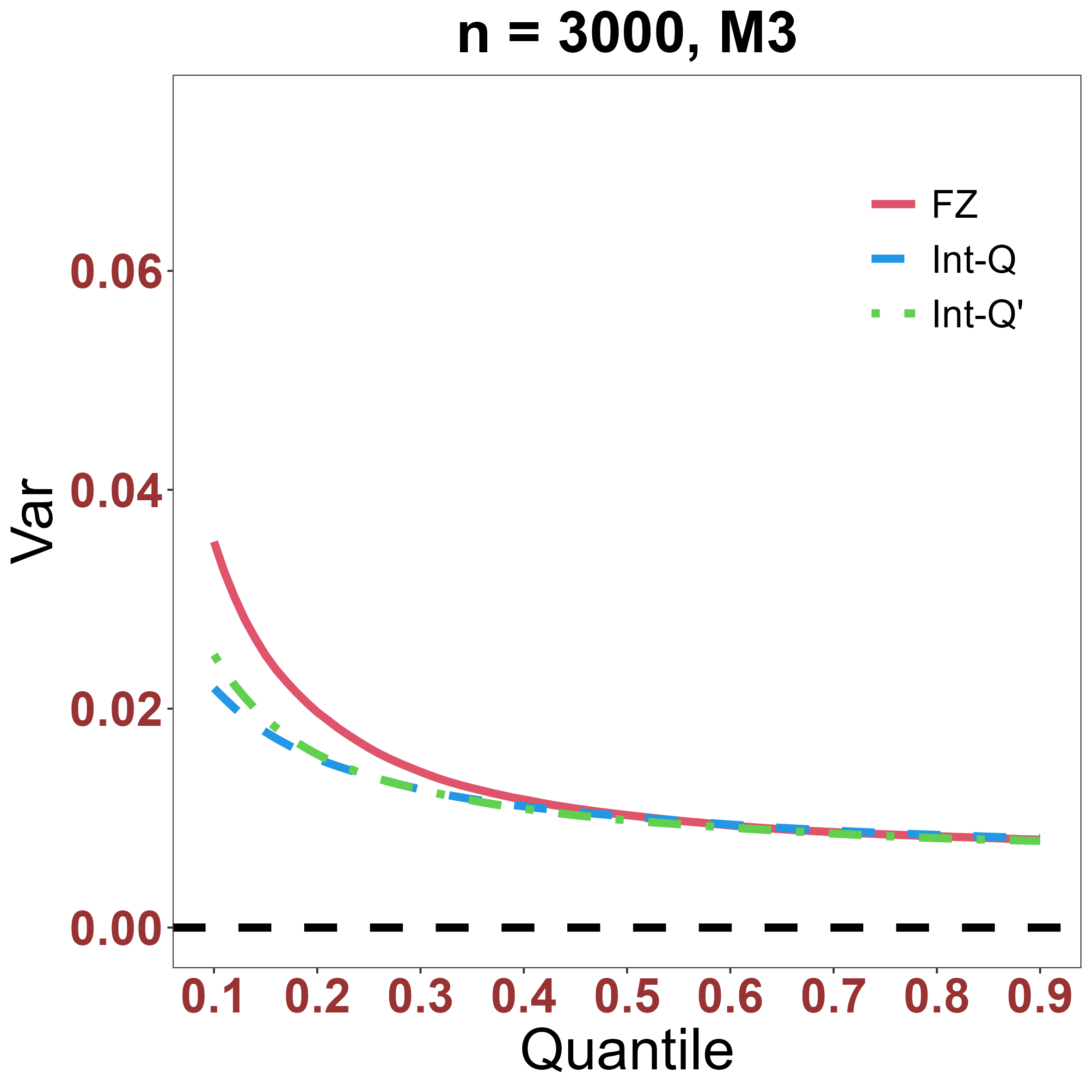}	
		\includegraphics[width = 3.9cm, height = 3.2cm]{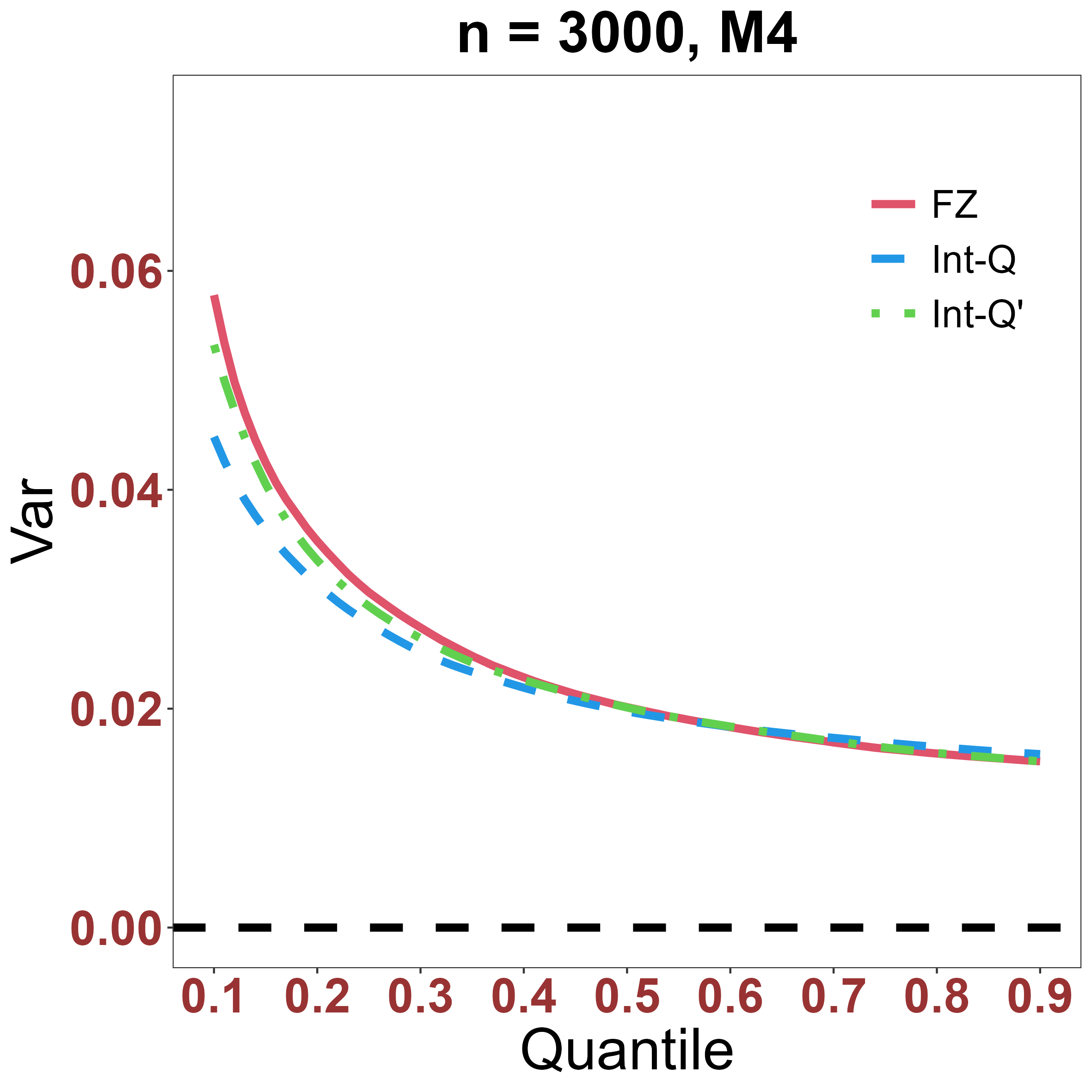}
	}	
	\mbox{
		\includegraphics[width = 3.9cm, height = 3.2cm]{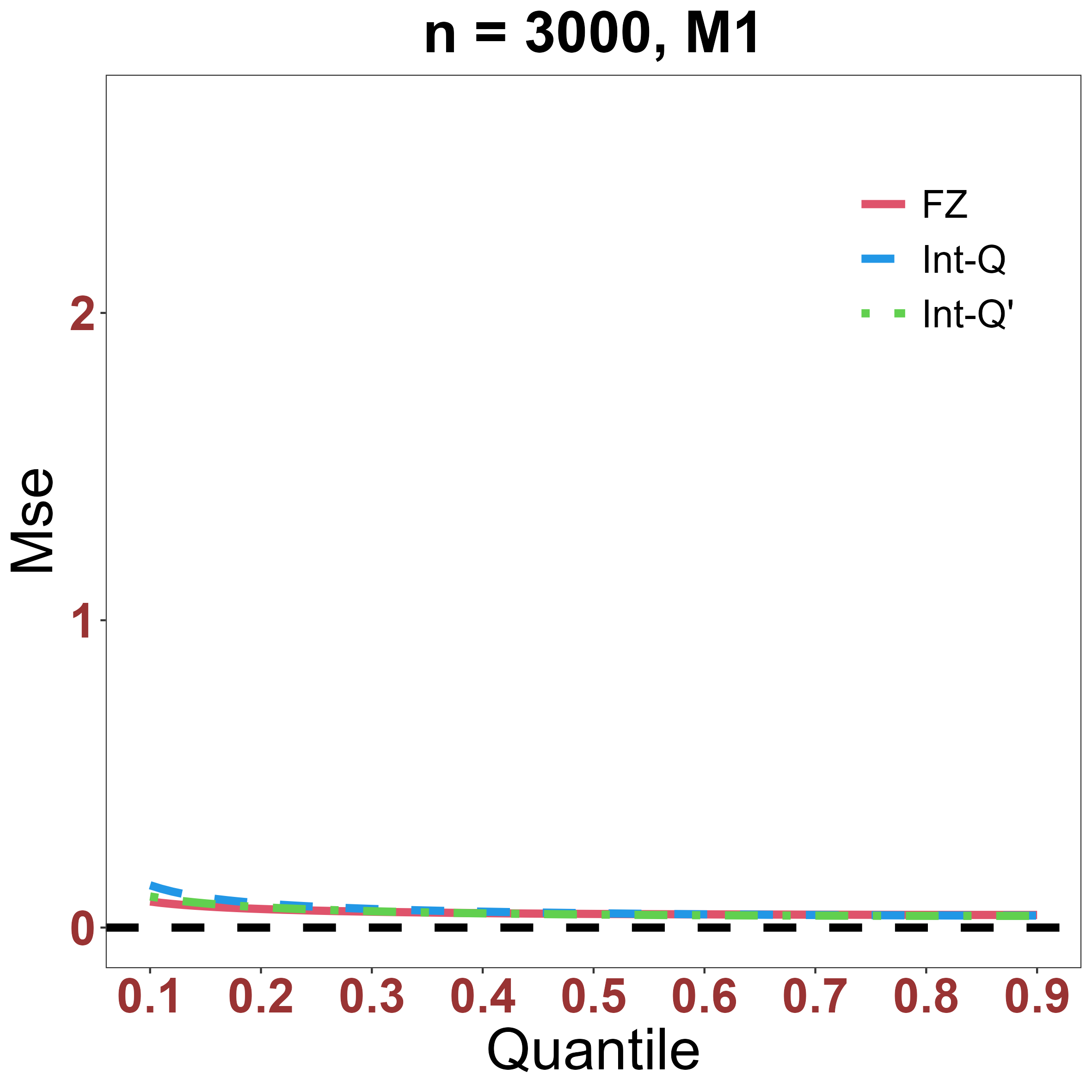}
		\includegraphics[width = 3.9cm, height = 3.2cm]{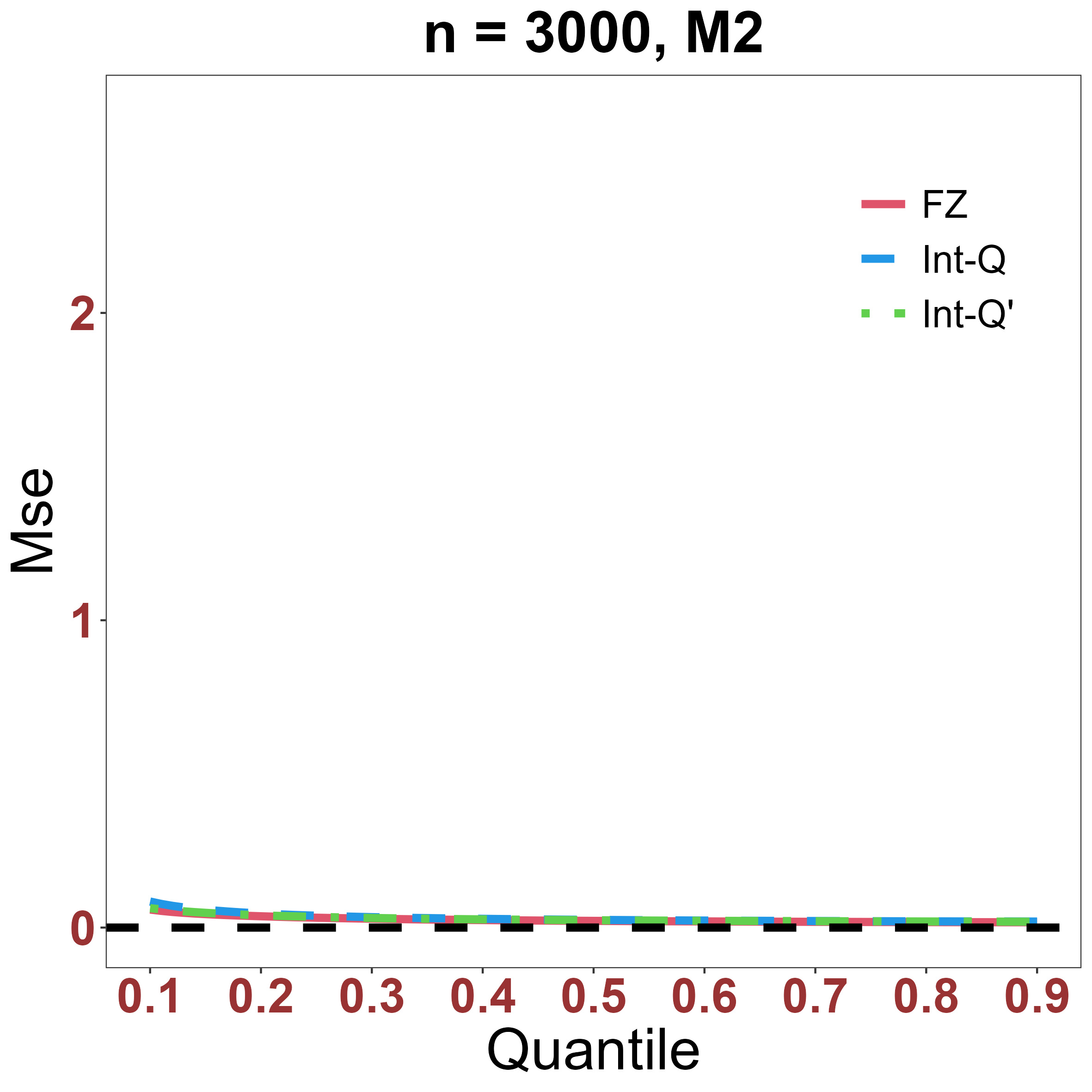}
		\includegraphics[width = 3.9cm, height = 3.2cm]{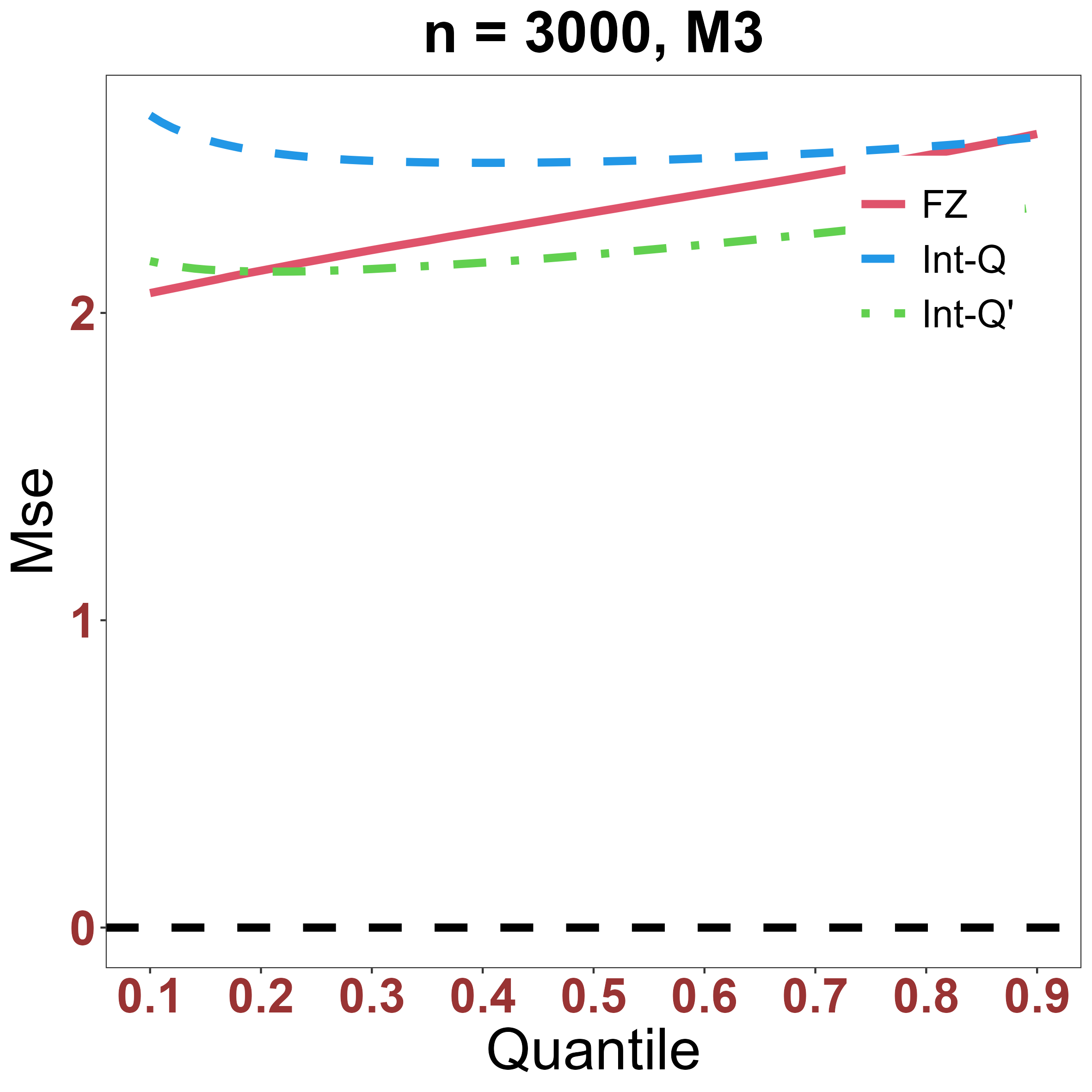}	
		\includegraphics[width = 3.9cm, height = 3.2cm]{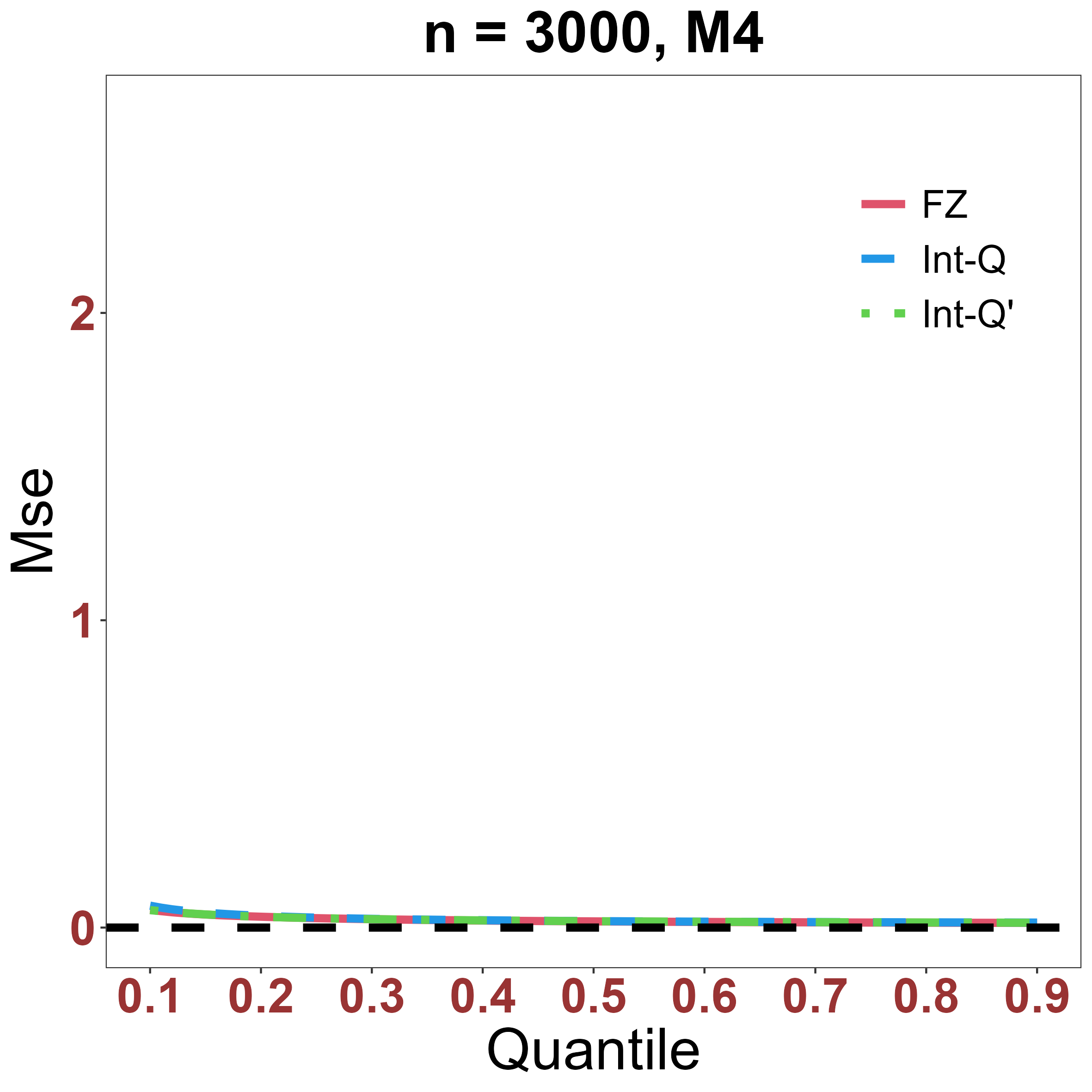}
	}	
	%\end{subfigure}
	\caption{Bias, variance and MSE of the CTATE estimator using the FZ loss (FZ) and that using the trimmed integrated-QTE based estimators (IntQ and IntQ') when $\rho=0.5$ and $n=3,000$.}
	\label{figure6}
\end{figure}

\begin{figure}[ht]
	%\begin{subfigure}
	\centering
	\mbox{
		\includegraphics[height=8cm,width=8cm]{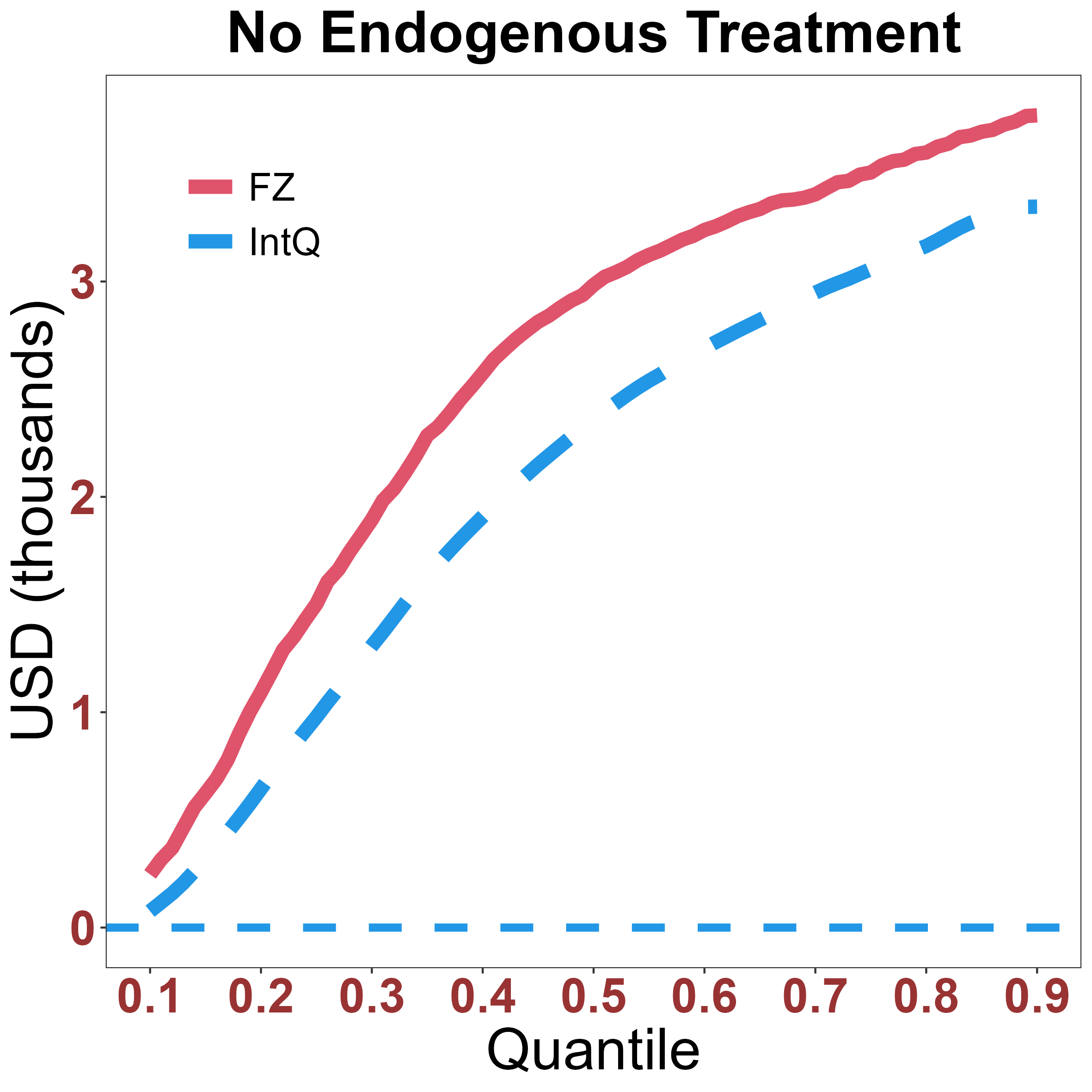}
		\includegraphics[height=8cm,width=8cm]{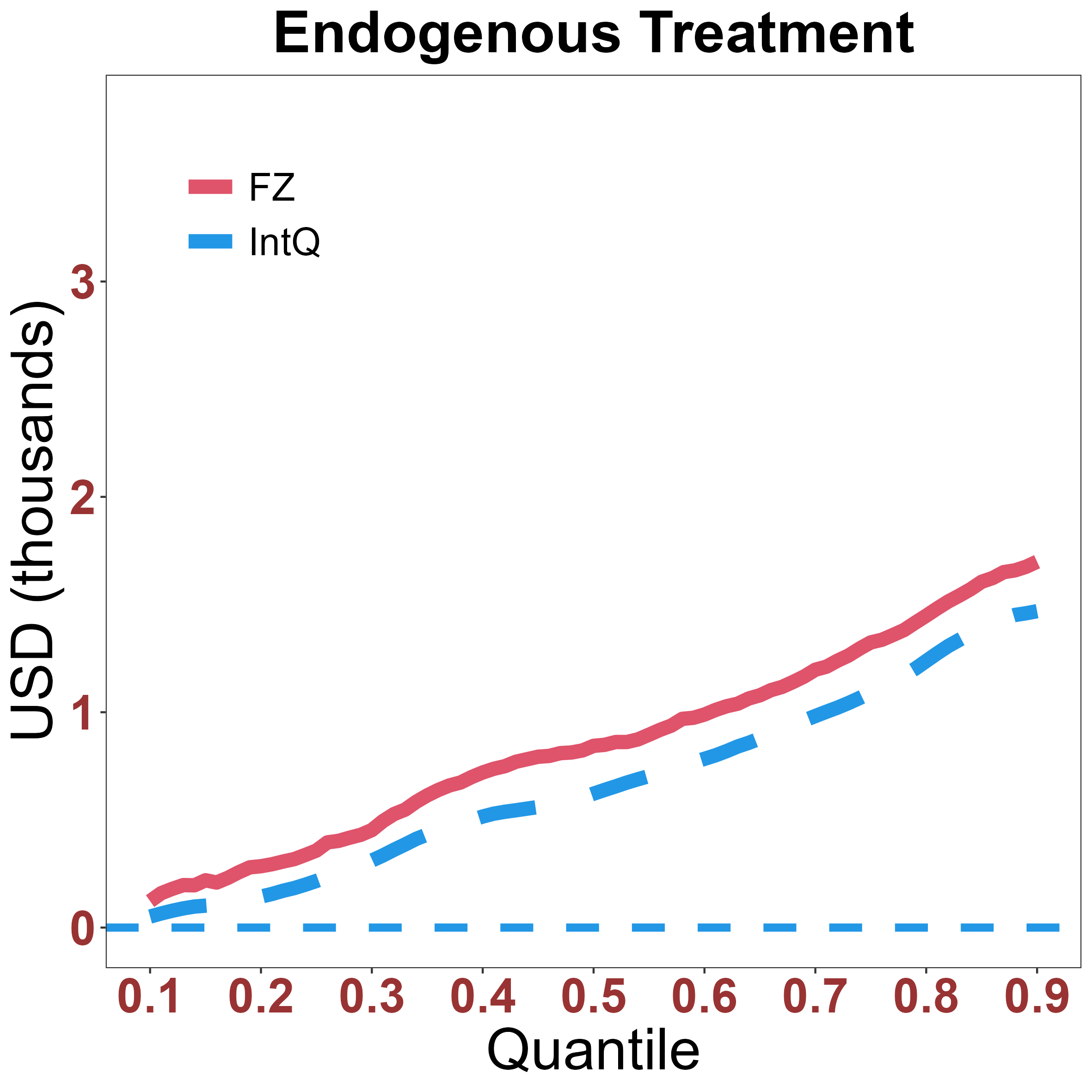}
	}	
	\mbox{	
		\includegraphics[height=8cm,width=8cm]{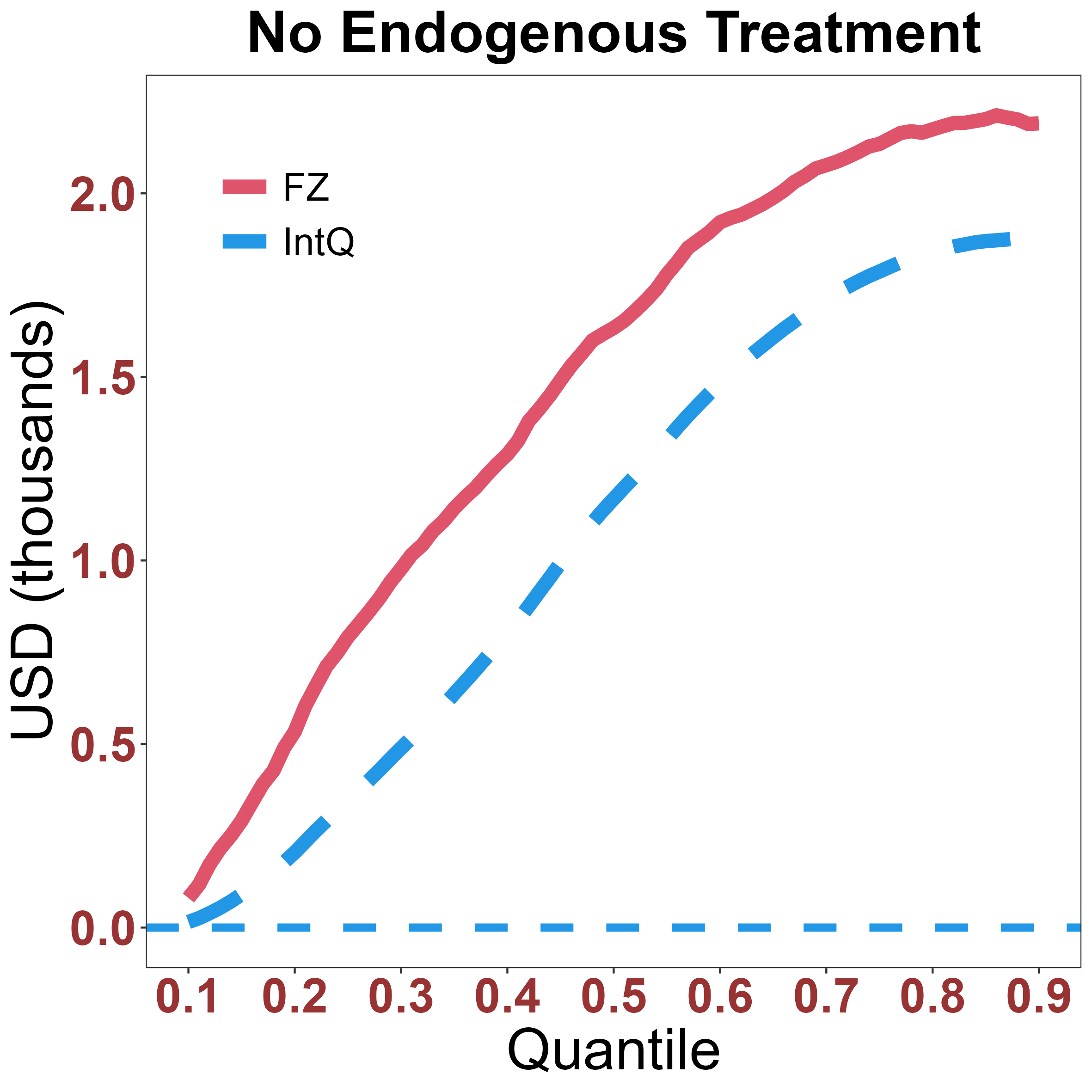}
		\includegraphics[height=8cm,width=8cm]{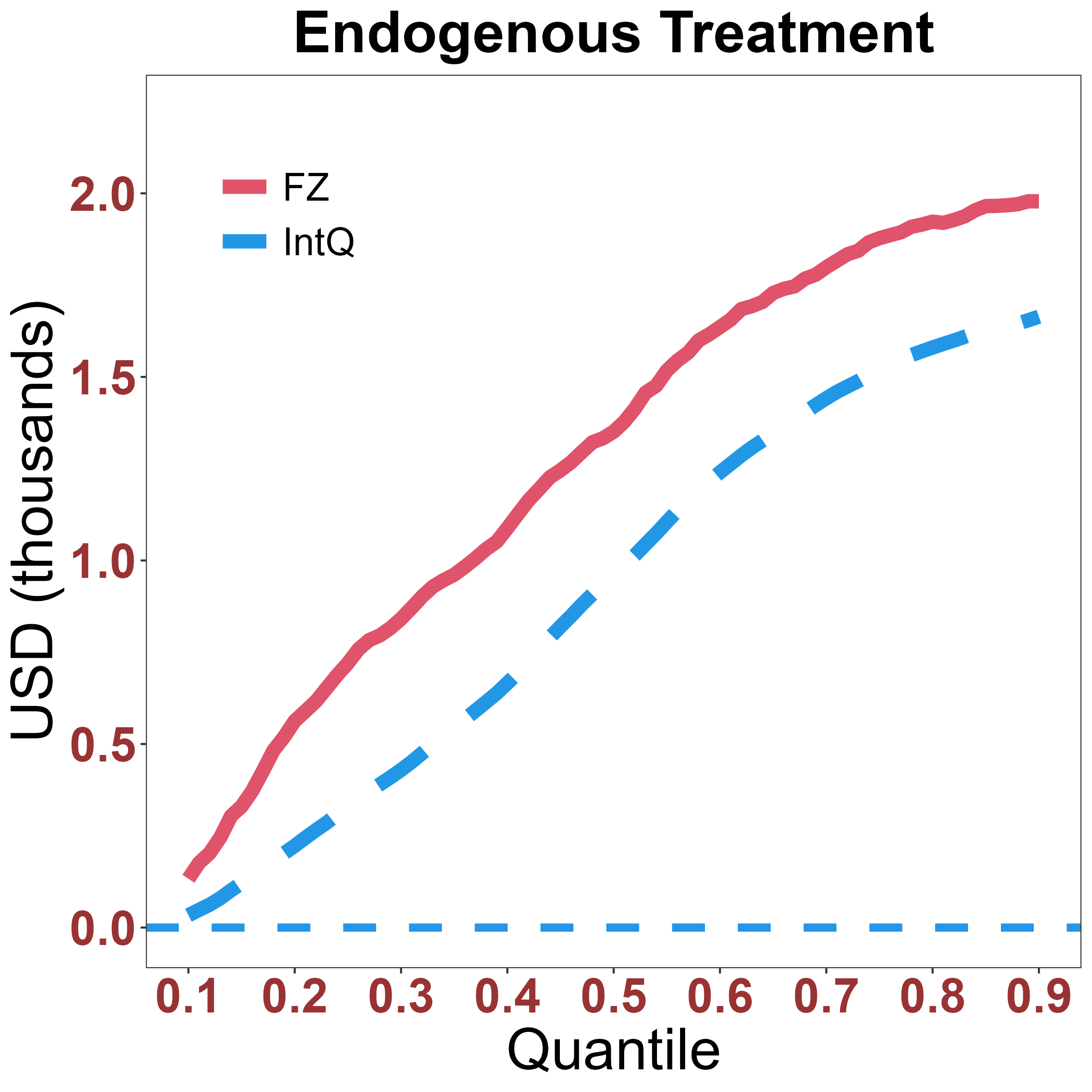}
	}
	%end{subfigure}
	\caption{Comparisons of CTATE estimates for compliers from using the FZ loss and from using the trimmed integrated-QTE based estimator (IntQ). Upper panel: Adult men's earnings. Lower panel: Adult women's earnings.}
	\label{figure7}
\end{figure}

\begin{figure}[ht]
	%\begin{subfigure}
	\centering
	\mbox{
		\includegraphics[height=8cm,width=8cm]{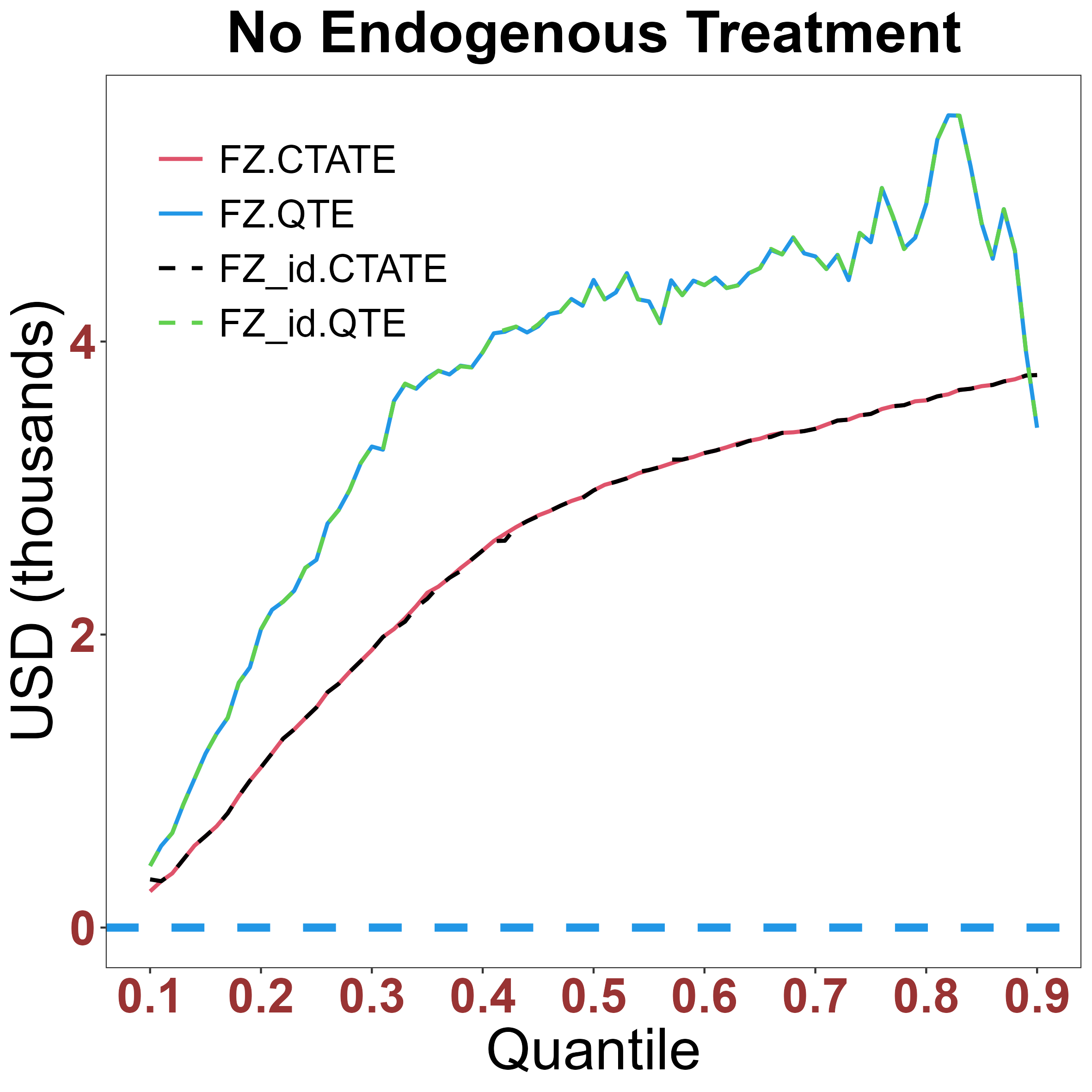}
		\includegraphics[height=8cm,width=8cm]{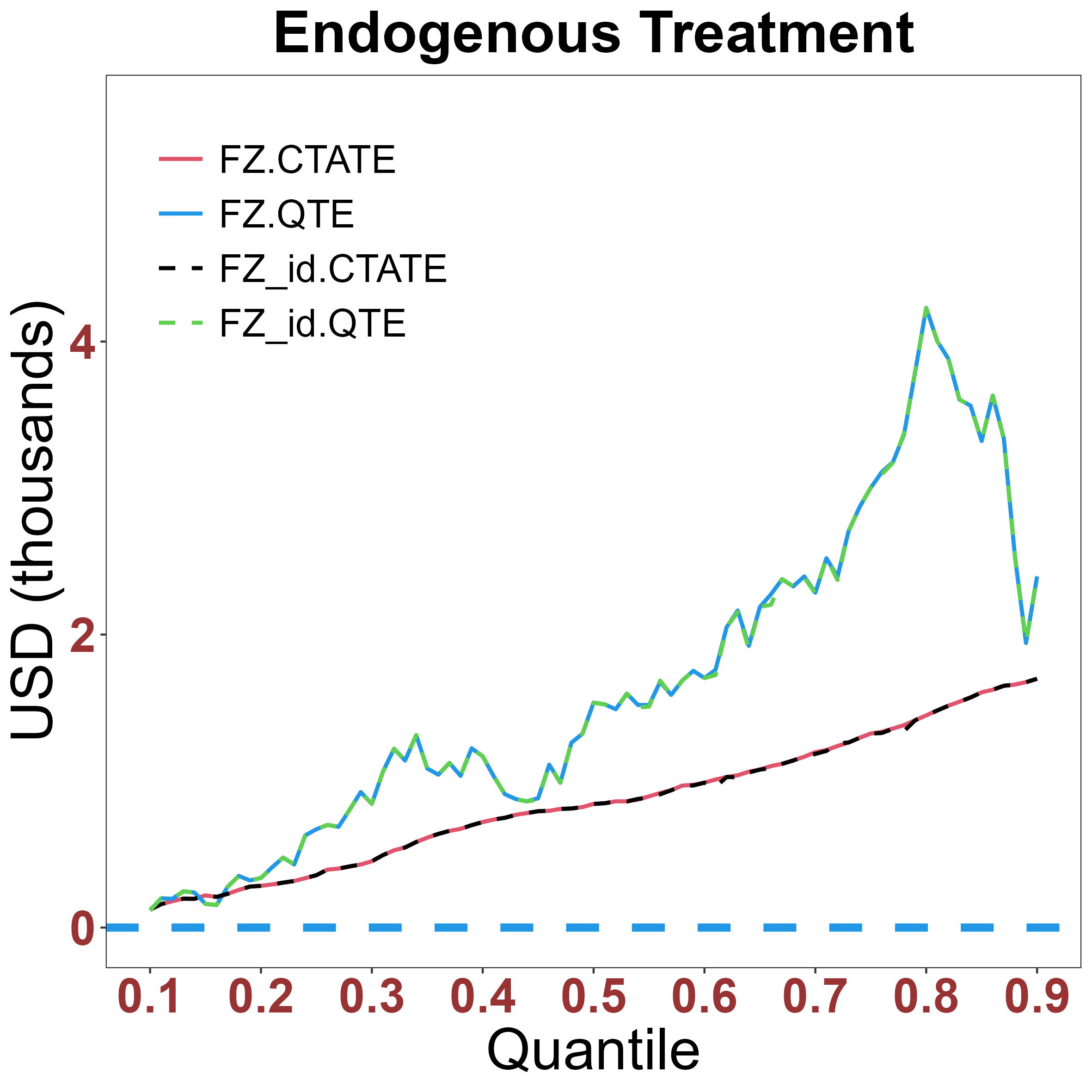}
	}	
	\mbox{	
		\includegraphics[height=8cm,width=8cm]{figure/m_id_nw.png}
		\includegraphics[height=8cm,width=8cm]{figure/m_id.png}
	}
	%end{subfigure}
	\caption{Comparisons of the CTATE and QTE estimates for compliers from using the FZ loss with $G_{1}(t)=t$ (FZ\_id) and with $G_{1}(t)=0$ (FZ). The function $G_{2}(t)=\ln(1+\exp(t))$. Upper panel: Adult men's earnings. Lower panel: Adult women's earnings.}
	\label{figure8}
\end{figure}

\begin{figure}[ht]
	%\begin{subfigure}
	\centering
	\mbox{
		\includegraphics[height=8cm,width=8cm]{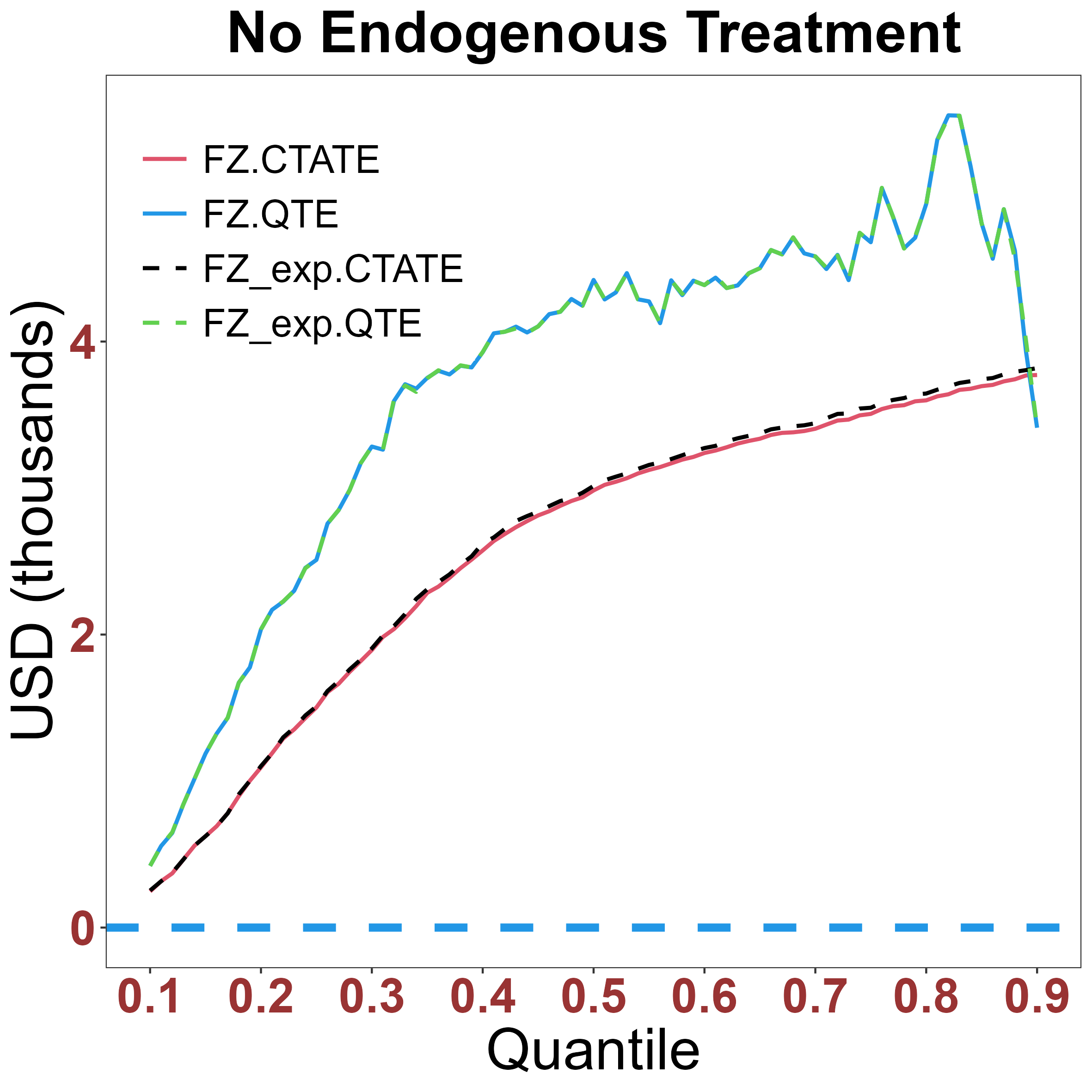}
		\includegraphics[height=8cm,width=8cm]{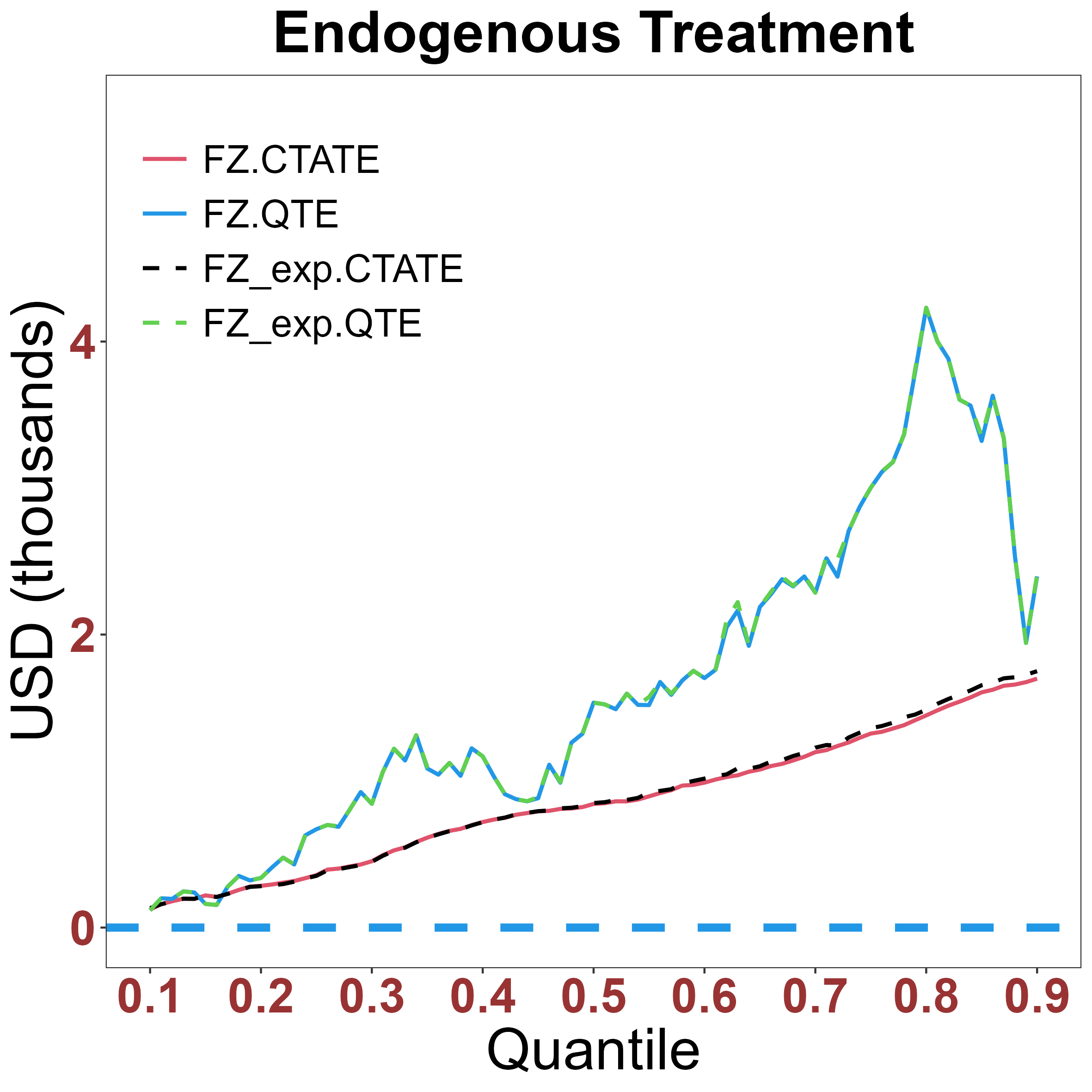}
	}	
	\mbox{	
		\includegraphics[height=8cm,width=8cm]{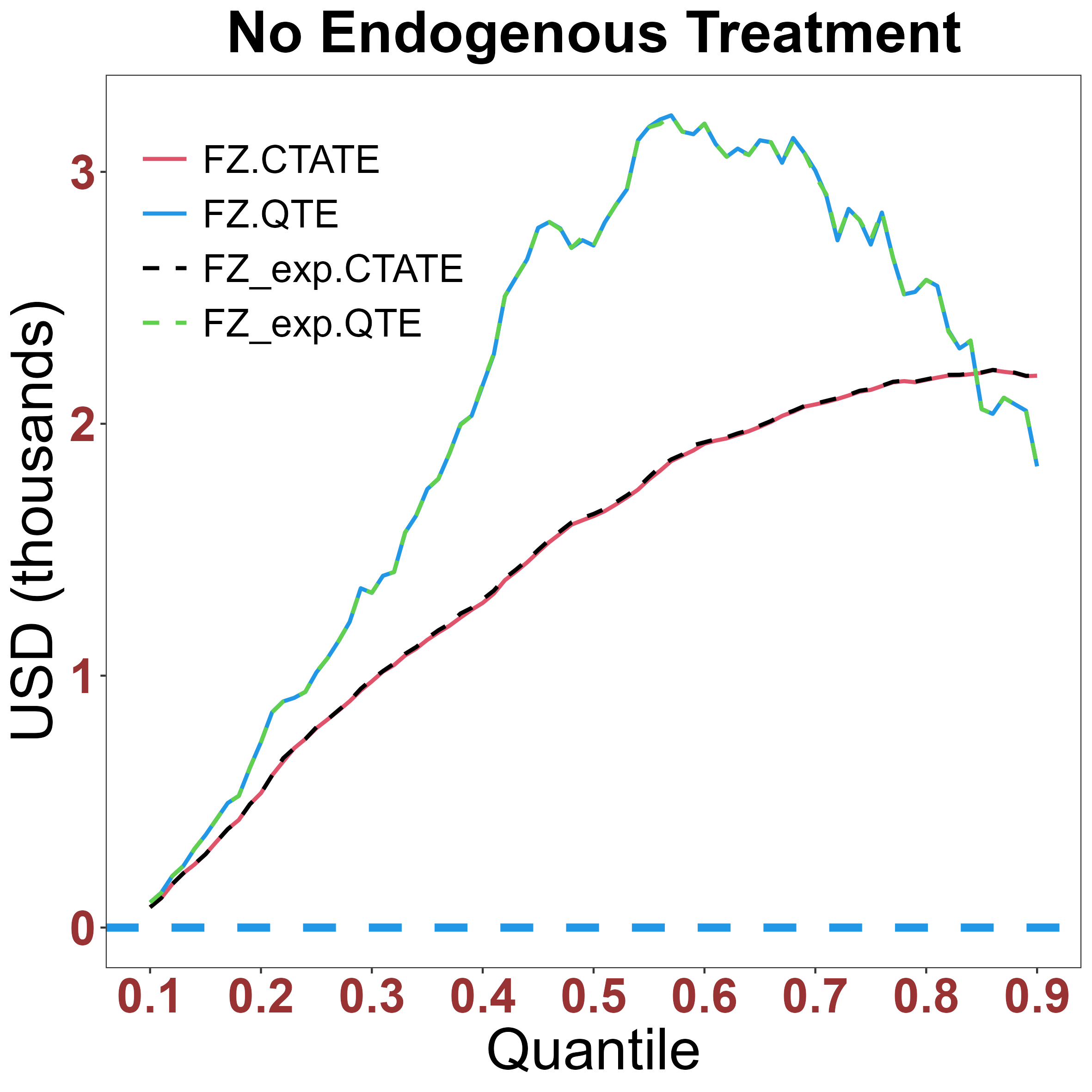}
		\includegraphics[height=8cm,width=8cm]{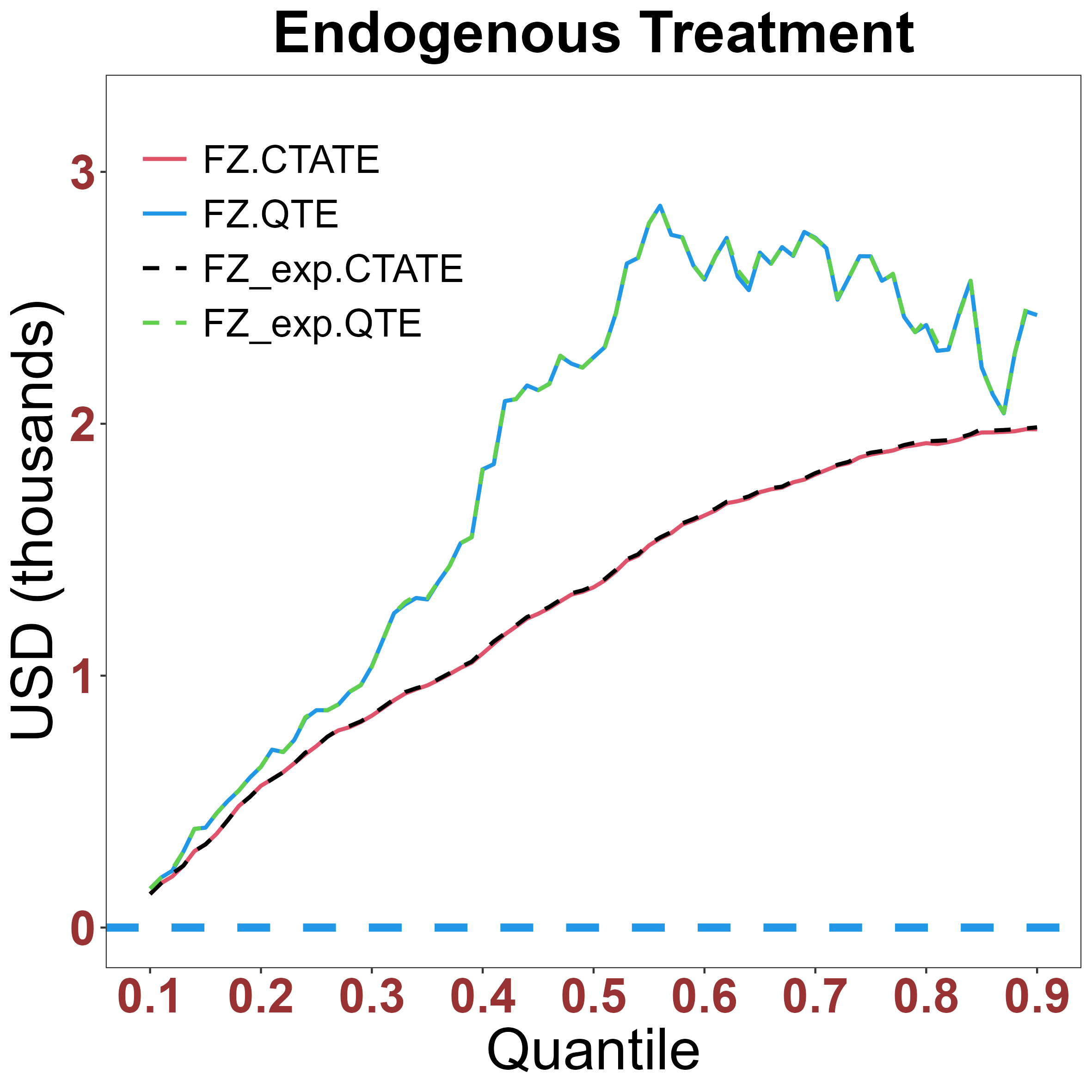}
	}
	%\end{subfigure}
	\caption{Comparisons of CTATE and QTE estimates for compliers from using the FZ loss with $G_{2}(t)=\exp(t)$ (FZ\_exp) and with $G_{2}(t)=\ln(1+\exp(t))$ (FZ). The function $G_{1}(t)=0$. Upper panel: Adult men's earnings. Lower panel: Adult women's earnings.}
	\label{figure9}
\end{figure}

\begin{figure}[ht]
	%\begin{subfigure}
	\centering
	\mbox{
		\includegraphics[height=8cm,width=8cm]{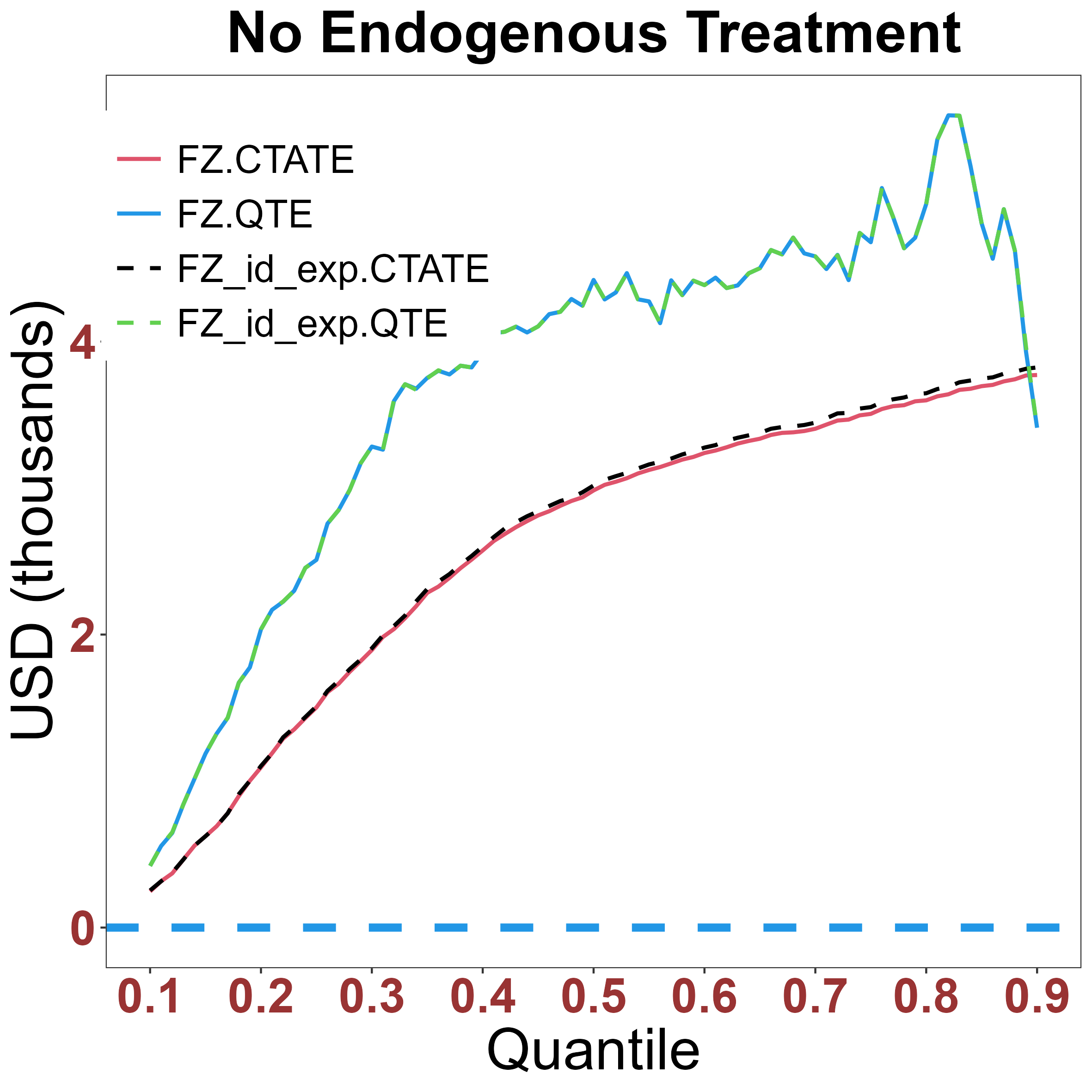}
		\includegraphics[height=8cm,width=8cm]{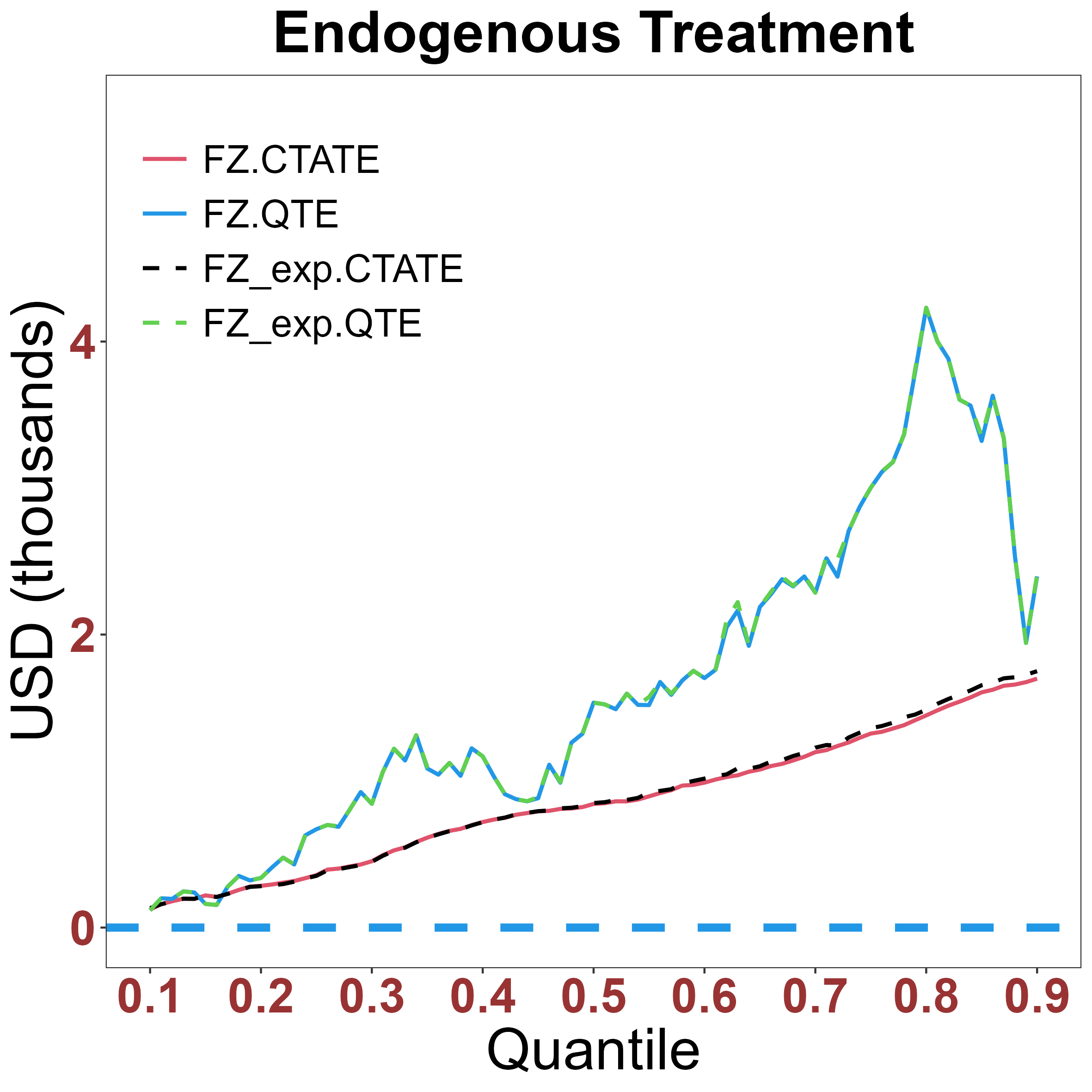}
	}	
	\mbox{	
		\includegraphics[height=8cm,width=8cm]{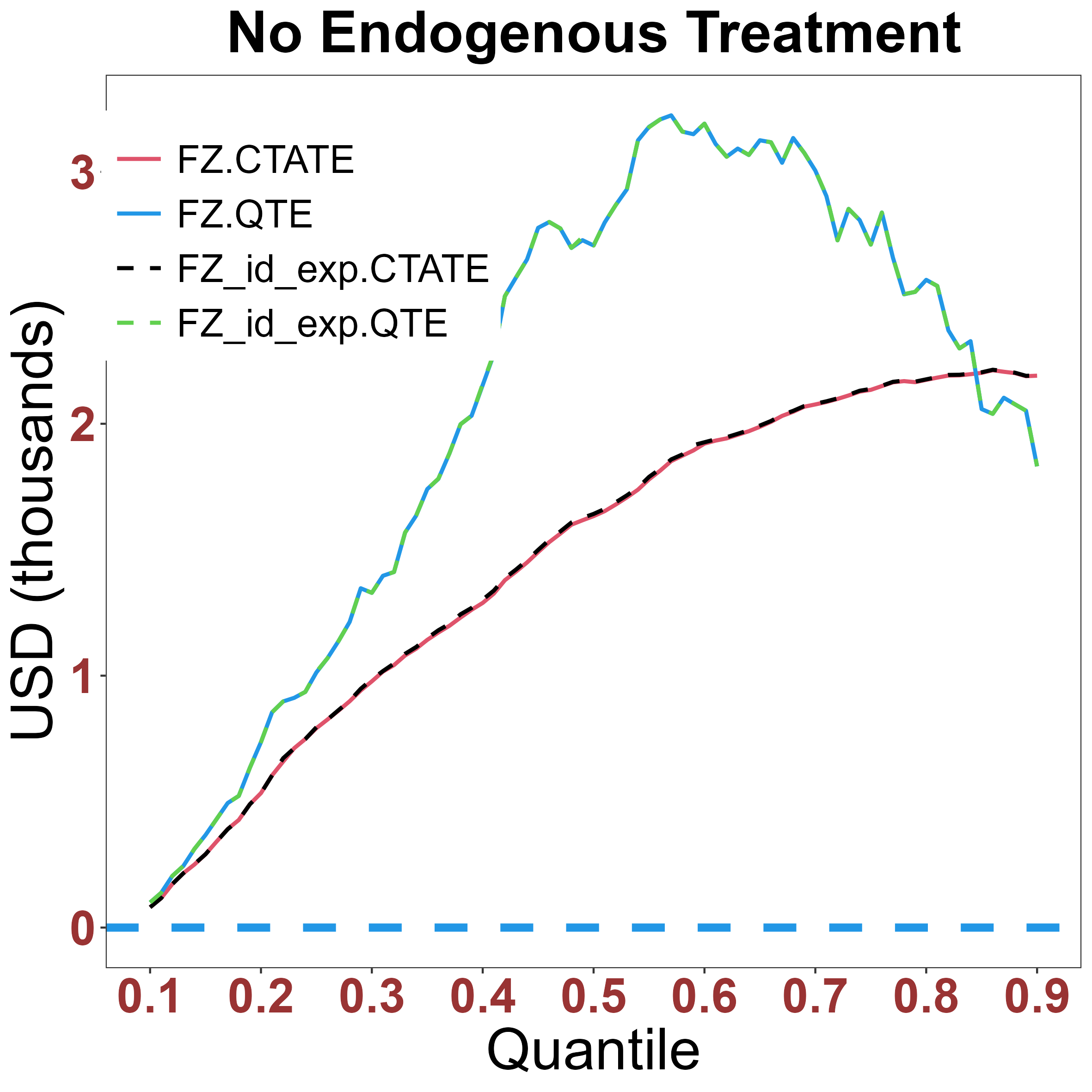}
		\includegraphics[height=8cm,width=8cm]{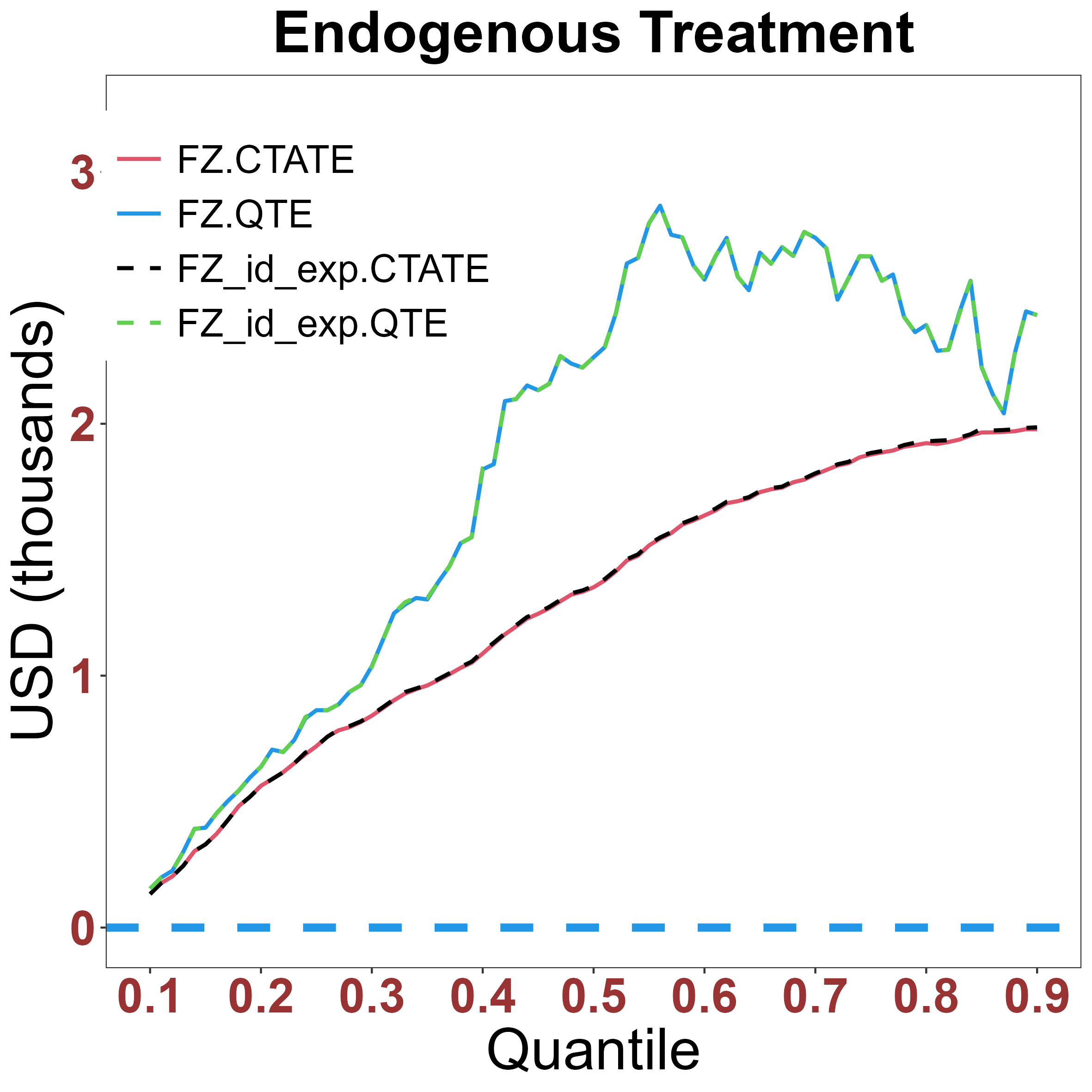}
	}
	%\end{subfigure}
	\caption{Comparisons of CTATE and QTE estimates for compliers from using the FZ loss with $G_{1}(t)=t$ and $G_{2}(t)=\exp(t)$ (FZ\_id\_exp) and with $G_{1}(t)=0$ and $G_{2}(t)=\ln(1+\exp(t))$ (FZ). Upper panel: Adult men's earnings. Lower panel: Adult women's earnings.}
	\label{figure10}
\end{figure}

\begin{figure}[!htb]
	%\begin{subfigure}
	\centering
	\mbox{
		\includegraphics[width = 3.9cm, height = 3.2cm]{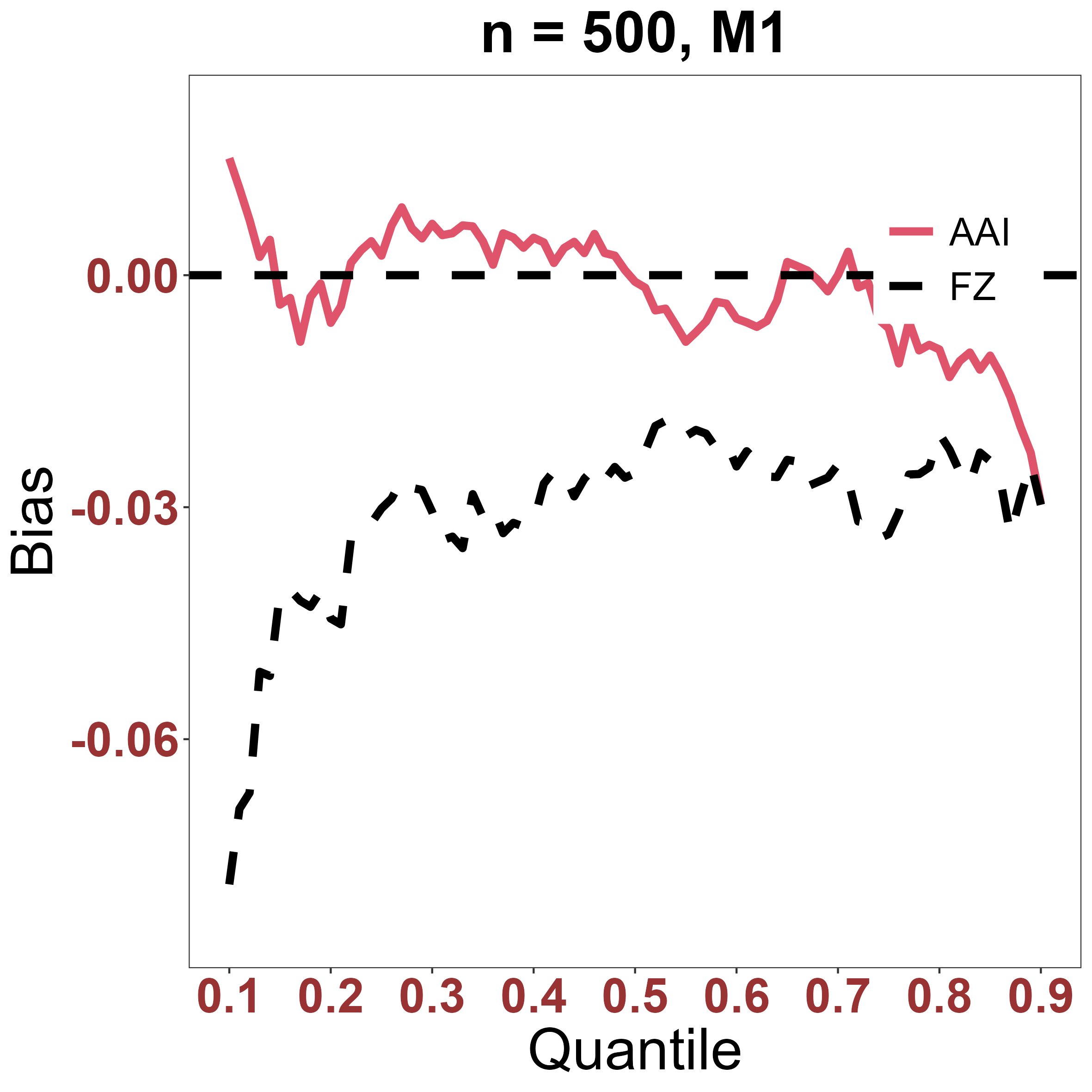}
		\includegraphics[width = 3.9cm, height = 3.2cm]{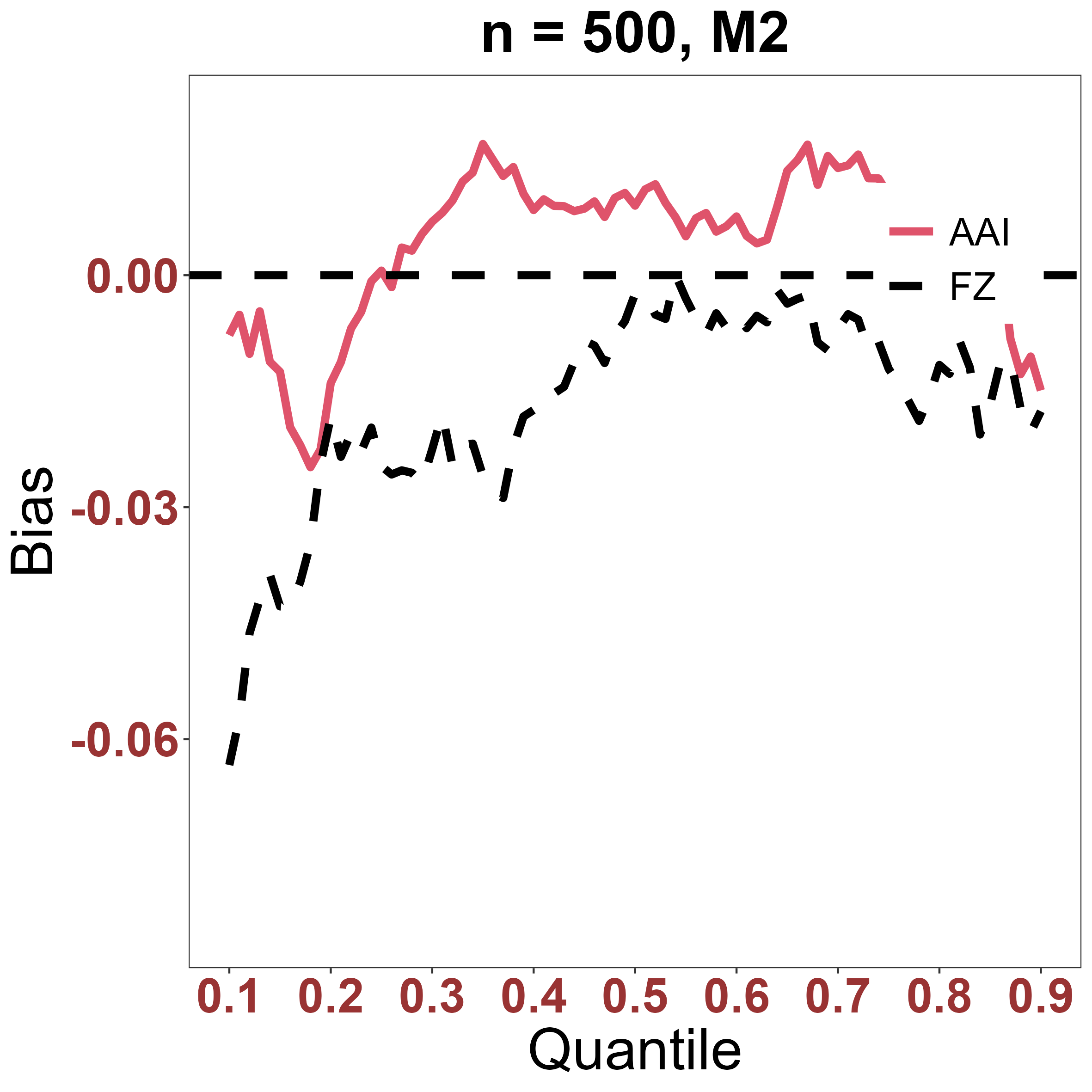}
		\includegraphics[width = 3.9cm, height = 3.2cm]{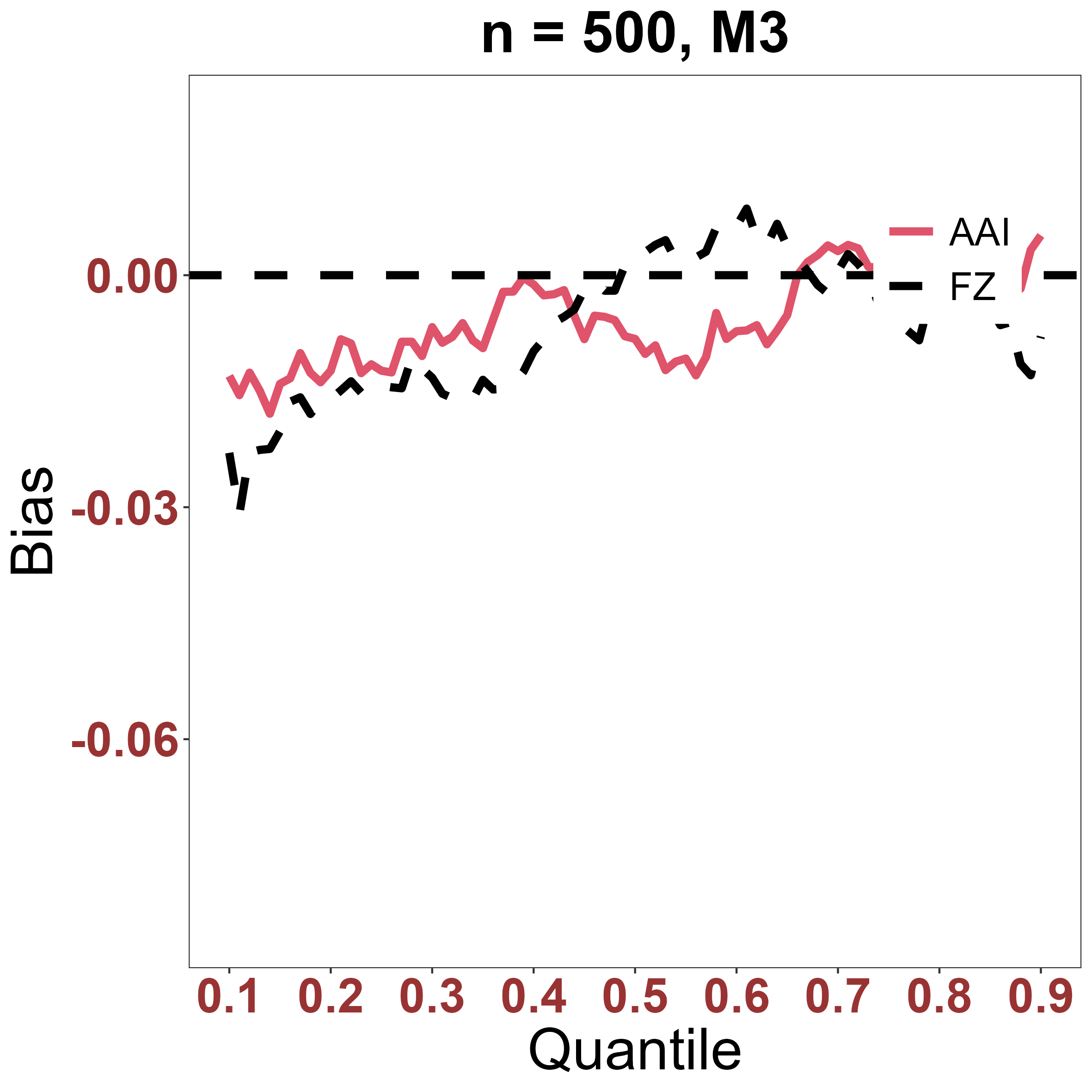}	
		\includegraphics[width = 3.9cm, height = 3.2cm]{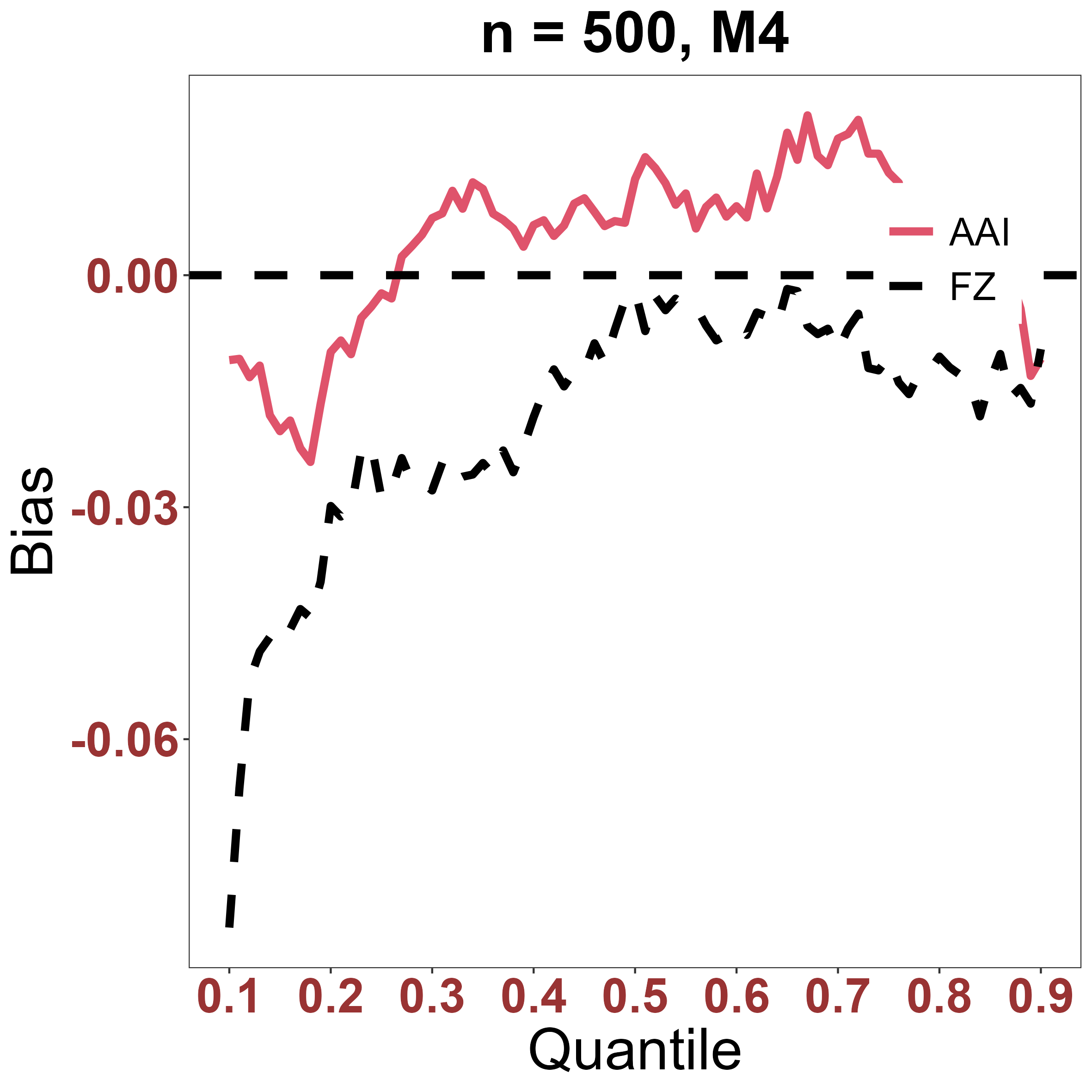}
	}	
	\mbox{	\includegraphics[width = 3.9cm, height = 3.2cm]{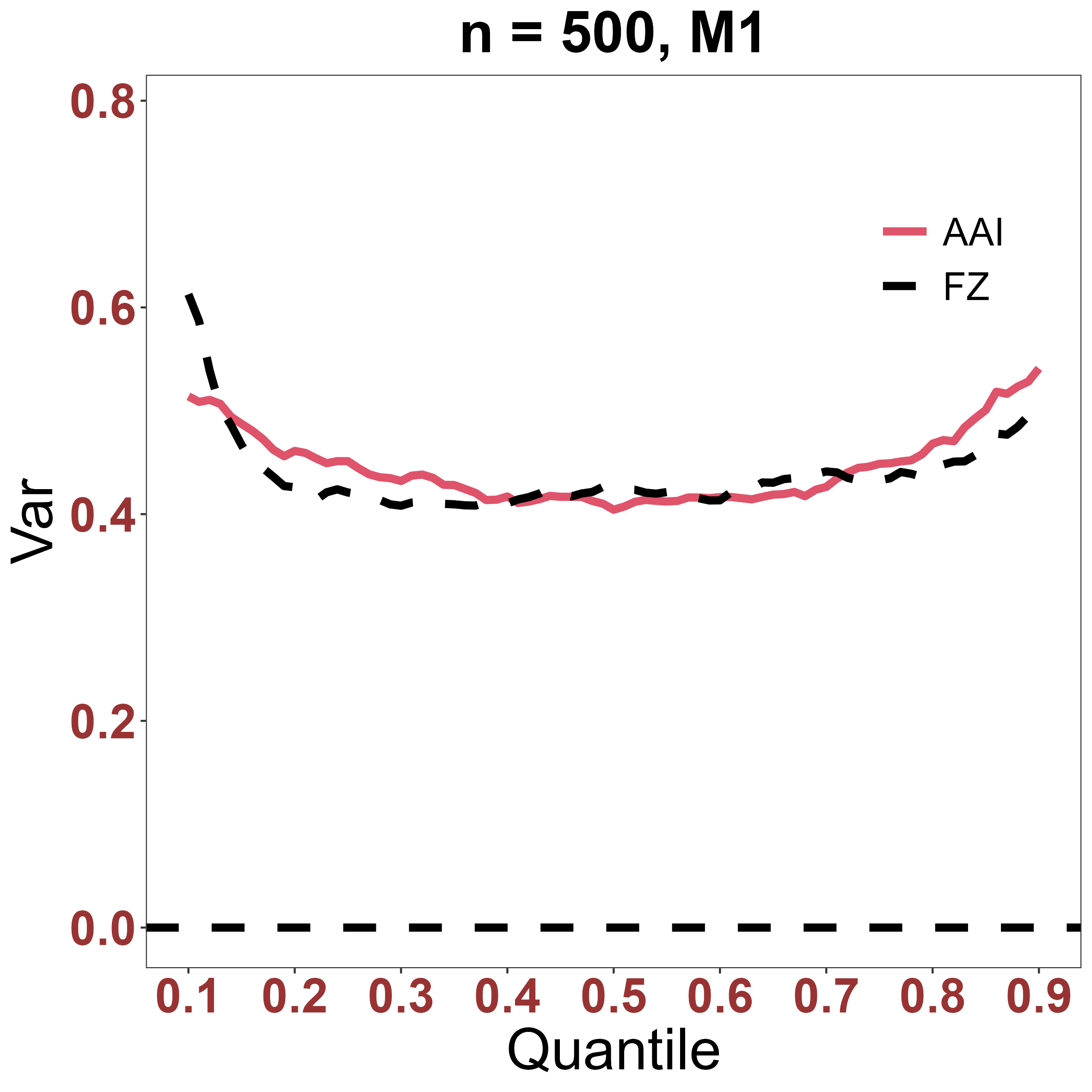}
		\includegraphics[width = 3.9cm, height = 3.2cm]{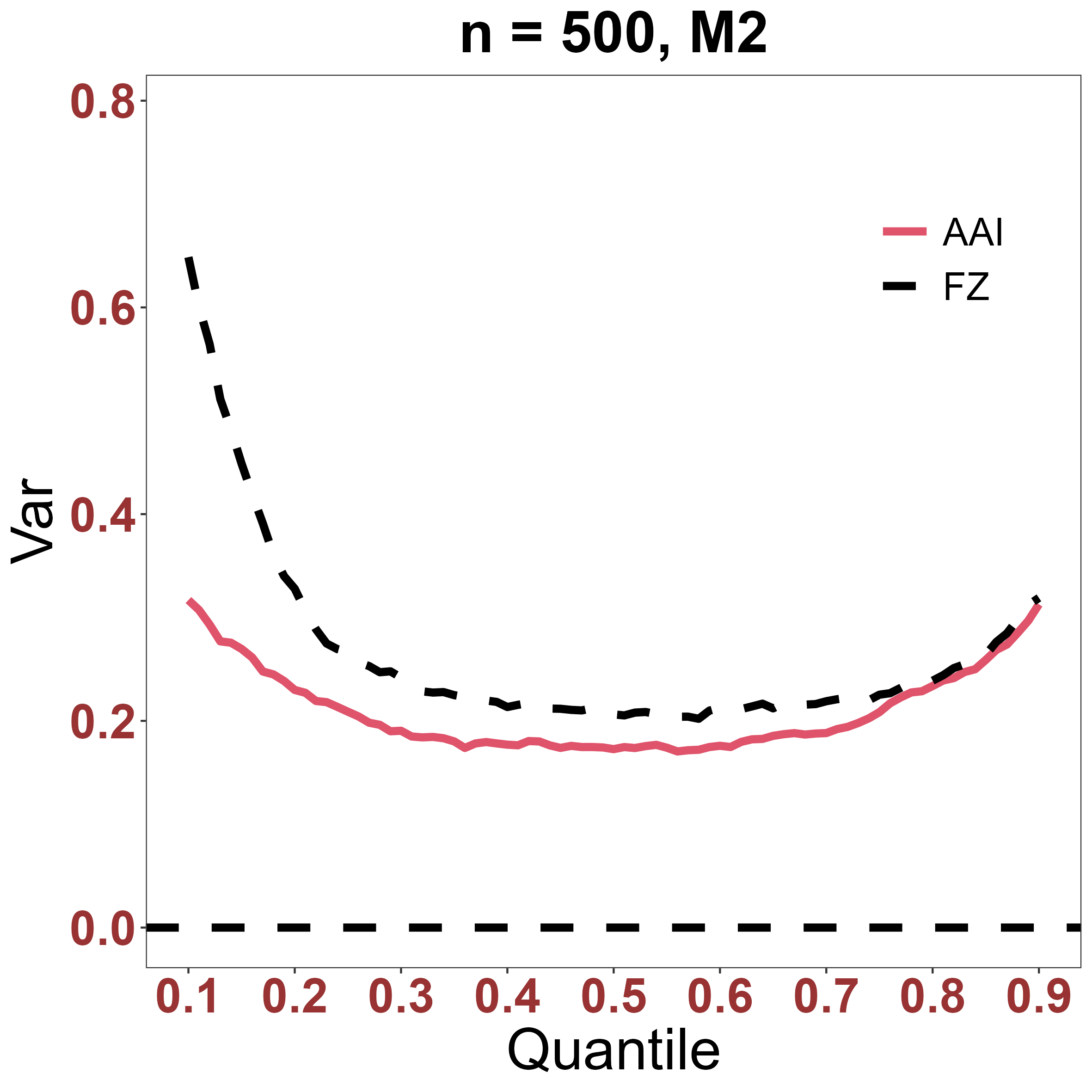}
		\includegraphics[width = 3.9cm, height = 3.2cm]{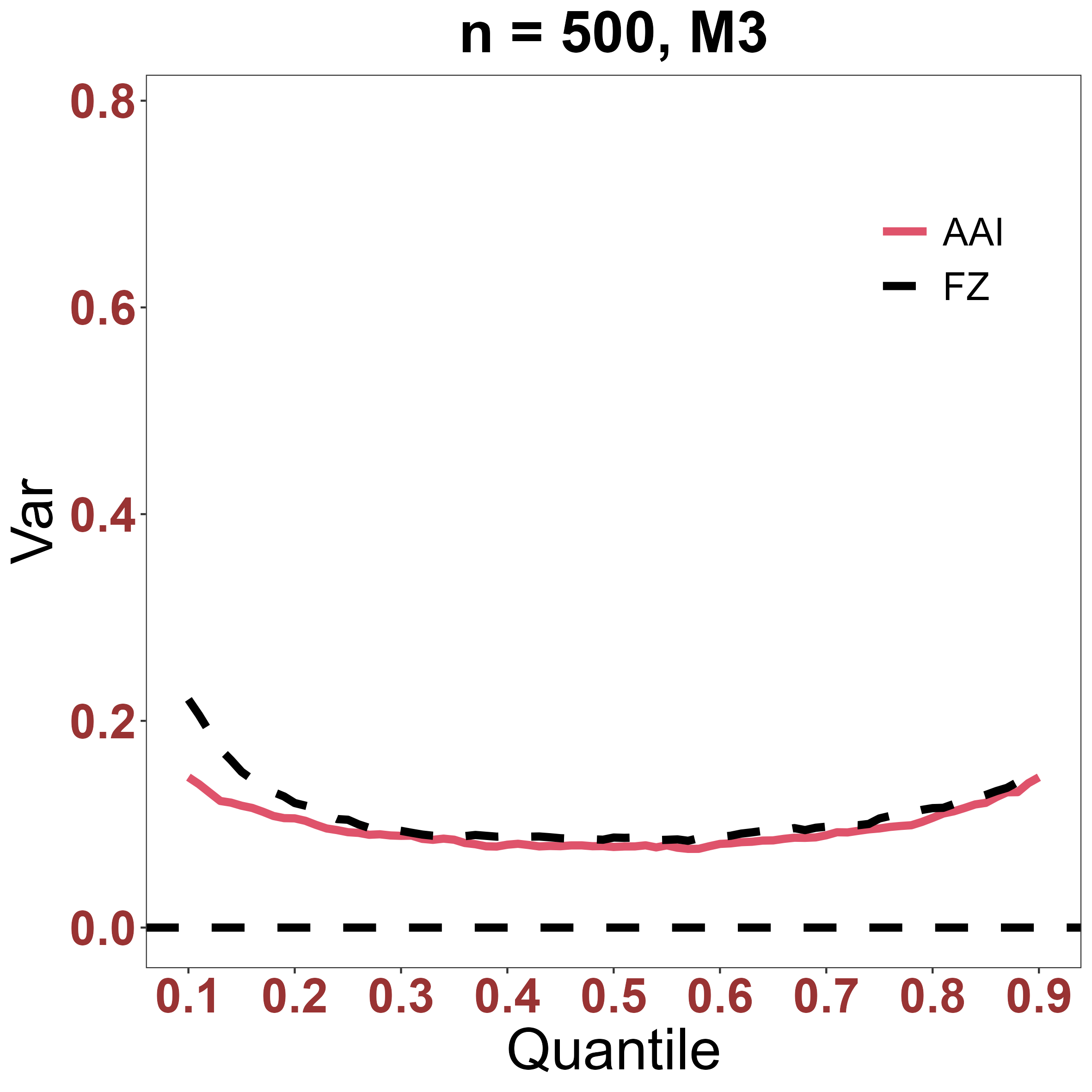}	
		\includegraphics[width = 3.9cm, height = 3.2cm]{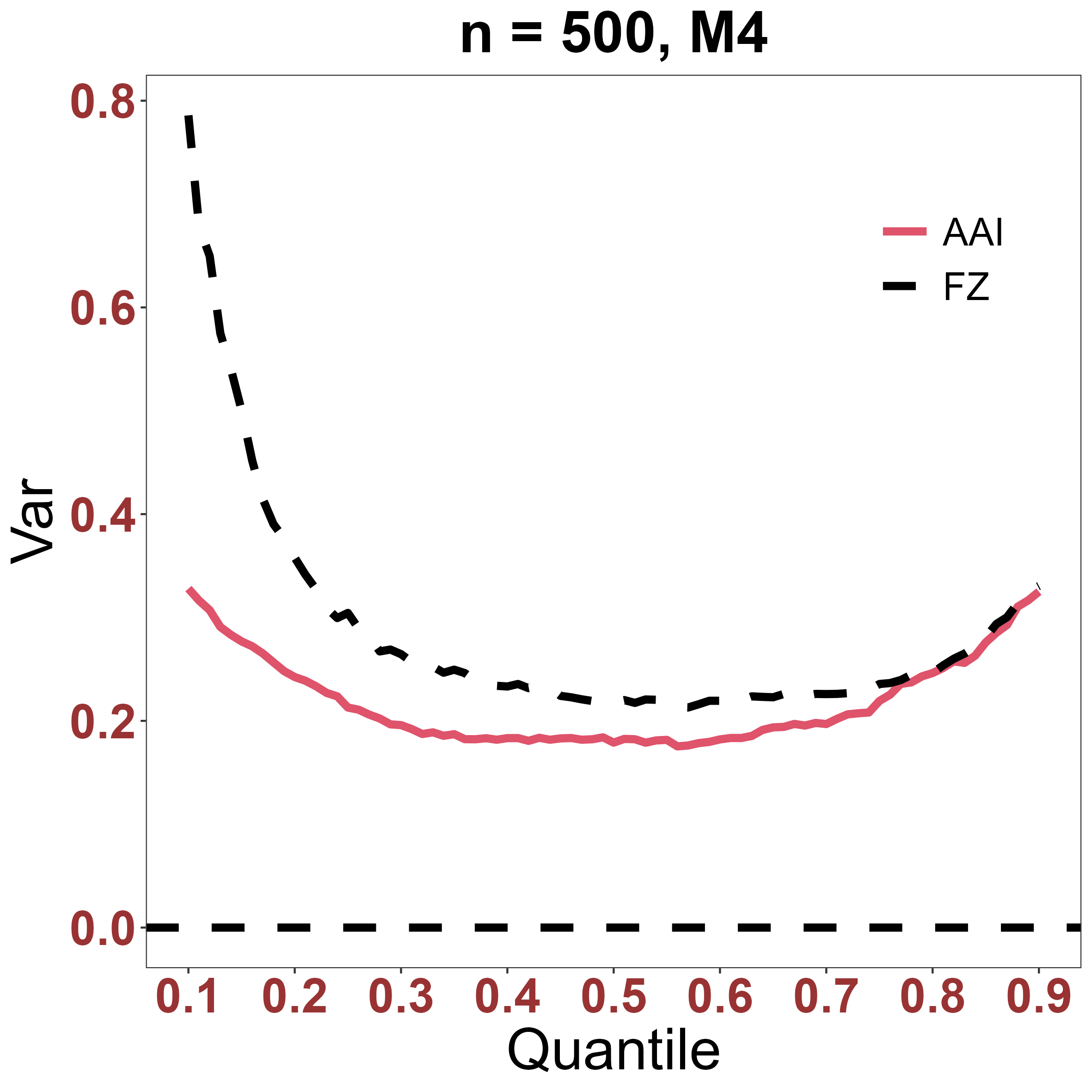}
	}	
	\mbox{	\includegraphics[width = 3.9cm, height = 3.2cm]{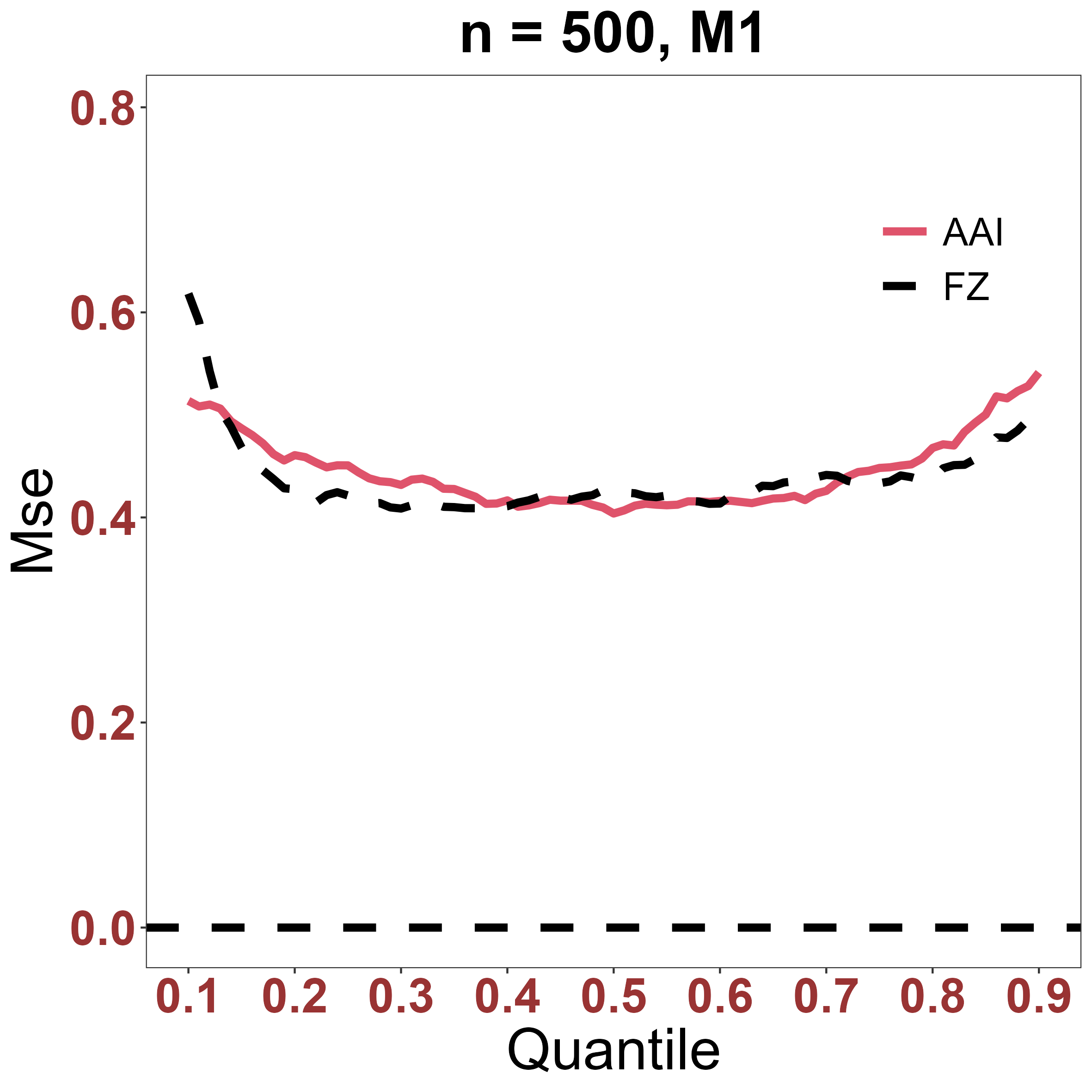}
		\includegraphics[width = 3.9cm, height = 3.2cm]{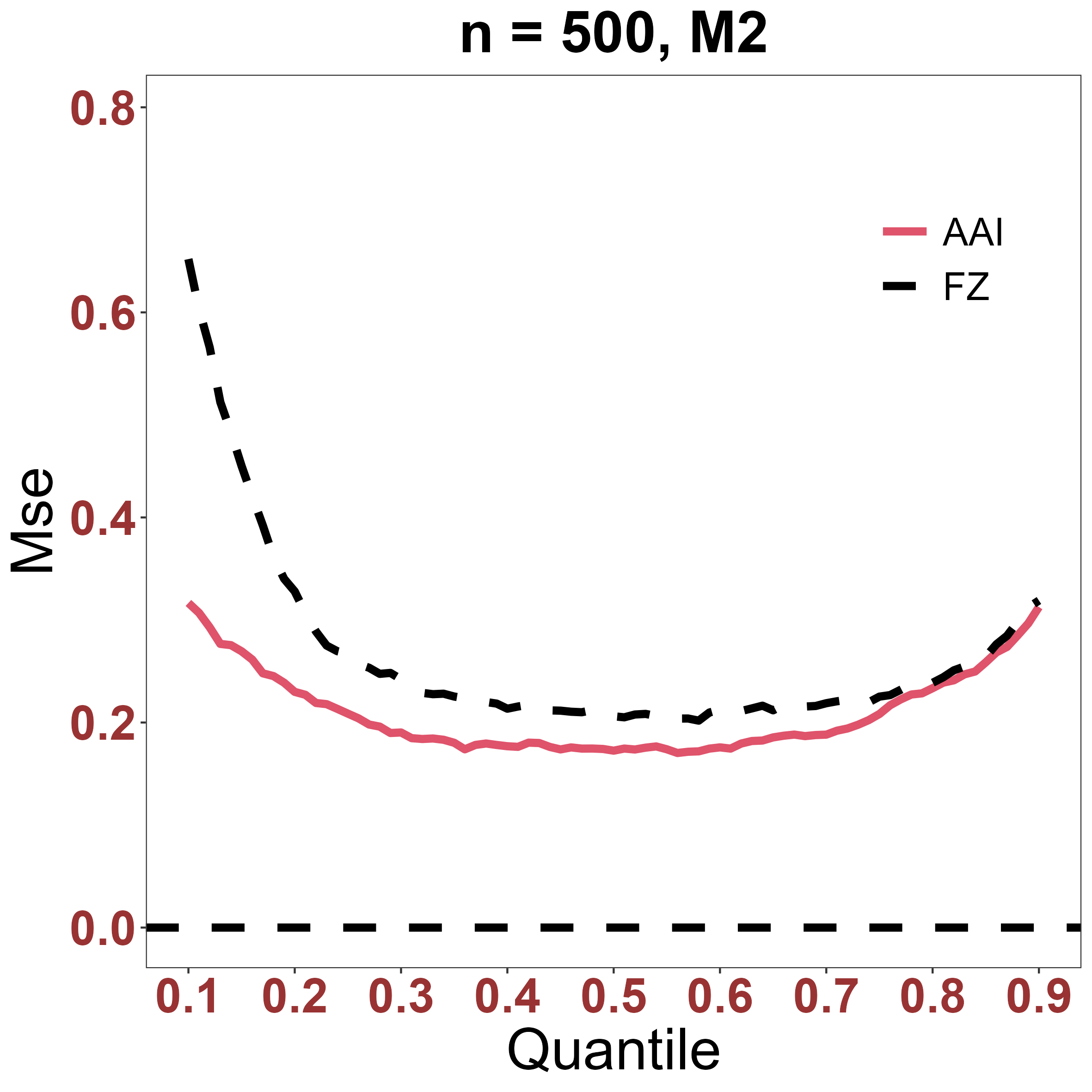}
		\includegraphics[width = 3.9cm, height = 3.2cm]{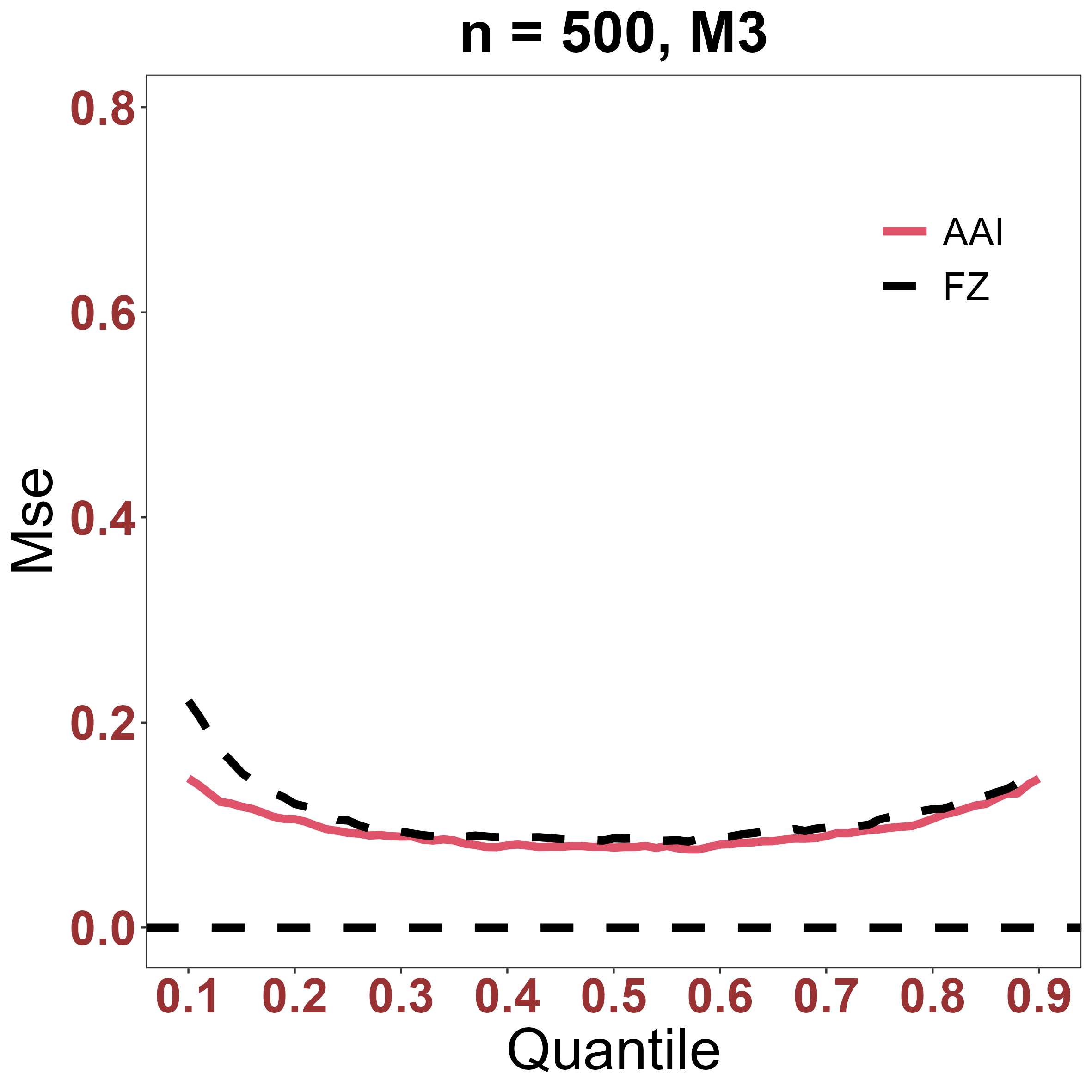}	
		\includegraphics[width = 3.9cm, height = 3.2cm]{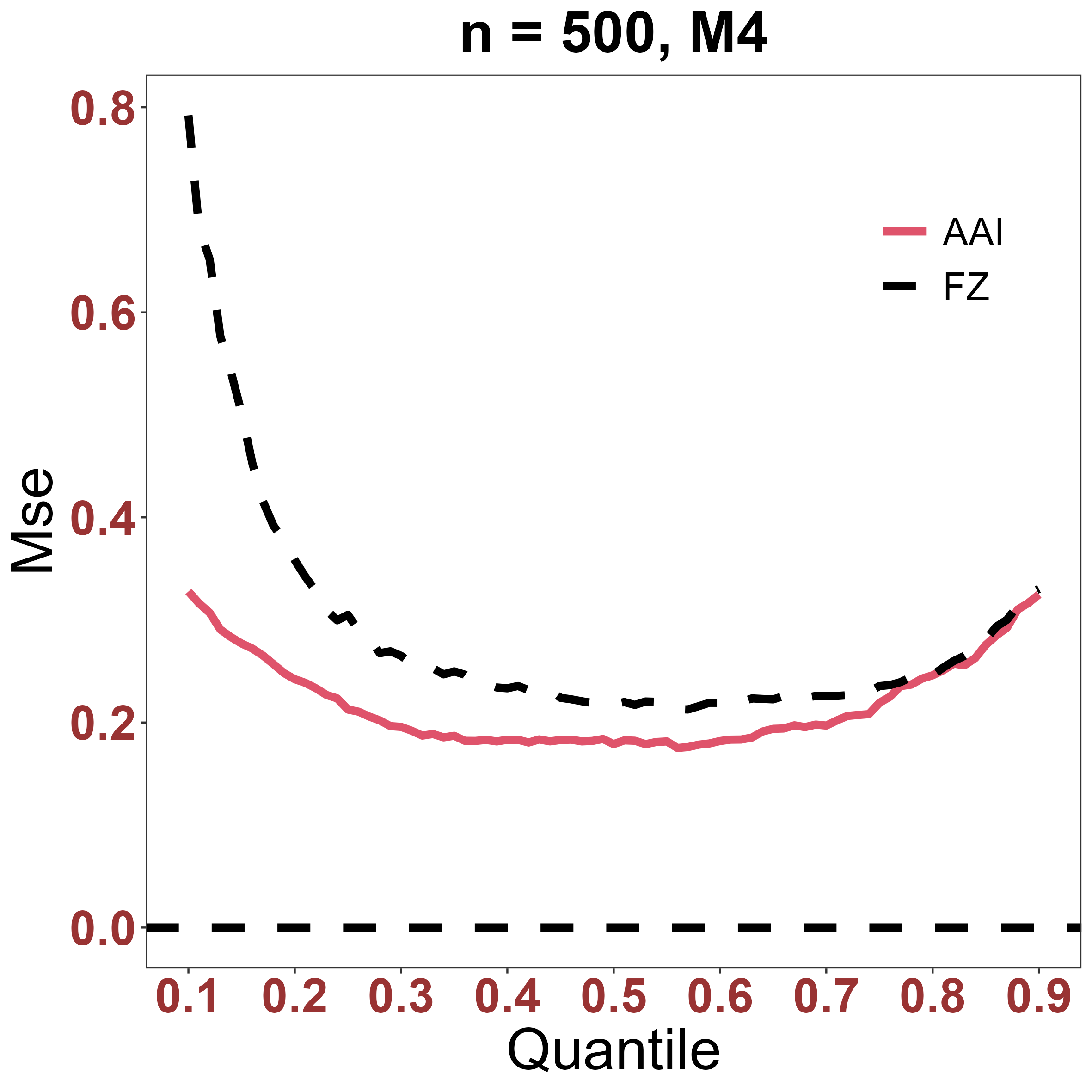}	
	}
	%\end{subfigure}
	\caption{Bias, variance and MSE of the QTE estimator using the weighted quantile regression (WQR) of \cite{AAI_2002} (AAI) and that using the FZ loss (FZ) when $\rho=0$ and $n=500$.}
	\label{figure11}
\end{figure}

\begin{figure}[!htb]
	%\begin{subfigure}
	\centering
	\mbox{
		\includegraphics[width = 3.9cm, height = 3.2cm]{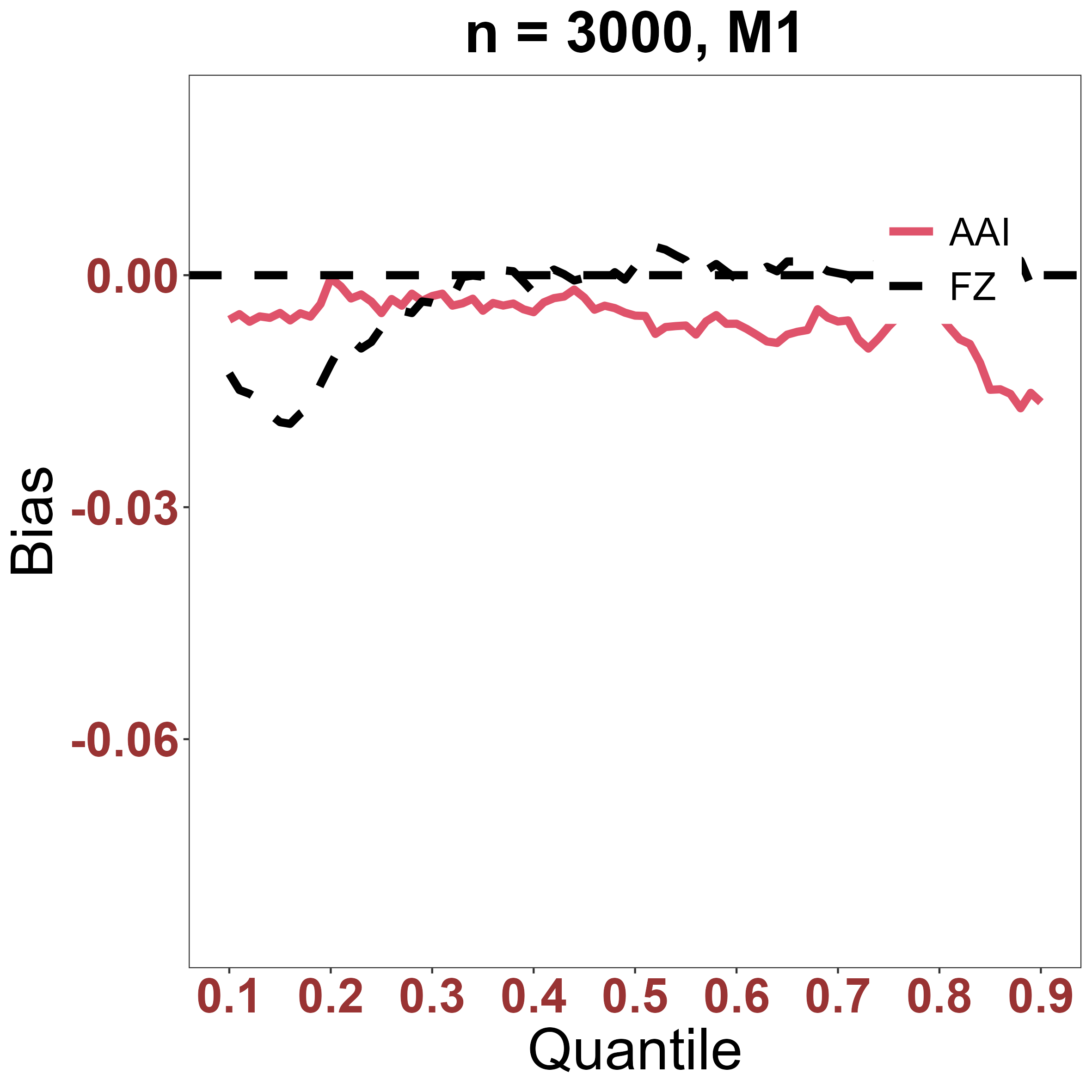}
		\includegraphics[width = 3.9cm, height = 3.2cm]{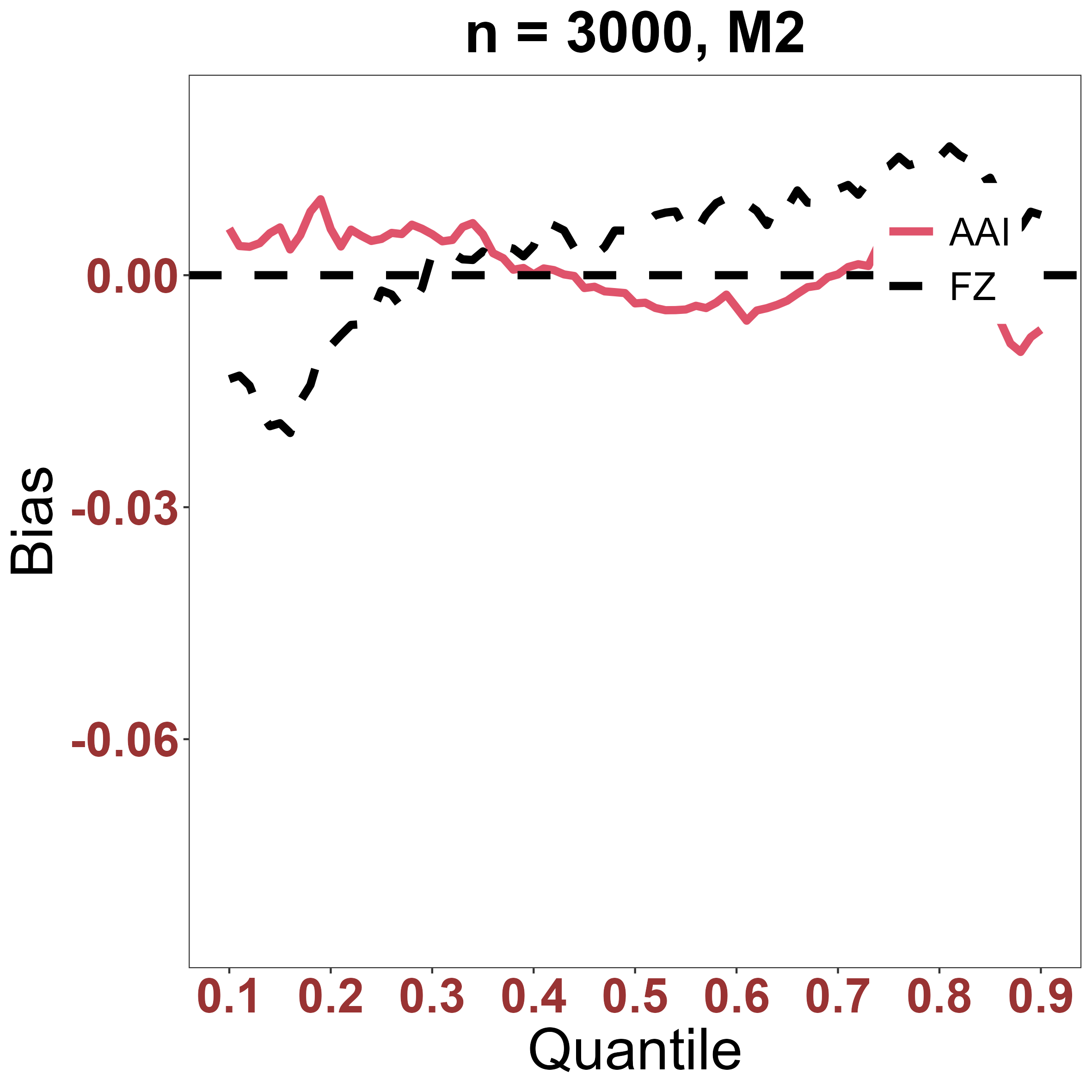}
		\includegraphics[width = 3.9cm, height = 3.2cm]{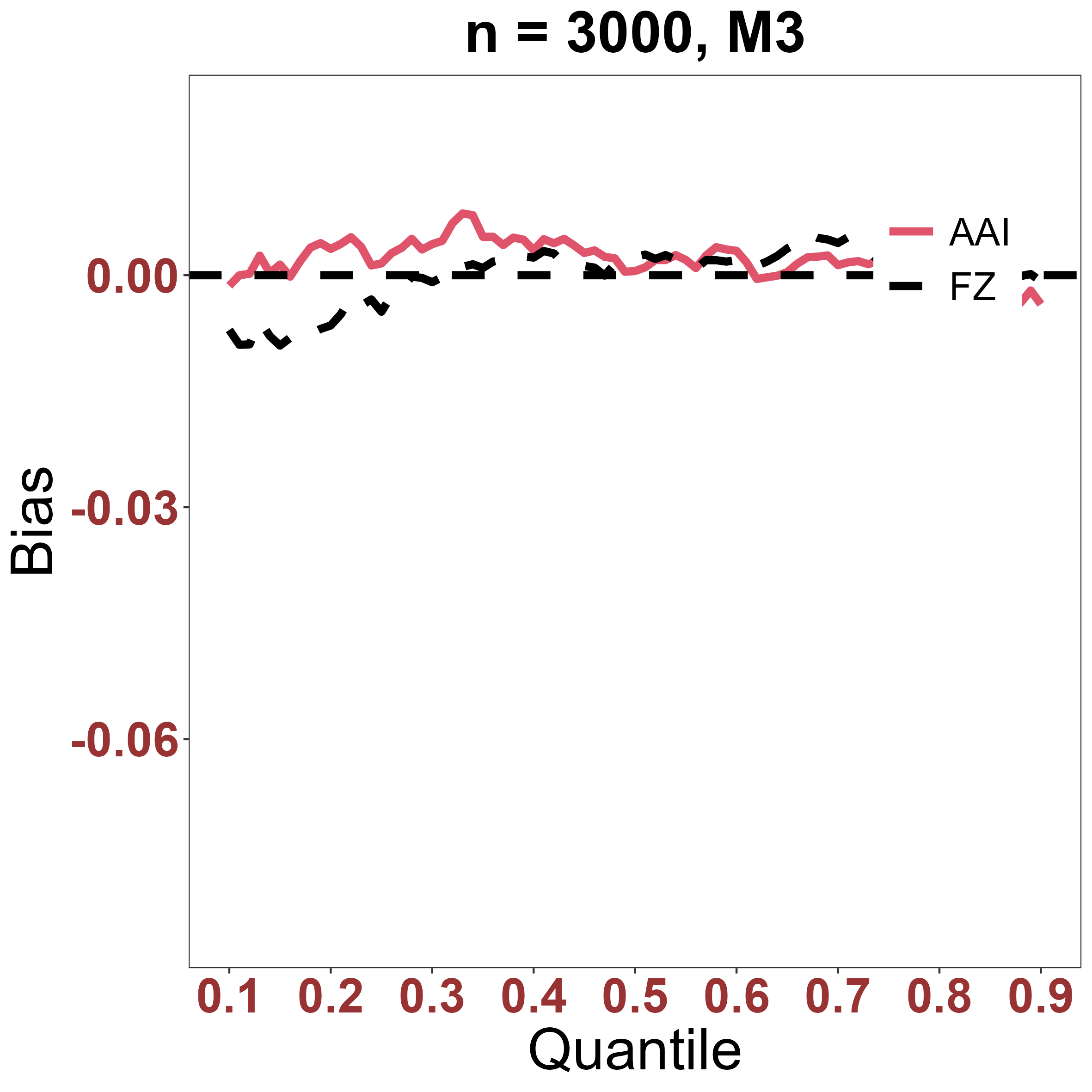}	
		\includegraphics[width = 3.9cm, height = 3.2cm]{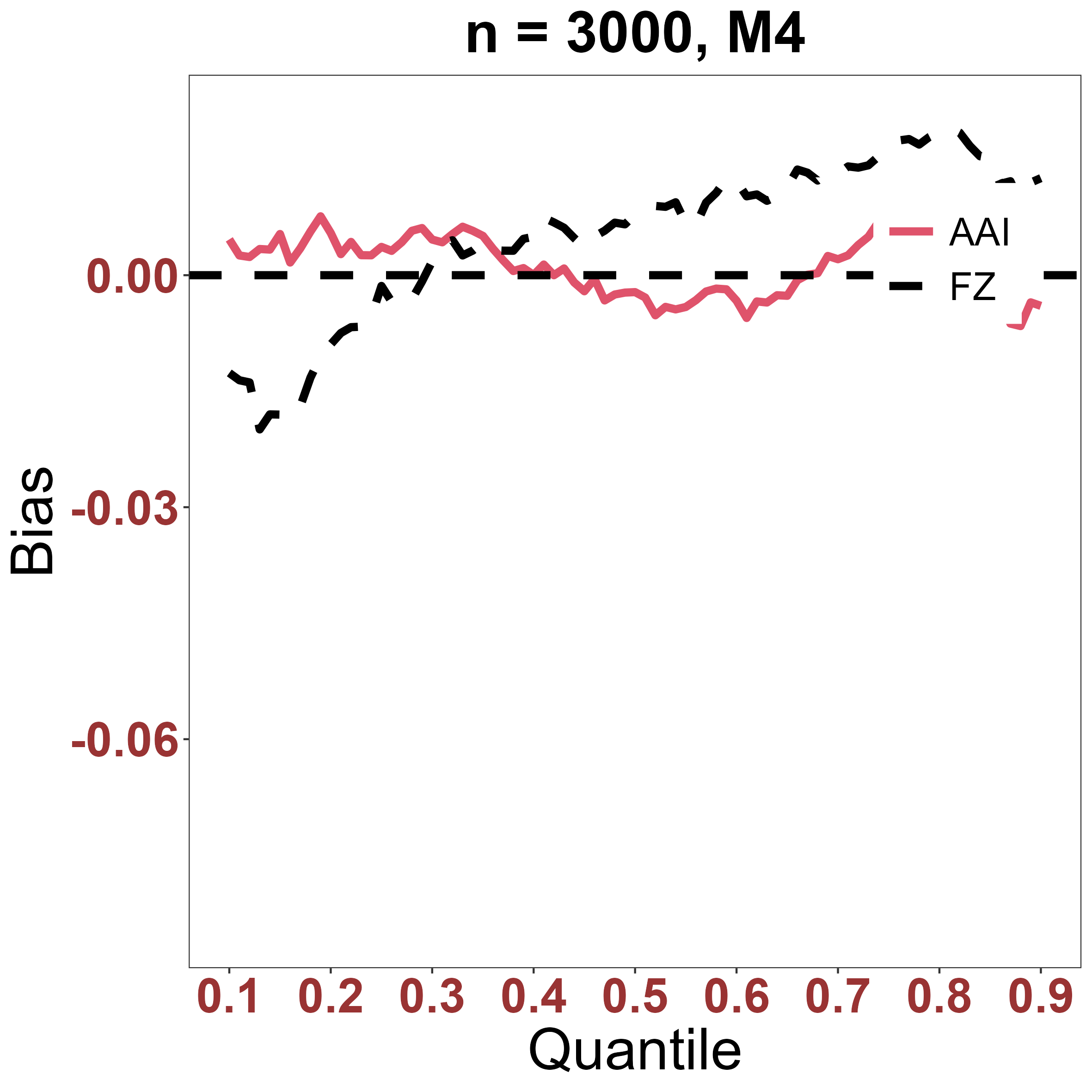}
	}	
	\mbox{
		\includegraphics[width = 3.9cm, height = 3.2cm]{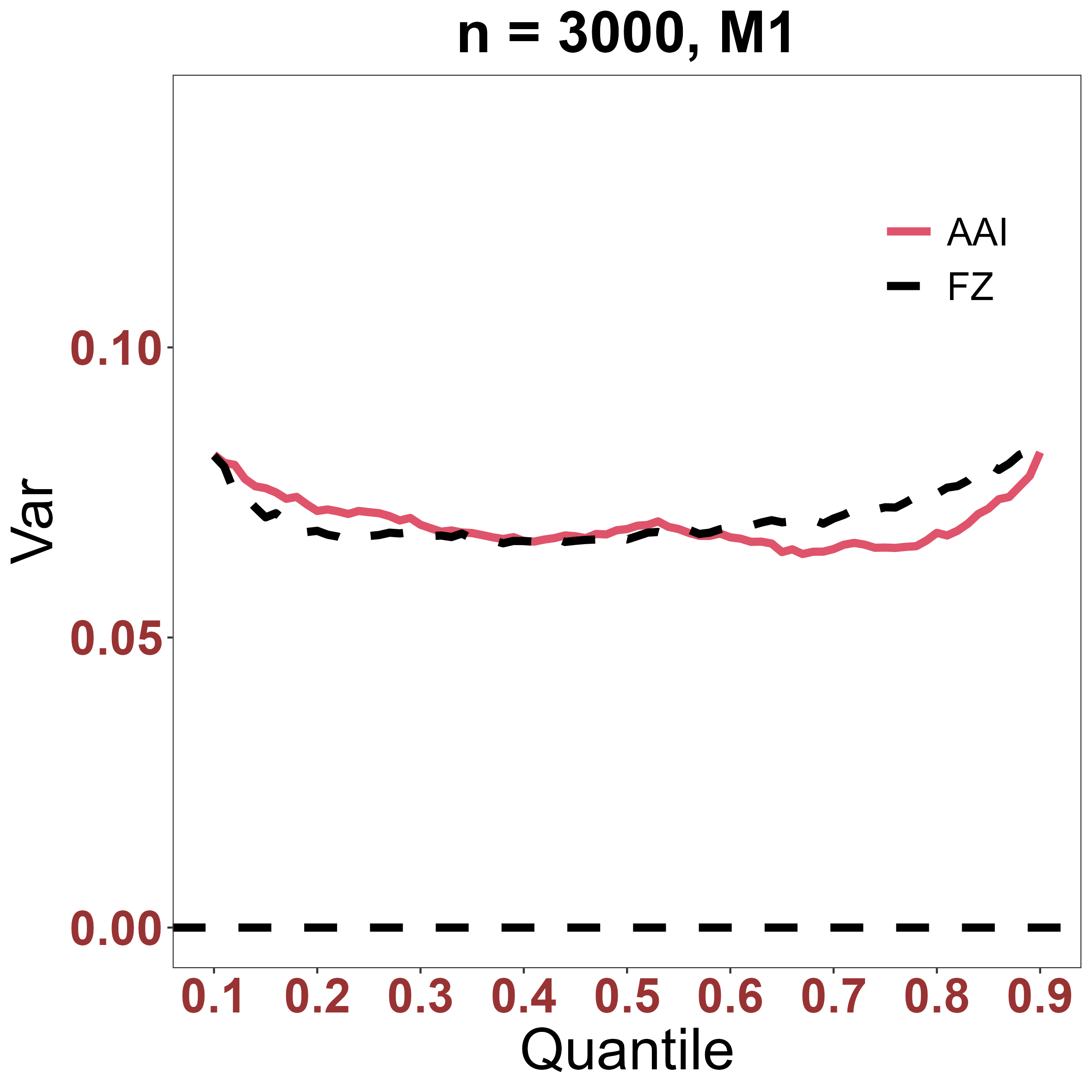}
		\includegraphics[width = 3.9cm, height = 3.2cm]{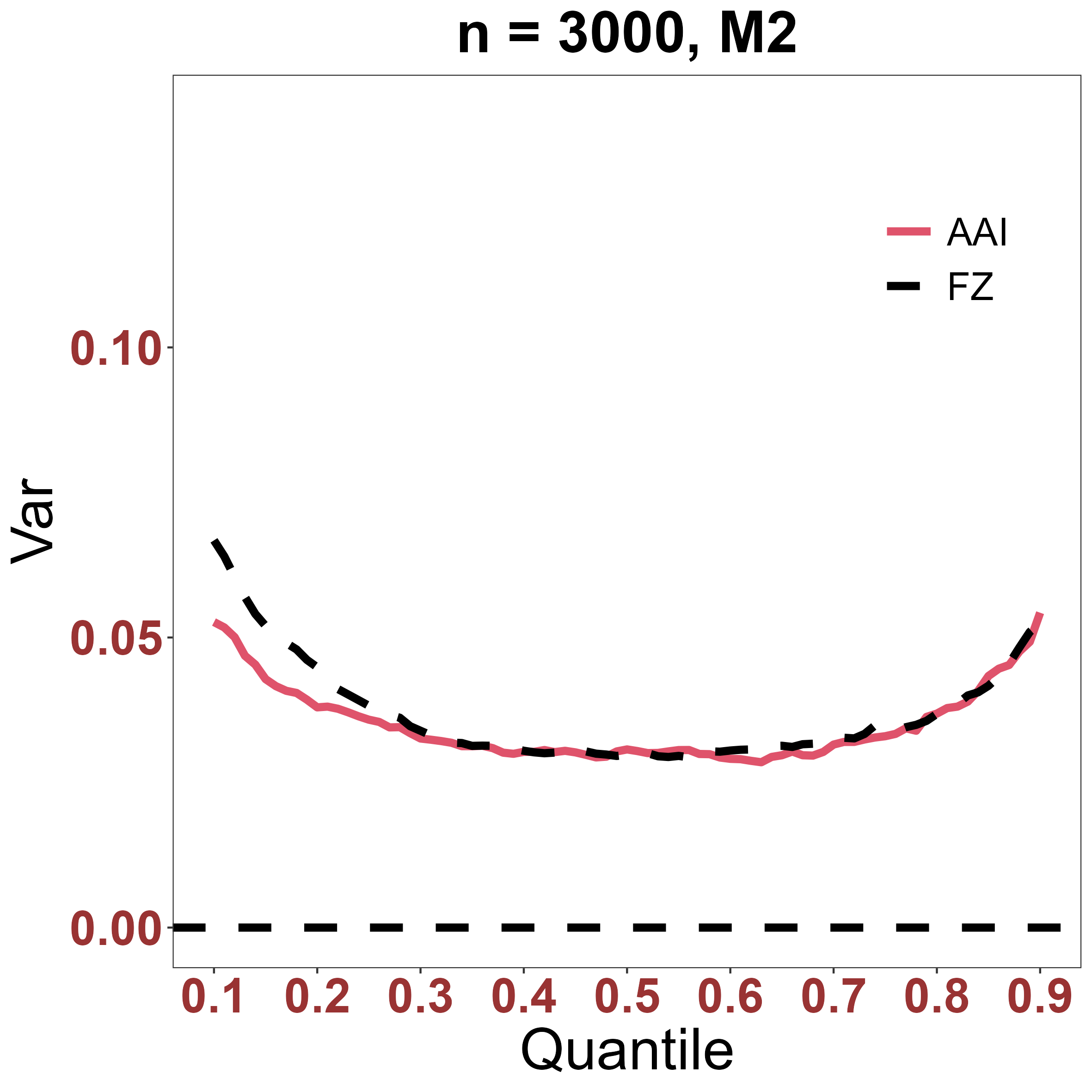}
		\includegraphics[width = 3.9cm, height = 3.2cm]{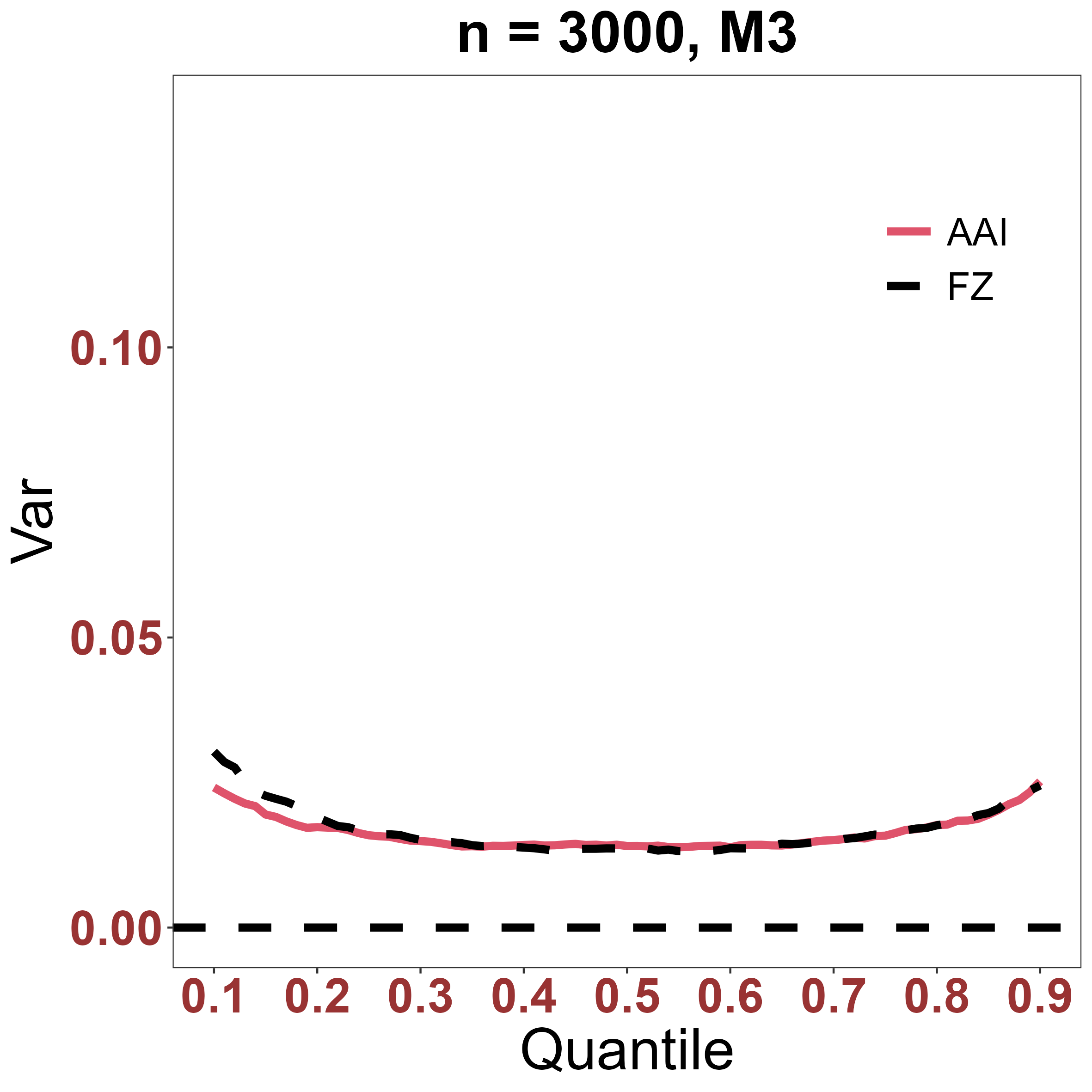}	
		\includegraphics[width = 3.9cm, height = 3.2cm]{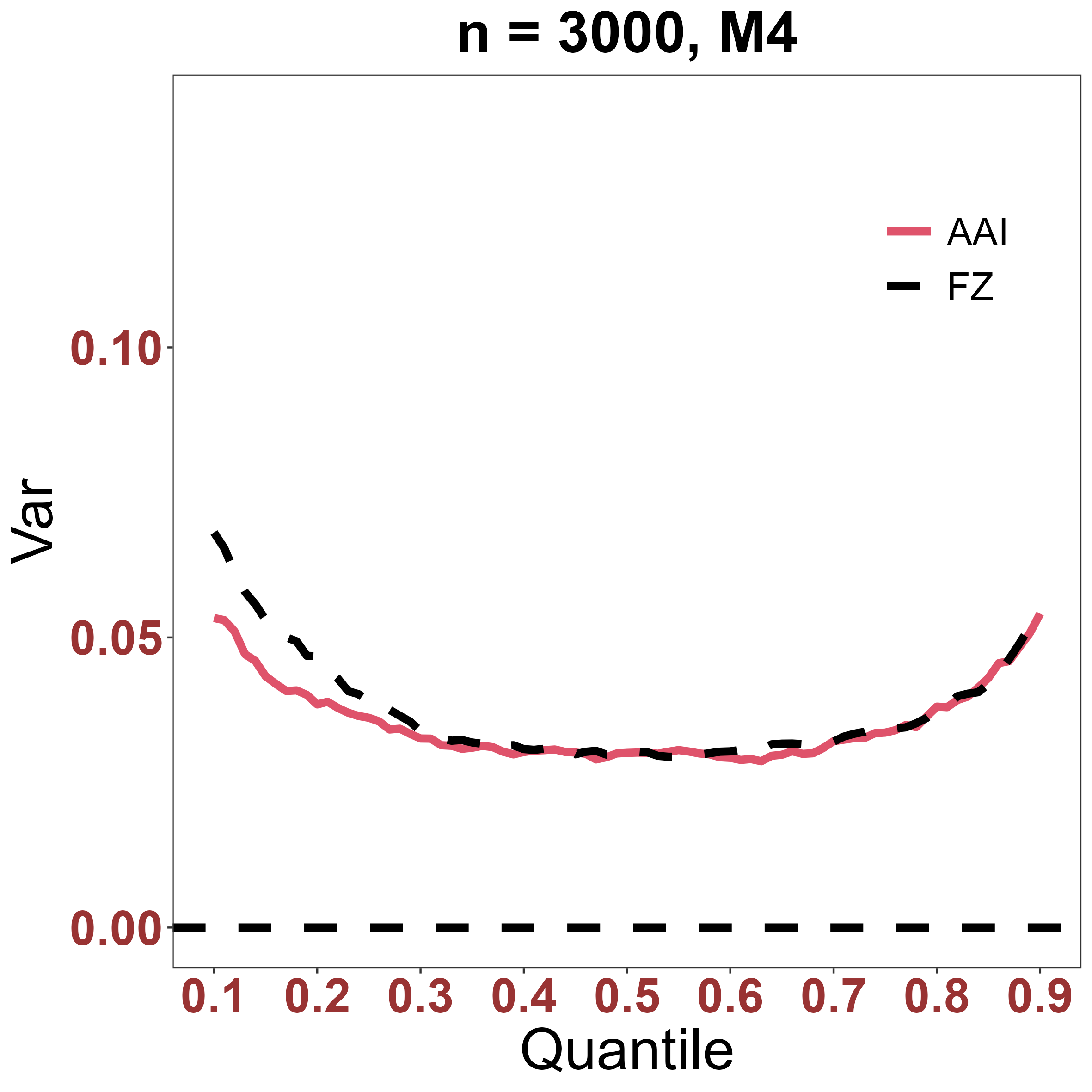}
	}	
	\mbox{
		\includegraphics[width = 3.9cm, height = 3.2cm]{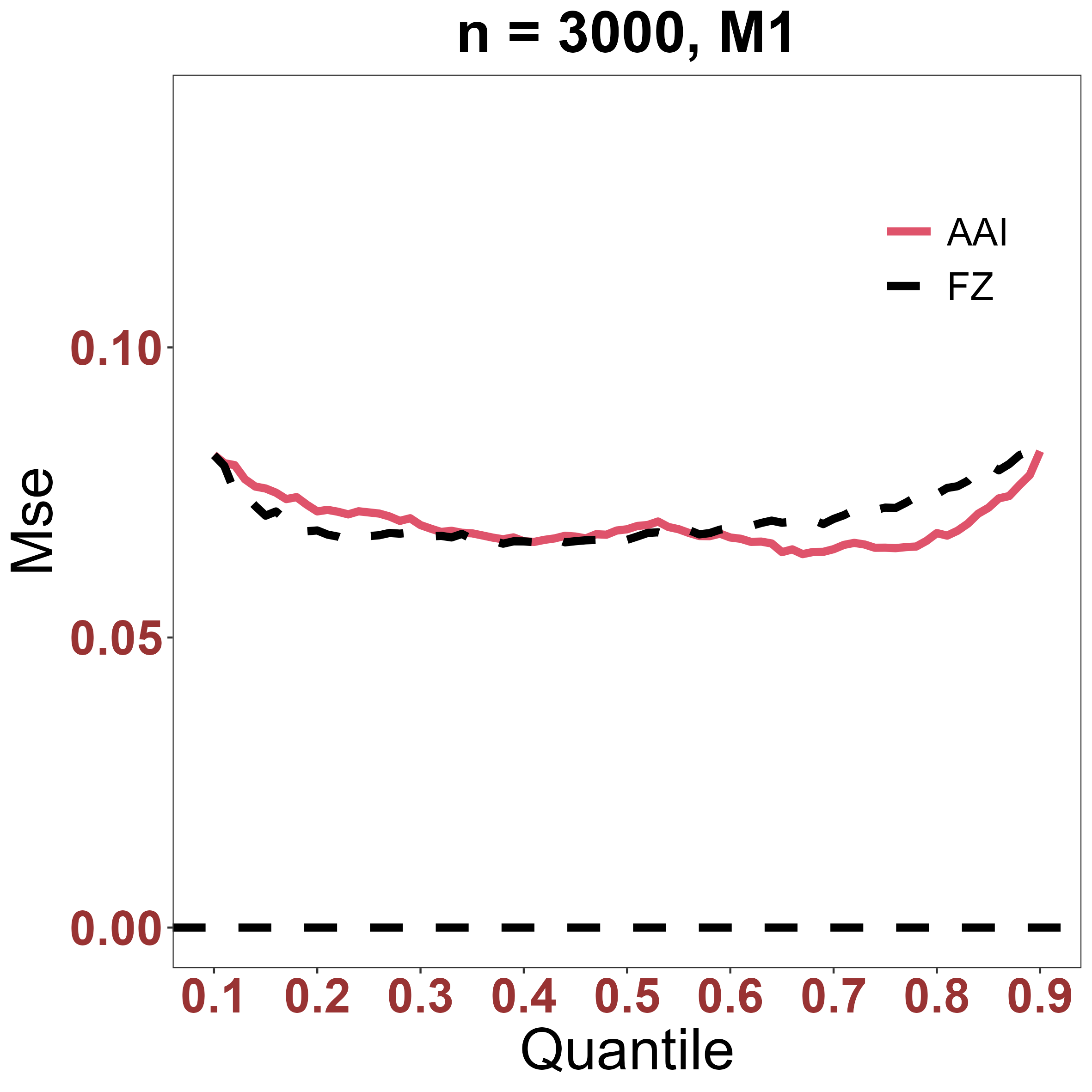}
		\includegraphics[width = 3.9cm, height = 3.2cm]{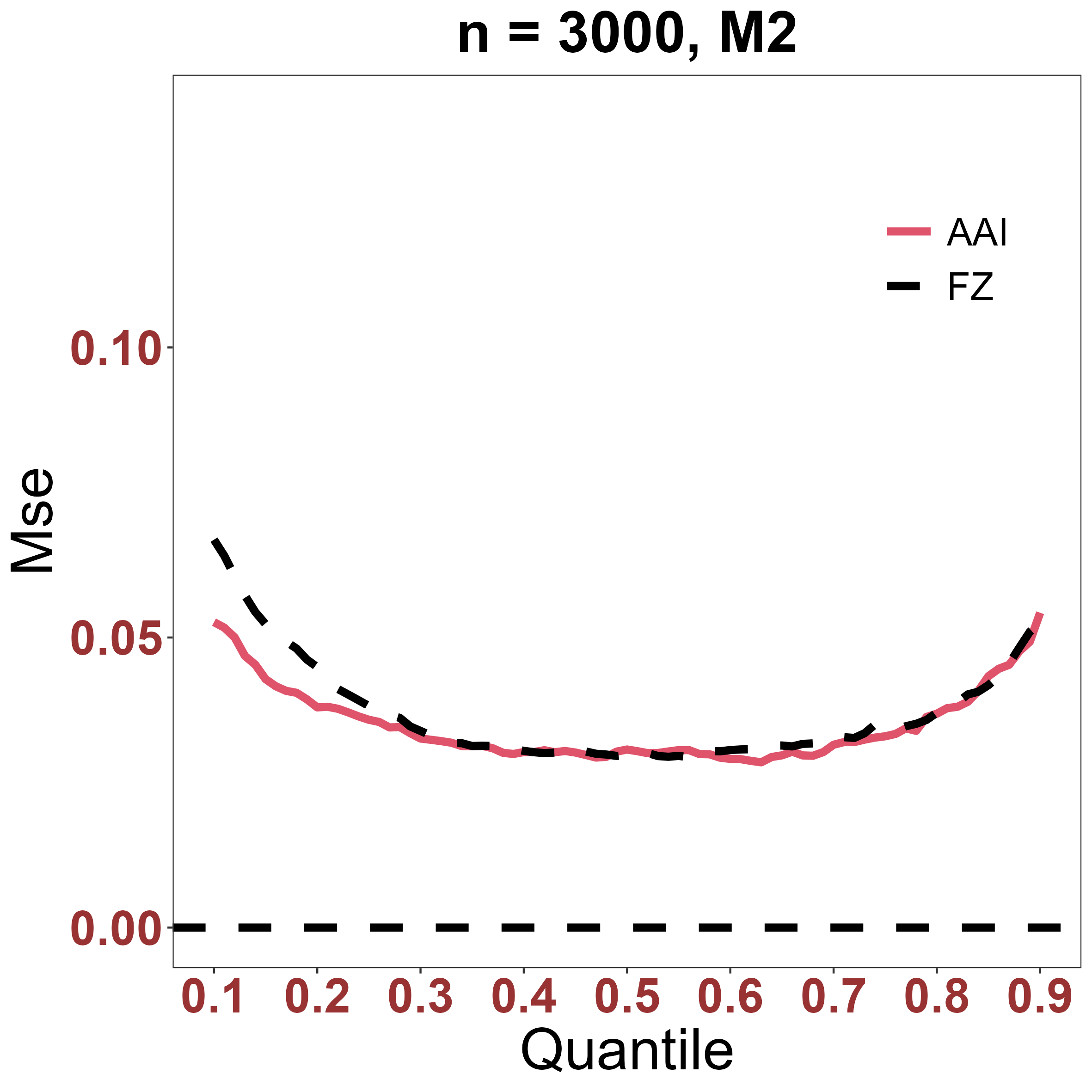}
		\includegraphics[width = 3.9cm, height = 3.2cm]{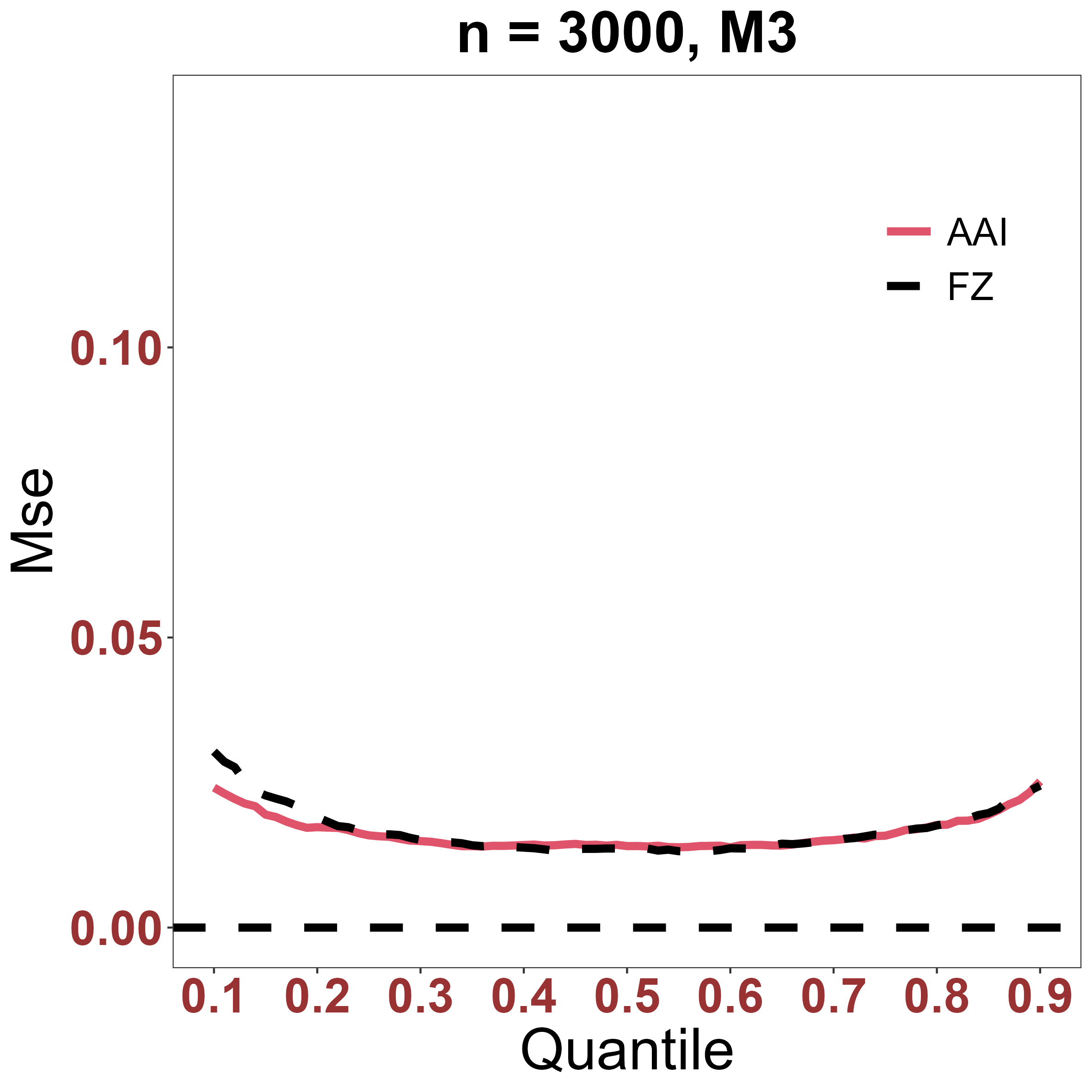}	
		\includegraphics[width = 3.9cm, height = 3.2cm]{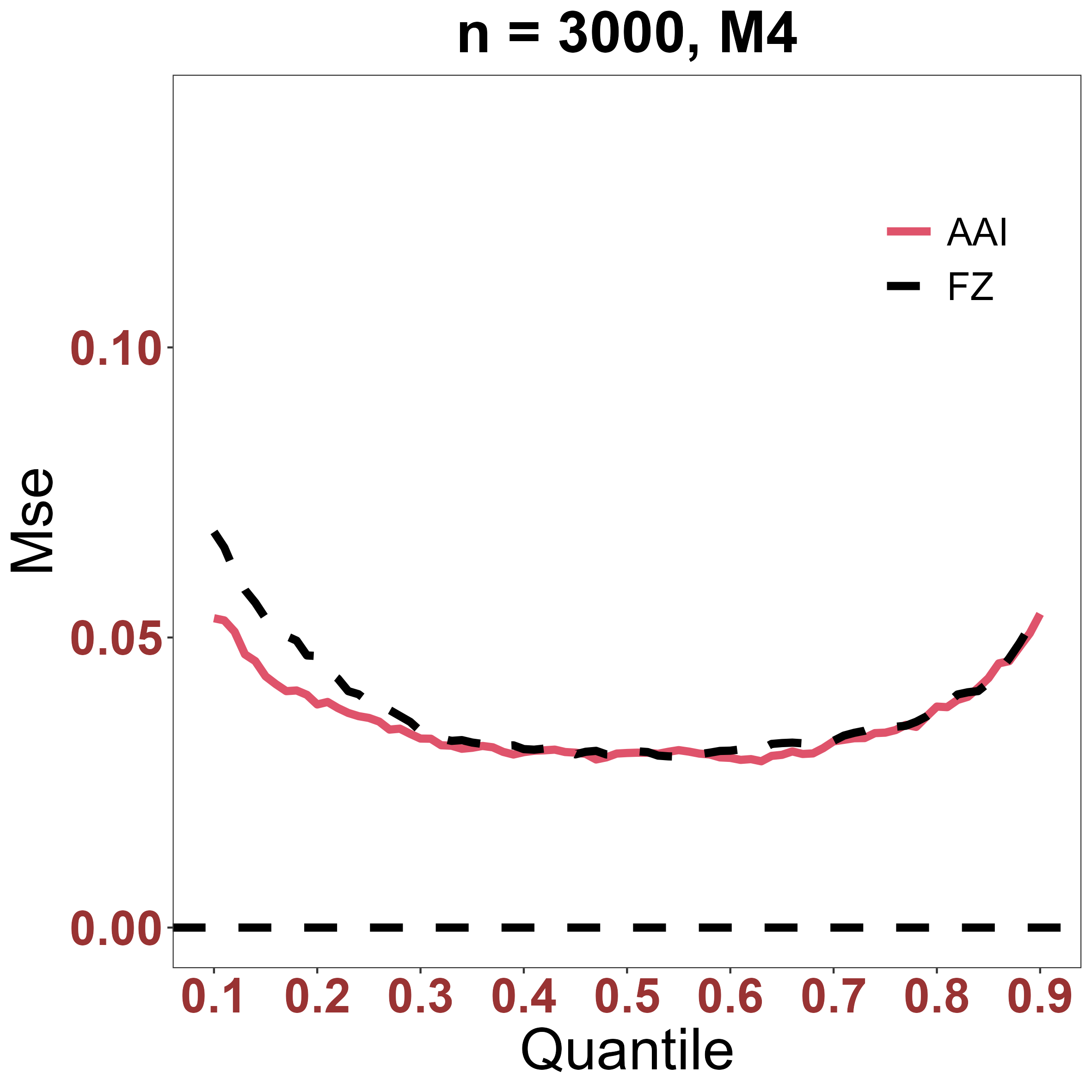}
	}	
	%\end{subfigure}
	\caption{Bias, variance and MSE of the QTE estimator using the weighted quantile regression (WQR) of \cite{AAI_2002} (AAI) and that using the FZ loss (FZ) when $\rho=0$ and $n=3,000$.}
	\label{figure12}
\end{figure}

\begin{figure}[!htb]
	%\begin{subfigure}
	\centering
	\mbox{
		\includegraphics[width = 3.9cm, height = 3.2cm]{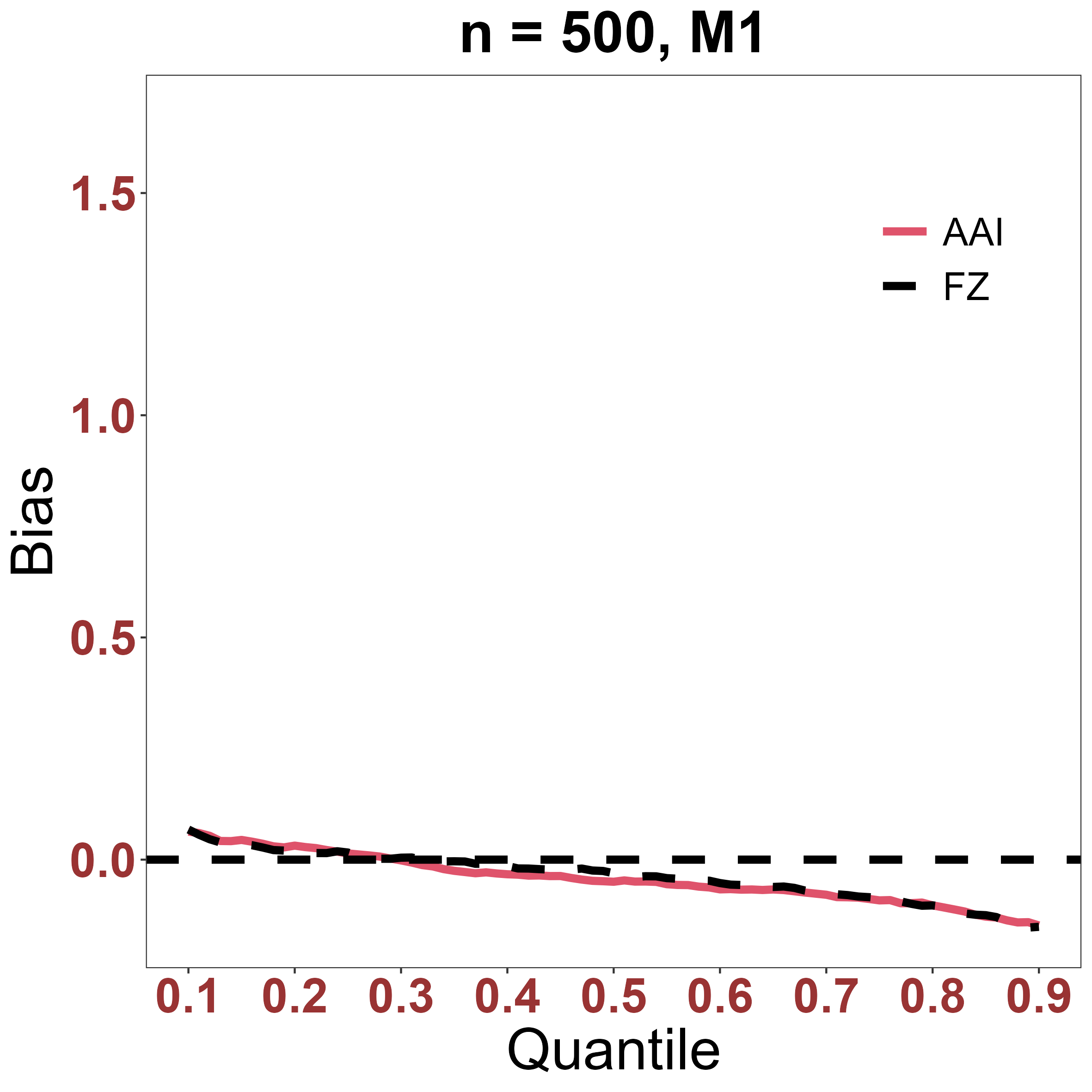}
		\includegraphics[width = 3.9cm, height = 3.2cm]{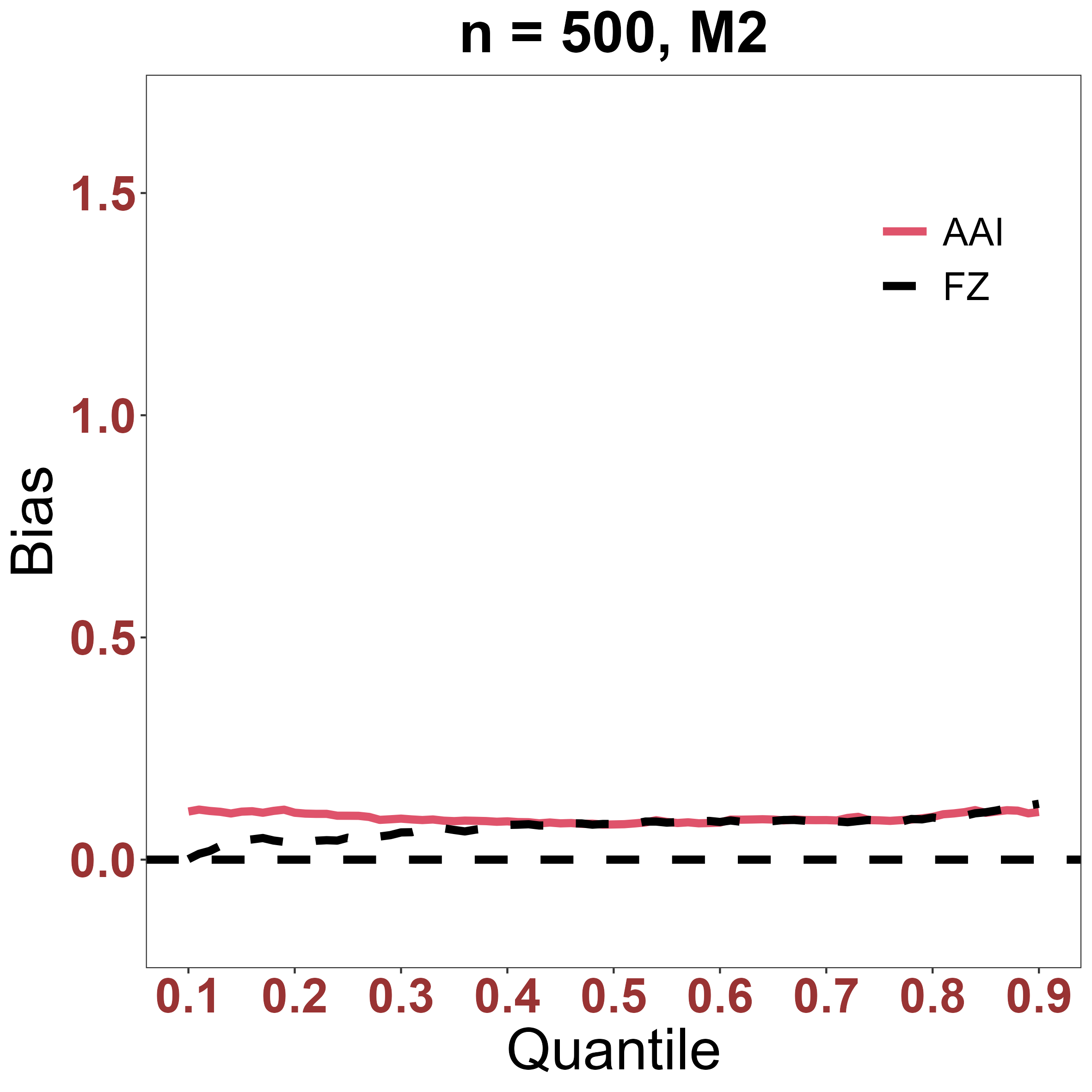}
		\includegraphics[width = 3.9cm, height = 3.2cm]{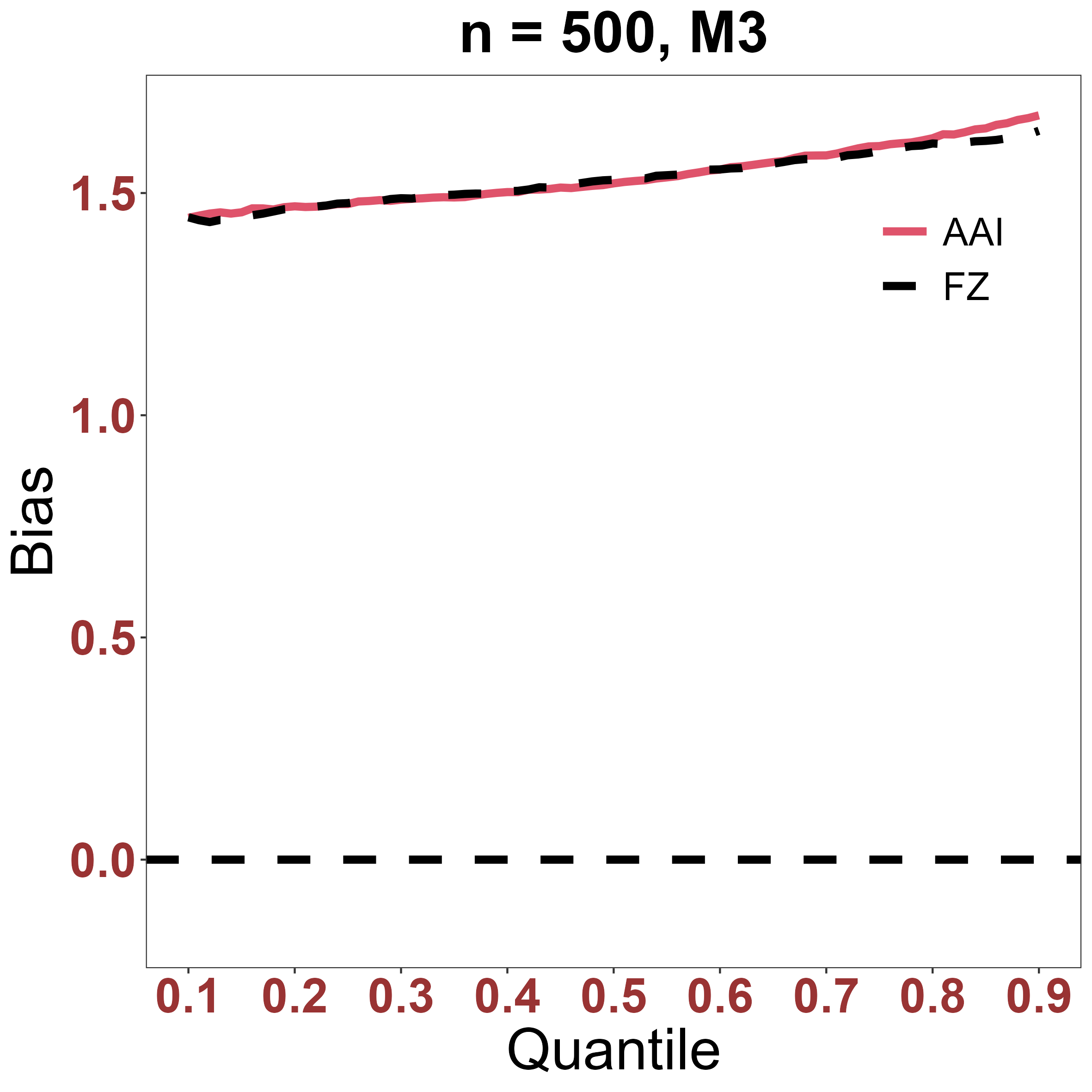}	
		\includegraphics[width = 3.9cm, height = 3.2cm]{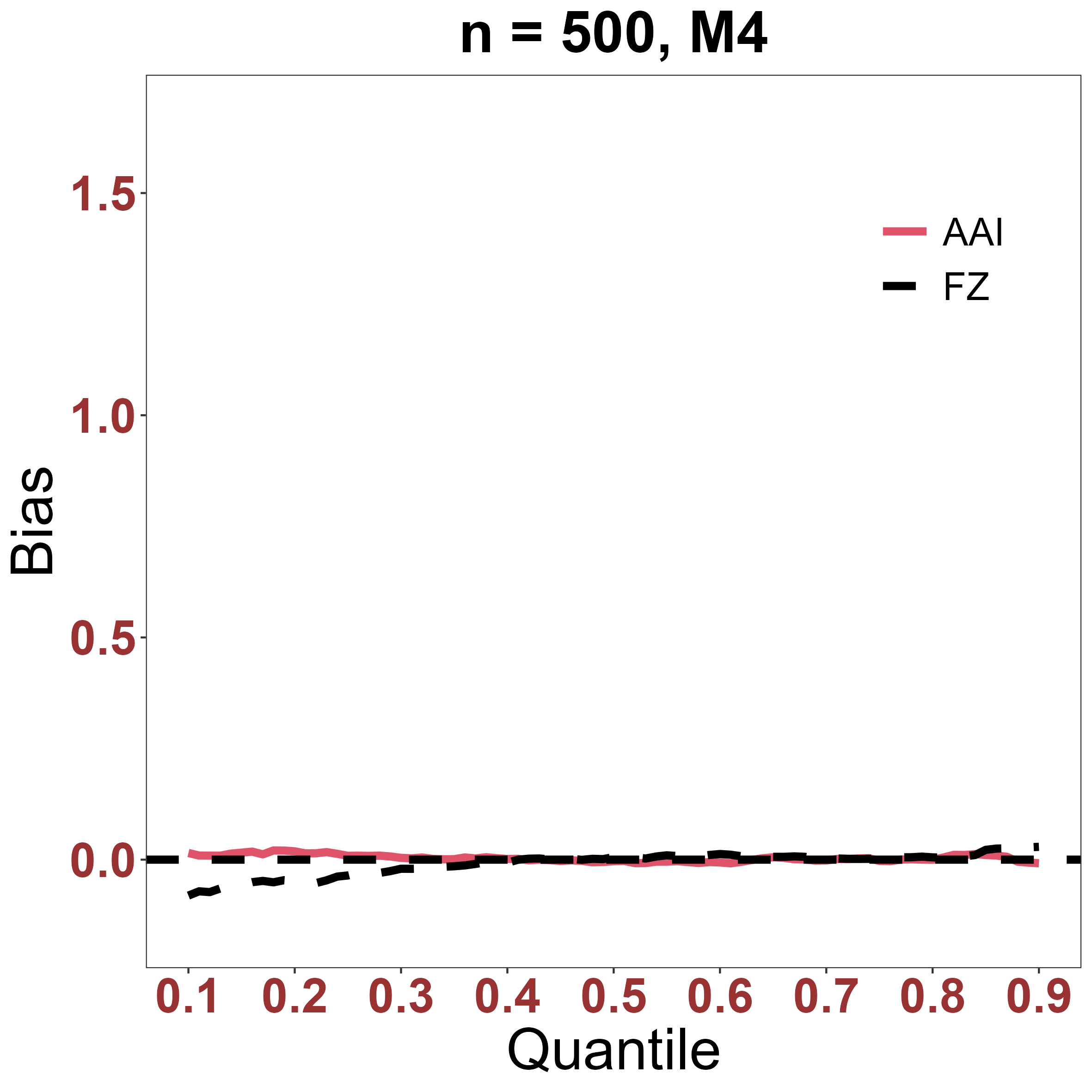}
	}	
	\mbox{	\includegraphics[width = 3.9cm, height = 3.2cm]{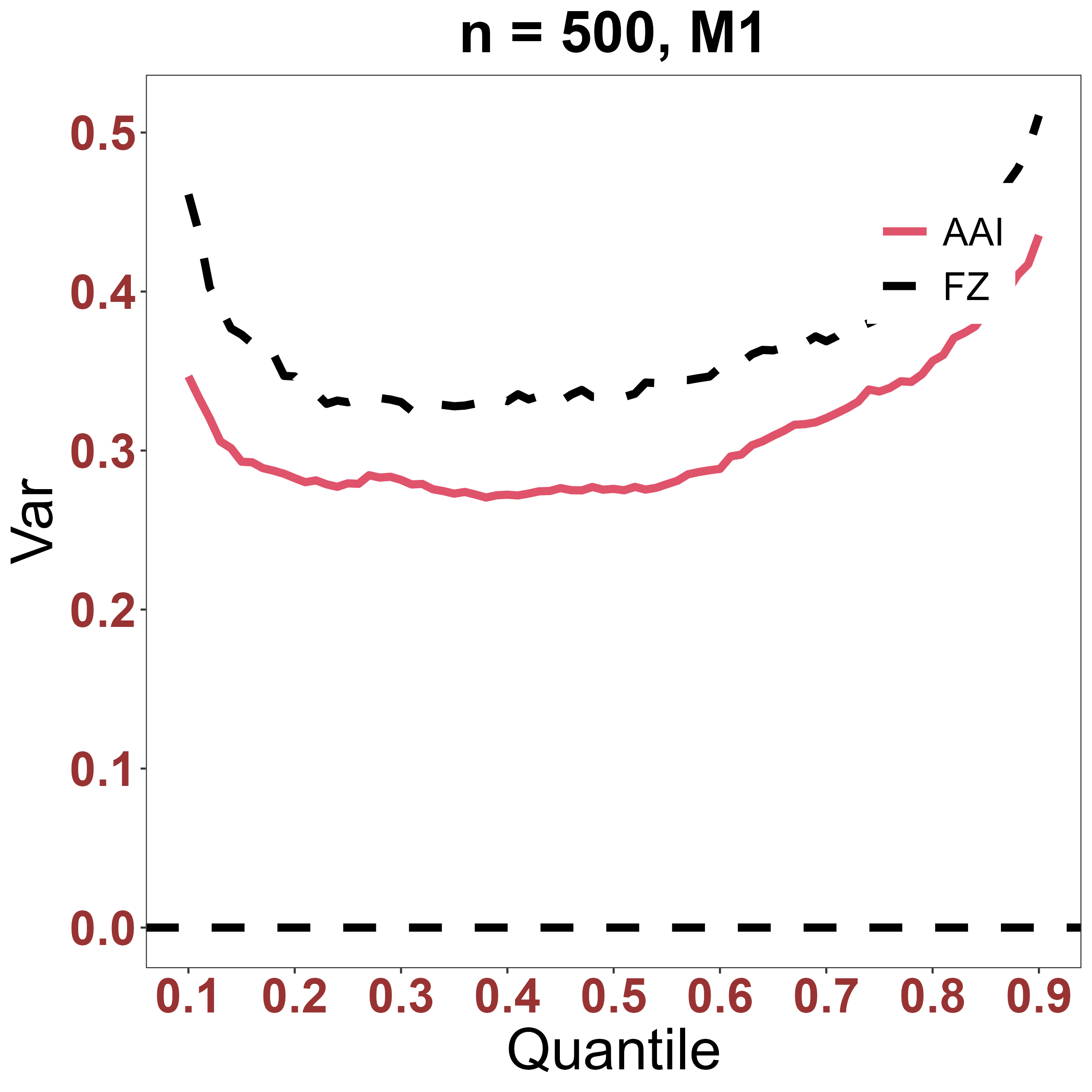}
		\includegraphics[width = 3.9cm, height = 3.2cm]{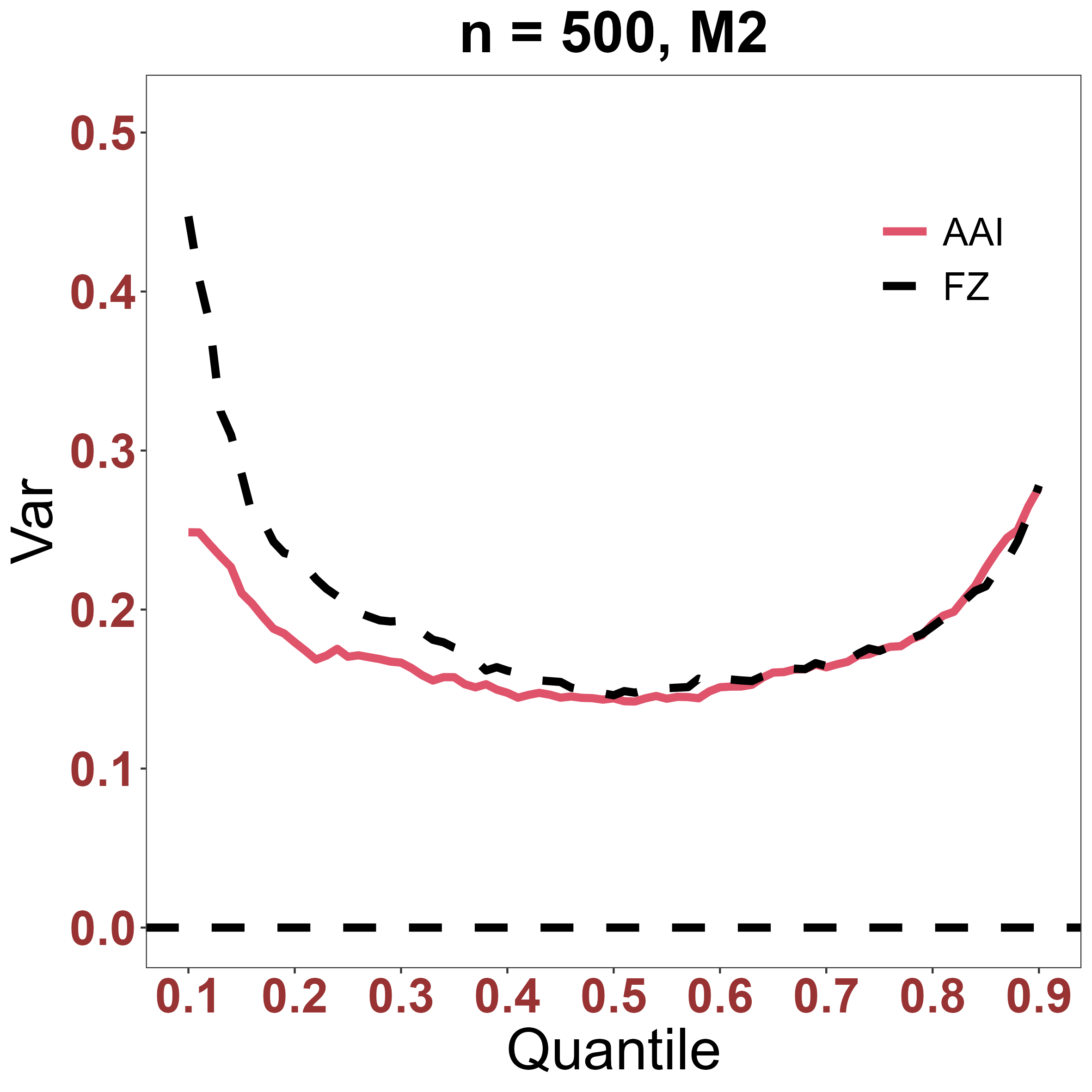}
		\includegraphics[width = 3.9cm, height = 3.2cm]{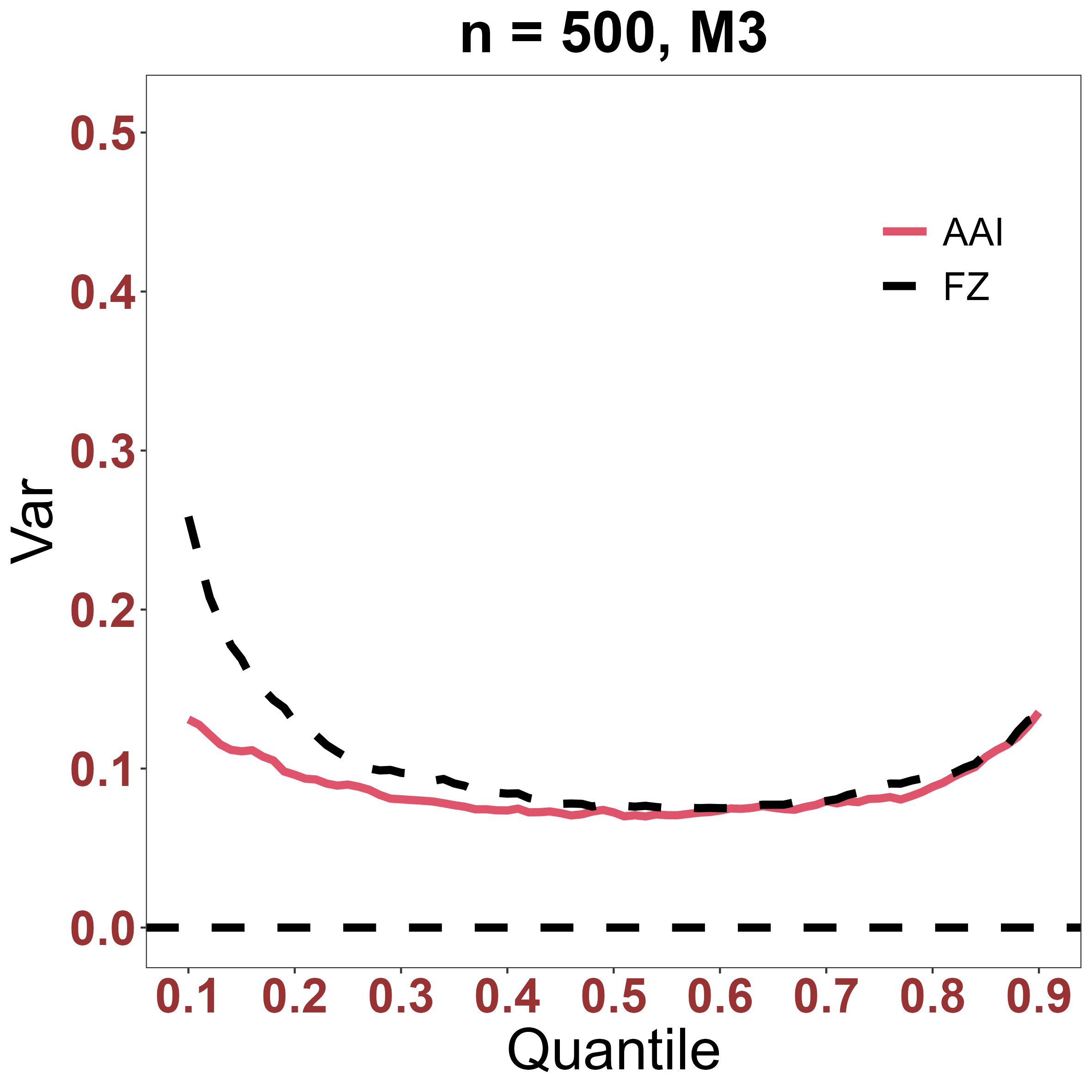}	
		\includegraphics[width = 3.9cm, height = 3.2cm]{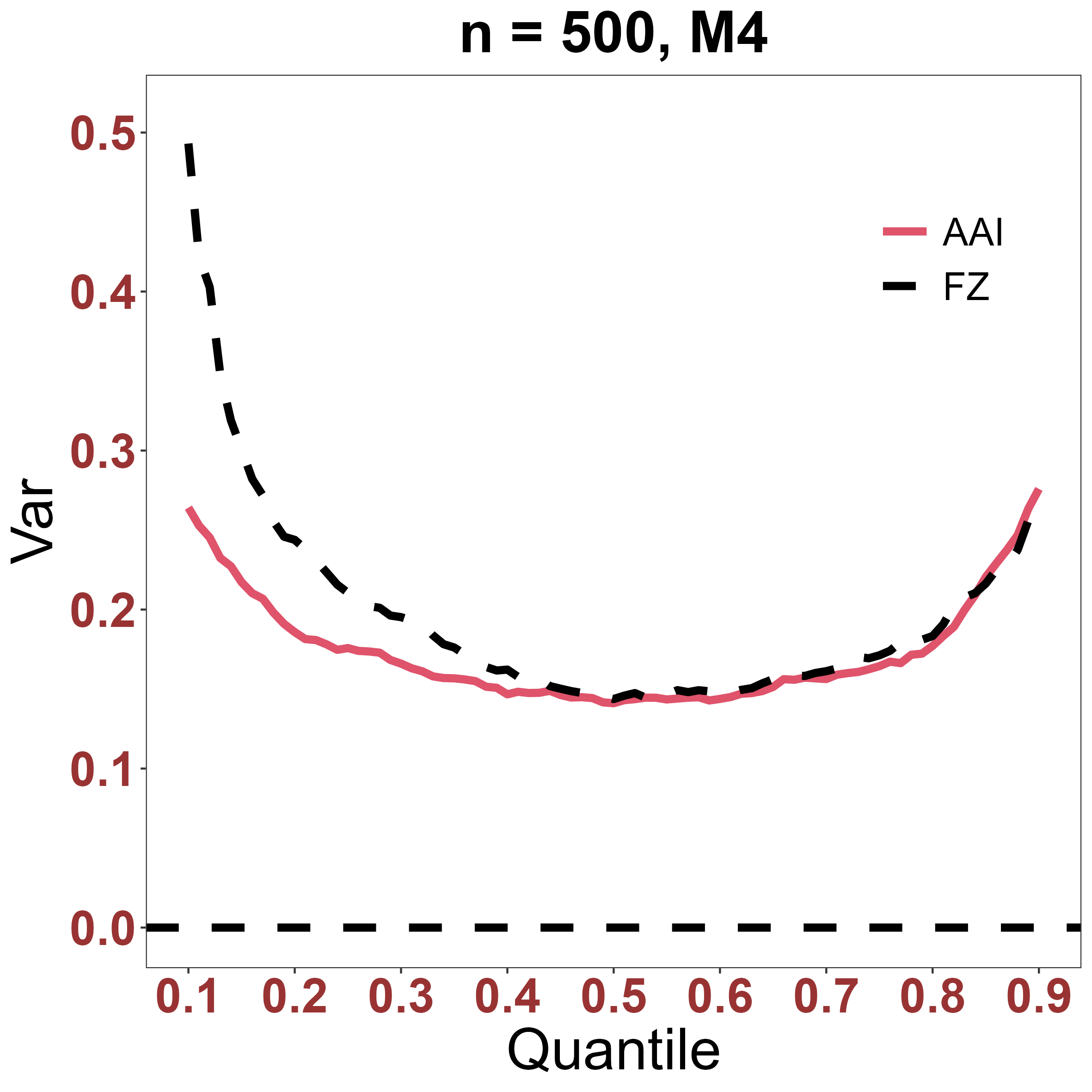}
	}	
	\mbox{	\includegraphics[width = 3.9cm, height = 3.2cm]{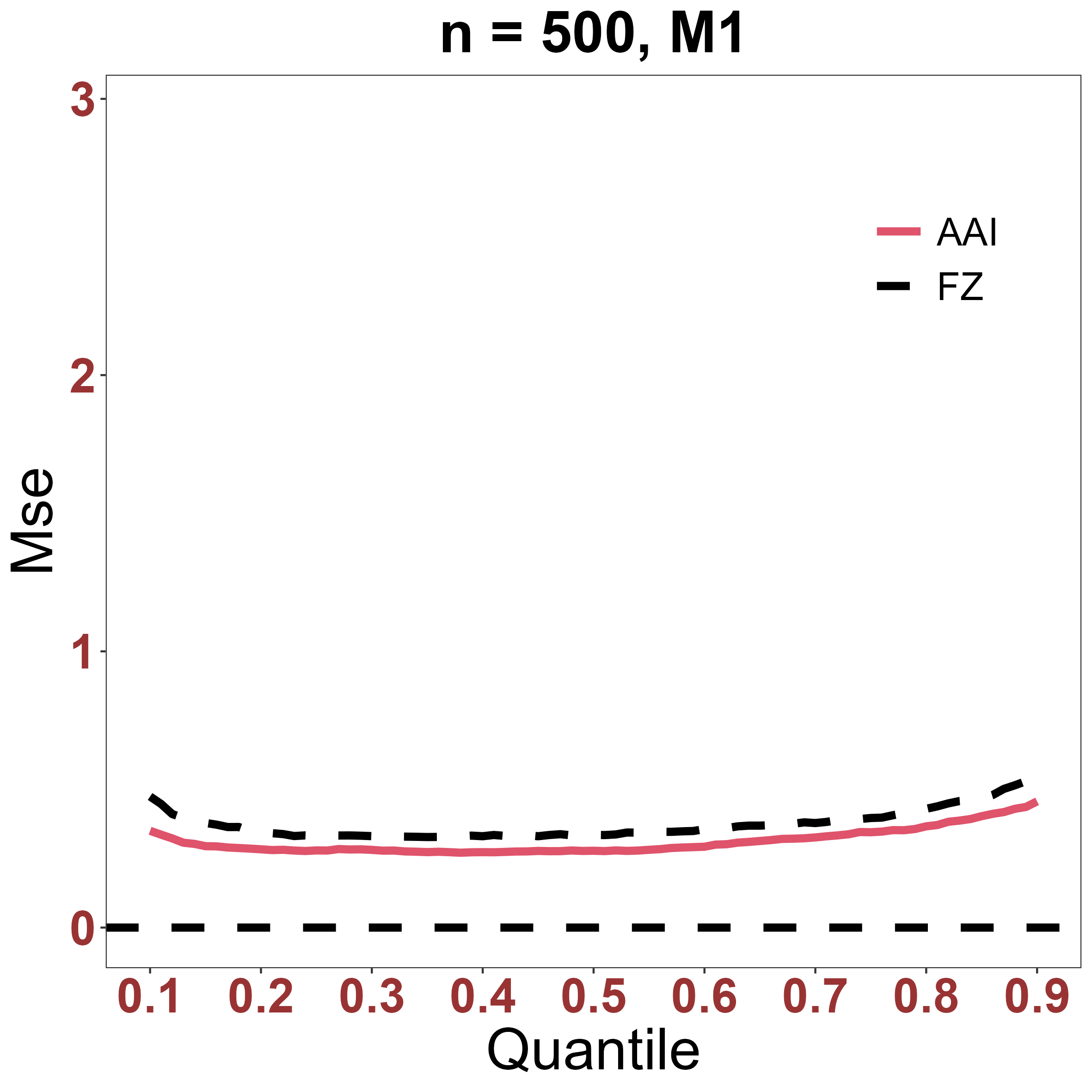}
		\includegraphics[width = 3.9cm, height = 3.2cm]{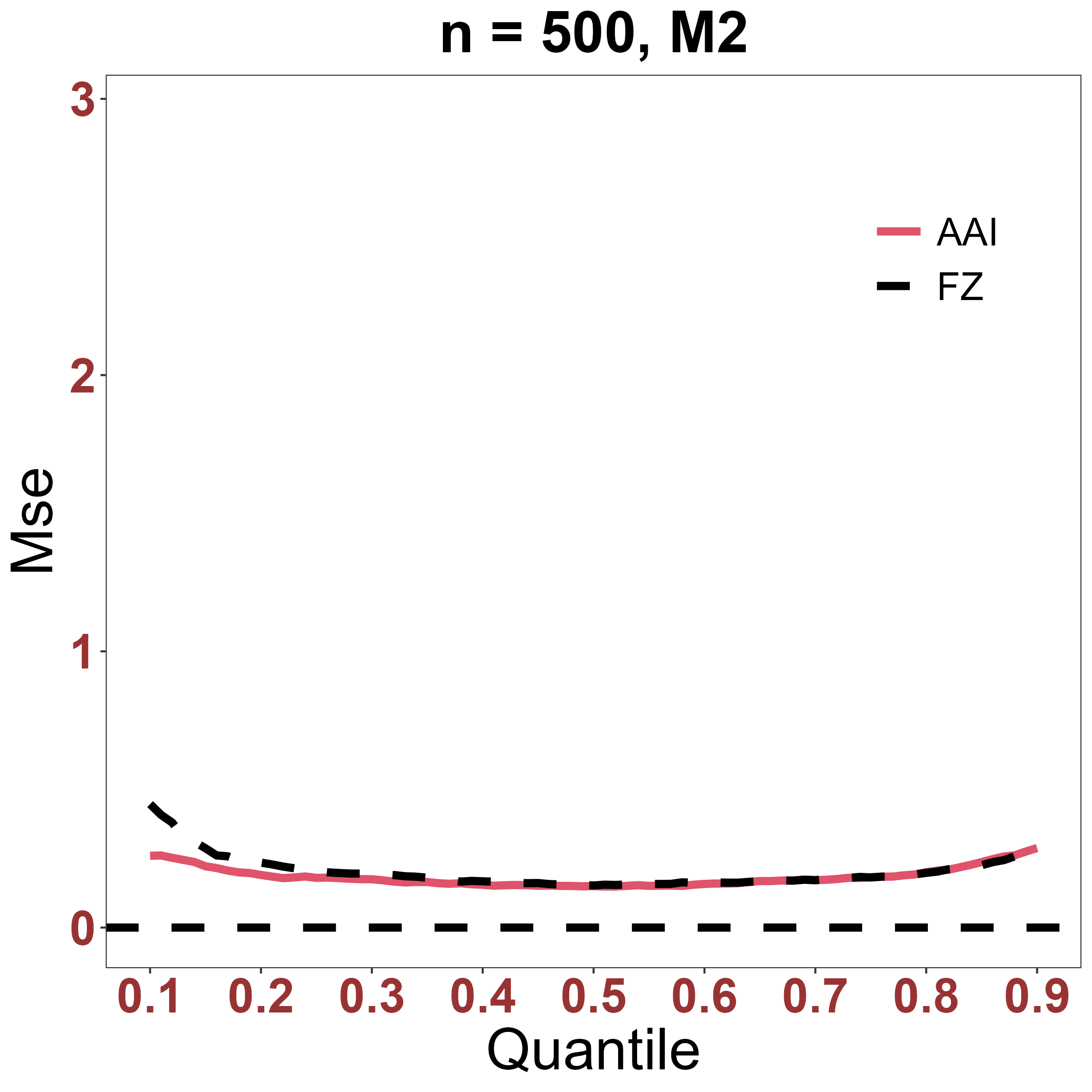}
		\includegraphics[width = 3.9cm, height = 3.2cm]{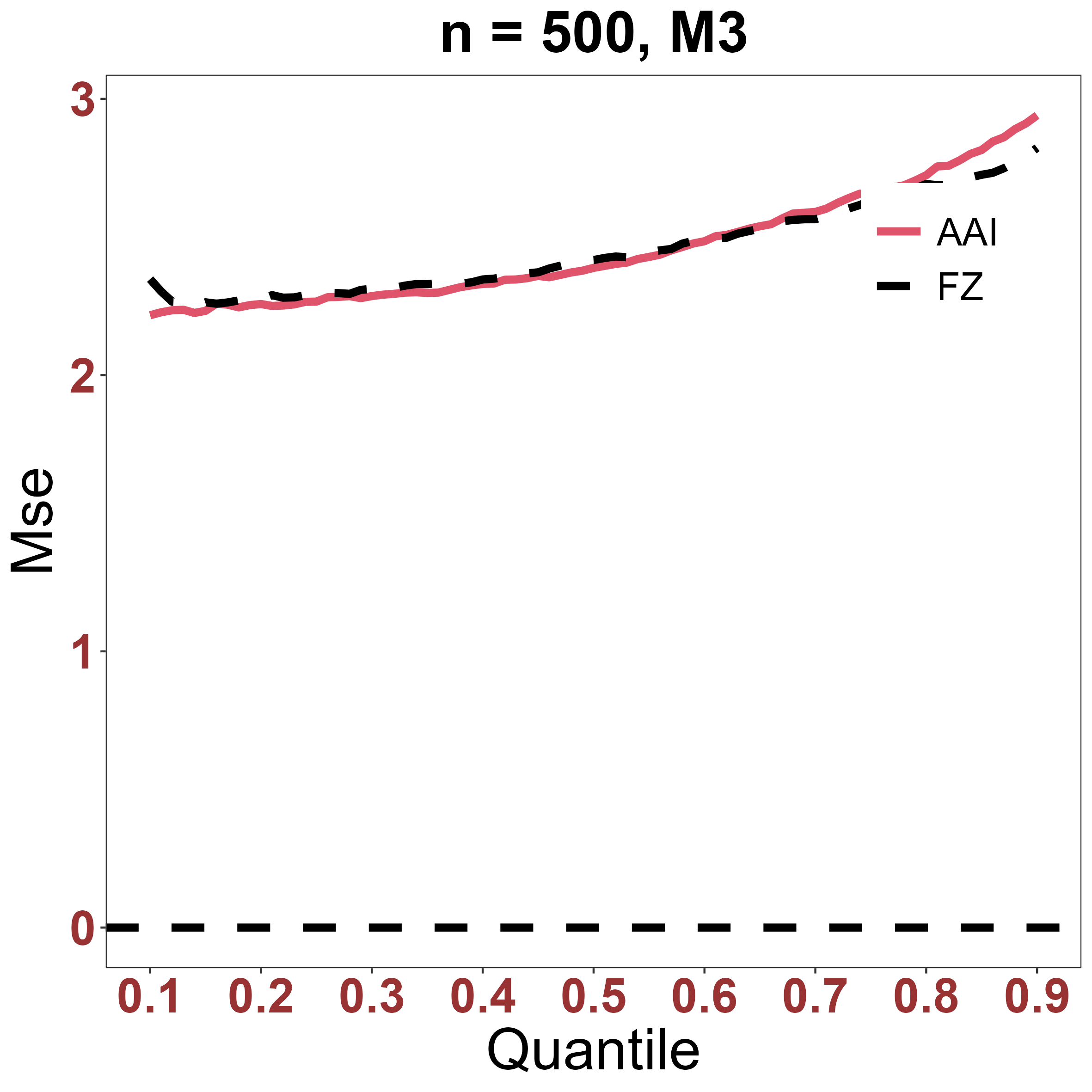}	
		\includegraphics[width = 3.9cm, height = 3.2cm]{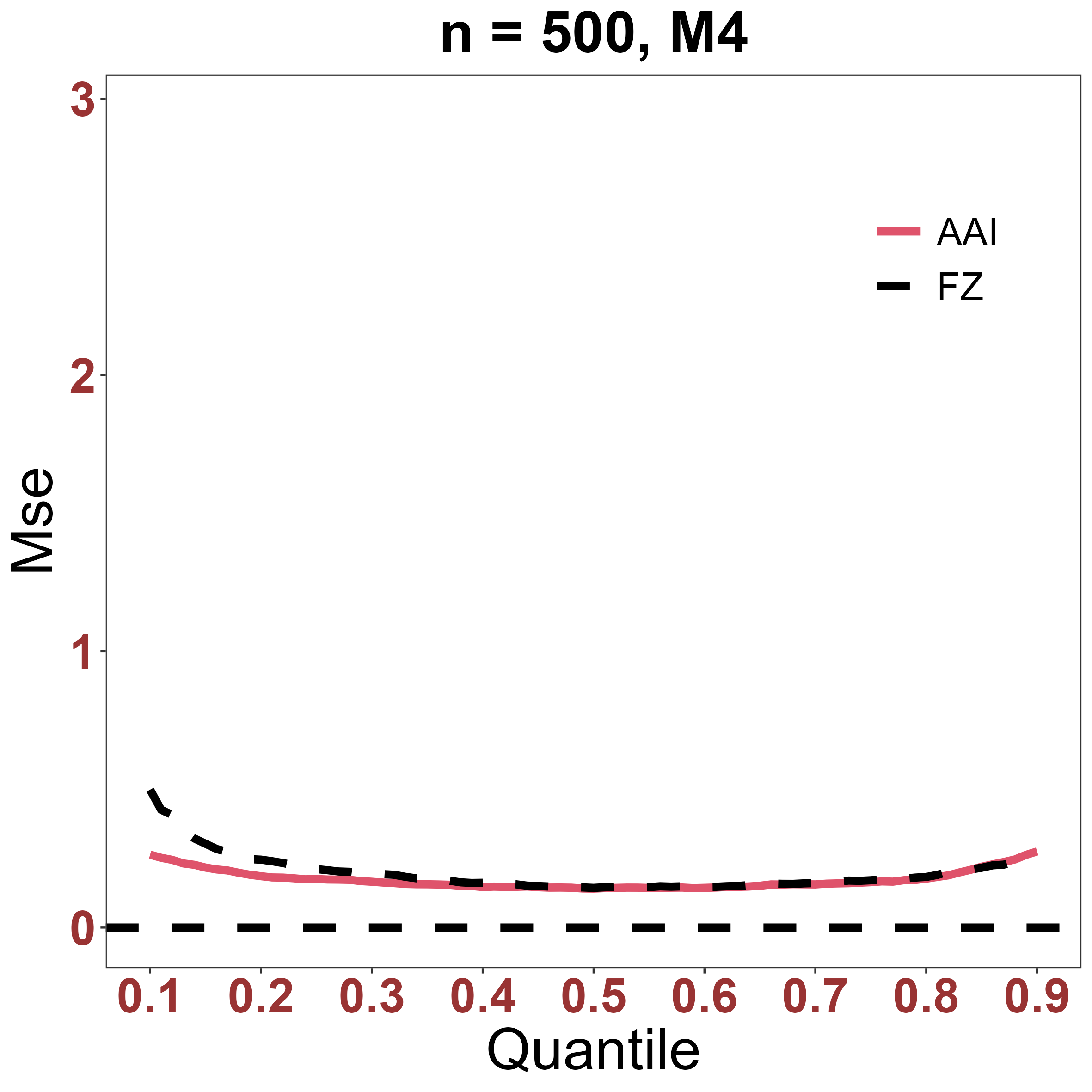}	
	}
	%\end{subfigure}
	\caption{Bias, variance and MSE of the QTE estimator using the weighted quantile regression (WQR) of \cite{AAI_2002} (AAI) and that using the FZ loss (FZ) when $\rho=0.5$ and $n=500$.}
	\label{figure13}
\end{figure}

\begin{figure}[!htb]
	%\begin{subfigure}
	\centering
	\mbox{
		\includegraphics[width = 3.9cm, height = 3.2cm]{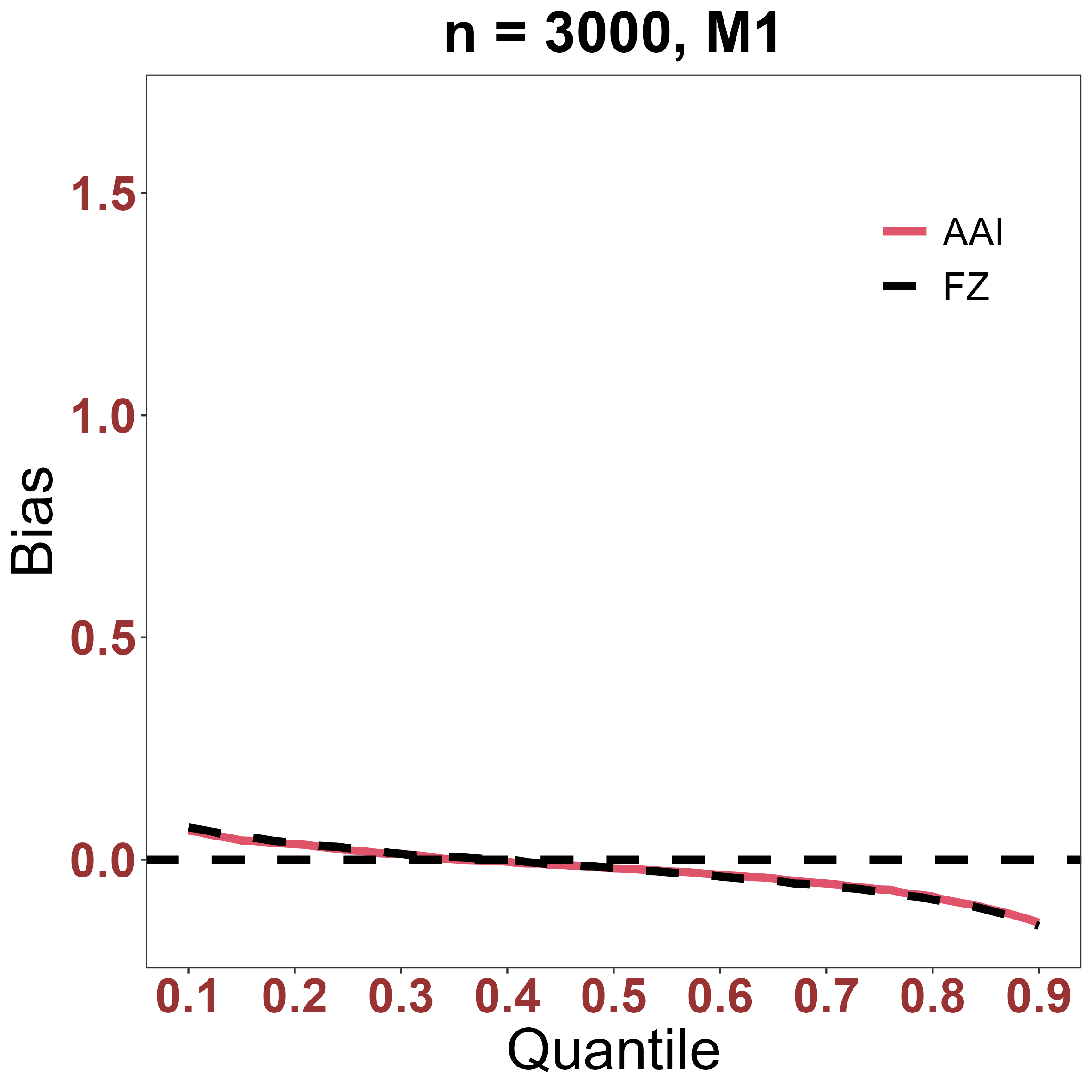}
		\includegraphics[width = 3.9cm, height = 3.2cm]{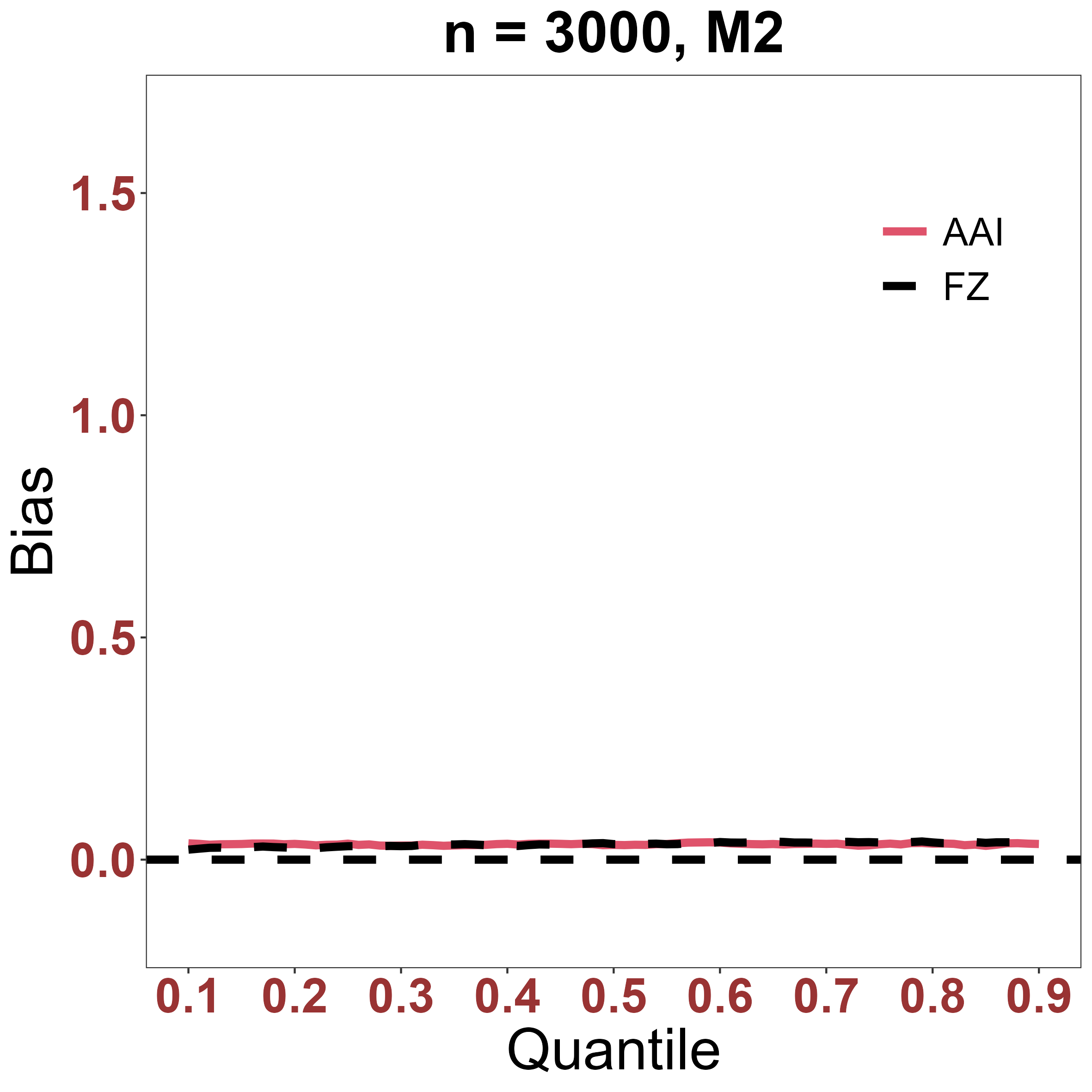}
		\includegraphics[width = 3.9cm, height = 3.2cm]{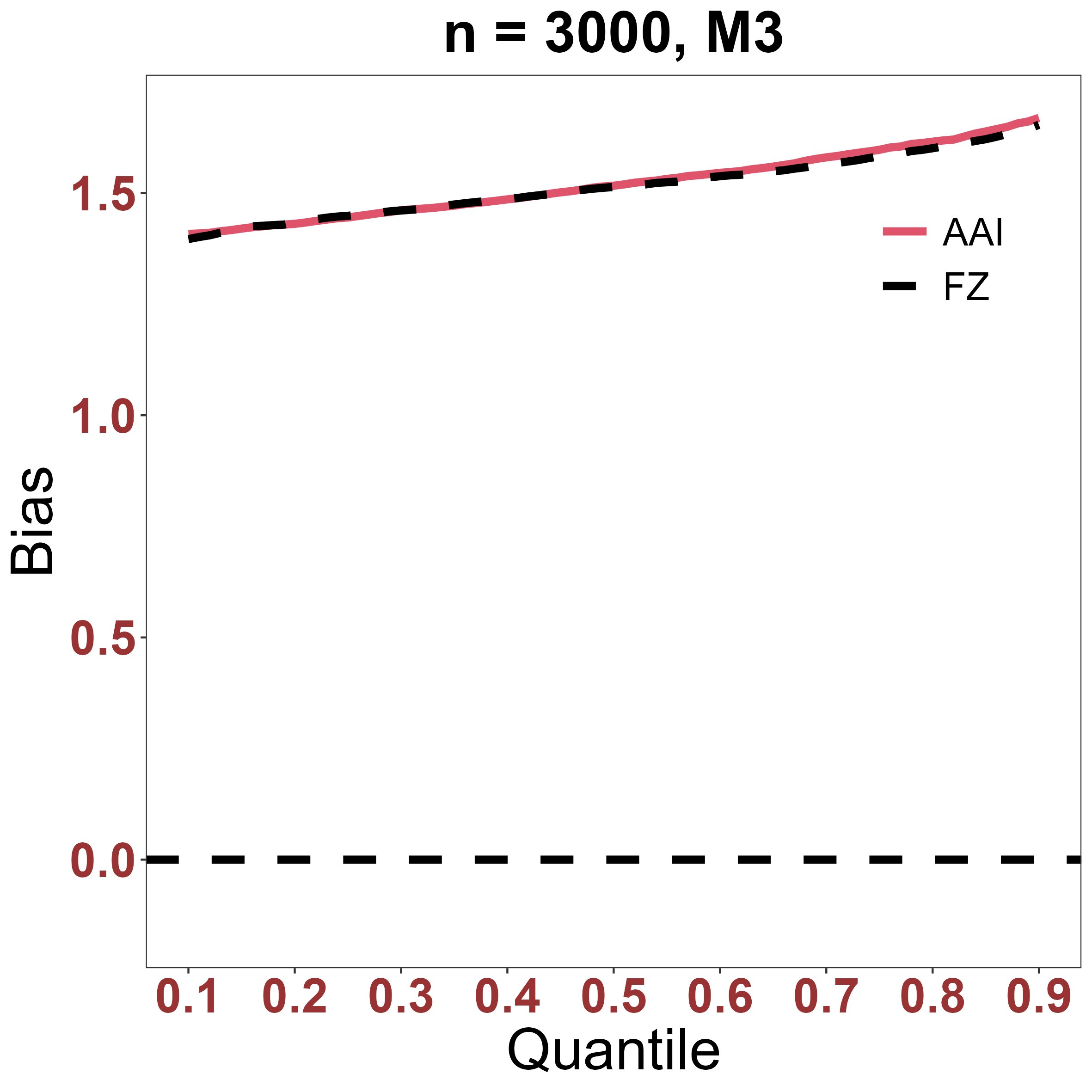}	
		\includegraphics[width = 3.9cm, height = 3.2cm]{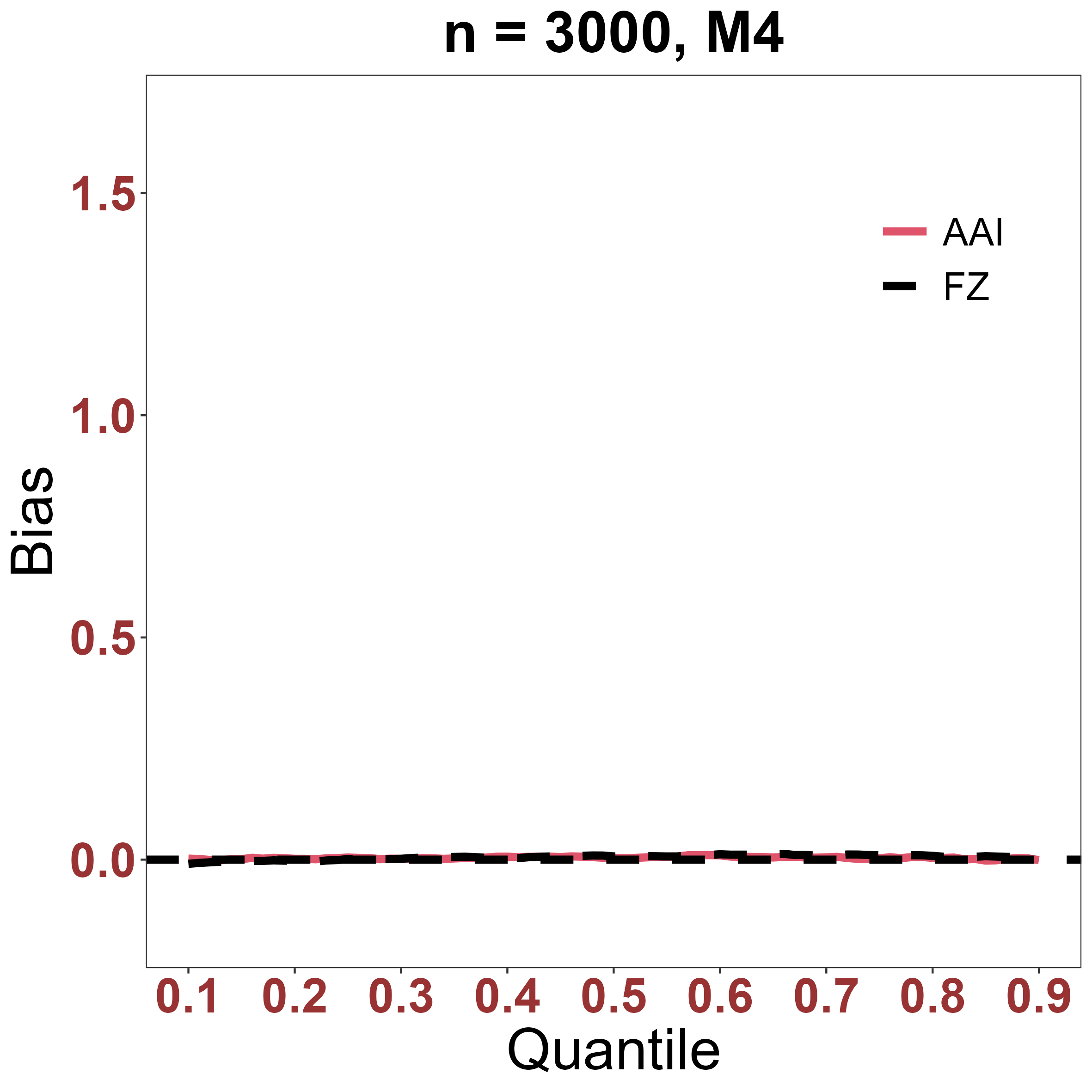}
	}	
	\mbox{
		\includegraphics[width = 3.9cm, height = 3.2cm]{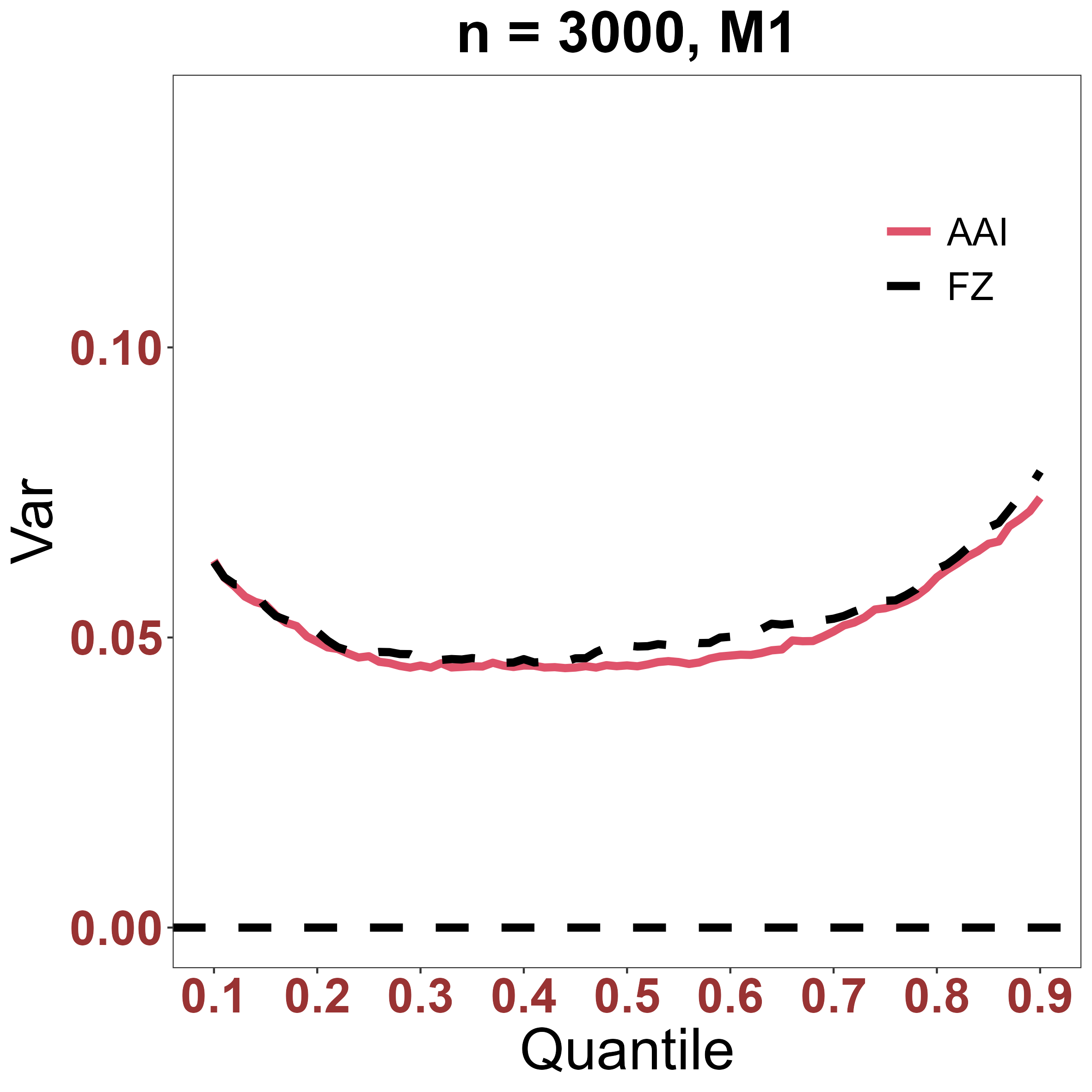}
		\includegraphics[width = 3.9cm, height = 3.2cm]{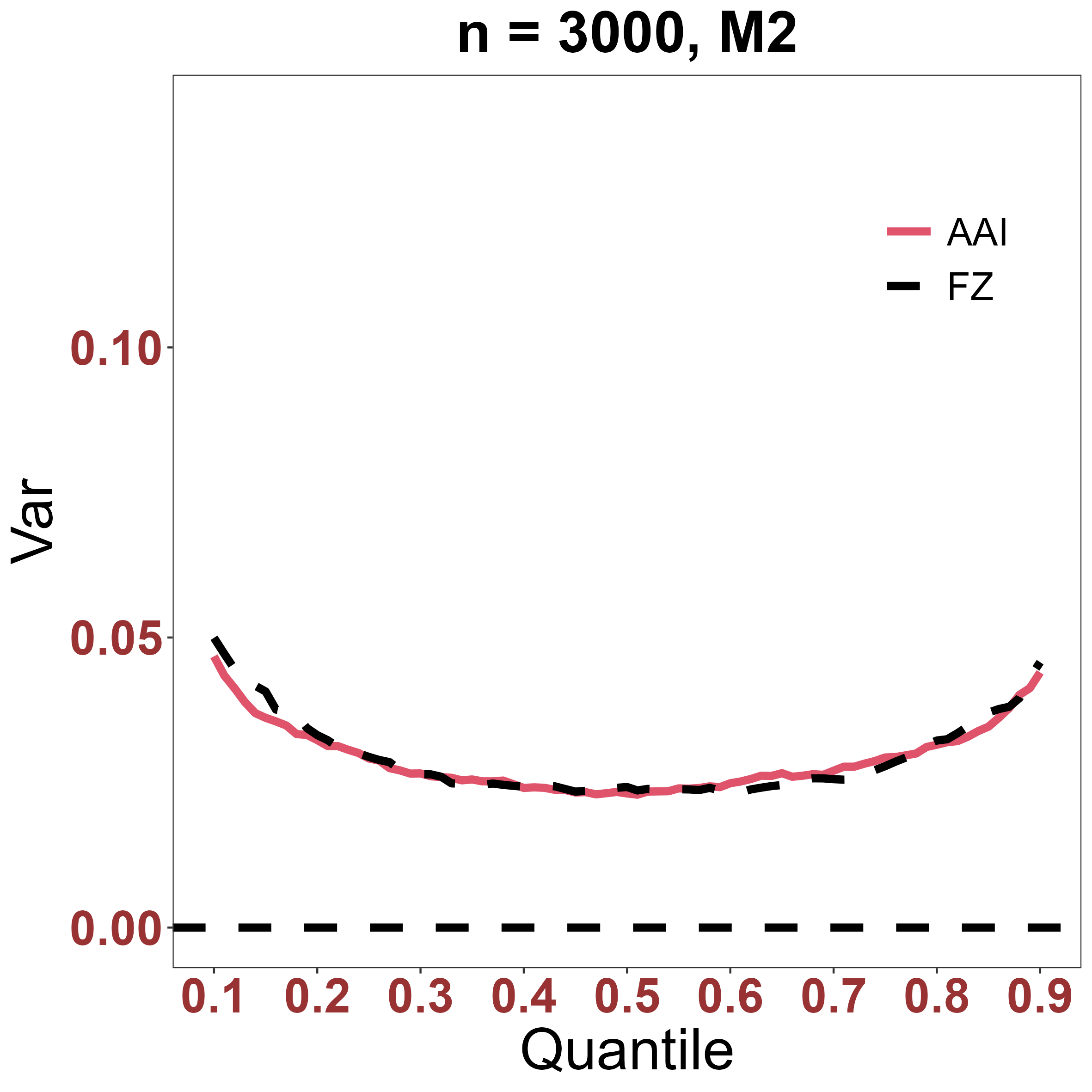}
		\includegraphics[width = 3.9cm, height = 3.2cm]{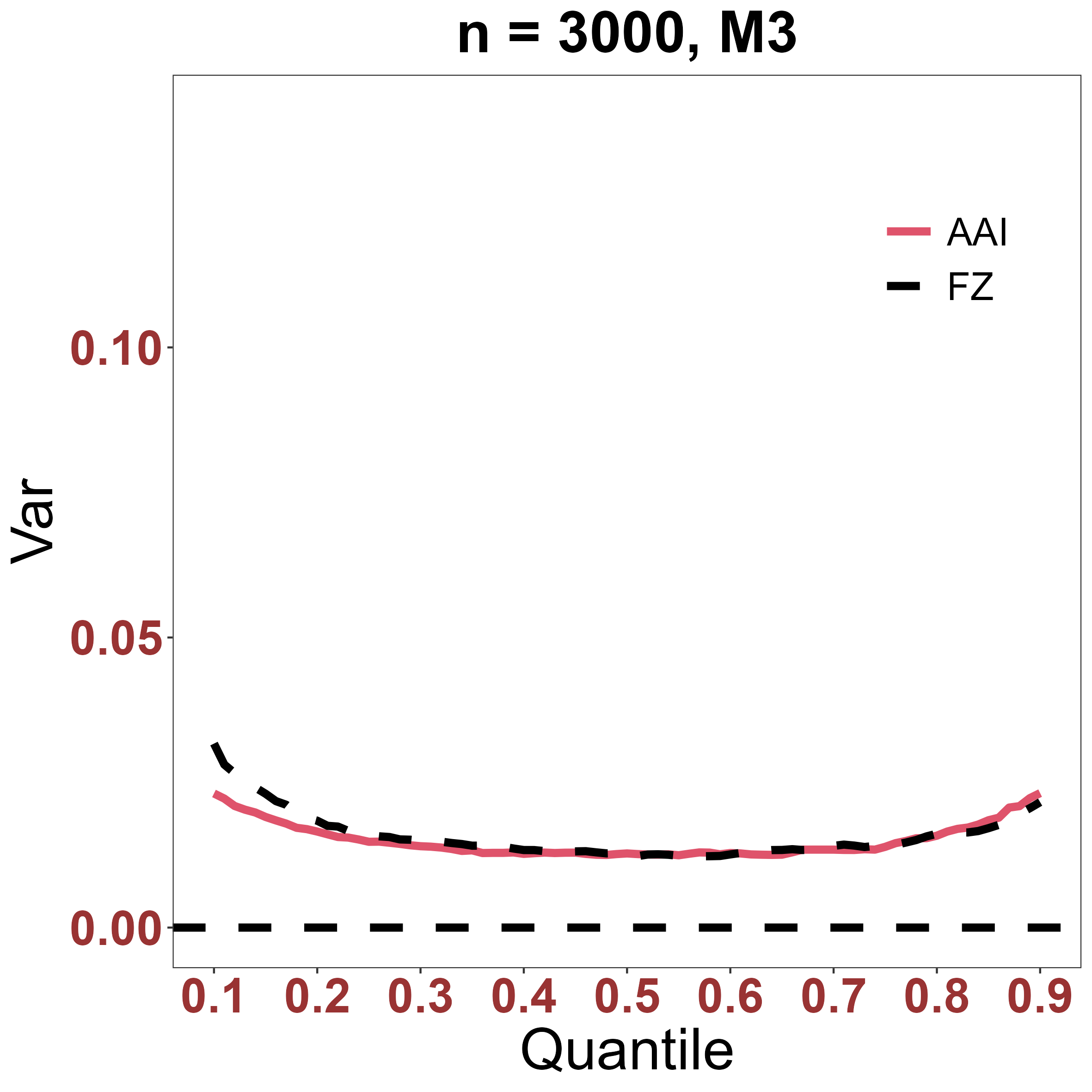}	
		\includegraphics[width = 3.9cm, height = 3.2cm]{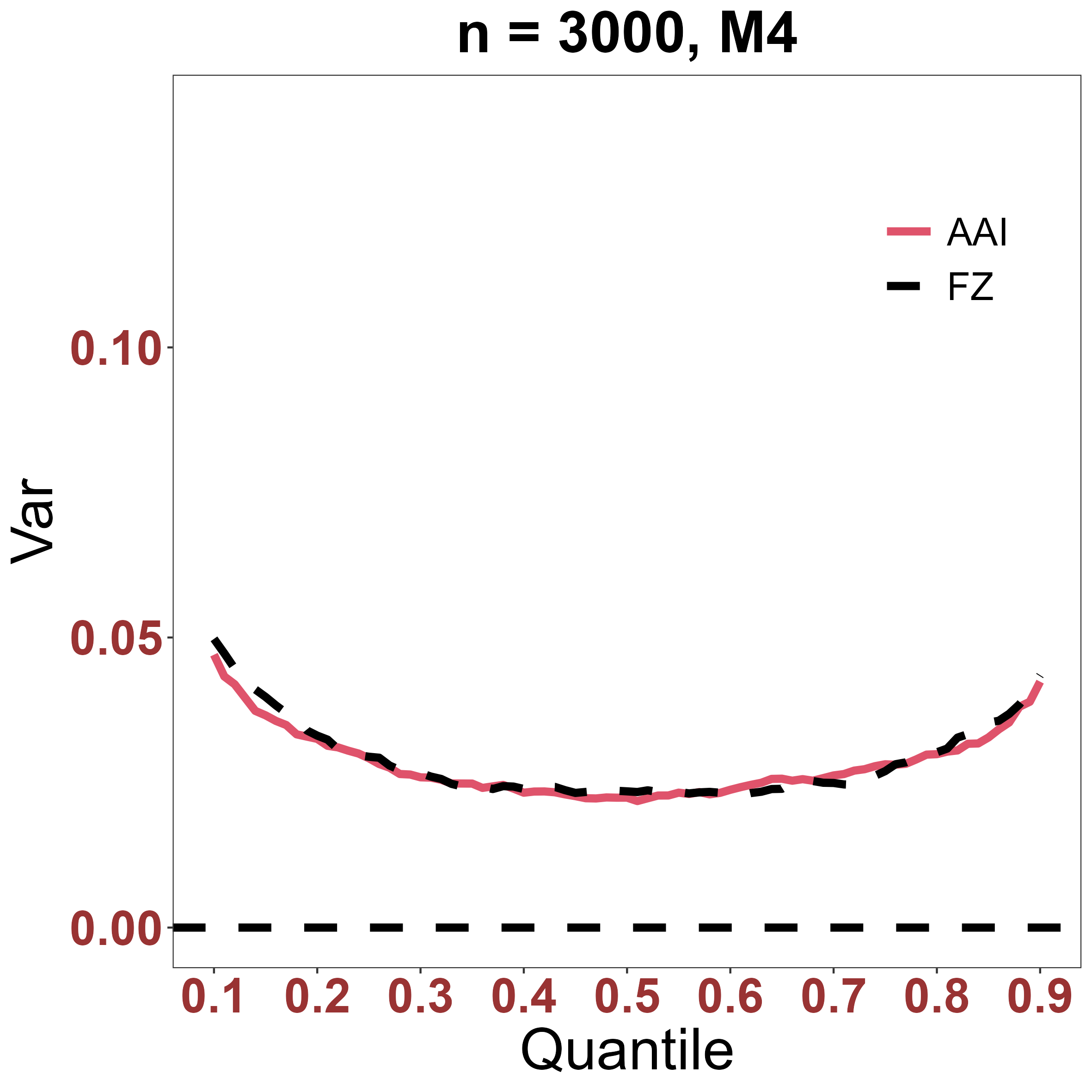}
	}	
	\mbox{
		\includegraphics[width = 3.9cm, height = 3.2cm]{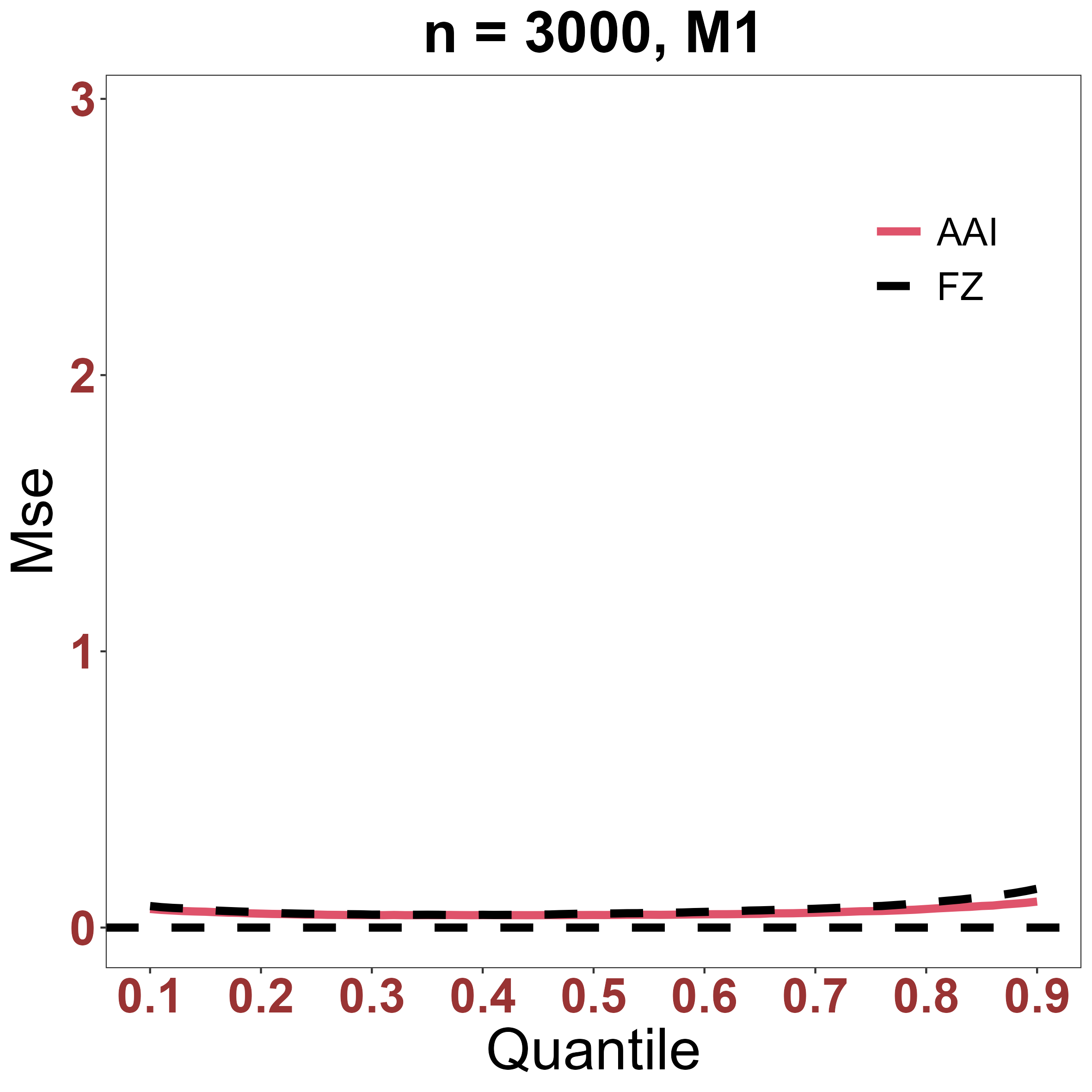}
		\includegraphics[width = 3.9cm, height = 3.2cm]{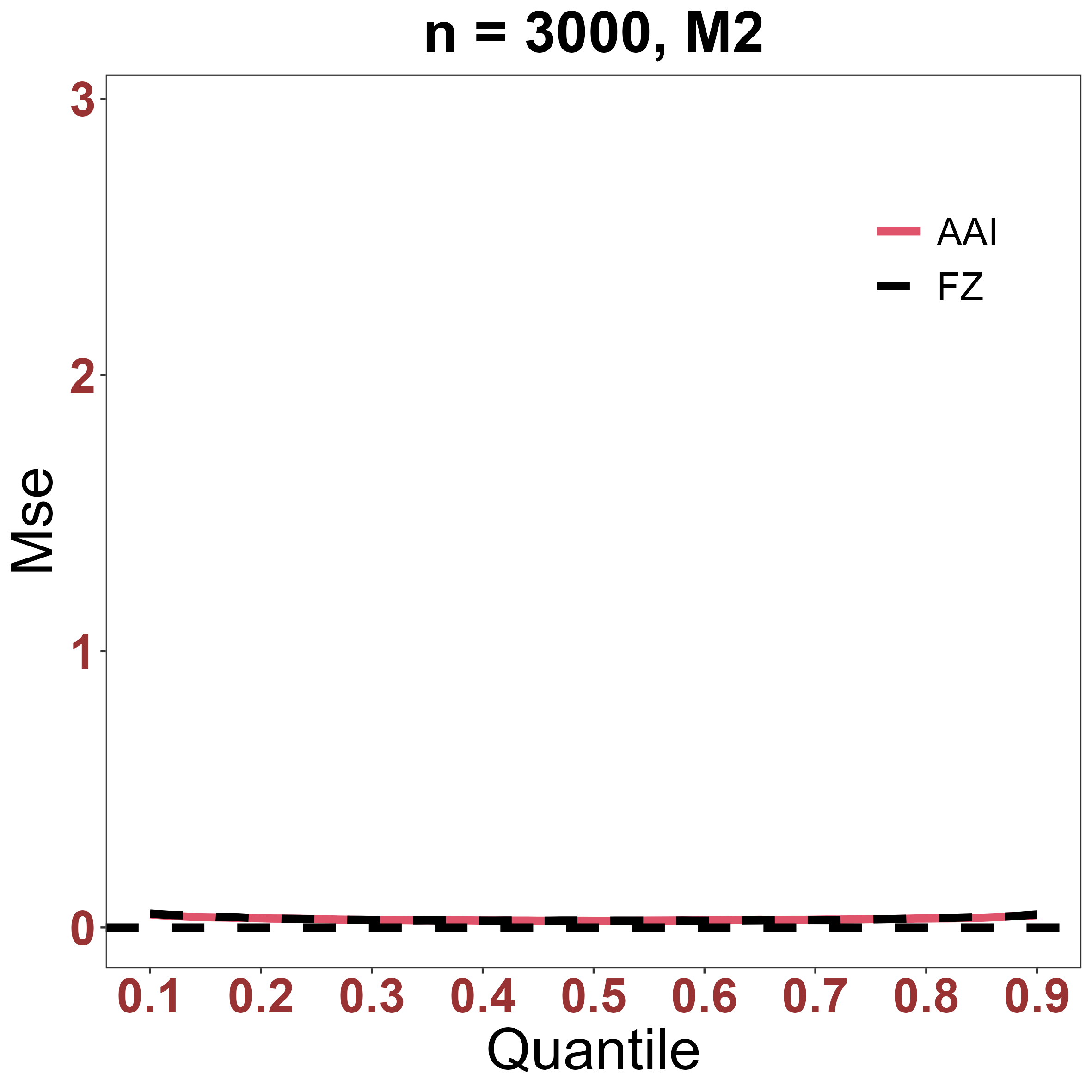}
		\includegraphics[width = 3.9cm, height = 3.2cm]{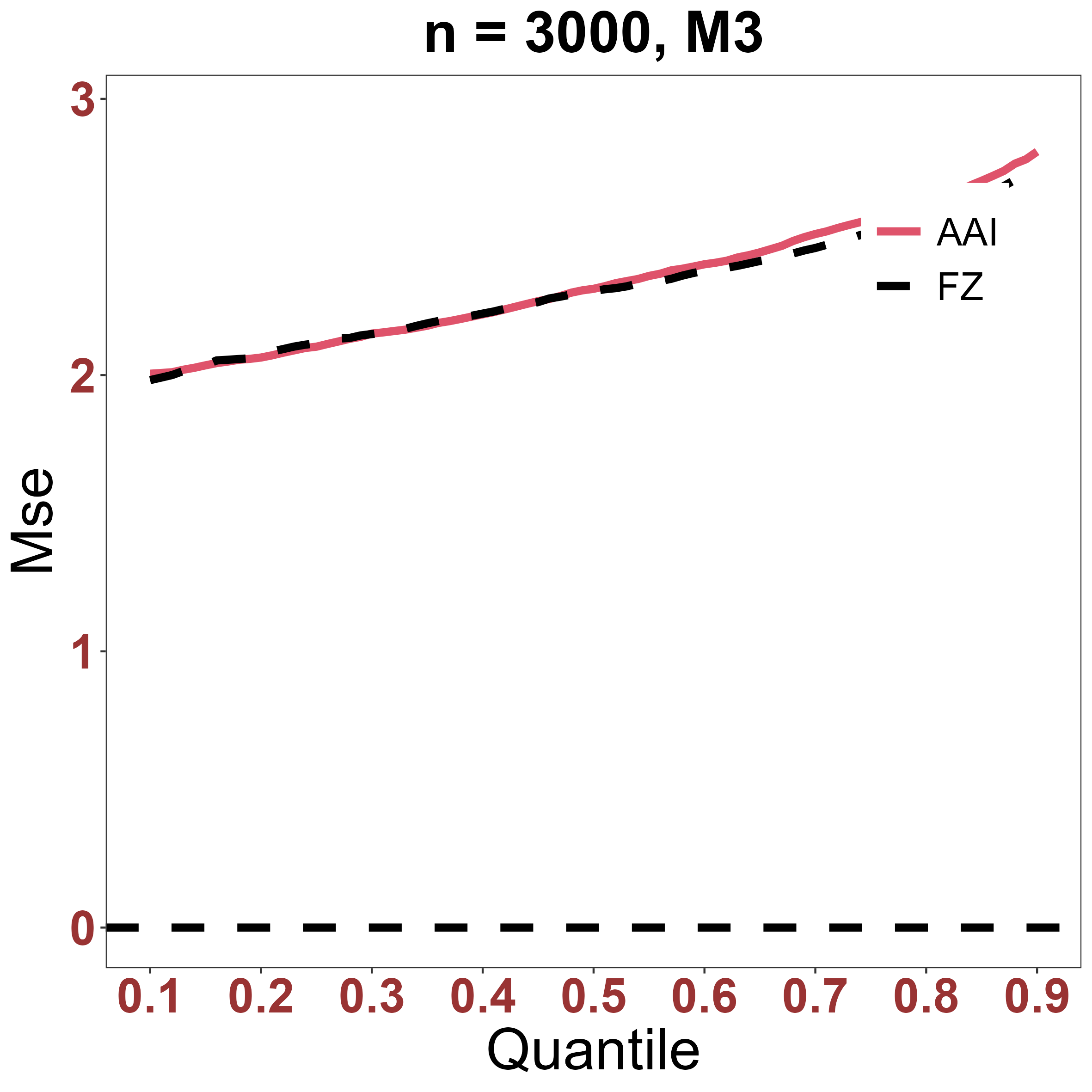}	
		\includegraphics[width = 3.9cm, height = 3.2cm]{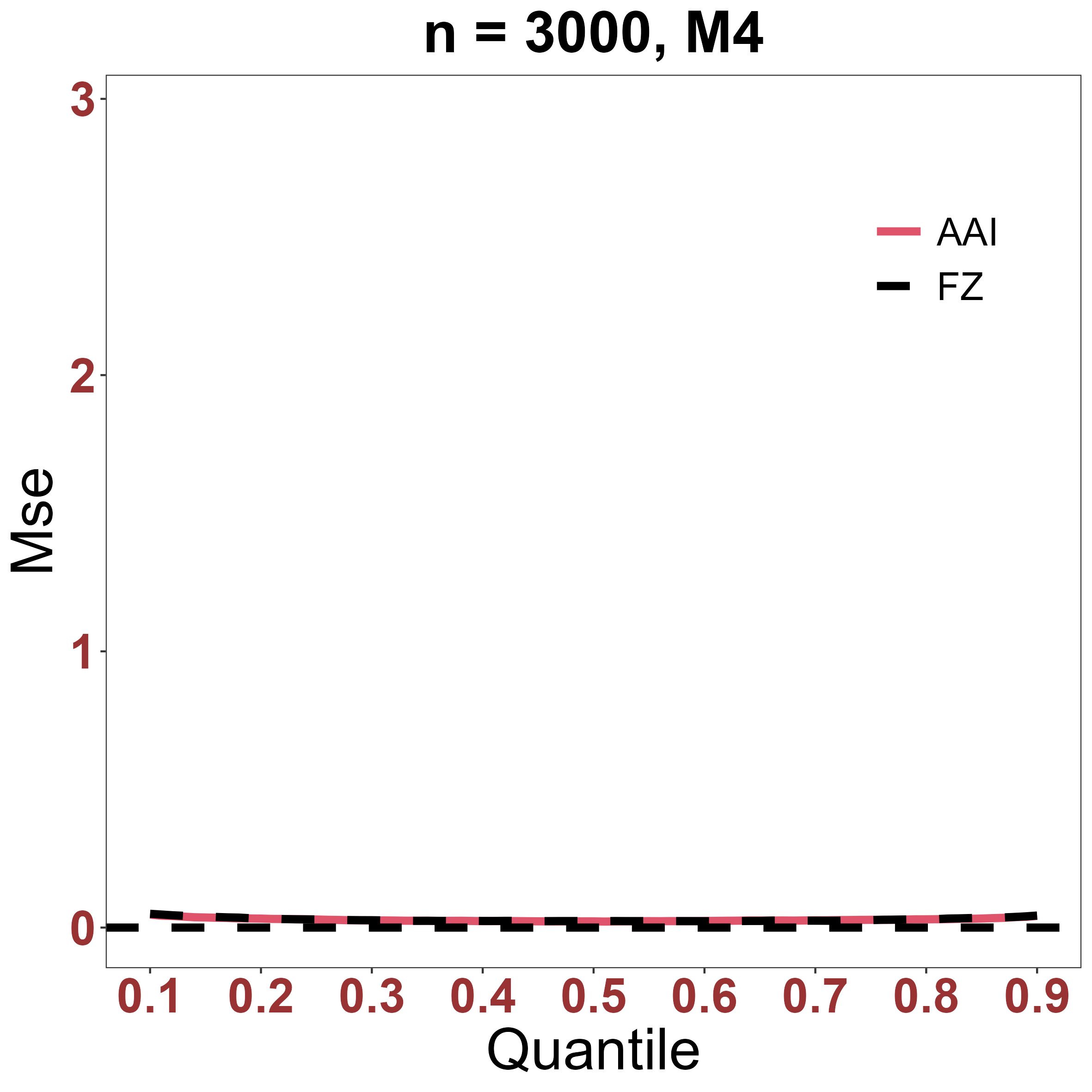}
	}	
	%\end{subfigure}
	\caption{Bias, variance and MSE of the QTE estimator using the weighted quantile regression (WQR) of \cite{AAI_2002} (AAI) and that using the FZ loss (FZ) when $\rho=0.5$ and $n=3,000$.}
	\label{figure14}
\end{figure}
\clearpage

\begin{figure}[ht]
	%\begin{subfigure}
	\centering
	\mbox{
		\includegraphics[height=8cm,width=8cm]{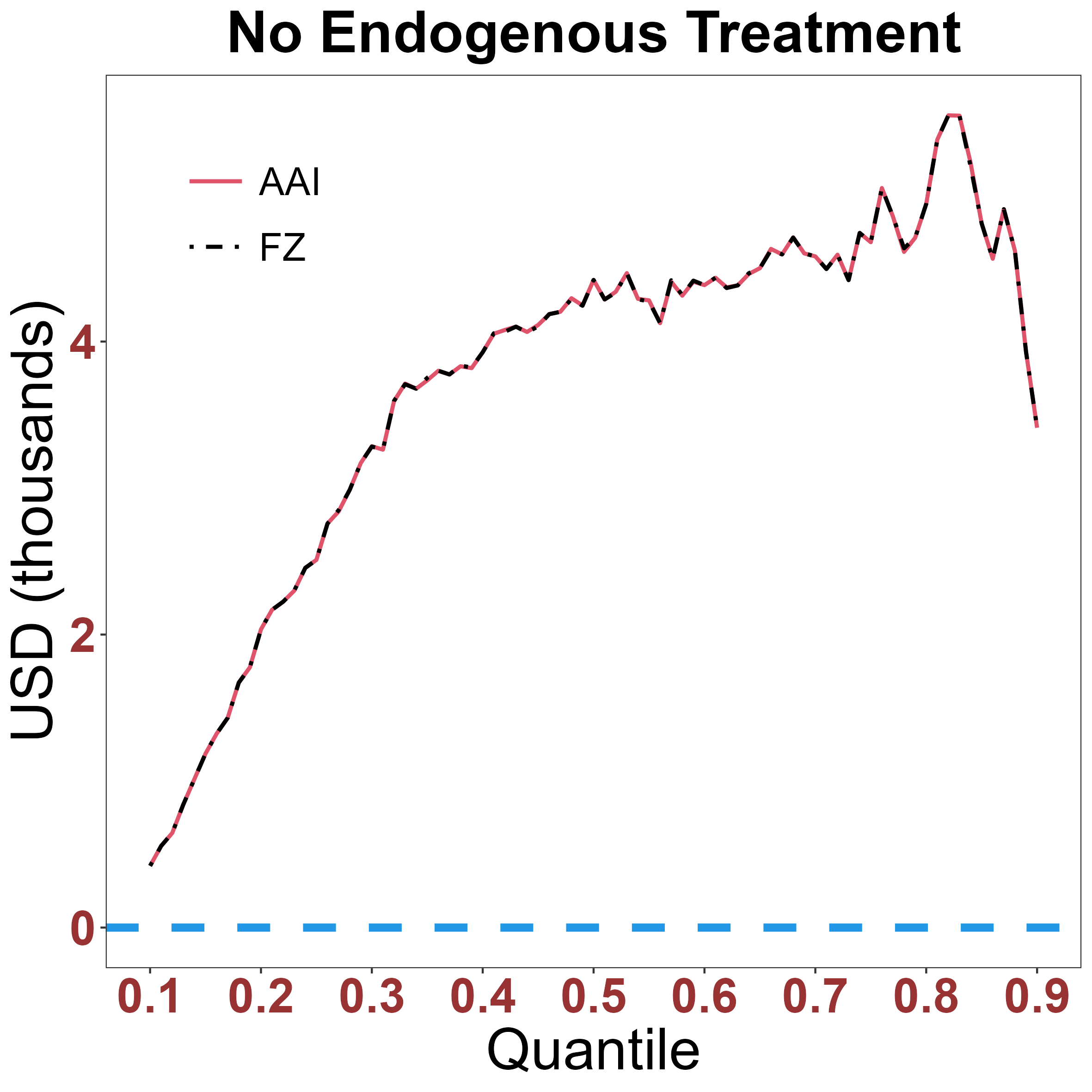}
		\includegraphics[height=8cm,width=8cm]{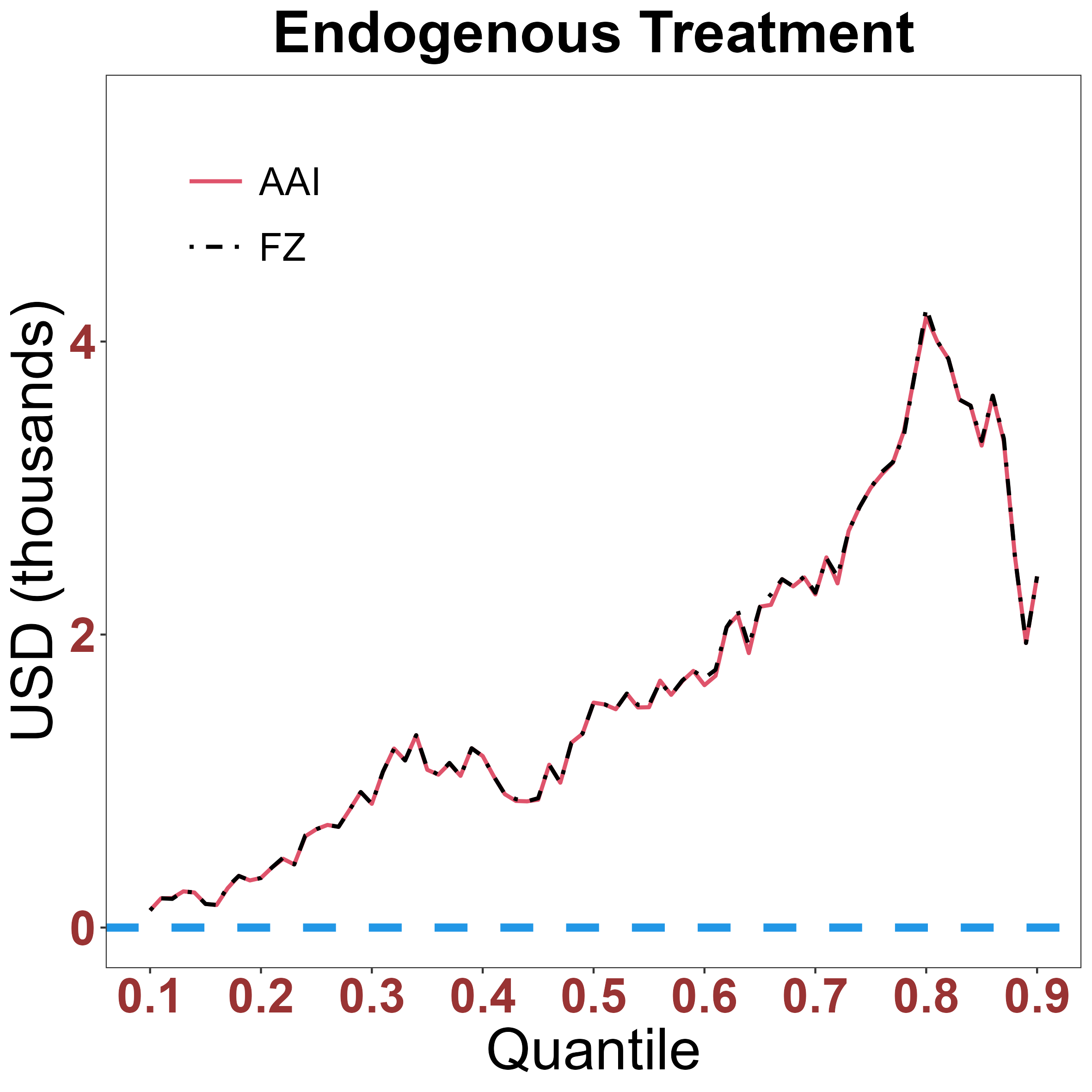}
	}	
	\mbox{	
		\includegraphics[height=8cm,width=8cm]{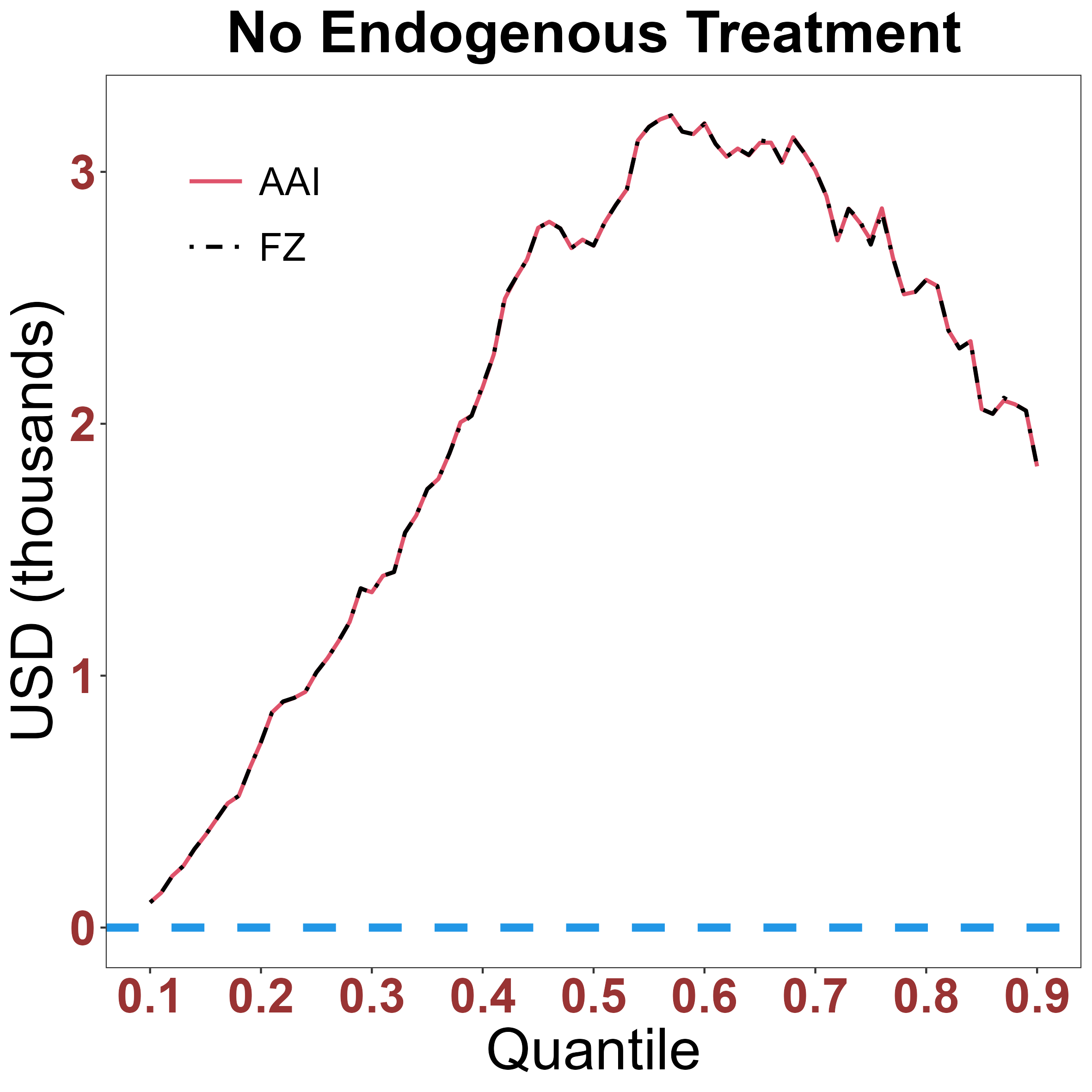}
		\includegraphics[height=8cm,width=8cm]{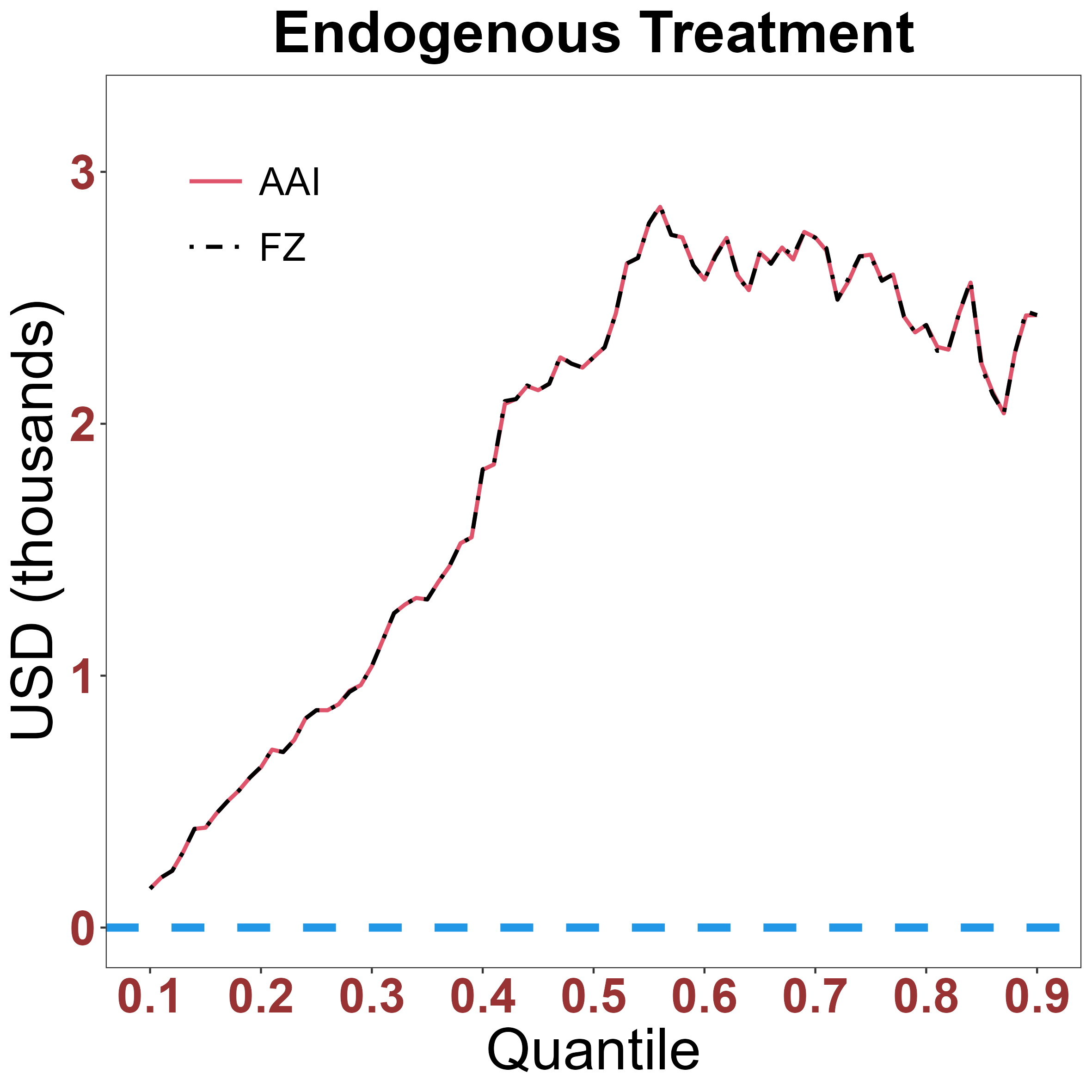}
	}
	%\end{subfigure}
	\caption{Comparisons of the QTE estimates for compliers from using weighted quantile regression (WQR) of \citet{AAI_2002} (AAI) and from using the FZ loss (FZ). Upper panel: Adult men's earnings. Lower panel: Adult women's earnings.}
	\label{figure15}
\end{figure}

\clearpage
\bibliographystyle{ECTA}
\bibliography{ref_ctate}
\end{document}